\documentclass[11pt]{article}
\usepackage{amssymb}
\usepackage[sectionbib]{natbib}
\usepackage{array,epsfig,rotating}
\usepackage{amsmath}
\usepackage{amsfonts}
\usepackage{multirow}
\usepackage{amsthm}
\usepackage{array}
\usepackage{graphicx}
\usepackage{graphics}
\usepackage{epsfig,latexsym,verbatim}
\usepackage{color}
\usepackage{multicol}
\usepackage{threeparttable}
\usepackage{float}
\usepackage{cancel}
\usepackage{url}
\usepackage{ulem}
\usepackage{rotating}
\usepackage{hyperref}
\usepackage{xr}
\usepackage{secdot}

\def\bs{\boldsymbol}

\definecolor{red}{rgb}{0,0,0}
\definecolor{green}{rgb}{0,0,0}
\definecolor{blue}{rgb}{0,0,0}

\def\red{\textcolor{red}}
\newcommand{\diag}{\mathrm{diag}}
\DeclareMathOperator*{\argmin}{arg\,min}

\def\bs{\boldsymbol}

\floatstyle{ruled}
\newtheorem{theorem}{Theorem}
\newtheorem{lemma}{Lemma}
\newtheorem{corollary}{Corollary}

\theoremstyle{definition}

\begin{document}

\renewcommand{\baselinestretch}{1.2}
\markboth{\hfill{\footnotesize\rm Shan Yu, Guannan Wang, Li Wang, and Lijian Yang}\hfill}
{\hfill {\footnotesize\rm Multivariate Spline Estimation and Inference for Image-On-Scalar Regression} \hfill}
\renewcommand{\thefootnote}{}
$\ $\par \fontsize{10.95}{14pt plus.8pt minus .6pt}\selectfont
\vspace{0.8pc} \centerline{\large\bf Multivariate Spline Estimation and Inference for Image-On-Scalar Regression}
\vspace{.4cm} \centerline{ Shan Yu$^{a}$, Guannan Wang$^{b}$, Li Wang$^{c}$ and Lijian Yang$^{d}$
\footnote{\emph{Address for correspondence}: Li Wang, Department of Statistics and the Statistical Laboratory, Iowa State University, Ames, IA, USA. Email: lilywang@iastate.edu}} \vspace{.4cm} \centerline{\it $^{a}$University of Virginia, $^{b}$College of William \& Mary, $^{c}$Iowa State University and $^{d}$Tsinghua University} \vspace{.55cm}
\fontsize{9}{11.5pt plus.8pt minus .6pt}\selectfont


\begin{quotation}
\noindent {\it Abstract:}
Motivated by recent analyses of data in biomedical imaging studies, we consider a class of image-on-scalar regression models for imaging responses and scalar predictors. We propose using flexible multivariate splines over triangulations to handle the irregular domain of the objects of interest on the images, as well as other characteristics of images. The proposed estimators of the coefficient functions are proved to be root-$n$ consistent and asymptotically normal under some regularity conditions. We also provide a consistent and computationally efficient estimator of the covariance function. Asymptotic pointwise confidence intervals and data-driven simultaneous confidence corridors for the coefficient functions are constructed. Our method can simultaneously estimate and make inferences on the coefficient functions, while incorporating spatial heterogeneity and spatial correlation. A highly efficient and scalable estimation algorithm is developed. Monte Carlo simulation studies are conducted to examine the finite-sample performance of the proposed method, which is then applied to the spatially normalized positron emission tomography data of the Alzheimer's Disease Neuroimaging Initiative.

\vspace{9pt}
\end{quotation}

\fontsize{10.95}{14pt plus.8pt minus .6pt}\selectfont
\thispagestyle{empty}

\label{sec:introduction}
\setcounter{equation}{0}
\noindent \textbf{1. Introduction} \vskip 0.1in
\renewcommand{\thefigure}{1.\arabic{figure}} \setcounter{figure}{0}

\noindent Medical and public health studies collect massive amount of imaging data using methods such as functional magnetic resonance imaging (fMRI), positron emission tomography (PET) imaging, computed tomography (CT), and ultrasonic imaging. Much of these data can be characterized as functional data. Compared with traditional one-dimensional (1D) functional data, these imaging data are complex, high-dimensional, and structured, which poses challenges to traditional statistical methods. 

We propose a unifying approach to characterize the varying associations between imaging responses and  a set of explanatory variables. Three types of statistical methods are widely used to investigate such associations. The first category includes the univariate approaches and pixel-/voxel-based methods \citep{Worsley:04,Stein:etal:10, Hibar:etal:15}, which take each pixel/voxel as a basic analytic unit. Because all pixels/voxels are treated as independent, a major drawback of these methods is that they ignore correlation between the pixels/voxels. The second category is the tensor regression. This approach considers an image as a multi-dimensional array \citep{Zhou:Li:Zhu:13, Li:Zhang:17}, which is then changed to a vector to perform the regression. However, doing so naively yields an ultra-high dimensionality and requires a novel dimension-reduction technique and highly scalable algorithms \citep{Li:Zhang:17}. The third category is the functional data analysis (FDA) approach, in which an image is viewed as the realization of a function defined on a given domain \citep{Zhu:Li:Kong:12, Zhu:Fan:Kong:14, Reiss:Goldsmith:Shang:Ogden:17}. Using an FDA, we are able to combine information both across and within functions. 

We adopt the FDA approach in this study. Functional linear models (FLMs) are widely used to model the regression relationship between a response and some set of predictors from multiple subjects. In the literature \citep{Ramsay:Silverman:2005, Muller:05, Morris:15, Wang:Chiou:Muller:16}, FLMs are often categorized based on whether the outcome, the predictor, or both are functional: (i) functional predictor regression (scalar-on-function) \citep{Cardot:Ferraty:Sarda:99, Cardot:Ferraty:Sarda:03, Hall:Horowitz:07}; (ii) functional response regression (function-on-scalar) \citep{Morris:Carroll:06, Reiss:Huang:Mennes:10, Staicu:Crainiceanu:Carroll:10, Zhu:Fan:Kong:14, Zhang:Wang:15,Chen:Delicado:Muller:17}; and (iii) function-on-function regression \citep{Ramsay:Dalzell:91, Yao:Muller:Wang:05b,  Senturk:Muller:10, Wu:Muller:11}.  

Motivated by the structure of brain imaging data, we propose a novel image-on-scalar regression model with spatially varying coefficients that captures the varying associations between imaging phenotypes and a set of explanatory variables. Figure \ref{FIG:01} shows a schematic diagram of the proposed modeling approach. Specifically, let $\Omega$ be a two-dimensional bounded domain, and let $\bs{z}=(z_1,z_2)$ be the location point on $\Omega$. For the $i$th  subject, $i=1,\ldots,n$, let $Y_i(\bs{z})$ be the imaging measurement at location $\bs{z}\in \Omega$, and let $X_{i\ell}$, for $\ell=0,1,\ldots,p$, with $X_{i0}\equiv 1$, be scalar predictors, for example, clinic variables (such as age and sex) and genetic factors. The spatially varying coefficient regression characterizes the associations between imaging measures and covariates, and is given by the following model:
\begin{equation*}
Y_{i}(\bs{z})=\widetilde{\mathbf{X}}_i^{\top}\bs{\beta}^{o}(\bs{z})+\eta_{i}(\bs{z})+\sigma(\bs{z})\varepsilon_{i}(\bs{z}), ~ i=1,\ldots,n, ~ \bs{z}\in \Omega,
\label{model1}
\end{equation*} 
where $\widetilde{\mathbf{X}}_i=(X_{i0},X_{i1},\ldots,X_{ip})^{\top}$, $\bs{\beta}^{o}=(\beta_{0}^{o},\beta_{1}^{o},\ldots,\beta_{p}^{o})^{\top}$ is a vector of some unknown bivariate functions, $\eta_i(\bs{z})$ characterizes the individual image variations, $\varepsilon_i(\bs{z})$ represents additional measurement errors, and $\sigma(\bs{z})$ is a positive deterministic function. In the following, we assume that $\eta_i(\bs{z})$ and $\varepsilon_{i}(\bs{z})$ are mutually independent. Moreover, we assume that $\eta_i(\bs{z})$, for $i=1,\ldots,n$, are independent and identically distributed (i.i.d.) copies of an $L_2$ stochastic process with mean zero and covariance function $G_{\eta}(\bs{z},\bs{z}^{\prime})=\text{cov}\{\eta_{i}(\bs{z}), \eta_{i}(\bs{z}')\}$. Furthermore,  $\varepsilon_i(\bs{z})$, for $i=1,\ldots,n$, are i.i.d. copies of a stochastic process with zero mean. and covariance function $G_{\varepsilon}(\bs{z},\bs{z}^{\prime})=\text{cov}\{\varepsilon_{i}(\bs{z}), \varepsilon_{i}(\bs{z}')\}=I(\bs{z}=\bs{z}^{\prime})$. 

\begin{figure}[t]
	\begin{center}
			\includegraphics[scale=0.45]{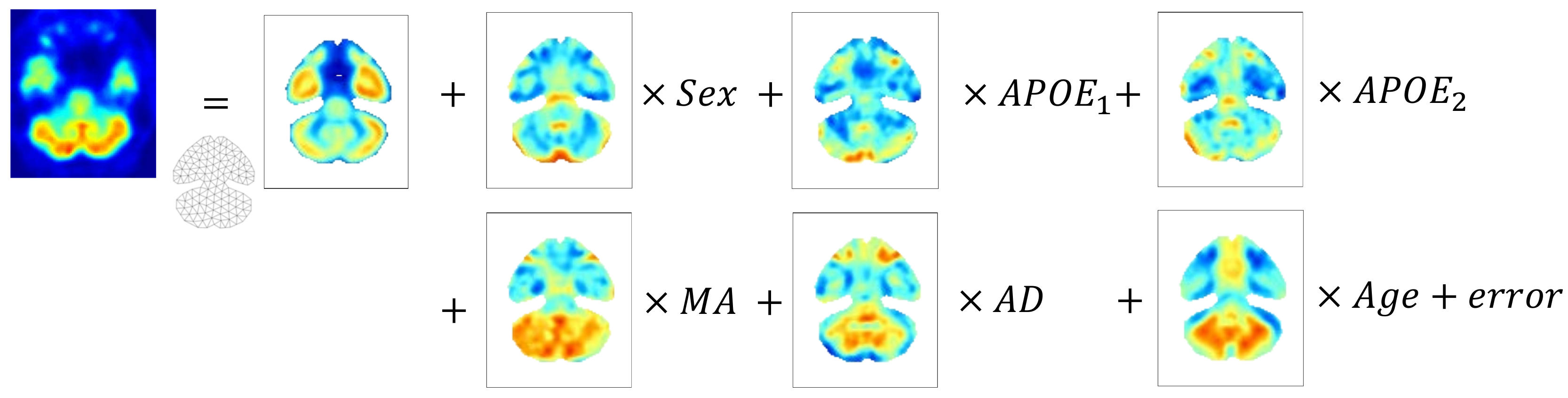}
	\end{center} \vspace{-.3in}
	\caption{A schematic diagram of proposed modeling approach.}
	\label{FIG:01}
\end{figure}

For a 1D function-on-scalar regression, Chapter 13 of \cite{Ramsay:Silverman:2005} provides a common model-fitting strategy, in which the coefficient functions are expanded using some sets of basis functions, and the basis coefficients are estimated using the ordinary least squares with{method}. However, it is not trivial to extend this to an image-on-scalar regression, particularly with biomedical imaging responses. For biomedical images, the objects (e.g., organs) on the images are usually irregularly shaped (e.g., breast tumors). Another example is that of brain images, as shown in Figure \ref{FIG:01}, especially slices from the bottom and the top of the brain. Even though some images seem to be \red{rectangular}, the true signal comes \red{only} from the domain of an object, and the image contains \red{only noise} outside the boundary of the object. Many smoothing methods, \red{such as}, tensor product smoothing \citep{Reiss:Goldsmith:Shang:Ogden:17, Chen:Delicado:Muller:17}, kernel smoothing \citep{Zhu:Fan:Kong:14}, and wavelet smoothing \citep{Morris:Carroll:06}, provide poor \red{estimations} over difficult regions \red{because they smooth}  inappropriately across boundary features,  \red{referred to as} the ``leakage" problem in the smoothing literature; see \cite{Ramsay:02}  and \cite{Sangalli:Ramsay:Ramsay:13}. Next, for technical reasons, imaging data often have different visual qualities. \red{The general characteristics} of medical images are determined and limited by the technology for each specific modality. \red{As a result,} \red{t}here is a great interest in \red{developing} a flexible method \red{with varying} smoothness to adaptively smooth biomedical imaging data.

In this \red{study}, we tackle the above challenges using bivariate splines on triangulations \citep{Lai:Wang:13} to effectively model the spatially nonstationary relationship and preserve the important features (shape, smoothness) of the imaging data. a triangulation can represent any two-dimensional (2D) geometric domain \red{effectively because} any polygon can be decomposed into triangles. We study the asymptotic properties of the bivariate spline estimators of the coefficient functions, and show that our spline estimators are root-$n$ consistent and asymptotically normal. The asymptotic results are used as a guideline to construct pointwise confidence intervals (PCIs) and simultaneous confidence corridors (SCCs; also referred to as ``simultaneous confidence \red{bands}/regions") for the true coefficient functions. 
Figure \ref{FIG:02} \red{shows the} proposed inferential approach. \red{Our method is statistically more efficient than} the tensor regression \citep{Li:Zhang:17} and the three-stage estimation \citep{Zhu:Fan:Kong:14}, \red{because} it is able to accommodate complex domains of arbitrary shape and adjust the individual smoothing needs of different coefficient functions \red{using} multiple smoothing parameters.  In addition, our method does not rely on estimating the spatial similarity and adaptive weights repeatedly, as in \cite{Zhu:Fan:Kong:14}; thus, it is much simpler. 

\begin{figure}[t]
	\begin{center}
			\includegraphics[scale=0.45]{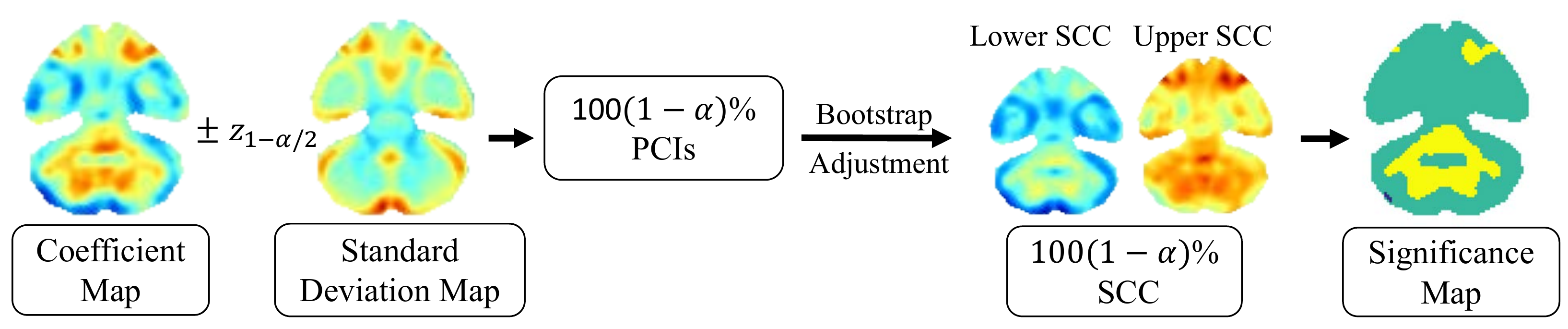}
	\end{center} \vspace{-.3in}
	\caption{A schematic diagram of proposed inferential approach.}
	\label{FIG:02}
\end{figure}

The remainder of the paper is structured as follows. Section 2 describes the spline estimators for the coefficient functions, and establish\red{es} their asymptotic properties. Section 3 describes the bootstrap method used to construct the SCC and how to estimate the unknown variance functions involved in the SCC. Section 4 presents the implementation of the proposed estimation and inference. Section 5 reports our findings from two  simulation studies. In Section 6, we illustrate the proposed method using PET data provided by the Alzheimer's Disease Neuroimaging Initiative (ADNI). Section 7 concludes the paper. All technical proofs of the theoretical results and additional numerical results are deferred to the Appendices A and B.

\vskip 0.1in \noindent \textbf{2. Models and Estimation Method} \vskip 0.1in
\renewcommand{\thetable}{2.\arabic{table}} \setcounter{table}{0} 
\renewcommand{\thefigure}{2.\arabic{figure}} \setcounter{figure}{0}
\renewcommand{\theequation}{2.\arabic{equation}} \setcounter{equation}{0} 

\noindent \textbf{2.1. Image-on-scalar regression model} \vskip .10in

Let $\bs{z}_j\in \Omega$ be the center point of the $j$th pixel in the domain $\Omega$, and let $Y_{ij}$ be the  imaging response of subject $i$ at location $j$. the actual data set consists of $\{(Y_{ij}, \widetilde{\mathbf{X}}_i, \bs{z}_j), i=1,\ldots,n,j=1,\ldots,N\}$, which can be modeled \red{as follows}:
\begin{equation}
Y_{ij}=\sum_{\ell =0}^{p}X_{i\ell}\beta_{\ell}^{o}(\bs{z}_{j})+\eta_{i}(\bs{z}_{j})+\sigma(\bs{z}_j)\varepsilon_{ij}.
\label{model2}
\end{equation}

Denote the eigenvalues and eigenfunctions of the covariance operator $G_{\eta}(\bs{z},\bs{z}^{\prime})$ as $\left\{\lambda_{k}\right\}_{k=1}^{\infty}$ \red{and} $\left\{\psi_{k}(\bs{z})\right\}_{k=1}^{\infty}$, respectively, where $\lambda_{1}\geq \lambda_{2}\geq \cdots \geq 0$, $\sum_{k=1}^{\infty}\lambda_{k}<\infty $, and $\left\{\psi_{k}\right\}_{k=1}^{\infty}$ forms an orthonormal basis of $L^{2}\left(\Omega\right) $.\ It follows from spectral theory that  $G_{\eta}(\bs{z},\bs{z}^{\prime}) =\sum_{k=1}^{\infty}\lambda_{k}\psi_{k}(\bs{z})\psi_{k}(\bs{z}^{\prime})$. The $i$th trajectory $\left\{\eta_{i}(\bs{z}), \bs{z} \in \Omega\right\} $ allows the Karhunen--Lo\'{e}ve $L^{2}$ representation \citep{Li:Hsing:2010, Sang:Huang:2012}: $\eta_{i}(\bs{z})=\sum_{k=1}^{\infty}\lambda_k^{1/2}\xi_{ik}\psi_{k}(\bs{z})$, $\lambda_k^{1/2}\xi_{ik}=\int_{\bs{z}\in \Omega}\eta_{i}(\bs{z})\psi_{k}(\bs{z})d\bs{z}$, where the random coefficients  $\xi_{ik}$ are uncorrelated random variables with mean \red{zero} and $E(\xi_{ik}\xi_{ik^{\prime}})=I(k=k^{\prime})$, referred to as the $k$th functional principal component score (FPCA) of the $i$th subject. Thus, the response measurements in (\ref{model2}) can be represented as follows:  
\begin{equation}
Y_{ij}=\sum_{\ell =0}^{p}\beta_{\ell}^{o}(\bs{z}_{j})
X_{i\ell}+\sum_{k=1}^{\infty}\lambda_k^{1/2}\xi_{ik}\psi_{k}(\bs{z}_{j})+ \sigma(\bs{z}_j)\varepsilon_{ij}.
\label{model3}
\end{equation}

\vskip .10in \noindent \textbf{2.2. Spline approximation over triangulations and penalized regression} \vskip .10in
\label{SUBSEC:triangulations}

Note that the objects of interest on many biomedical images are often distributed over an irregular domain $\Omega$. \red{Triangulation} is an effective strategy to handle such data. For example, the spatial smoothing problem over difficult regions in \cite{Ramsay:02} and \cite{Sangalli:Ramsay:Ramsay:13} was solved \red{using} the finite element method (FEM) on triangulations, which \red{was} developed \red{primarily} to solve partial differential equations. \red{Here}, we approximate each coefficient function in (\ref{model3}) \red{using} bivariate splines over triangulations \citep{Lai:Schumaker:07}. The idea is to approximate each function $\beta_{\ell}(\cdot)$ \red{using} Bernstein basis polynomials that are piecewise polynomial functions over a 2D triangulated domain. Compared with the FEM, the proposed approach is appealing in the sense that \red{its} spline functions \red{are} more flexible and \red{it uses} various smoothness \red{settings} to better approximate the coefficient functions. In this section, we briefly introduce the \red{triangulation technique} and describe the bivariate penalized spline smoothing (BPST) method \red{used to approximate} the spatial data.

Triangulation is an effective tool to deal with data distributed over difficult regions with complex boundaries and/or interior holes. In the following, we use $T$ to denote a triangle \red{that} is a convex hull of three points not located \red{on} one line. A collection $\triangle=\{T_1,\ldots,T_H\}$ of $H$ triangles is called a triangulation of $\Omega=\cup_{h=1}^{H} T_{h}$, provided that any nonempty intersection between a pair of triangles in $\triangle$ is either a shared vertex or a shared edge. Given a triangle $T\in \triangle$, let $|T|$ be its longest edge length and $\varrho_{T}$ be the radius of the largest disk inscribed in $T$. Define the shape parameter of $T$ as the ratio $\pi_{T}=|T|/\varrho_T$. When $\pi_T$ is small, the triangles are relatively uniform in the sense that all angles of the triangles in $\triangle$ are relatively the same. Denote the size of $\triangle$ by $|\triangle|=\max \{|T|,T \in \triangle \}$, \red{that is}, the length of the longest edge of $\triangle$. For an integer $r\geq 0$, let $\mathcal{C}^r(\Omega)$ be the collection of all $r$th continuously differentiable functions over $\Omega$. Given $\triangle$, let $\mathcal{S}_{d}^{r}(\triangle )=\{s\in \mathcal{C}^{r}(\Omega ):s|_{T}\in \mathbb{P}_{d}(T), T \in \triangle \}$ be a spline space of degree $d$ and smoothness $r$ over $\triangle $, where $s|_{T}$ is the polynomial piece of spline $s$ restricted on triangle $T$, and  $\mathbb{P}_{d}$ is the space of all polynomials of degree less than or equal to $d$. Note that the major difference between the FEM and the BPST is the flexibility of the smoothness, $r$, and the degree of the polynomials, $d$. Specifically, the FEM in \cite{Sangalli:Ramsay:Ramsay:13} requires that $r=0$ and $d=1$ or $2$, \red{whereas the} BPST allows smoothness $r\ge 0$ and various degrees of polynomials.

We use Bernstein basis polynomials to represent the bivariate splines. For any $\ell=0,1,\ldots,p$, denote by $\triangle_{\ell}$ the triangulation of the $\ell$th component. Define 
\begin{equation*}
\mathcal{G}^{(p+1)}\equiv \mathcal{G}^{(p+1)}(\triangle_{0}\times \cdots \times \triangle_{p})=\left\{\bs{g}=(g_0,\ldots,g_p)^{\top},
g_{\ell}\in \mathcal{S}_d^r(\triangle_{\ell}),\ell=0,\ldots, p\right\},
\label{DEF:G_space}
\end{equation*}
and let $\{B_{\ell m}\}_{m \in \mathcal{M}_{\ell}}$ be the set of degree-$d$ bivariate Bernstein basis polynomials for $\mathcal{S}_{d}^{r}(\triangle_{\ell})$, where $\mathcal{M}_{\ell}$ \red{is} an index set of Bernstein basis polynomials. Denote by $\mathbf{B}_{\ell}$ the evaluation matrix of the Bernstein basis polynomials for the $\ell$th component, and \red{let} the $j$th row of $\mathbf{B}_{\ell}$ is given by $\mathbf{B}_{\ell}^{\top}(\bs{z}_{j})=\{B_{\ell m}(\bs{z}_{j}), m\in \mathcal{M}_{\ell}\}$. \red{We approximate} each $\beta_{\ell}(\cdot)$ \red{using}
$\beta_{\ell}(\bs{z}_{j})\approx \mathbf{B}_{\ell}^{\top}(\bs{z}_{j})\bs{\gamma}_{\ell}$, for $\ell=0,1,\ldots,p$,
where $\bs{\gamma}_{\ell}^{\top} =(\gamma_{\ell m},m \in \mathcal{M}_{\ell})$ is the spline coefficient vector.

Penalized spline smoothing has gained \red{in}  popularity over the last two decades; see \cite{Hall:Opsomer:05,Claeskens:Krivobokova:Opsomer:09,Schwarz:Krivobokova:16}. To define the penalized spline method, for any direction $z_{q}$, $q=1,2$, let $\nabla_{z_{q}}^{v}s(\bs{z})$ denote the $v$th--order derivative in the direction $z_{q}$ at the point $\bs{z}$. We consider the following penalized least squares problem: 
\[
\min_{(\beta_0,\ldots,\beta_p)^{\top}\in \mathcal{G}^{(p+1)}} \sum_{i=1}^{n}\sum_{j=1}^{N}\left\{Y_{ij}-\sum_{\ell =0}^{p}X_{i\ell}\beta_{\ell}(\bs{z}_{j})
 \right\}^{2}+\sum_{\ell =0}^{p}\rho_{n,\ell}\mathcal{E}(\beta_{\ell}),
\]
where $\mathcal{E}(s)= \sum_{T\in\triangle}\int_{T} \sum_{i+j=2}
\binom{2}{i}(\nabla_{z_{1}}^{i}\nabla_{z_{2}}^{j}s)^{2}dz_{1}dz_{2}$ is the roughness penalty, and $\rho_{n,\ell}$ is the penalty parameter for the $\ell$th function.

To satisfy the smoothness condition of the splines, we need to impose some linear constraints on the spline coefficients  $\bs{\gamma}_{\ell}$: $\mathbf{H}_{\ell}\bs{\gamma}_{\ell}=\mathbf{0}$, for $\ell=0,1,\ldots,p$. Thus, we have to minimize the following constrained least squares: 
 \[
\sum_{i=1}^{n}\sum_{j=1}^{N}\left\{Y_{ij}-\sum_{\ell =0}^{p}X_{i\ell}\mathbf{B}_{\ell}^{\top}(\bs{z}_{j})
 \bs{\gamma}_{\ell}\right\}^{2}+\sum_{\ell =0}^{p}\rho_{n,\ell}\bs{\gamma}_{\ell}^{\top}\mathbf{P}_{\ell}\bs{\gamma}_{\ell}, \mathrm{~subject~to~} \mathbf{H}_{\ell}\bs{\gamma}_{\ell}=0,
\]
where $\mathbf{P}_{\ell}$ is the block diagonal penalty matrix satisfying $\bs{\gamma}_{\ell}^{\top}\mathbf{P}_{\ell}\bs{\gamma}_{\ell}=\mathcal{E}(\mathbf{B}_{\ell}^{\top}\bs{\gamma}_{\ell})$.

We first remove the constraint \red{using a} QR decomposition of the transpose of the constraint matrix $\mathbf{H}_{\ell}$. Applying a QR decomposition on $\mathbf{H}_{\ell}^{\top}$, \red{we have}
$\mathbf{H}_{\ell}^{\top}=\mathbf{Q}_{\ell}\mathbf{R}_{\ell}=(\mathbf{Q}_{\ell,1}~\mathbf{Q}_{\ell,2})
\binom{\mathbf{R}_{\ell,1}}{\mathbf{R}_{\ell,2}}$, where $\mathbf{Q}_{\ell}$ is an orthogonal matrix and $\mathbf{R}_{\ell}$ is an upper triangular matrix. the submatrix $\mathbf{Q}_{\ell,1}$ \red{represents} the first $r$ columns of $\mathbf{Q}_{\ell}$, where $r$ is the rank of matrix $\mathbf{H}_{\ell}$, and $\mathbf{R}_{\ell,2}$ is a matrix of zeros. We reparametrize \red{this} using $\bs{\gamma}_{\ell} = \mathbf{Q}_{\ell,2}\bs{\theta}_{\ell}$, for some $\bs{\theta}_{\ell}$\red{.} \red{Then,} it is guaranteed that $\mathbf{H}_{\ell}\bs{\gamma}_{\ell}= \mathbf{0}$. \red{Thus,} the minimization problem is converted to \red{the following} conventional penalized regression problem, without restrictions:
\begin{equation}
\sum_{i=1}^{n}\sum_{j=1}^{N}\left\{Y_{ij}-\sum_{\ell =0}^{p}X_{i\ell}\mathbf{B}_{\ell}^{\top}(\bs{z}_{j})
\mathbf{Q}_{\ell,2}\bs{\theta}_{\ell}\right\}^{2}+\sum_{\ell =0}^{p}\rho_{n,\ell}\bs{\theta}_{\ell}^{\top}\mathbf{D}_{\ell}\bs{\theta}_{\ell},
\label{EQ:PLS}
\end{equation}
where $\mathbf{D}_{\ell}=\mathbf{Q}_{\ell,2}^{\top}\mathbf{P}_{\ell}\mathbf{Q}_{\ell,2}$.

Let $\widetilde{\mathbf{Y}}_i=(Y_{i1},Y_{i2},\ldots, Y_{iN})^{\top}$, $\mathbf{B}_{\ell}(\bs{z})=\{B_{\ell m}(\bs{z}),m\in \mathcal{M}_{\ell}\}^{\top}$, $\mathbb{Y}=(\widetilde{\mathbf{Y}}_1^{\top},\ldots, \widetilde{\mathbf{Y}}_n^{\top})^{\top}$,
and $\mathbb{U}=(\mathbf{U}_{11},\mathbf{U}_{12},\ldots, \mathbf{U}_{nN})^{\top}$,
where
\begin{equation}
\mathbf{U}_{ij}=\{X_{i0}\mathbf{B}_{0}(\bs{z}_{j})^{\top}
\mathbf{Q}_{0, 2}, X_{i1}\mathbf{B}_1(\bs{z}_j)^{\top}\mathbf{Q}_{1, 2}, \cdots, X_{ip}\mathbf{B}_{p}(\bs{z}_{j})^{\top}
\mathbf{Q}_{p, 2}\}^{\top}.
\label{EQ:U_ij}
\end{equation}
Let
$\bs{\theta}=(\bs{\theta}_{0}^{\top},\bs{\theta}_{1}^{\top},\ldots,\bs{\theta}_{p}^{\top})^{\top}$
and $\mathbb{D}(\rho_{n,0},\ldots,\rho_{n,p})=\diag\{\rho_{n,0}\mathbf{D}_0,\ldots,\rho_{n,p}\mathbf{D}_{p}\}$. Minimizing (\ref{EQ:PLS}) is then equivalent to minimizing
$\left\|\mathbb{Y}-\mathbb{U}\bs{\theta}\right\|^{2}+\bs{\theta}^{\top}\mathbb{D}(\rho_{n,0},\ldots,\rho_{n,p})\bs{\theta}$. \red{Hence,}
\begin{equation*}
\widehat{\bs{\theta}}=(\widehat{\bs{\theta}}_{0}^{\top},\widehat{\bs{\theta}}_{1}^{\top},
\ldots,\widehat{\bs{\theta}}_{p}^{\top})^{\top}
=\{\mathbb{U}^{\top}\mathbb{U}+\mathbb{D}(\rho_{n,0},\ldots,\rho_{n,p})\}^{-1}
\mathbb{U}^{\top} \mathbb{Y}.
\label{DEF:Thetahat}
\end{equation*}
Thus, the estimators of $\bs{\gamma}_{\ell}$ and $\beta_{\ell}(\cdot)$ are
\begin{equation}
\widehat{\bs{\gamma}}_{\ell}=\mathbf{Q}_{\ell,2}\widehat{\bs{\theta}}_{\ell}, ~~ \widehat{\beta}_{\ell}(\bs{z})=\mathbf{B}_{\ell}(\bs{z})^{\top}\widehat{\bs{\gamma}}_{\ell}.
\label{DEF:beta_l_hat}
\end{equation}


\vskip .10in \noindent \textbf{2.3. Asymptotic properties of the BPST estimators} \vskip .10in

This section \red{examines} the asymptotics of the proposed estimators. Given random variables $U_{n}$ for $n\geq 1$, we write $U_{n}=O_{P}(b_{n})$ if $\lim\nolimits_{c\rightarrow \infty}\lim \sup_{n}P(|U_{n}|\geq cb_{n})=0$. Similarly, we write $U_{n}=o_{P}(b_{n})$ if $\lim_{n}P(|U_{n}|\geq cb_{n})=0$, for any constant $c>0$.
Next, to facilitate discussion, we introduce some notation of norms. For any function $g$ over the closure of domain $\Omega$, denote $\left\| g\right\|_{L ^2(\Omega)}^{2}=\int_{\Omega} g^2(\bs{z})d\bs{z}$ \red{as} the regular $L_2$ norm of $g$, and $\Vert g\Vert_{\infty ,\Omega} =\sup_{\bs{z}\in \Omega} |g(\bs{z})|$ \red{as} the supremum norm of $g$. \red{Further denote $\|\bs{g}\|_{\upsilon,\infty,\Omega}=\max_{0\leq \ell\leq p}|g_{\ell}|_{\upsilon,\infty,\Omega}$,  where $|g|_{\upsilon,\infty,\Omega}=\max_{i+j=\upsilon}\Vert \nabla_{z_{1}}^{i}\nabla_{z_{2}}^{j}g\Vert_{\infty ,\Omega}$ is the maximum norm of all $\upsilon$th\red{--}order derivatives of $g$ over $\Omega $.} Let $\mathcal{W}^{d,\infty}(\Omega)=\left\{g:|g|_{k,\infty, \Omega}<\infty, 0\le k\le d \right\}$ be the standard Sobolev space. Next, we introduce some technical conditions.

\begin{itemize}
\item[(A1)] For any $\ell=0,\ldots, p$, $\beta_{\ell}^{o}(\cdot)\in \mathcal{W}^{d+1 ,\infty}(\Omega )$, for an integer $d \ge 1$.

\vspace{-0.2in}
\item[(A2)] For any $i=1,\ldots,n$, $j=1,\ldots,N$, $\varepsilon_{ij}$'s are independent with mean \red{zero and} variance \red{one}, and for any $k\geq 1$, $\xi_{ik}$ are uncorrelated random variables with mean \red{zero} and variance \red{one}.

\vspace{-0.2in}
\item[(A3)] For any $\ell=0,1,\ldots, p$, there exists a positive constant $C_{\ell}$, such that $E|X_{\ell}|^8\leq C_{\ell}$. The eigenvalues of $\bs{\Sigma}_{X}=E(\bs{X}\bs{X}^{\top})$ are bounded away from \red{zero} and infinity.

\vspace{-0.2in}
\item[(A4)] The function $\sigma(\bs{z})\in \mathcal{C}^{(1)}(\Omega)$, with $0<c_{\sigma}\leq \sigma(\bs{z}) \leq C_{\sigma}\leq \infty$, for any $\bs{z}\in \Omega$; for any $k$, $\psi_{k}(\bs{z})\in \mathcal{C}^{(1)}(\Omega)$ and $0<c_{G}\leq G_{\eta}(\bs{z},\bs{z})\leq C_{G}\leq \infty$, for any $\bs{z}\in \Omega$.

\vspace{-0.2in}
\item[(A5)] Let $|\underline{\triangle}|=\min_{0\leq \ell \leq p}|\triangle_{\ell}|$ and  $|\overline{\triangle}|=\max_{0\leq \ell \leq p}|\triangle_{\ell}|$. the triangulations $\triangle_{\ell}$ satisify that $\limsup_n (|\overline{\triangle}|/|\underline{\triangle}|)<\infty$. The triangulations are $\pi$-quasi-uniform; that is, there exists a positive constant $\pi$, such that $\max_{0\leq \ell \leq p}\{(\min_{T\in \triangle_{\ell}}\varrho_T)^{-1} |\triangle_{\ell}|\} \leq \pi$.

\vspace{-0.2in}
\item[(A6)] As $N \rightarrow \infty$, $n \rightarrow \infty$, for some $0<\kappa<1$, $N^{-1}n^{1/(d+1)+\kappa} \rightarrow 0$, $n^{1/2} |\overline{\triangle}|^{d+1}\rightarrow 0$, $N^{1/2}|\underline{\triangle}| \rightarrow \infty$, and the smoothing parameters satisfy that $n^{-1/2}N^{-1}|\underline{\triangle}|^{-3} \rho_{n} \rightarrow 0$, where $\rho_n=\max_{0\leq \ell\leq p}\rho_{n,\ell}$. 
\end{itemize}

The above assumptions are mild conditions that \red{are} satisfied in many practical situations. Assumption (A1) describes the \red{usual} requirement on the coefficient functions \red{described} in the literature \red{on} nonparametric estimation. Assumption (A1) can be relaxed to Assumption (A1$'$) in Section 2.4, which only requires $\beta_{\ell}^{o}(\cdot)\in \mathcal{C}^{(0)}(\Omega)$ when dealing with imaging data with sharp edges; see Section 2.4. Assumptions (A1) and (A2) are similar \red{to} Assumptions (A1) and (A2) in  \cite{Gu:Wang:Wolfgang:Yang:2014} and Assumptions (A1)--(A3) in \cite{Huang:Wu:Zhou:2004}. Assumption (A3) is \red{analogous} to Assumption (A5) in \cite{Gu:Wang:Wolfgang:Yang:2014}, ensuring that $X_{i\ell}$ \red{is} not multicollinear. Assumption (A5) requires that \red{$\triangle_{\ell}$ be of similar size}, and suggests the use of more uniform triangulations with smaller shape parameters. Assumption (A6) implies that the number of pixels for each image $N$ diverges to infinity and the sample size $n$ grows as $N\rightarrow \infty$, a well-developed asymptotic scenario for dense functional data \citep{Li:Hsing:2010}. Assumption (A6) also describes the requirement of the growth rate of the dimension of the spline spaces relative to the sample size and the image resolution. This assumption is easily satisfied \red{because} images measured using \red{current} technology are usually of \red{sufficiently high} resolution.

The following theorem provides the $L_2$ convergence rate of $\widehat{\beta}_{\ell}(\cdot)$, for $\ell=0,1,\ldots,p$. a detailed \red{proof  is} given in Appendix A.

\begin{theorem}
\label{THM:beta-convergence}
Suppose Assumptions (A1)--(A5) hold \red{and} $N^{1/2}|\underline{\triangle}|\rightarrow \infty$ as $N\rightarrow \infty$. \red{Then,} for any $\ell=0,1,\ldots,p$, the BPST estimator $\widehat{\beta}_{\ell}(\cdot)$ is consistent and satisfies 
$\Vert \widehat{\beta}_{\ell}-\beta_{\ell}^{o}\Vert_{L ^2(\Omega)}=O_{P}\left\{
\frac{\rho_{n}}{nN|\underline{\triangle}|^{3}}\|\bs{\beta}^{o}\|_{2,\infty}
+\left(1+\frac{\rho_{n}}{nN|\underline{\triangle}|^{5}}\right)
|\overline{\triangle}|^{d +1}\|\bs{\beta}^{o}\|_{d+1,\infty}
+n^{-1/2}\right\}$.
\end{theorem}

Theorem \ref{THM:beta-normality} states the asymptotic normality of $\widehat{\beta}_{\ell}$ at any given point $\bs{z}\in \Omega$, for $\ell=0,1,\ldots,p$. See Appendix A for \red{a detailed proof}. Denote
\begin{equation}
\bs{\Xi}_{n}(\bs{z})=\widetilde{\mathbb{B}}(\bs{z})^{\top}
E\left\{\bs{\Gamma}_{n,\rho}^{-1}\frac{1}{n^2N^2}\sum_{i=1}^{n}\sum_{j,j^{\prime}=1}^{N}
\mathbf{U}_{ij}\mathbf{U}_{ij'}^{\top}
G_{\eta}(\bs{z}_j,\bs{z}_{j^{\prime}})\bs{\Gamma}_{n,\rho}^{-1} \right\}\widetilde{\mathbb{B}}(\bs{z}),
\label{DEF:Xi_n}
\end{equation}
where $\mathbf{U}_{ij}$ and $\bs{\Gamma}_{n,\rho}$ are given in (\ref{EQ:U_ij}) and (\ref{DEF:Gamma_rho}), \red{respectively},  in Appendix 1,  $\widetilde{\mathbf{B}}_{\ell}(\bs{z})=\mathbf{Q}_{2,\ell}^{\top}\mathbf{B}_{\ell}(\bs{z})$ for $\ell=0, \ldots, p$, and 
$\widetilde{\mathbb{B}}(\bs{z})=\diag\{\widetilde{\mathbf{B}}_0(\bs{z}),
\cdots,\widetilde{\mathbf{B}}_p(\bs{z})\}$.

\begin{theorem}
\label{THM:beta-normality}
Suppose Assumptions (A1)--(A6) hold. If for any $\ell=0,1,\ldots, p$, $|X_{i\ell}| \leq C_{\ell} < \infty$, \red{then}
$\bs{\Xi}_{n}^{-1/2}(\bs{z}) \{\widehat{\bs{\beta}}(\bs{z})-\bs{\beta}^{o}(\bs{z}) \}\overset{\mathcal{L}}{\longrightarrow}N\left(\mathbf{0}, \mathbf{I}_{(p+1)\times (p+1)}\right)$ as $N \rightarrow \infty$ \red{and} $n\rightarrow \infty$, where $\bs{\Xi}_{n}(\bs{z})$ is given in (\ref{DEF:Xi_n}). Furthermore, there exist positive constants $c_V<C_V<+\infty$, such that
$c_V n^{-1}\left(1+\frac{\rho_{n}}{nN |\underline{\triangle}|^{4}}\right)^{-2}
\leq \mathrm{Var}\{\widehat{\beta}_{\ell}(\bs{z})\} \leq C_V n^{-1}$, for any $\ell=0,1,\ldots, p$.
\end{theorem}

\vskip .10in \noindent \textbf{2.4. Piecewise constant spline over triangulation smoothing} \vskip .10in

Many imaging data can be regarded as a noisy version of a piecewise-smooth function of $\bs{z}\in \Omega$ with sharp edges, which often reflect the functional or structural changes. The penalized bivariate spline smoothing method introduced, in Section 2.2, assumes some \red{degree} of smoothness over the entire image. To relax this assumption \red{while preserving} the features of sharp edges, we make the following less stringent assumption on the smoothness of the coefficient functions: 
\begin{itemize}
\item[(A1$'$)] For any $\ell=0,\ldots, p$, the bivariate function $\beta_{\ell}^{o}(\cdot)\in \mathcal{C}^{(0)}(\Omega)$.
\end{itemize}

For the estimation, we consider the piecewise constant spline over triangulation (PCST) method. For any $\ell=1,\ldots, p$, denote by $\mathcal{PC}(\triangle_{\ell})$ the space of piecewise constant functions over each $T_{m}$, for $m\in \mathcal{M}_{\ell}$. The bivariate spline basis functions of $\mathcal{PC}(\triangle_{\ell})$ are denoted as $\{B_{\ell m}(\bs{z})\}_{m\in \mathcal{M}_{\ell}}$, which are simply indicator functions over triangle $T_{m}$, $B_{\ell m}(\bs{z})=I(\bs{z} \in T_m)$, $m\in \mathcal{M}_{\ell}$. Assumption (A1$'$) controls the bias of the piecewise constant spline estimator for $\beta_{\ell}^{o}$ and leads to the estimation consistency.

When using the constant bivariate spline basis functions, \red{we have} $\mathcal{E}(s)=0$ for all $s\in \mathcal{PC}(\triangle)$, and for any $\bs{z}\in \Omega$, $\mathbf{B}_{\ell}(\bs{z})\mathbf{B}_{\ell}(\bs{z})^{\top}=\diag\{B_{\ell m}^{2}(\bs{z}),m\in \mathcal{M}_{\ell}\}$. Then, we can write  $\widehat{\bs{\gamma}}_{m}=(\widehat{\gamma}_{0m},\widehat{\gamma}_{1m},\ldots,\widehat{\gamma}_{pm})^{\top}=\widehat{\mathbf{V}}_{m}^{-1}\left\{(nN)^{-1}\sum_{i=1}^{n}\sum_{j=1}^{N}B_{\ell m}(\bs{z}_{j})X_{i\ell}Y_{ij}\right\}_{\ell =0}^{p}$, where 
\begin{equation}
\widehat{\mathbf{V}}_{m}=\frac{1}{nN}\sum_{j=1}^{N}B_{\ell m}^{2}(\bs{z}_{j})\sum_{i=1}^{n}\widetilde{\mathbf{X}}_{i}\widetilde{\mathbf{X}}_{i}^{\top}
=\left\{\frac{1}{nN}\sum_{i=1}^{n}\sum_{j=1}^{N}B_{\ell m}^{2}(\bs{z}_{j})X_{i\ell}X_{i\ell^{\prime}}\right\}_{\ell,\ell^{\prime}=0}^{p}.
\label{DEF:V-m}
\end{equation}
By simple linear algebra, for any $\ell=0,\ldots,p$, the  PCST estimator \red{is given by}
\begin{equation}
\widehat{\beta}_{\ell}^{\mathrm{c}}(\bs{z})=\sum_{m\in \mathcal{M}_{\ell}} \widehat{\gamma}_{\ell m}B_{\ell m}(\bs{z}).
\label{EQN:betahat_c}
\end{equation}

For any $\bs{z}\in \Omega$, define  the index of the triangle containing $\bs{z}$ as $m(\bs{z})$; \red{that is,} $m(\bs{z})=m$ if $\bs{z}\in T_{m}$. Then, $\widehat{\beta}_{\ell}(\bs{z})=\widehat{\bs{\gamma}}_{\ell m(\bs{z})}$ and
$\widehat{\bs{\beta}}^{\mathrm{c}}(\bs{z})=(\widehat{\beta}_{0}^{\mathrm{c}}(\bs{z}),\ldots,\widehat{\beta}_{p}^{\mathrm{c}}(\bs{z}))^{\top}=(\widehat{\gamma}_{0m(\bs{z})},\ldots, \widehat{\gamma}_{pm(\bs{z})})^{\top}=\widehat{\bs{\gamma}}_{m(\bs{z})}$. 
For any $\bs{z}\in \Omega$, denote 
\begin{equation}
\bs{\Sigma}_{n}(\bs{z})=n^{-1} \bs{\Sigma}_{X}^{-1} G_{\eta}\left(\bs{z},\bs{z}\right).
\label{DEF:Sigma(z)}
\end{equation}

Theorem \ref{THM:multinormal} shows the asymptotic normality of the piecewise constant estimators $\widehat{\bs{\beta}}(\bs{z})$. See the Appendix A for detailed proofs. To obtain the asymptotic variance-covariance function, we also need the following assumption:

\begin{itemize}
\item[(C1)] The variables $ \xi _{ik}$ and $\varepsilon _{ij}$ are independent and satisfy $E\left\vert \xi _{ik}\right\vert ^{4+\delta _{1}}<+\infty $ for some $\delta _{1}>0$, and $E\left\vert \varepsilon _{ij}\right\vert ^{4+\delta_{2}}< \infty$  for some $\delta _{2}>0$. 
\end{itemize}

\begin{theorem}
\label{THM:multinormal}
Under Assumptions (A1$'$), (A2)--(A5), and (C1), as $N \rightarrow \infty$ \red{and} $n\rightarrow \infty$, if for some $0<\kappa<1$, $N^{-1}n^{1+\kappa} \rightarrow 0$, $N^{-1/2} \ll |\underline{\triangle}| \leq |\overline{\triangle}| \ll n^{1/4}N^{-1/2}$, and $\|\sum_{k=1}^{\infty}\lambda_{k}^{1/2}\psi_{k}\|_{\infty} < \infty$, then for any $\bs{z}\in \Omega$, $\bs{\Sigma}_{n}^{-1/2}(\bs{z}) \{\widehat{\bs{\beta}}^{\mathrm{c}}(\bs{z})-\bs{\beta}^{o} (\bs{z}) \}\overset{\mathcal{L}}{\longrightarrow}N\left(\mathbf{0},\mathbf{I}_{(p+1)\times (p+1)}\right) $, where $\bs{\Sigma}_n(\bs{z})$ is (\ref{DEF:Sigma(z)}); $\mathrm{pr}\left\{(\sigma_{n,\ell\ell}^{\mathrm{c}})^{ -1}(\bs{z})\left\vert \widehat{\beta}_{\ell}(\bs{z})-\beta_{\ell}^{o}(\bs{z})\right\vert \leq Z_{1-\alpha/2}\right\} \rightarrow 1-\alpha$, for any $\alpha \in (0,1)$, as $N\rightarrow \infty$, $n\rightarrow \infty$, where $\sigma_{n,\ell\ell}^{\mathrm{c}}(\bs{z})$ is the square root of the $(\ell,\ell)$th entry of the matrix $\bs{\Sigma}_n(\bs{z})$, and $Z_{1-\alpha/2}$ is the \red{$100\left(1-\alpha/2\right)$th} percentile of the standard normal distribution.
\end{theorem}

\vskip .10in \noindent \textbf{3. Variance Function Estimation and Simultaneous Confidence Corridors} \vskip 0.1in
\renewcommand{\thetable}{3.\arabic{table}} \setcounter{table}{0} 
\renewcommand{\thefigure}{3.\arabic{figure}} \setcounter{figure}{0}
\renewcommand{\theequation}{3.\arabic{equation}} \setcounter{equation}{0} 

 \noindent \textbf{3.1. Estimation of the variance function} \vskip .10in

Define the estimated residual $\widehat{R}_{ij}=Y_{ij}-\sum_{\ell=0}^{p}X_{i\ell}\widehat{\beta}_{\ell}(\bs{z}_j)$ or $Y_{ij}-\sum_{\ell=0}^{p}X_{i\ell}\widehat{\beta}_{\ell}^{\mathrm{c}}(\bs{z}_j)$, for any $i=1,\ldots,n$, $j=1,\ldots,N$. 
We \red{apply} the bivariate spline smoothing method to $\{(\widehat{R}_{ij},\bs{z}_{j})\}_{j=1}^{N}$. \red{Specifically}, we define 
\begin{equation}
\widehat{\eta}_{i}(\bs{z})=\argmin_{g_{i}\in \mathcal{S}_{d}^{r}(\triangle_{\eta})} \sum_{j=1}^{N}\left\{\widehat{R}_{ij}-g_{i}(\bs{z}_{j})\right\}^{2}, ~i=1,\ldots, n, 
\label{DEF:eta_i_hat}
\end{equation}
as the spline estimator of $\eta_{i}(\bs{z})$, where the triangulation $\triangle_{\eta}$ may \red{differ from that} introduced in Section 2 when estimating $\beta_{\ell}^{o}(\bs{z})$. Next, let $\widehat{\epsilon}_{ij}=\widehat{R}_{ij}-\widehat{\eta}_{i}(\bs{z}_{j})$. Define the \red{estimators} of $G_{\eta}(\bs{z},\bs{z}^{\prime})$ and $\sigma^2(\bs{z}_j)$ as 
\begin{equation}
\widehat{G}_{\eta}(\bs{z},\bs{z}^{\prime})=n^{-1}\sum_{i=1}^{n}
\widehat{\eta}_i(\bs{z})\widehat{\eta}_i(\bs{z}^{\prime}) \red{\textrm{~and~}}
\widehat{\sigma}^2(\bs{z}_j)=n^{-1}\sum_{i=1}^{n}
\widehat{\epsilon}_{ij}\widehat{\epsilon}_{ij},
\label{DEF:G_sigma_hat}
\end{equation}
\red{respectively.} In general, for spline estimators ($d\geq 0$), denote $\widehat{\bs{\Xi}}_{n}(\bs{z}) = \left\{\widehat{\sigma}_{n,\ell\ell^{\prime}}^{2}(\bs{z})\right\}_{\ell,\ell^{\prime}=0}^{p}$, where 
\begin{align}
\label{DEF:Xi_n_hat}
\widehat{\bs{\Xi}}_{n}(\bs{z}) \!=\!\frac{1}{n^2N^2} \widetilde{\mathbb{B}}(\bs{z})^{\top}\!
\sum_{i=1}^{n}\Bigg\{\sum_{j,j^{\prime}=1}^{N}
\bs{\Gamma}_{n,\rho}^{-1} \mathbf{U}_{ij}\mathbf{U}_{ij'}^{\top}
\widehat{G}_{\eta}(\bs{z}_j,\bs{z}_{j^{\prime}})\bs{\Gamma}_{n,\rho}^{-1} + \sum_{j=1}^N \mathbf{U}_{ij}\mathbf{U}_{ij}^{\top} \widehat{\sigma}^2(\bs{z}_j)  \Bigg\} \widetilde{\mathbb{B}}(\bs{z}).
\end{align}
\red{Note} that the estimation can be much simplified if PCST smoothing is applied. In this case, the variance-covariance matrix $\bs{\Sigma}_n(\bs{z})$ can be simply estimated \red{using}
\begin{equation*} 
\widehat{\bs{\Sigma}}_n(\bs{z})=\left\{(\widehat{\sigma}_{n,\ell\ell^{\prime}}^{\mathrm{c}})^{2}(\bs{z})\right\}_{\ell,\ell^{\prime}=0}^{p} =\frac{1}{n}\left(n^{-1}\sum_{i=1}^{n}\widetilde{\mathbf{X}}_{i}\widetilde{\mathbf{X}}_{i}^{\top}\right)^{-1}
\left\{\widehat{G}_{\eta}(\bs{z},\bs{z})+\frac{\widehat{\sigma}^2(\bs{z})}{NA_{m(\bs{z})}}\right\},
\label{DEF:cov matrix_hat}
\end{equation*}
where $A_{m(\bs{z})}$ is the area of triangle $T_{m(\bs{z})}$ divided by the area of the domain.
The following conditions (C2)--(C3) are required for the bivariate spline approximation in the covariance estimation and \red{to establish} the estimation consistency. The proofs of the results in this section are provided in the Appendix A.

\begin{itemize}
\item[(C2)] For any $k\geq 1$, $\psi _{k}(\bs{z}) \in \mathcal{W}^{s+1,\infty}$ for an integer $s\geq 0$, and for a sequence $\{K_n\}_{n=1}^{\infty}$ of increasing positive integers with $\lim_{n} K_n\rightarrow \infty$, $|\triangle_{\eta}|^{s+1} \sum_{k=1}^{K_n}\lambda_k^{1/2}\|\psi_k\|_{s+1,\infty}\rightarrow 0$ as $N\rightarrow \infty$, $n\rightarrow \infty$.

\item[(C3)] As  $N \rightarrow \infty$, $n \rightarrow \infty$, for some $0<\kappa<1$, $N^{-1}n^{1/(d+1)+\kappa} \rightarrow 0$, $N|\triangle_{\eta}|^2\rightarrow \infty$, and $n|\triangle_{\eta}|^2/(\log n)^{1/2} \to \infty$.
\end{itemize}

Assumption (C2) concerns the bounded smoothness of the principal components \red{that bound} the bias terms in the spline covariance estimator.

\begin{theorem}
\label{THM:Ghat-G}
Under Assumptions (A1)--(A6) \red{and} (C1)--(C3), $\widehat{G}_{\eta}(\bs{z},\bs{z}^{\prime})$ uniformly converges to
$G_{\eta}(\bs{z},\bs{z}^{\prime})$ in probability; \red{that is},
$\sup_{(\bs{z},\bs{z}^{\prime})\in \Omega^{2}}|\widehat{G}_{\eta}(\bs{z},\bs{z}^{\prime})
-G_{\eta}(\bs{z},\bs{z}^{\prime})|=o_{P}(1)$.
\end{theorem}

\begin{corollary}
\label{COR:variance}
Under Assumptions (A1)--(A6), (C1)--(C3), the estimator of $\widehat{\bs{\Sigma}}_n(\bs{z})$ uniformly converges to
to $\bs{\Sigma}_n(\bs{z})$ in probability\red{; that is,} $\sup_{\bs{z}\in \Omega}|\widehat{\bs{\Sigma}}_{n}(\bs{z})-\bs{\Sigma}_{n}(\bs{z})|=o_{P}(1)$.
\end{corollary}

Denote
\begin{equation}
\widehat{\sigma}_{n,\ell\ell}^{\mathrm{c}}(\bs{z})=n^{-1/2}\left[\mathbf{e}_{\ell}^{\top}\left(n^{-1}
\sum_{i=1}^{n}\widetilde{\mathbf{X}}_{i}\widetilde{\mathbf{X}}_{i}^{\top}\right)^{-1}\mathbf{e}_{\ell}
\left\{\widehat{G}_{\eta}(\bs{z},\bs{z})+\frac{\widehat{\sigma}^2(\bs{z})}{NA_{m(\bs{z})}}\right\}\right]^{1/2}.
\label{EQ:sigma_hat_l}
\end{equation}
\red{From} Corollary \ref{COR:variance}, $\widehat{\sigma}^{\mathrm{c}}_{n,\ell\ell}(\bs{z})$ is a consistent estimator of $\sigma^{\mathrm{c}}_{n,\ell\ell}(\bs{z})$ in (\ref{DEF:Sigma(z)}).

\vskip .10in \noindent \textbf{3.2. Bootstrap simultaneous confidence corridors (SCCs)} \vskip .10in

From Theorems \ref{THM:beta-normality} \red{and} \ref{THM:multinormal} and Slutzky's Theorem, we have the following asymptotic PCIs.
\begin{corollary}
\label{COR:confidence-interval}
 (a) For the BPST estimators, under Assumptions (A1)--(A6), for any $\ell=0,\ldots,p$, $\alpha \in (0,1)$, as  $N\rightarrow \infty$, $n\rightarrow \infty$, an asymptotic $100(1-\alpha)\%$ PCI for $\beta_{\ell}^{o}(\bs{z})$, is $\widehat{\beta}_{\ell}(\bs{z})\pm \sigma_{n,\ell\ell}(\bs{z})Z_{1-\alpha/2}$, for any $\bs{z}\in \Omega$, where $\sigma_{n,\ell\ell}^{2}(\bs{z})$ is the $(\ell,\ell)$th entry of the matrix $\bs{\Xi}_{n}^{-1/2}(\bs{z})$, and $Z_{1-\alpha/2}$ is the \red{$100\left(1-\alpha/2\right)$th} percentile of the standard normal distribution. 
 
(b) For the PCST estimators, under Assumptions  (A1$'$) \red{and} (A2)--(A6),  if for some $0<\kappa<1$, $N^{-1}n^{1+\kappa} \rightarrow 0$, an asymptotic $100(1-\alpha)\%$ PCI for $\beta_{\ell}^{o}(\bs{z})$ is $\widehat{\beta}_{\ell}^{\mathrm{c}}(\bs{z})\pm \sigma^{\mathrm{c}}_{n,\ell\ell}(\bs{z})Z_{1-\alpha/2}$, for any $\bs{z}\in \Omega$, where $\sigma_{n,\ell\ell}^{c}(\bs{z})$ is the standard deviation function of  $\widehat{\beta}_{\ell}^{\mathrm{c}}(\bs{z})$ in Theorem \ref{THM:multinormal}.
\end{corollary}

Next, we introduce a simple bootstrap approach to extend the PCIs to the SCCs. Our approach is based on the nonparametric bootstrap method used in \cite{Hall:Horowitz:13}. We triangulate the domain $\Omega$ \red{using} quasi-uniform triangles, \red{obtaining} a set of approximate $100(1-\alpha)\%$ \red{PCIs}. In the following, $\alpha_0$ \red{denotes} the nominal confidence level of the desired SCCs. We recalibrate the PCIs using the following bootstrap method.

\begin{itemize}
\item[Step 1.] Based on $\left\lbrace (\widetilde{\mathbf{X}}_i, Y_{ij})\right\rbrace_{j=1, i=1}^{N,n}$, obtain the coefficient functions $\beta_{\ell}^{o}(\bs{z})$ \red{using} the BPST estimators $\widehat{\beta}_{\ell}(\bs{z})$ in (\ref{DEF:beta_l_hat}) or the PCST estimators $\widehat{\beta}_{\ell}^{\mathrm{c}}(\bs{z})$ in (\ref{EQN:betahat_c}), for $\ell=0,\ldots,p$. Let  $\widehat{\mu}(\bs{z})=\sum_{\ell=0}^p X_{i\ell}\widehat{\beta}_{\ell}(\bs{z})$ or $\sum_{\ell=0}^p X_{i\ell}\widehat{\beta}_{\ell}^{\mathrm{c}}(\bs{z})$. 

\item[Step 2.] Obtain  $\widehat {\eta}_i(\bs{z})$ and $\widehat{\varepsilon}_{ij}$ presented in (\ref{DEF:eta_i_hat})--(\ref{DEF:G_sigma_hat}),  and estimate $G_{\eta}(\bs{z},\bs{z})$, $\sigma^2(\bs{z})$, and $\sigma_{n,\ell\ell}^{2}(\bs{z})$ \red{using}  $\widehat{G}_{\eta}(\bs{z},\bs{z})$ and $\widehat{\sigma}^2(\bs{z})$ in (\ref{DEF:G_sigma_hat}) and $\widehat{\sigma}_{n,\ell\ell}^{2}(\bs{z})$ in (\ref{DEF:Xi_n_hat}) or (\ref{EQ:sigma_hat_l}), respectively.

\item[Step 3.]  Obtain \red{an} adjusted nominal confidence level $\widehat{\alpha}_{\ell}(\alpha_0)$. 
\begin{itemize}
\item[(i)]  Generate an independent random sample $\delta_{i}^{(b)}$ and $\delta_{ij}^{(b)}$ from $\{-1,1\}$ with probability 0.5 each, and define $Y_{ij}^{\ast(b)}=\widehat{\mu}(\bs{z}_j)+\delta_{i}^{(b)} \widehat {\eta}_i(\bs{z}_j) + \delta_{ij}^{(b)} \widehat{\varepsilon}_{ij}$.

\item[(ii)] Based on $\left\lbrace (\widetilde{\mathbf{X}}_i, Y_{ij}^{\ast (b)})\right\rbrace_{j=1, i=1}^{N,n}$, obtain $\widehat{\beta}_\ell^{*(b)}(\bs{z})$ using (\ref{DEF:beta_l_hat})  or (\ref{EQN:betahat_c}), and calculate $\widehat{\sigma}_{n,\ell\ell}^{*(b)}$ using (\ref{DEF:Xi_n_hat}) or (\ref{EQ:sigma_hat_l}).
    
\item[(iii)] Construct SCCs for the resampled data $\left\lbrace (\widetilde{\mathbf{X}}_i, Y_{ij}^{\ast (b)})\right\rbrace_{j=1, i=1}^{N,n}$: $\mathcal{B}_{(b)}^*(\alpha),\ b=1,\cdots, B$,
\[
\mathcal{B}^{*}_{(b)}(\alpha)=\{(\bs{z},y): \bs{z} \in \Omega, \widehat{\beta}_{\ell}^{*(b)}(\bs{z}) - \widehat{\sigma}_{n,\ell\ell}^{*(b)}(\bs{z}) Z_{1-\alpha/2} \leq y \leq \widehat{\beta}_{\ell}^{*(b)}(\bs{z}) + \widehat{\sigma}_{n,\ell\ell}^{*(b)}(\bs{z}) Z_{1-\alpha/2}\}.
\]
\item[(iv)] Estimate the coverage rate
$\tau_{\ell}(\bs{z}_j, \alpha) = P\{(\bs{z}_j, \widehat{\beta}_{\ell}(\bs{z}_j)) \in \mathcal{B}^*(\alpha)|\mathbb{X}\}$ using $\widehat{\tau}_{\ell}(\bs{z}_j, \alpha) = \frac{1}{B} \sum_{b=1}^{B} I\{(\bs{z}_j, \widehat{\beta}_{\ell}(\bs{z}_j)) \in \mathcal{B}_{(b)}^*(\alpha)\}$.

\item[(v)] Find the root of the equation $\widehat{\tau}_{\ell}(\bs{z}_j, \alpha) = 1-\alpha_0$, for $j=1,\ldots,N$, and denote \red{it} as $\{\widehat{\alpha}_{\ell}(\bs{z}_j, \alpha_0)\}_{j=1}^{N}$. The root can be found using the grid method by repeating the last two steps for different values of $\alpha$.

\item[(vi)] Take the minimum of $\{\widehat{\alpha}_{\ell}(\bs{z}_j, \alpha_0)\}_{j=1}^{N}$ and denote it as $\widehat{\alpha}_{\ell}\equiv\widehat{\alpha}_{\ell}(\alpha_0)$.
\end{itemize}

\item[Step 4.] Construct the final SCCs: $\mathcal{B}(\widehat{\alpha}_{\ell})=\{(\bs{z},y): \bs{z} \in \Omega, \widehat{\beta}_{\ell}(\bs{z}) - \widehat{\sigma}_{n,\ell\ell}(\bs{z}) Z_{1-\widehat{\alpha}_{\ell}/2} \leq y \leq \widehat{\beta}_{\ell}(\bs{z}) + \widehat{\sigma}_{n,\ell\ell}(\bs{z}) Z_{1-\widehat{\alpha}_{\ell}/2}\}$. 
\end{itemize}

\vskip .10in \noindent \textbf{4. Implementation} \vskip .10in

The proposed procedure can be implemented using our R package ``FDAimage" \citep{FDAimage}, in which the bivariate spline basis \red{is} generated \red{using} the R package ``BPST" \citep{BPST}. When the response imaging seems to be a realization from some smooth function, we suggest using \red{the smoothing} parameter $r=1$ and degree $d\geq5$, which achieves full estimation power asymptotically \citep{Lai:Schumaker:07}. In contrast, if there are sharp edges on the images, we suggest considering the PCST presented in Section 2.4.

Selecting suitable values \red{for the }smoothing parameters is important to good model fitting. To select $\rho_{n,\ell}$, for $\ell=0,\ldots,p$, we used $K$-fold cross-validation (CV). The individuals are randomly partitioned into $K$ groups, \red{where} one group is retained as a test set, and the remaining $K-1$ groups are used as training sets. The CV process is repeated $K$ times (the folds), with each of the $K$ groups used exactly once as the validation data. Then, the $K$-fold CV score is
\[
{\rm CV}(\rho_{n,0},\ldots,\rho_{n,p})= K^{-1}\sum_{k=1}^{K}(|\mathcal{V}_k|N)^{-1}\sum_{i \in \mathcal{V}_k} \sum_{j=1}^N\{Y_{ij}-\widetilde{\mathbf{X}}_i^{\top}\widehat{\bs{\beta}}_{-k}(\bs{z}_j)\}^2,
\]
where $\mathcal{V}_k$ is the $k$th testing set for $k=1,\ldots, K$, and $\widehat{\bs{\beta}}_{-k}$ is the corresponding estimator after removing the $k$th testing set. We \red{use} $K = 5$ in \red{our} numerical examples.

To determine an optimal triangulation, the criterion usually considers the shape, size, or number of triangles. In terms of shape, a ``good" triangulation usually refers to \red{one} with well-shaped triangles without small angles \red{and/or} obtuse angles. Therefore, for a given number of triangles, \cite{Lai:Schumaker:07} and \cite{Lindgren:etal:11} recommended selecting the triangulation according to ``max-min" criterion, which maximizes the minimum angle of all the angles of the triangles in the triangulation. With respect to the number of triangles, our numerical studies show that a lower limit of the number of triangles is necessary to capture the features of the images. \red{However}, once this minimum number has been reached, refining the triangulation \red{further} usually has little effect on the fitting process. In practice, when using higher-order BPST smoothing, we suggest taking the number of triangles as $H_n=\min\{\lfloor c_1 n^{1/(2d+2)}N^{1/2} \rfloor,  N/10\}$, where $c_1$ is a tuning parameter. \red{We} find that $c_1\in[0.3, 2.0]$ works well in our numerical studies. \red{When} using the PCST, we suggest taking the number of triangles as $H_n=\min\{\lfloor c_2 n^{-1/4} N \rfloor,  N/2\} \mathrm{, with}~c_2\in[0.3,2.0]$. Once $H_n$ is chosen, \red{we} can build the triangulation using typical triangulation construction methods, such as Delaunay \red{triangulation} and \red{DistMesh} \citep{Persson:Strang:04}.

\vskip .10in \noindent \textbf{5. Simulation Studies} \vskip 0.1in
\renewcommand{\thetable}{5.\arabic{table}} \setcounter{table}{0} 
\renewcommand{\thefigure}{5.\arabic{figure}} \setcounter{figure}{0}
\renewcommand{\theequation}{5.\arabic{equation}} \setcounter{equation}{0} 

In this section, we conduct two Monte Carlo simulation studies using our R package ``FDAimage" \citep{FDAimage} to examine the finite\red{-}sample performance of the proposed methodology. The triangulations used \red{here} can be found \red{in the data set} in the ``FDAimage" package. To illustrate the performance of our estimation method, we compare the proposed spline method with the kernel method proposed by \cite{Zhu:Fan:Kong:14} (Kernel) and the tensor regression method of \cite{Li:Zhang:17} (Tensor). To implement the kernel method, we use the R Package \textit{SVCM}, which is publicly available at \url{https://github.com/BIG-S2/SVCM}. \red{For the tensor method}, the accompanying \red{MATLAB} code at \url{https://ani.stat.fsu.edu/~henry/TensorEnvelopes_html.html} is used. We compare the proposed method with the tensor regression approach in \cite{Li:Zhang:17} and the three-stage FDA approach in \cite{Zhu:Fan:Kong:14}.

\vskip .10in \noindent \textbf{5.1. Example 1} \vskip .10in

To illustrate the advantage of the proposed method over a complex domain, we study the horseshoe domain in \cite{Sangalli:Ramsay:Ramsay:13}. The response images are generated from the following model: $Y_{ij}=\beta_0^o(\bs{z}_j)+X_{i}\beta_1^o(\bs{z}_j)+\eta_i(\bs{z}_j)+\sigma\varepsilon_{ij}$, for $i=1,\ldots,n$, $j=1,\ldots,N$, and $\bs{z}_j\in \Omega$. To understand the advantages and disadvantages of different methods, we consider two types of coefficient functions in the above image-on-scalar regression model: (I) functions with jumps; and (II) smooth functions. The true coefficient functions are \red{shown} in Figure \ref{FIG:coeff_true_simu1}.

\begin{figure}[t]
	\begin{center}
		\begin{tabular}{ccccc} 
		\multicolumn{2}{c}{Case I (jump functions)}&\multicolumn{2}{c}{Case II (smooth functions)}\\
			\includegraphics[scale=0.38]{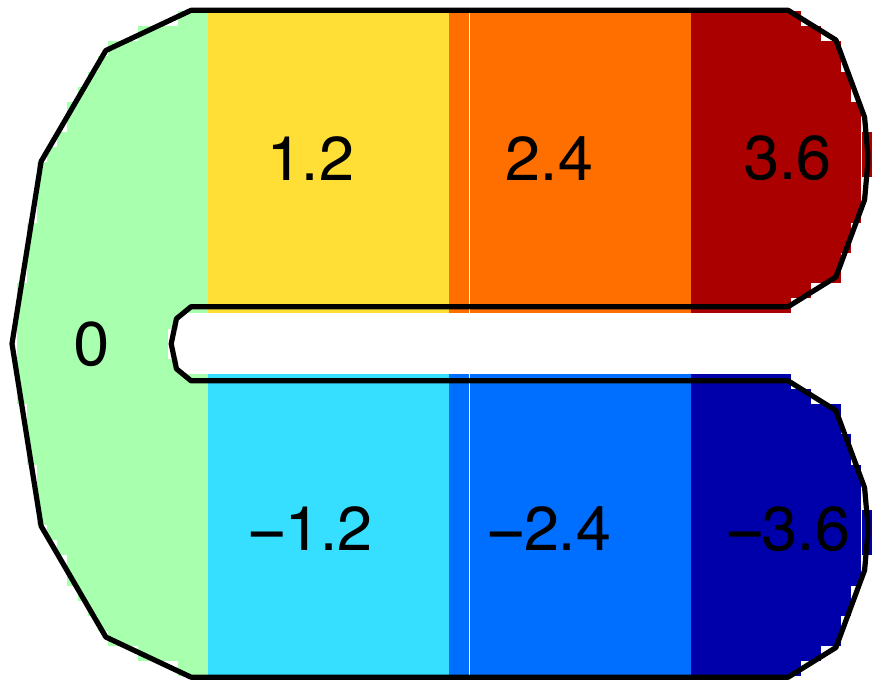} & 
			\includegraphics[scale=0.38]{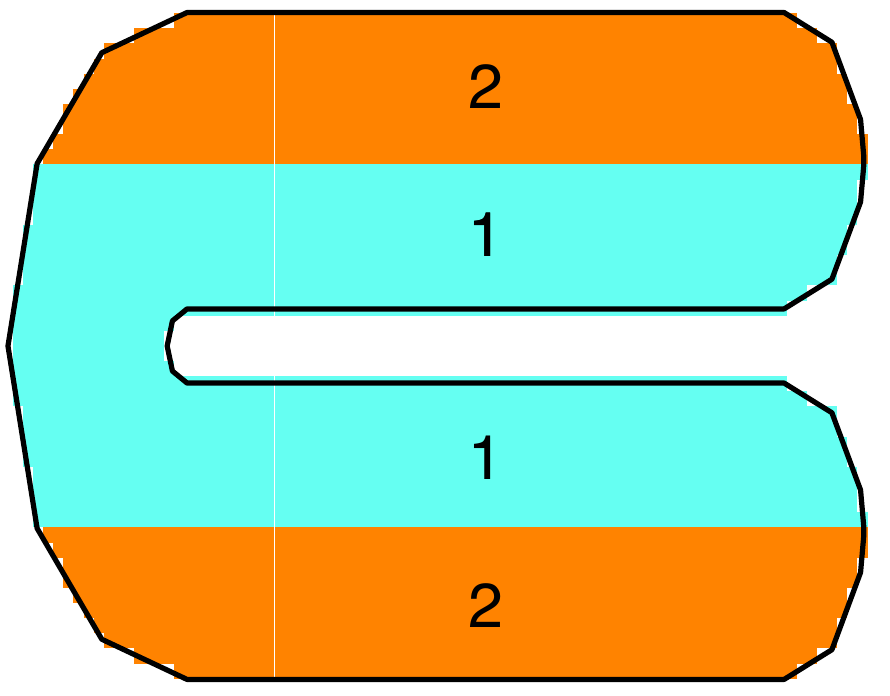}~~& 
			~~\includegraphics[scale=0.38]{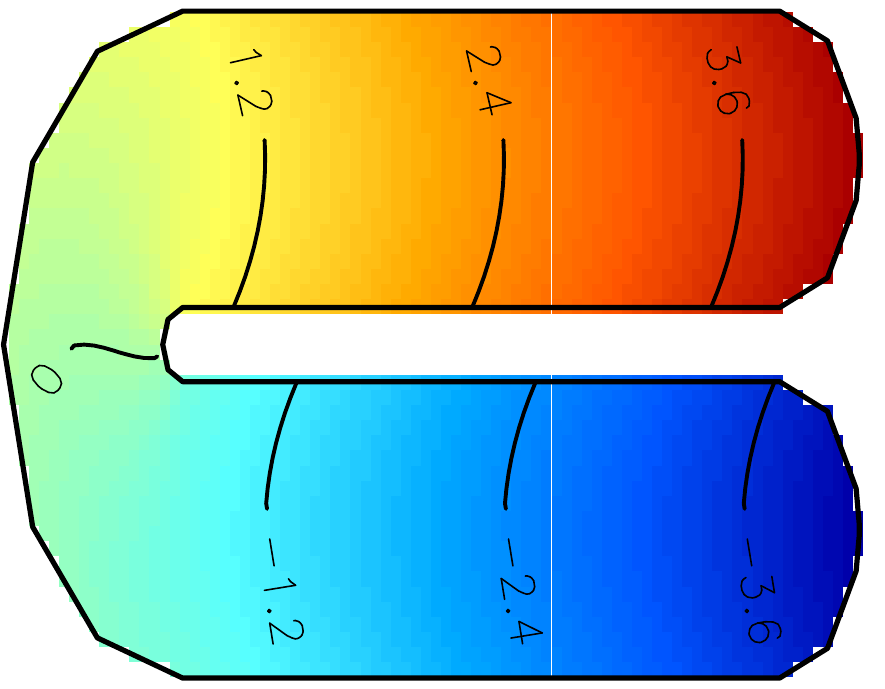}&
			\includegraphics[scale=0.38]{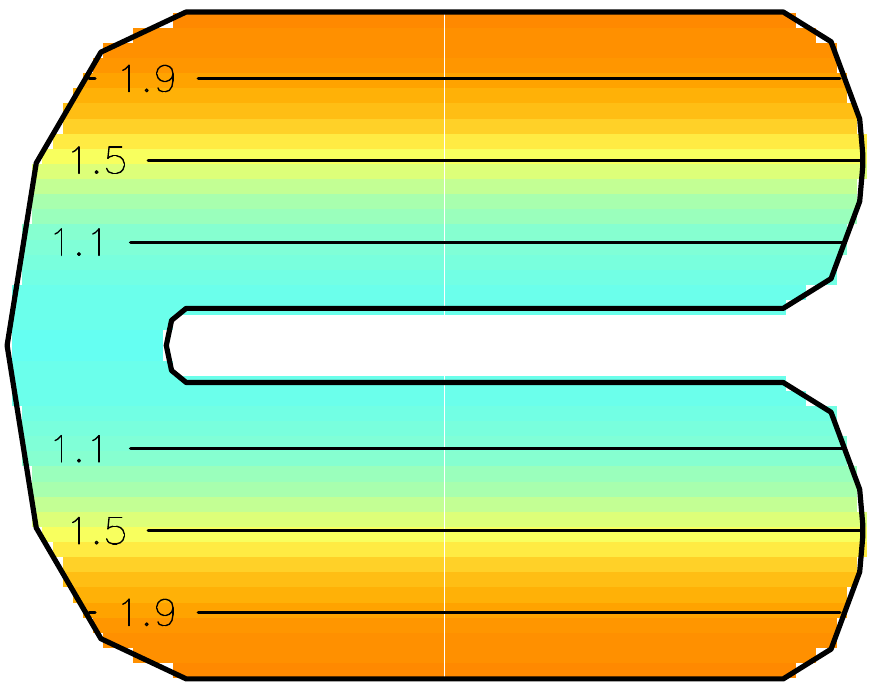} \\[-5pt]
			$\beta_0^{o}$&$\beta_1^{o}$&$\beta_0^{o}$&$\beta_1^{o}$
		\end{tabular}	
	\end{center} \vspace{-.2in}
	\caption{The true coefficient functions in Simulation Example 1.}
	\label{FIG:coeff_true_simu1}
\end{figure}

For each image, we set the resolution as $100 \times 50$ (pixels). The true signal falls only within the horseshoe domain (3182 pixels)\red{;} outside the domain \red{is} pure noise. We generate the scalar covariate $X_{i}\sim N\left(0,1\right)$, and then truncate it by $[-3,+3]$. We set $\eta_i(\bs{z})=\sum_{k=1}^{2}\lambda_k^{1/2}\xi_{ik}\psi_k(\bs{z})$, where $(\lambda_1, \lambda_2)=(0.1,0.02) \text{~or~} (0.2,0.05)$ \red{and} $\xi_{i1}$ and $\xi_{i2} \sim N(0,1)$, $\psi_1(\bs{z})=c_1\sin(2\pi z_1)$, and $\psi_2(\bs{z})=c_2\cos(2\pi z_2)$.  Let $c_1=0.56$ \red{and} $c_2=0.61$, \red{such} that $\psi_1$ and $\psi_2$ are orthonormal functions on $\Omega$. The measurement error $\varepsilon_{ij}$ is independently generated from $N(0,1)$ and $\sigma=1.0,~2.0$. 
\begin{table}[t]
\begin{center}
\caption{Estimation errors of the coefficient estimators, $\sigma=2.0$. \label{TAB:eg1_01}}
\renewcommand\arraystretch{0.75}
\scalebox{0.9}{\begin{tabular}{cccccccccc}\\ \hline\hline
Function &\multirow{2}{*}{$n$} &\multirow{2}{*}{Method} &\multicolumn{2}{c}{$\lambda_1=0.03,~\lambda_2=0.006$}& &\multicolumn{2}{c}{$\lambda_1=0.2,~\lambda_2=0.05$}  \\ \cline{4-5} \cline{7-8}
Type &&&$\beta_0$ &$\beta_1$ & &$\beta_0$ &$\beta_1$\\ \hline

\multirow{8}{*}{Jump} &\multirow{4}{*}{50} &BPST &0.0139 &0.0182 & &0.0145 &0.0189 \\
& & PCST &0.0088 &0.0090 & &0.0094 &0.0097 \\
& & Kernel &0.0801 &0.0819 & &0.0807 &0.0826 \\ 
& & Tensor &0.0799 &0.0248 & &0.0799 &0.0254\\ \cline{2-8}
	
&\multirow{4}{*}{100} &BPST &0.0090 &0.0118 & &0.0093 &0.0122\\
& &PCST &0.0044 &0.0044 & &0.0047 &0.0047 \\
& &Kernel &0.0400 &0.0405 & &0.0403 &0.0409 \\ 
& &Tensor &0.0395 &0.0166 & &0.0399 &0.0171 \\ \hline

\multirow{8}{*}{Smooth} &\multirow{4}{*}{50} &BPST &0.0026 &0.0032 & &0.0032 &0.0041 \\
&  &PCST &0.0088 &0.0090 & &0.0119 &0.0139 \\
&  &Kernel &0.0801 &0.0819 & &0.0807 &0.0826 \\ 
&  &Tensor &0.0799 &0.0256 & &0.0806 &0.0271\\ \cline{2-8}
	
&\multirow{4}{*}{100} &BPST &0.0016 &0.0019 & &0.0019 &0.0022 \\
& & PCST &0.0070 &0.0086 & &0.0073 &0.0090 \\
& & Kernel &0.0400 &0.0405 & &0.0403 &0.0409 \\ 
& & Tensor &0.0399 &0.0168 & &0.0402 &0.0179\\ \hline\hline
\end{tabular}}
\end{center}
\end{table}

To fit the model, we consider the BPST and PCST methods presented in Section 2. To obtain the BPST estimators, we set $d=5$ and $r=0$ when generating the bivariate spline basis functions. Figure \ref{FIG:triangulations_simu1} in the Appendix B
illustrates the triangulations used for the BPST and PCST. The triangulation  used for the BPST ($\triangle_1$)  contains 90 triangles (73 vertices), and the triangulation used for the PCST ($\triangle_2$)  contains 346 triangles (226 vertices).

We quantify the estimation accuracy of the coefficient functions using the mean squared error (MSE). Table \ref{TAB:eg1_01} provides the average MSE (across 500  Monte Carlo experiments) for two types of coefficient functions. To save space, we present the results for $\sigma=2.0$ \red{only}; the results for $\sigma=1.0$ are presented in  Table \ref{TAB:eg1_01_2}  in the Appendix B. As expected, the estimation accuracy of all the methods improves as the sample size increases or the noise level decreases. In both scenarios, the BPST and PCST outperform the other two competitors, reflecting the advantage of our method over a complex domain. When the true coefficient functions are smooth, the BPST provides the best estimation, followed by the PCST. On the other hand, when the true coefficient function contains jumps, the PCST provides a better result. For \red{the tensor} regression, the estimator of $\beta_1^{o}(\cdot)$ is much more accurate than \red{that} of $\beta_0^{o}(\cdot)$,  \red{owing to} the design of the coefficient function. Figure \ref{FIG:coeff_true_simu1} \red{shows} that,  in contrast to the intercept function of $\beta_0^{o}(\cdot)$, the true slope function of $\beta_1^{o}(\cdot)$ is still smooth across the complex boundary. Moreover, when the coefficient function is smooth across the boundary, the estimation accuracy is also affected by the domain of the true signal. The performance of the kernel method is not affected by the design of the coefficient functions. \red{Instead, it depends} heavily on the noise level,  \red{owing} to the three-stage structure.

\vskip .10in \noindent \textbf{5.2. Example 2} \vskip .10in

In this example, we \red{simulate the data by considering} the domains of the \red{fifth} and 35th slices of the brain images illustrated in Section 6 as the domain $\Omega$. We generate  response images based on a set of smooth coefficient functions from the following model: $Y_{ij}=\sum_{\ell=0}^{2}X_{i\ell}\beta_{\ell}^{o}(\bs{z}_{j})+\eta_{i}(\bs{z}_{j})+\sigma\varepsilon_{ij}$, for $i=1,\ldots,n$, $j=1,\ldots,N$, and $\bs{z}_j\in \Omega$, where $\beta_0^{o}(\bs{z})=5\{(z_1-0.5)^2+(z_2-0.5)^2\}$,  $\beta_1^{o}(\bs{z})=-1.5z_1^3+1.5z_2^3$ and $\beta_2^{o}(\bs{z})=2-2\exp[-8\{(z_1-0.5)^2+(z_2-0.5)^2\}]$. the true coefficient images are shown in the first \red{columns} of Figures \ref{FIG:EST_SCC} and \ref{FIG:EST_SCC_s35} in the Appendix B for the fifth and 35th slices, respectively. For each image, we simulate the data at all $79\times 95$ pixels. To mimic real brain images, the true signals are \red{generated only} on the pixels/voxels (3476 or 5203 pixels in total) within the brain domain; outside the boundary of the brain, the image contains only noise. We set $X_{i0}=1$ and generate $\widetilde{\mathbf{X}}_i=(X_{i1},X_{i2})^\top \sim N\left(\mathbf{0},\boldsymbol{\Sigma}\right)$, with $\boldsymbol{\Sigma}=\binom{1.0~~0.5}{0.5~~1.0}$ and $X_{i\ell}$ truncated by $[-3,+3]$. For the error terms, we set $\eta_i(\bs{z})=\sum_{k=1}^{2}\lambda_k^{1/2}\xi_{ik}\psi_k(\bs{z})$, where $\xi_{i1}$ and $\xi_{i2} \sim N(0,1)$, $\psi_1(\bs{z})=1.488\{\sin(\pi z_1)-1.5\}$, $\psi_2(\bs{z})=1.939\cos(2\pi z_2)$, \red{and} $(\lambda_1, \lambda_2)=(0.1,0.02) \text{~or~} (0.2,0.05)$. The measurement error $\varepsilon_{ij}$ is independently generated from $N(0,1)$ and $\sigma =0.5,~1.0$. \red{To conserve space}, we \red{show} only the results for the domain of the fifth slice for $\sigma=1.0$ \red{here}. The results for $\sigma=0.5$ and \red{those} based on  the domain of the 35th slice are shown in Appendix B. 

\red{Because} the functions in this example are smooth, for the bivariate spline approach, we consider \red{only} the BPST method. To further study the effect of different triangulations, we consider  $\triangle_3$ and $\triangle_4$; see Figure \ref{FIG:triangulations_simu2} in the Appendix B. Similarly to Section 5.1, we summarize the MSE for different coefficient functions based on 500 Monte Carlo experiments in Table \ref{TAB:estimation}. Columns 2--5 in Figure \ref{FIG:EST_SCC} in the Appendix B show the estimated coefficient functions using the kernel, tensor and BPST methods, respectively. Table \ref{TAB:estimation} and Figure \ref{FIG:EST_SCC} in the Appendix B show that the estimation accuracy \red{improves} for all methods  as the sample size increases or the noise level decreases. In all settings, the BPST method has the smallest MSE compared with the kernel and tensor methods, reflecting the advantage of our method in estimating the coefficient functions and, hence, the regression function. \red{Because  the kernel and tensor} methods are \red{both} designed for a rectangle domain, the estimation accuracy can be affected by the \red{noise} outside the domain. \red{Futhermore,} the MSE is invariable across two triangulations, \red{thus,} $\triangle_{3}$ might be sufficient to capture the feature in the \red{data set}. \red{This} also implies that when this minimum number of triangles is reached, further refining the triangulation \red{has} little effect on the fitting process, but makes the computational burden unnecessarily heavy.

\begin{table}[t]
\begin{center}
\caption{Estimation errors of the coefficient function estimators, $\sigma=1.0$. \label{TAB:estimation}}
\renewcommand\arraystretch{0.75}
\scalebox{0.9}{\begin{tabular}{lccccccccccr}\\ \hline\hline
\multirow{2}{*}{$n$}  &\multirow{2}{*}{Method} &\multicolumn{3}{c}{$\lambda_1=0.1,~\lambda_2=0.02$}& \multicolumn{3}{c}{$\lambda_1=0.2,~\lambda_2=0.05$} \\ \cline{3-5} \cline{6-8}
	&&$\beta_0$ &$\beta_1$ &$\beta_2$ & $\beta_0$ &$\beta_3$ &$\beta_2$\\ \hline
	\multirow{4}{*}{50} &BPST($\triangle_{3}$)&0.003&0.005&0.005&0.007&0.011&0.010 \\
	&BPST($\triangle_{4}$)&0.003&0.005&0.005& 0.007&0.010&0.009\\
	&Kernel &0.023 & 0.032 & 0.032 & 0.026 & 0.037 & 0.037 \\
	&Tensor &0.023& 0.013&0.019&0.026 &  0.017 & 0.024\\ \hline
	
	\multirow{4}{*}{100} &BPST($\triangle_{3}$) & 0.002&0.002&0.002 &0.003&0.005&0.005\\
	&BPST($\triangle_{4}$)& 0.002&0.002&0.002 &0.003&0.004&0.004 \\
	&Kernel &0.011& 0.015&0.015 & 0.013 & 0.018& 0.018\\
	&Tensor & 0.011&0.007& 0.011& 0.013  &  0.009  &  0.013\\ \hline\hline
	\end{tabular}}
\end{center}  \vspace{-.2in}
\end{table}

Finally, we illustrate the finite-sample performance of the proposed SCCs for the coefficient functions described in Section 3. In particular, we report the empirical coverage probabilities of the nominal 95\% SCCs using triangulation $\triangle_3$. We evaluate the coverage of the proposed SCCs over all pixels on the interior of $\Omega$, and test whether the true functions are entirely covered by the SCCs at these pixels. Table \ref{TAB:coverage} summarizes the empirical coverage rate (ECR) for 500 Monte Carlo experiments of the 95\% SCCs and the average width of the SCCs. The results clearly show \red{that} the ECRs of the SCCs are well approximated to 95\%, \red{particularly as} the sample size increases. Table \ref{TAB:coverage}  also reveals that the SCCs tend to be narrower when the sample size becomes larger or the noise level \red{decreases}.

\begin{table}[t]
\begin{center}
\renewcommand*{\arraystretch}{0.75}
\caption{The coverage rate of the $95\%$ SCCs for the coefficient functions. \label{TAB:coverage}} 
\scalebox{0.9}{\begin{tabular}{lcccccccccr}\\\hline\hline
\multirow{2}{*}{$n$}&\multirow{2}{*}{$\lambda$}&\multirow{2}{*}{$\sigma$}&\multicolumn{3}{c}{Coverage}&\multicolumn{3}{c}{Width}\\\cline{4-9}
&&&$\beta_0$&$\beta_1$&$\beta_2$&$\beta_0$&$\beta_1$&$\beta_2$\\\hline
\multirow{4}{*}{50}&\multirow{2}{*}{(0.1,0.02)}&0.5&0.976&0.928&0.938&0.332&0.362&0.377\\
&&1.0&0.976&0.940&0.952&0.358&0.392&0.413\\
\cline{2-9}
&\multirow{2}{*}{(0.2,0.05)}&0.5&0.962&0.918&0.932&0.445&0.497&0.513\\
&&1.0&0.970&0.930&0.940&0.478&0.527&0.544\\
\hline
\multirow{4}{*}{100}&\multirow{2}{*}{(0.1,0.02)}&{0.5}&0.970&0.956&0.956&0.234&0.250&0.267\\
&&1.0&0.978&0.968&0.978&0.262&0.285&0.297\\
\cline{2-9}
&\multirow{2}{*}{(0.2,0.05)}&0.5&0.956&0.958&0.936&0.313&0.348&0.357\\
&&1.0&0.966&0.964&0.954&0.344&0.378&0.389\\
\hline\hline
\end{tabular}}
\end{center}  \vspace{-.2in}
\end{table}


\vskip 0.1in  \noindent \textbf{6. ADNI Data Analysis} \vskip 0.1in
\renewcommand{\thetable}{5.\arabic{table}} \setcounter{table}{0} 
\renewcommand{\thefigure}{5.\arabic{figure}} \setcounter{figure}{0}
\renewcommand{\theequation}{5.\arabic{equation}} \setcounter{equation}{0} 

To illustrate the proposed method, we consider the spatially normalized FDG (fludeoxyglucose) PET data of the Alzheimer's Disease Neuroimaging Initiative (ADNI). As pointed out in \cite{Marcus:Mena:Subramaniam:14}, FDG-PET images have been shown to be a promising modality for detecting functional brain changes in Alzheimer's Disease (AD). The data can be obtained from the ADNI database at \url{http://adni.loni.usc.edu/}. The database contains spatially normalized PET images of 447 subjects. Of these 447 subjects, 112 have normal cognitive functions, considered to be the control group, 213 are diagnosed as mild cognitive impairment (MCI), and 122 are diagnosed as AD. Table \ref{TAB:diagnosis} in the Appendix B summarizes the distribution of patients by diagnosis status and sex.

In this study, we examine several patient-level features\red{:} (i) demographical features, such as age (Age) and sex (Sex); (ii) a dummy variable for the abnormal diagnosis status ``MA" ($1=$ ``AD" or ``MCI", zero otherwise); (iii) a dummy variable for ``AD'' ($1=$ ``AD," zero otherwise); and (iv) dummy variables \red{for the} APOE genotype, the strongest genetic risk factor for ``AD"; see \cite{CorderEtAl:93}. We code APOE$_1$ as a dummy variable for subjects with one epsilon 4 allele, and APOE$_2$ as subjects who have two alleles. 

Noting that the PET images are 3D, we select the 5th, 8th, 15th, 35th, 55th, 62nd, and 6fifth horizontal slices (bottom to up) of the brain from a total of 68 slices to illustrate our method. Each slice of the image contains 79$\times$95 pixels, but the domains of \red{different brain slices} are quite different. Specifically, the domain boundary for the bottom slices and upper slices are much more complex than the slices in the middle; more examples can be found in Figure \ref{FIG:APP-TRI} in the Appendix B. For each slice, we consider the following image-on-scalar regression:
\begin{align*}
Y_{i}(\bs{z}_j)=&\beta_0(\bs{z}_j) + \beta_1(\bs{z}_j)\textrm{MA}_i + \beta_2(\bs{z}_j)\textrm{AD}_i + \beta_3(\bs{z}_j)\textrm{Age}_i + \beta_4(\bs{z}_j)\textrm{Sex}_i \\
&+\beta_5(\bs{z}_j)\textrm{APOE}_{1i} +\beta_6(\bs{z}_j)\textrm{APOE}_{2i} +\eta_i(\bs{z}_j) +\sigma(\bs{z}_j)\varepsilon_{i}(\bs{z}_j),~i=1,\ldots,n.
\end{align*}

We fit the above model using the BPST method for each slice; see Figure \ref{FIG:APP-TRI} in the Appendix B for the set of triangulations used for the BPST method. The image maps in Figure \ref{FIG:APP-EST} \red{ and Figures} \ref{FIG:APP-EST2} and \ref{FIG:APP-EST3} in the Appendix B present the estimated coefficient functions using the BPST ($d=5$, $r=1$) method. To evaluate the predictive performance, Table \ref{TAB:CV_ADNI} reports the 10-fold CV (parts of the images are left out as training sets) MSPE results for the BPST method, kernel method in \cite{Zhu:Fan:Kong:14}, and tensor regression method in \cite{Li:Zhang:17}. The table shows that the MSPEs of the BPST method are uniformly smaller than those of the kernel method and tensor regression methods.

\begin{table}[t]
\begin{center}
\renewcommand*{\arraystretch}{0.75}
\caption{10-fold CV results for the ADNI dataset. ($\times 10^{-2}$) \label{TAB:CV_ADNI}}  
\scalebox{0.9}{\begin{tabular}{ccccccccc}\\ \hline \hline 
Method &Slice 5 &Slice 8 &Slice 15 &Slice 35 &Slice 55 &Slice 62 &Slice 65\\ \hline
BPST &1.4508 &1.4809 &1.5013 &1.5633 &2.0693 &2.3020 &2.6239\\
Kernel &1.4533 &1.4828 &1.5021 &1.5638 &2.0715 &2.3060 &2.6303\\
Tensor &1.5010 &1.5260 &1.5400 &1.5900 &2.1000 &2.3340 &2.6400\\ \hline\hline 
\end{tabular}} 
\end{center}  \vspace{-.2in}
\end{table}

\begin{figure}
	\begin{center}
		\begin{tabular}{cccccccc} 
			\multicolumn{2}{c}{Intercept} & \multicolumn{2}{c}{MA} \\ 
			\includegraphics[scale=0.49]{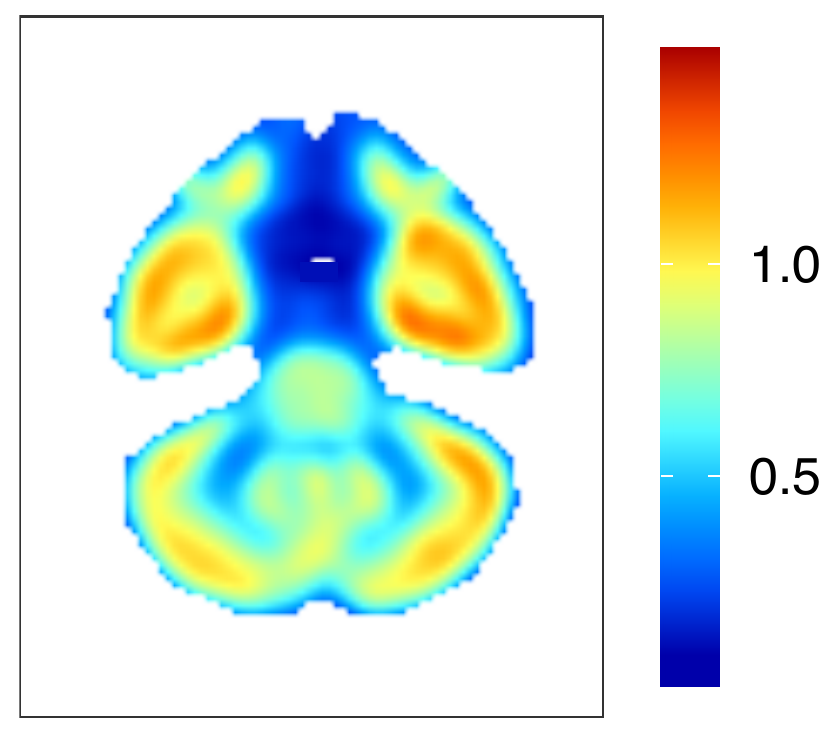} \!\!\!\!\! & \!\!\!\!\!
			\includegraphics[scale=0.282]{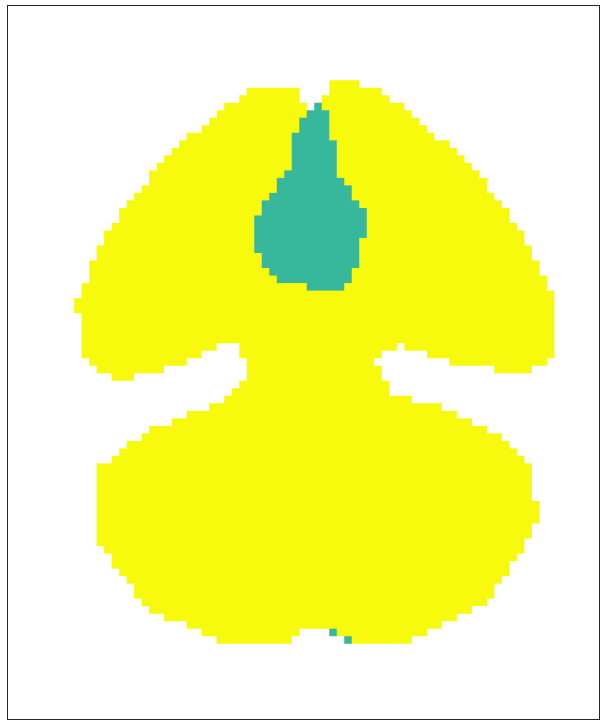}~&~
			\includegraphics[scale=0.49]{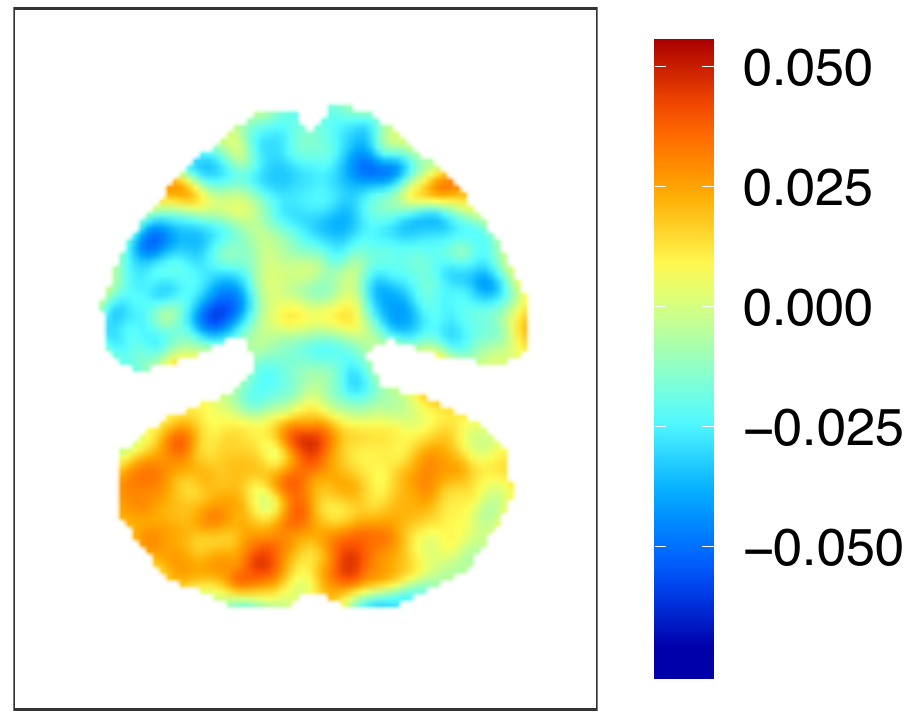} \!\!\!\!\! & \!\!\!\!\!
			\includegraphics[scale=0.282]{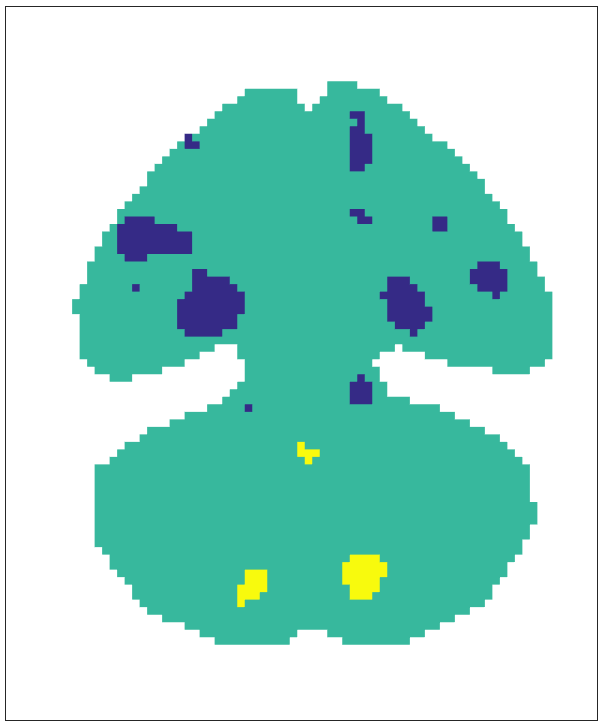} \!\!\!\!\! & \!\!\!\!\\[-10pt]
			\multicolumn{2}{c}{AD} & \multicolumn{2}{c}{Age}\\ 
			\includegraphics[scale=0.49]{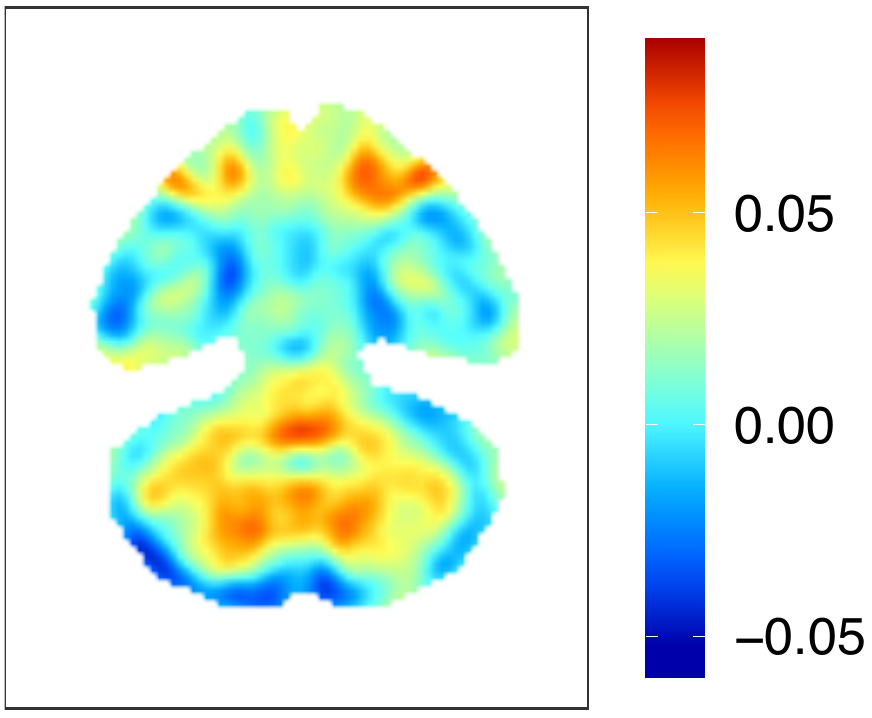} \!\!\!\!\! & \!\!\!\!\!
			\includegraphics[scale=0.282]{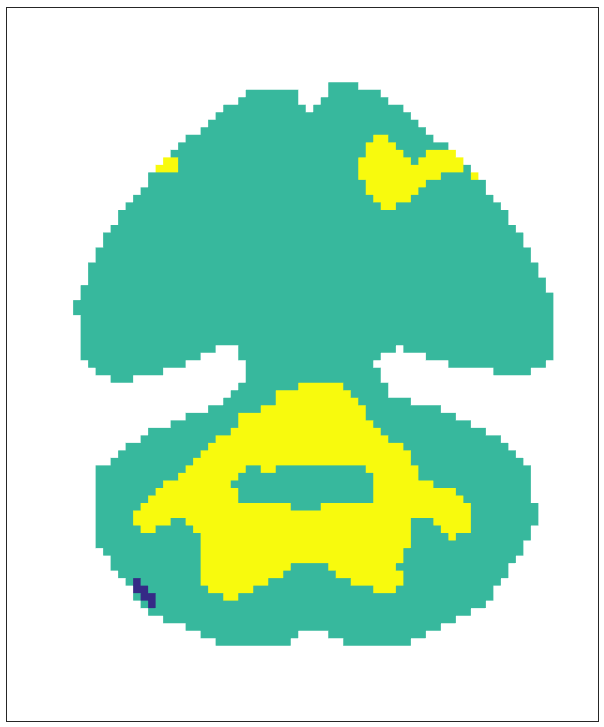} ~&~
			\includegraphics[scale=0.49]{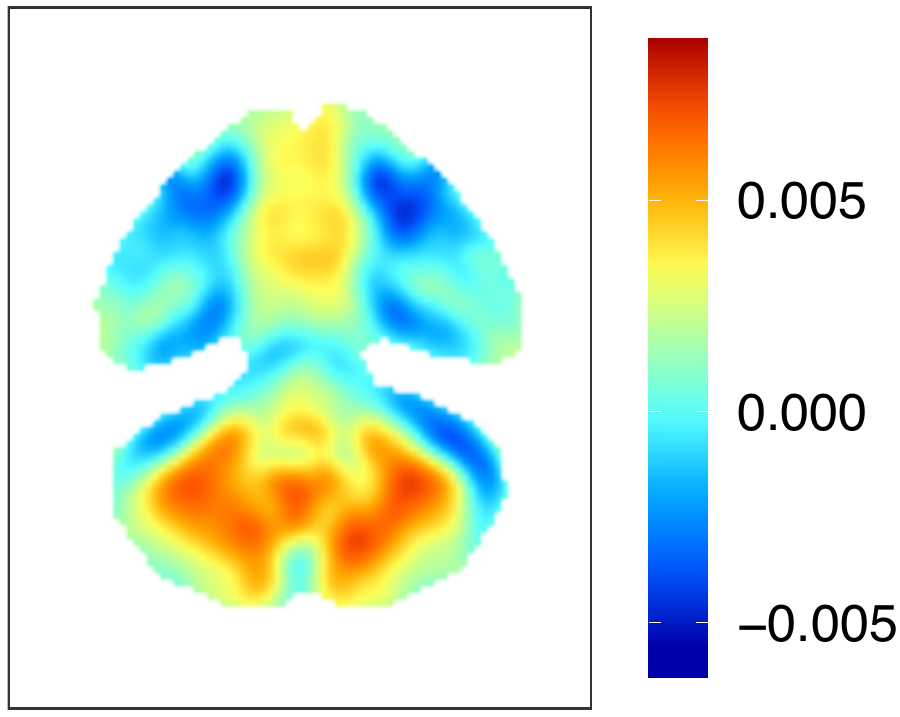}\!\!\!\!\! & \!\!\!\!\!
			\includegraphics[scale=0.282]{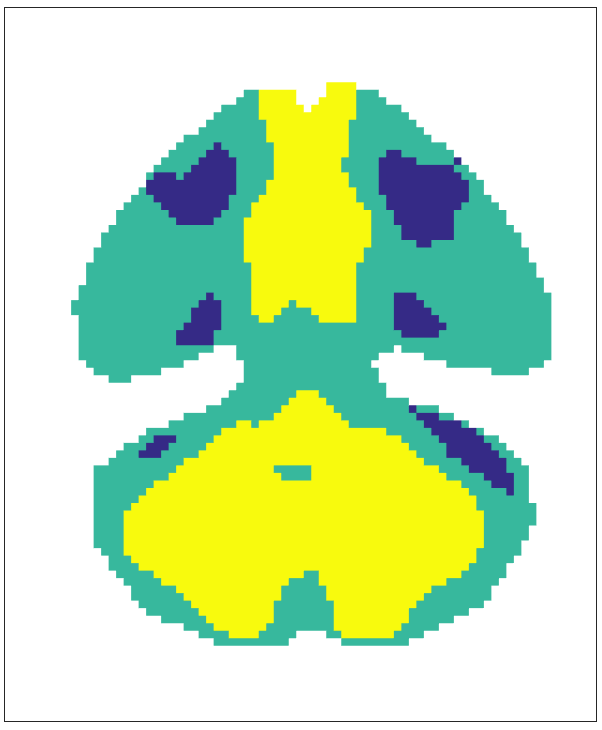}\!\!\!\!\! & \!\!\!\!\\[-10pt]
			\multicolumn{2}{c}{Sex} & \multicolumn{2}{c}{$\textrm{APOE}_{1}$}\\ 
			\includegraphics[scale=0.49]{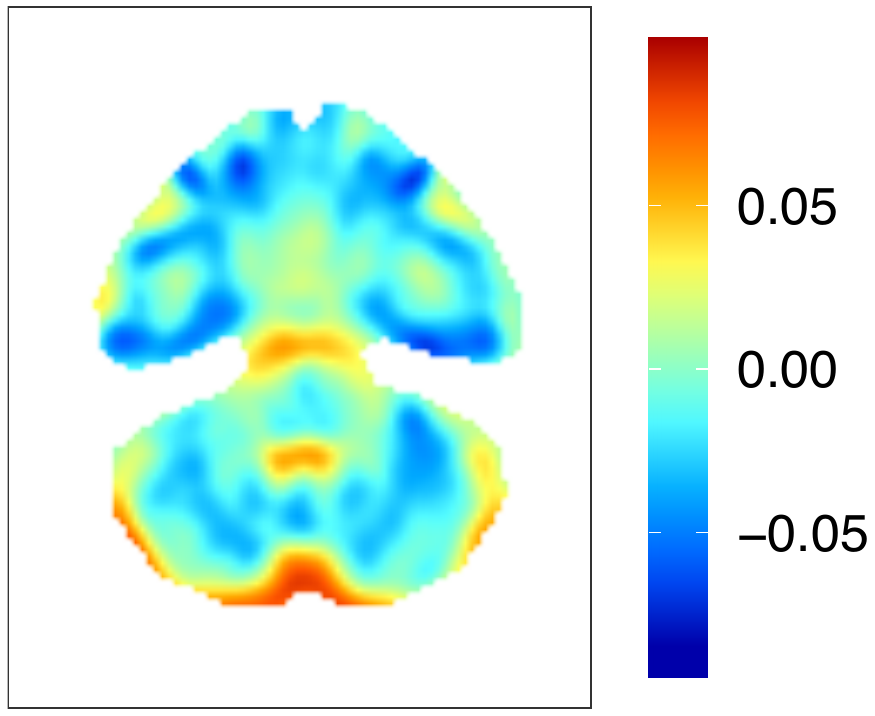} \!\!\!\!\! & \!\!\!\!\!
			\includegraphics[scale=0.282]{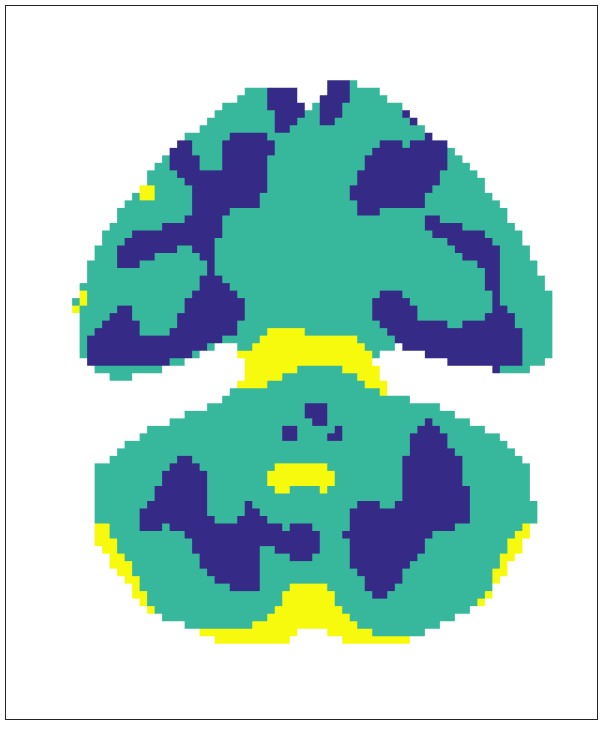}~&~
			\includegraphics[scale=0.49]{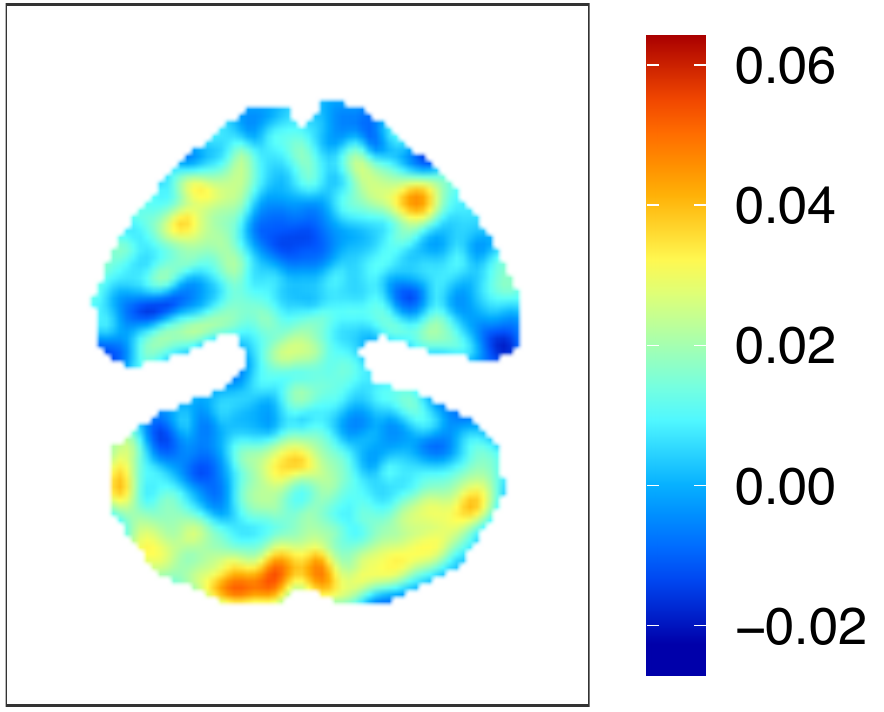} \!\!\!\!\! & \!\!\!\!\!
			\includegraphics[scale=0.282]{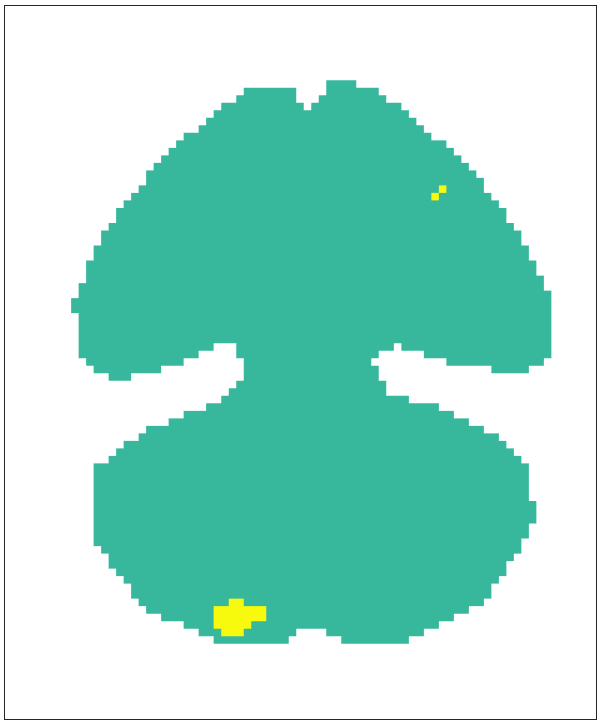} \!\!\!\!\! & \!\!\!\!\\[-10pt]
			\multicolumn{2}{c}{$\textrm{APOE}_{2}$}\\
			\includegraphics[scale=0.49]{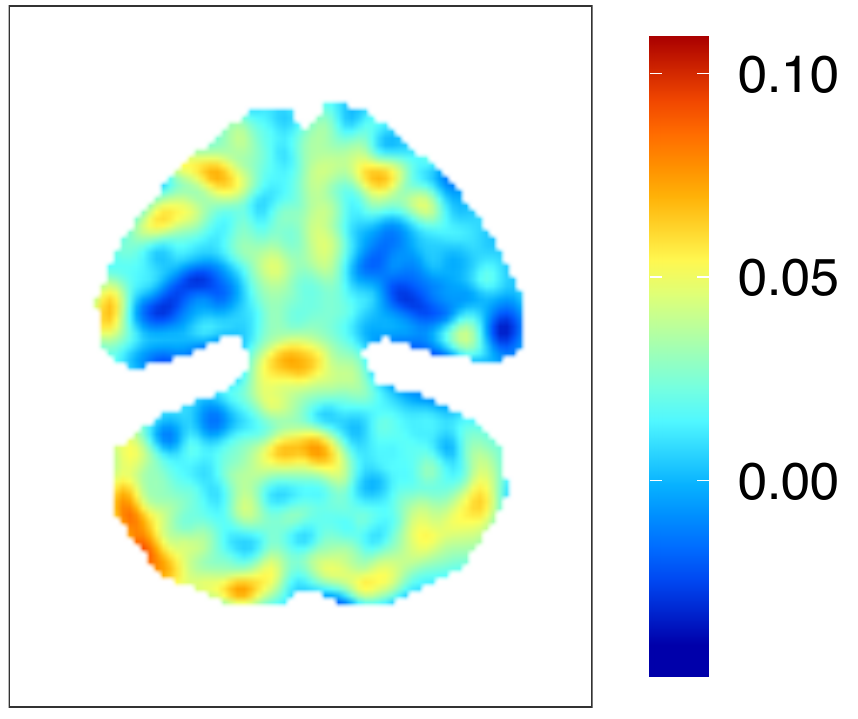} \!\!\!\!\! & \!\!\!\!
			\includegraphics[scale=0.282]{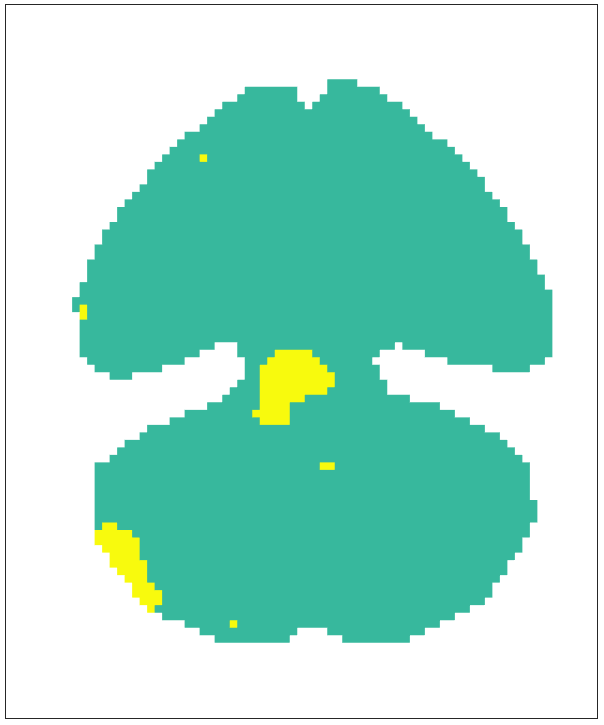}
			\\[-10pt]
	\end{tabular}	
	\end{center}
	\caption{The BPST estimate and significance map of the coefficient functions for the fifth slice of the PET images. The yellow and blue colors in the significance map indicate the regions \red{in which} zero is below the lower SCC or above the upper SCC, respectively.}
	\label{FIG:APP-EST}
\end{figure}

Next, we construct the 95\% SCCs to check whether the covariates are significant. The yellow and blue colors on the ``significance" map in Figure \ref{FIG:APP-EST} indicate the regions in which zero is below the lower SCC or above the upper SCC, respectively. Using these estimated coefficient functions and the 95\% SCCs, we can assess the impact of the covariates on the response images. Taking the fifth slice as an example,  the main impact of ``AD" on \red{in the} PET images is an increase in activity in the cerebellum compared with a normal individual. The cerebellum obtains information from the sensory systems, spinal cord, and other parts of the brain, and then regulates motor movements, resulting in smooth and balanced muscular activities. The significance map of ``Age" also shows an increase in activity in the cerebellum, and ``Sex"  shows different effects in the male and female brain images. The significance maps of the covariates for all other slices of the PET image are shown in Figures \ref{FIG:APP-SCC2} --  \ref{FIG:APP-SCC3} in the Appendix B. From these figures, we can see that the effect of the covariates on the brain activity level varies between slices, depending on the location of the slice\red{;} see the Appendix B for further details.
\vskip 0.1in  \noindent \textbf{7. Conclusion} \vskip 0.1in

\red{We examine} a class of image-on-scalar regression models to efficiently explore the spatial nonstationarity of a regression relationship between imaging responses and scalar predictors, allowing the regression coefficients to change with the pixels. We have proposed an efficient estimation procedure to carry out statistical inference. We have developed a fast and accurate method for estimating the coefficient images, while consistently estimating their standard deviation images. Our method provides coefficient maps and significance maps that highlight and visualize the associations with brain and the potential risk factors, adjusted for other patient-level features, \red{as well as permitting} inference. In addition, it allows an easy implementation of piecewise polynomial representations of various degrees and smoothness over an arbitrary triangulation, and therefore can handle irregular-shaped 2D objects with different visual qualities. This provides enormous flexibility, accommodating various types of nonstationarity that are commonly encountered in imaging data analysis. Our methodology is extendable to 3D images to fully realize its potential usefulness in biomedical imaging. Instead of using bivariate splines over triangulation, the trivariate splines over tetrahedral partitions introduced in \cite{Lai:Schumaker:07} could be well suited, \red{because} they have many properties in common with the bivariate splines over triangulation. However, \red{this} is a nontrivial task, because the computation is much more challenging for high-resolution 3D images than it is for 2D images, and thus warrants further investigation.


\vskip 0.1in  \noindent \textbf{Acknowledgements} \vskip 0.1in

The authors wish to thank the editor, associate editor, and two reviewers for their constructive comments and suggestions. Shan Yu's research was partially supported by the Iowa State University Plant Sciences Institute Scholars Program. Li Wang's research was partially supported by National Science Foundation grants DMS-1542332 and DMS-1916204. Lijian Yang's research was partially supported by National Natural Science Foundation of China award 11771240. The study data were obtained from the Alzheimer's Disease Neuroimaging Initiative (ADNI) database (\url{ADNI.loni.usc.edu}). As such, the investigators within the ADNI contributed to the design and implementation of  the ADNI and/or provided data,  but did not participate in the analysis or the writing of this report. A complete listing of  ADNI investigators can be found at: \url{http://adni.loni.usc.edu/wp-content/uploads/how_to_apply/ADNI_Acknowledgement_List.pdf}. 

\newpage

\fontsize{12}{14pt plus.8pt minus .6pt}\selectfont
\vskip 0.1in  \noindent \textbf{Appendix A} \vskip 0.1in 
\setcounter{equation}{0}
\setcounter{subsection}{0}
\def\thelemma{A.\arabic{lemma}}
\def\thetheorem{A.\arabic{theorem}}
\def\theequation{A.\arabic{equation}}

In the following, we use $c$, $C$, $c_1$, $c_2$, $C_1$, $C_2$, etc. as generic constants, which may be different even in the same line.  For any sequence $a_n$ and $b_n$, we write $a_n \asymp b_n$ if there exist two positive constants $c_1, c_2$ such that $c_1 |a_n| \le |b_n| \le c_2 |a_n|$, for all $n\ge 1$. For a real valued vector $\bs{a}$, denote $\|\bs{a}\|$ its Euclidean norm. For a matrix $\mathbf{A}=(a_{ij})$, denote $\|\mathbf{A}\|_{\infty}=\max_{i,j} |a_{ij}|$. For any positive definite matrix $\mathbf{A}$, let $\lambda_{\min}(\mathbf{A})$ and $\lambda_{\max}(\mathbf{A})$ be the smallest and largest eigenvalues of $\mathbf{A}$. For a vector valued function $\bs{g}=(g_0,\ldots,g_p)^{\top}$, denote $\|\bs{g}\|_{L _2(\Omega)}=\{\sum_{\ell =0}^{p} \|g_{\ell}\|_{L _2(\Omega)}^2\}^{1/2}$ and $\Vert \bs{g}\Vert_{\infty ,\Omega}=\max_{0\leq \ell\leq p} \Vert g_{\ell}\Vert_{\infty ,\Omega}$, where $\Vert g_{\ell}\Vert_{L _2,\Omega}$ and $\Vert g_{\ell}\Vert_{\infty ,\Omega}$ are the $L_2$ norm and supremum norm of $g_{\ell}$ defined at the beginning of Section 2.2. Further denote $\|\bs{g}\|_{\upsilon,\infty,\Omega}=\max_{0\leq \ell\leq p}|g_{\ell}|_{\upsilon,\infty,\Omega}$, where $|g_{\ell}|_{\upsilon,\infty,\Omega}=\max_{i+j=\upsilon}\Vert \nabla_{z_{1}}^{i}\nabla_{z_{2}}^{j}g_{\ell}(\bs{z})\Vert_{\infty ,\Omega}$. For notation simplicity, we drop the subscript $\Omega$ in the rest of the paper.
For $\bs{g}^{(1)}(\bs{z})=(g_0^{(1)}(\bs{z}), \ldots, g_p^{(1)}(\bs{z}))^{\top}$ and $\bs{g}^{(2)}(\bs{z})=(g_0^{(2)}(\bs{z}), \ldots, g_p^{(2)}(\bs{z}))^{\top}$, define the empirical inner product as
\begin{equation}
\langle \bs{g}^{(1)}, \bs{g}^{(2)}\rangle_{n,N} = \frac{1}{nN}\sum_{\ell,\ell^{\prime}=0}^{p}\sum_{i=1}^{n}\sum_{j=1}^N X_{i\ell}X_{i\ell^{\prime}} g_{\ell}^{(1)}(\bs{z}_j) g_{\ell^{\prime}}^{(2)}(\bs{z}_j),
\label{DEF:empirical_product}
\end{equation}
and the theoretical inner product as
\begin{equation}
\langle \bs{g}^{(1)}, \bs{g}^{(2)}\rangle = \sum_{\ell,\ell^{\prime}=0}^{p}  E(X_{\ell}X_{\ell^{\prime}})\int_{\Omega}g_{\ell}^{(1)}(\bs{z})g_{\ell^{\prime}}^{(2)}(\bs{z})d\bs{z},
\label{DEF:theoretical_product}
\end{equation}
and denote the corresponding empirical and theoretical norms $\|\cdot\|_{n,N}$ and $\|\cdot\|$.

Furthermore, let $\| \cdot\|_{\mathcal{E}}$ be the norm introduced by the inner product $\langle \cdot, \cdot\rangle_{\mathcal{E}}$, where, for $\bs{g}^{(1)}(\bs{z})$ and $\bs{g}^{(2)}(\bs{z})$,
\begin{equation*}
\left\langle \bs{g}^{(1)},\bs{g}^{(2)}\right\rangle_{\mathcal{E}}\!=\!\sum_{\ell,\ell^{\prime}=0}^{p}\int_{\Omega}\left\{\sum_{i+j=2}
\binom{2}{i}
(\nabla_{z_{1}}^{i}\nabla_{z_{2}}^{j}g_{\ell}^{(1)})\right\}\!\left\{\sum_{i+j=2}
\binom{2}{i}
(\nabla_{z_{1}}^{i}\nabla_{z_{2}}^{j}g_{\ell^{\prime}}^{(2)})\right\}dz_{1}dz_{2}.
\label{DEF:energy_product}
\end{equation*}

Let $A(\Omega)$ be the area of the domain $\Omega$, and without loss of generality, we assume $A(\Omega)=1$ in the rest of the article. Note that the triangulation for different coefficient function can be different from each other.
For notational convenience in the proof below, we consider a common triangulation for all the explanatory variables: $\mathbf{B}_{0}(\bs{z})=\mathbf{B}_{1}(\bs{z})=\cdots=\mathbf{B}_{p}(\bs{z})=\mathbf{B}(\bs{z})$, and $\beta_{\ell}(\bs{z}_{j})=\mathbf{B}^{\top}(\bs{z}_{j})\bs{\gamma}_{\ell}$.

\vskip .10in \noindent \textbf{A.1. Properties of bivariate splines} \vskip .10in

We cite two important results from \cite{Lai:Schumaker:07}.

\begin{lemma}[Theorem 2.7, \cite{Lai:Schumaker:07}]
\label{LEM:normequity}
Let $\{B_m\}_{m\in \mathcal{M}}$ be the Bernstein polynomial basis for spline space $\mathcal{S}_{d}^{r}(\triangle)$ defined over a $\pi$-quasi-uniform triangulation $\triangle$.  Then there exist positive constants $c$, $C$ depending on the smoothness $r$, $d$, and the shape parameter $\pi$ such that
$c|\triangle|^{2}\sum_{m\in \mathcal{M}}\gamma_{m}^{2}\leq
\left\Vert \sum_{m\in \mathcal{M}} \gamma_{m}B_{m}\right\Vert_{L _2}^{2}\leq C|\triangle|^{2}\sum_{m\in \mathcal{M}}\gamma_{m}^{2}$.
\end{lemma}

\begin{lemma}[Theorems 10.2 and 10.10, \cite{Lai:Schumaker:07}]
\label{LEM:appord}
Suppose that $|\triangle|$ is a $\pi$-quasi-uniform triangulation of a polygonal domian $\Omega$, and $g(\cdot) \in  \mathcal{W}^{d+1,\infty}(\Omega)$.
\begin{itemize}
\item[(i)] For bi-integer $(a_{1},a_{2})$ with $0\leq {a_{1}}+{a_{2}} \leq d $, there exists a spline $g^{\ast}(\cdot)\in \mathcal{S}_{d}^{0}(\triangle)$ such that $\Vert \nabla_{z_{1}}^{a_{1}}\nabla_{z_{2}}^{a_{2}}\left(g-g^{\ast}\right) \Vert_{\infty}\leq C|\triangle|^{d+1-a_{1}-a_{2}}|g|_{d+1,\infty}$, where $C$ is a constant depending on $d$, and the shape parameter $\pi$.
\item[(ii)] For bi-integer $(a_{1},a_{2})$ with $0\leq {a_{1}}+{a_{2}} \leq d $, there exists a spline $g^{\ast\ast}(\cdot)\in \mathcal{S}_{d}^{r}(\triangle)$ ($d\geq 3r+2$) such that $\Vert \nabla_{z_{1}}^{a_{1}}\nabla_{z_{2}}^{a_{2}}\left(g-g^{\ast\ast}\right) \Vert_{\infty}\leq C|\triangle|^{d+1-a_{1}-a_{2}}|g|_{d+1,\infty}$, where $C$ is a constant depending on $d$, $r$, and the shape parameter $\pi$.
\end{itemize}
\end{lemma}
Lemma \ref{LEM:appord} shows that $\mathcal{S}_{d}^{0}(\triangle)$ has full approximation power, and $\mathcal{S}_{d}^{r}(\triangle)$ also has full approximation power if $d\geq 3r+2$. 
For any $g(\cdot)$ in Sobolev space $\mathcal{C}^{(0)}(\Omega)$, there exists a spline $g^{*}(\cdot)\in \mathcal{PC}(\triangle)$ such that $\Vert g-g^{\ast} \Vert_{\infty}\leq C|\triangle| \|g\|_{\infty}$.

\begin{lemma}
\label{LEM:norm}
Let $\bs{g}(\bs{z})=(g_0(\bs{z}), \ldots, g_p(\bs{z}))^{\top}$, where $g_{\ell}(\bs{z})=\sum_{m\in \mathcal{M}}\gamma_{\ell m} B_{m}(\bs{z})$. Then, under Assumptions (A3) and (A5), $\|\bs{g}\| \asymp \sum_{\ell =0}^{p} \|g_{\ell}\|_{L _2}$.
\end{lemma}
\begin{proof}
By  (\ref{DEF:empirical_product}), $ \|\mathbf{g}\|^2=\sum_{\ell,\ell^{\prime}=0}^{p}  E (X_{\ell}X_{\ell^{\prime}})  \int_{\Omega}g_{\ell}(\bs{z})g_{\ell^{\prime}}(\bs{z})d\bs{z}=\int_{\Omega} \mathbf{g}^{\top}(\bs{z})\bs{\Sigma}_{X}\mathbf{g}(\bs{z})d\bs{z}$. 
According to Assumptions (A3) and (A5), $\|\mathbf{g}\|^2 \asymp \int_{\Omega} \mathbf{g}^{\top}(\bs{z})\mathbf{g}(\bs{z})d\bs{z} \asymp \sum_{\ell =0}^{p} \|g_{\ell}\|_{L _2}$.
\end{proof}

\begin{lemma}
\label{LEM:integration}
Under Assumptions (A4) and (A5), for any Bernstein basis polynomials $B_{m}(\bs{z}),~m \in \mathcal{M}$, of degree $d\geq0$, we have
\begin{align}
&\max_{m\in\mathcal{M}}\left|\frac{1}{N}\sum_{j=1}^{N}B_{m}^{k}(\bs{z}_{j})
-\int_{\Omega}B_{m}^{k}(\bs{z})d\bs{z}\right|
= O\left(|\triangle|N^{-1/2}\right),~ 1 \leq k < \infty,
\label{EQ:B_integration}\\
&\max_{m, m'\in\mathcal{M}}\left|\frac{1}{N}\sum_{j=1}^{N}B_{m}(\bs{z}_{j})B_{m'}(\bs{z}_{j})
-\int_{\Omega}B_{m}(\bs{z})B_{m'}(\bs{z})d\bs{z}\right|
=  O\left(|\triangle|N^{-1/2}\right),
\label{EQ:B_integration2}\\
&\max_{m, m^{\prime} \in \mathcal{M}}\left|\frac{1}{N^2}\sum_{j,j^{\prime}=1}^{N}G_{\eta}(\bs{z}_{j},\bs{z}_{j^{\prime}})B_{m}(\bs{z}_{j})B_{m^{\prime}}(\bs{z}_{j^{\prime}})
\!-\!\int_{\Omega^{2} }G_{\eta}(\bs{z},\bs{z}^{\prime})B_m(\bs{z})B_{m^{\prime}}(\bs{z}^{\prime})
d\bs{z}d\bs{z}^{\prime}\right|\nonumber \\
&\quad=O\left(N^{-1/2} |\triangle|^3\right),
\label{EQ:G_integration}\\
&\max_{m\in \mathcal{M}}\left|\|\sigma B_m\|_{N,L_2}^2 \!-\! \|\sigma B_m\|_{L_2}^2\right|
=\max_{m\in \mathcal{M}}\left| \frac{1}{N}
\sum_{j=1}^{N}B_{m}^{2}(\bs{z}_{j})\sigma^{2}(\bs{z}_{j})
\!-\!\int_{\Omega}\sigma^{2}(\bs{z})B_{m}^{2}(\bs{z})d\bs{z}\right| \nonumber\\
&\quad=O\left(N^{-1/2} |\triangle|\right).
\label{EQ:sigma_integration}
\end{align}
\end{lemma}

\begin{proof}Note that there are $d^{\ast}=(d+1)(d+2)/2$ Bernstein basis polynomials on each triangle and $\int_{\Omega}B_m^k (\bs{z})d\bs{z}=\int_{T_{\lceil m /d^{\ast}\rceil}}B_m^k (\bs{z})d\bs{z}$, for any $k\geq 1$.

For piecewise constant basis functions, we have $B_{m}(\bs{z})=I(\bs{z}\in T_{m})$, then
\[
\left|\frac{1}{N}\sum_{j=1}^{N}B_{m}^{k}(\bs{z}_{j})-\int_{\Omega}B_{m}^{k}(\bs{z})d\bs{z}\right|
=\left|\frac{1}{N}\sum_{j=1}^{N}I(\bs{z}_{j}\in T_{m})-A(T_m)\right|.
\] 
According to Assumption (A5), 
\[
\max_{m\in\mathcal{M}}\left|\frac{1}{N}\sum_{j=1}^{N}B_{m}^{k}(\bs{z}_{j})-\int_{\Omega}B_{m}^{k}(\bs{z})d\bs{z}\right|\leq CN^{-1/2}|\triangle|.
\]
For any $j=1,\ldots,N$, let $\mathcal{V}_j$ be the $j$th pixel, and it is clear that
\[
\left|\frac{1}{N}\sum_{j=1}^{N}B_{m}^{k}(\bs{z}_{j})-\int_{\Omega}B_{m}^{k}(\bs{z})d\bs{z}\right|\leq
\left|\sum_{j=1}^{N}\int_{\mathcal{V}_j}\{B_{m}^{k}(\bs{z}_{j})-B_{m}^{k}(\bs{z})\}d\bs{z}\right|
+\int_{\Omega\setminus\cup\mathcal{V}_j}B_{m}^{k}(\bs{z})d\bs{z}.
\]

If $d\geq 1$, by the properties of bivariate spline basis functions in \cite{Lai:Schumaker:07}, 
$\int_{\Omega\setminus\cup\mathcal{V}_j}B_{m}^{k}(\bs{z})d\bs{z}=O(N^{-1/2}|\triangle|)$, and
\begin{align*}
\left|\sum_{j=1}^{N}\int_{\mathcal{V}_j}\{B_{m}^{k}(\bs{z}_{j})-B_{m}^{k}(\bs{z})\}d\bs{z}\right|
&\leq \sum_{\{j:\bs{z}_{j}\in {T_{\lceil m /d^{\ast}\rceil}}\}} \int_{\mathcal{V}_j}|B_{m}^{k}(\bs{z}_{j})-B_{m}^{k}(\bs{z})|d\bs{z}\\
&\leq C(N|\triangle|^2)\times N^{-1}\times (N^{-1/2}|\triangle|^{-1})\leq CN^{-1/2}|\triangle|.
\end{align*}
Thus, (\ref{EQ:B_integration}) holds. The proof of (\ref{EQ:B_integration2}) is similar to the proof (\ref{EQ:B_integration}), thus omitted.

Next, for any $m, m^{\prime} \in \mathcal{M}$,
\begin{align*}
&\frac{1}{N^2}\sum_{j=1}^{N}\sum_{j^{\prime}=1}^{N}G_{\eta}(\bs{z}_{j},\bs{z}_{j^{\prime}})B_{m}(\bs{z}_{j})B_{m^{\prime}}(\bs{z}_{j^{\prime}})
-\int_{\Omega^{2}}G_{\eta}(\bs{z},\bs{z}^{\prime})B_{m}(\bs{z})B_{m^{\prime}}(\bs{z}^{\prime})d\bs{z}d\bs{z}^{\prime}\\
&= \sum_{j=1}^{N}\sum_{j^{\prime}=1}^{N} \int_{\mathcal{V}_j\times \mathcal{V}_{j^{\prime}}}\left\{G_{\eta}(\bs{z}_{j},\bs{z}_{j^{\prime}})B_{m}(\bs{z}_{j})B_{m^{\prime}}(\bs{z}_{j^{\prime}})
-G_{\eta}(\bs{z},\bs{z}^{\prime})B_{m}(\bs{z})B_{m^{\prime}}(\bs{z}^{\prime})\right\}d\bs{z}d\bs{z}^{\prime}\\
&~+\int_{\Omega^2\setminus \cup_{j,j^{\prime}} \mathcal{V}_j\times \mathcal{V}_{j^{\prime}}}\left\{G_{\eta}(\bs{z}_{j},\bs{z}_{j^{\prime}})B_{m}(\bs{z}_{j})B_{m^{\prime}}(\bs{z}_{j^{\prime}})
-G_{\eta}(\bs{z},\bs{z}^{\prime})B_{m}(\bs{z})B_{m^{\prime}}(\bs{z}^{\prime})\right\}d\bs{z}d\bs{z}^{\prime}.
\end{align*}
As $N\rightarrow \infty$,
\[
\int_{\Omega^2\setminus \cup_{j,j^{\prime}} \mathcal{V}_j\times \mathcal{V}_{j^{\prime}}}
\!\!\!\!\!\!\!\!\left\{G_{\eta}(\bs{z}_{j},\bs{z}_{j^{\prime}})B_{m}(\bs{z}_{j})B_{m^{\prime}}(\bs{z}_{j^{\prime}})
-G_{\eta}(\bs{z},\bs{z}^{\prime})B_{m}(\bs{z})B_{m^{\prime}}(\bs{z}^{\prime})\right\}d\bs{z}d\bs{z}^{\prime}=O\left(\frac{|\triangle|^{3}}{\sqrt{N}}\right).
\]
Notice that
\begin{align*}
&\left|\sum_{j=1}^{N}\sum_{j^{\prime}=1}^{N} \int_{\mathcal{V}_j\times \mathcal{V}_{j^{\prime}}}\left\{G_{\eta}(\bs{z}_{j},\bs{z}_{j^{\prime}})B_{m}(\bs{z}_{j})B_{m^{\prime}}(\bs{z}_{j^{\prime}})
-G_{\eta}(\bs{z},\bs{z}^{\prime})B_{m}(\bs{z})B_{m^{\prime}}(\bs{z}^{\prime})\right\}d\bs{z}d\bs{z}^{\prime}\right|\\
&\quad \leq \sum_{\{(j,j^{\prime}):\bs{z}_{j}\in {T_{\lceil m /d^{\ast}\rceil}}, \bs{z}_{j^{\prime}}\in {T_{\lceil m /d^{\ast}\rceil}}\}} \int_{\mathcal{V}_j\times \mathcal{V}_{j^{\prime}}}\omega_{jj^{\prime}}(G_{\eta}K_{m}, 2N^{-1/2})d\bs{z}d\bs{z}^{\prime},
\end{align*}
where $K_{m}(\bs{z},\bs{z}^{\prime})=B_{m}(\bs{z})B_{m}(\bs{z}^{\prime})$ and
\[
\omega_{jj^{\prime}}(g,\varrho)=\sup_{\substack{(\bs{z}_{1},\bs{z}_{1}^{\prime}), (\bs{z}_{2},\bs{z}_{2}^{\prime})\in \mathcal{V}_j\times \mathcal{V}_{j^{\prime}},\\ \|\bs{z}_{1}-\bs{z}_{2}\|^{2}+\|\bs{z}_{1}^{\prime}-\bs{z}_{2}^{\prime}\|^{2}=\varrho^{2}}}
|g(\bs{z}_1,\bs{z}_{1}^{\prime})
-g(\bs{z}_2,\bs{z}_{2}^{\prime})|
\]
is the modulus of continuity of $g$ on $\mathcal{V}_j\times \mathcal{V}_{j^{\prime}}$.
Therefore, by Assumption (A4), we have
\begin{align*}
&\left|\sum_{j=1}^{N}\sum_{j^{\prime}=1}^{N} \int_{\mathcal{V}_j\times \mathcal{V}_{j^{\prime}}}\left\{G_{\eta}(\bs{z}_{j},\bs{z}_{j^{\prime}})B_{m}(\bs{z}_{j})B_{m^{\prime}}(\bs{z}_{j^{\prime}})
-G_{\eta}(\bs{z},\bs{z}^{\prime})B_{m}(\bs{z})B_{m^{\prime}}(\bs{z}^{\prime})\right\}d\bs{z}d\bs{z}^{\prime}\right|\\
&\qquad \leq (N|\triangle|^2)^2\times N^{-2}\times (N^{-1/2}|\triangle|^{-1}) =O(N^{-1/2}|\triangle|^{3}).
\end{align*}
Thus, (\ref{EQ:G_integration}) follows.

Finally, note that
\begin{align*}
&\left|\frac{1}{N}
\sum_{j=1}^{N}B_{m}^{2}(\bs{z}_{j})\sigma^{2}(\bs{z}_{j})
-\int_{\Omega}\sigma^{2}(\bs{z})d\bs{z}\right|
\leq \left|\sum_{j=1}^{N}\int_{\mathcal{V}_j}\{B_{m}^{2}(\bs{z}_{j})\sigma^{2}(\bs{z}_j)
-B_{m}^{2}(\bs{z})\sigma^{2}(\bs{z})\}d\bs{z}\right|\\
&\qquad +\int_{\Omega\setminus\cup\mathcal{V}_j}|B_{m}^{2}(\bs{z}_{j})\sigma^{2}(\bs{z}_j)
-B_{m}^{2}(\bs{z})\sigma^{2}(\bs{z})|d\bs{z}.
\end{align*}
It is easy to see that $\int_{\Omega\setminus\cup\mathcal{V}_j}|B_{m}^{2}(\bs{z}_{j})\sigma^{2}(\bs{z}_j)
-B_{m}^{2}(\bs{z})\sigma^{2}(\bs{z})|d\bs{z}=O(N^{-1/2}|\triangle|)$. Denote $\omega_j(g,\varrho)=\sup_{\bs{z},\bs{z}^{\prime}\in \mathcal{V}_j, \|\bs{z}-\bs{z}^{\prime}\|=\varrho}|g(\bs{z})-g(\bs{z}^{\prime})|$ is the modulus of continuity of $g$ on the $j$th pixel $\mathcal{V}_j$, then by Assumption (A4), we have
\begin{align*}
&\left|\sum_{j=1}^{N}\int_{\mathcal{V}_j}\{B_{m}^{2}(\bs{z}_{j})\sigma^{2}(\bs{z}_j)
-B_{m}^{2}(\bs{z})\sigma^{2}(\bs{z})\}d\bs{z}\right| \leq \sum_{\{j:\bs{z}_{j}\in {T_{\lceil m /d^{\ast}\rceil}}\}} \int_{\mathcal{V}_j}\omega_j(B_{m}^{2}\sigma^2,2N^{-1/2})d\bs{z}\\
&\qquad \leq C(N|\triangle|^2)\times N^{-1}\times (N^{-1/2}|\triangle|^{-1})\leq CN^{-1/2}|\triangle|.
\end{align*}
We obtain (\ref{EQ:sigma_integration}).
\end{proof}

\begin{lemma}
\label{LEM:inner product}
For any $m\in\mathcal{M}$, $0 \leq \ell,\ell^{\prime}\leq p$, let $\Phi_{m,\ell,\ell^{\prime}}= E(X_{\ell}X_{\ell^{\prime}})\int_{\Omega}B_{m}^{2}(\bs{z})d\bs{z} $. Suppose Assumptions (A3) and (A5) hold, and $N^{1/2}|\triangle|\rightarrow \infty$ as $N\rightarrow \infty$, then with probability $1$, one has 
\[
\max_{m\in\mathcal{M}}\max_{0 \leq \ell,\ell^{\prime}\leq p}\left| \frac{1}{nN}\sum_{i=1}^{n}\sum_{j=1}^{N}B_{m}^{2}(\bs{z}_{j})X_{i\ell}X_{i\ell^{\prime}}
-\Phi_{m,\ell,\ell^{\prime}}\right|=O\left\{n^{-1/2}|\triangle|^{2}(\log n)^{1/2}+N^{-1/2}|\triangle|\right\}.
\]
\end{lemma}
\begin{proof}
Let $\varsigma_{i,m}\equiv\varsigma_{i,m,\ell,\ell^{\prime}}=\frac{1}{N}\sum_{j=1}^{N}B_{m}^{2}(\bs{z}_{j})
X_{i\ell}X_{i\ell^{\prime}}$. If $N^{1/2}|\triangle|\rightarrow \infty$ as $N\rightarrow \infty$, then by (\ref{EQ:B_integration}), we can show that 
$ E(\varsigma_{i,m})=\frac{1}{N}\sum_{j=1}^{N}B_{m}^{2}(\bs{z}_{j}) E(X_{\ell}X_{\ell^{\prime}}) \asymp  |\triangle|^2$, and
$ E\left(\varsigma_{i,m}\right)^{2}=\left\{\frac{1}{N}\sum_{j=1}^{N}B_{m}^{2}(\bs{z}_{j})\right\}^{2}
 E\left(X_{\ell}X_{\ell^{\prime}}\right)^{2} \asymp |\triangle|^{4}$.

Next define a sequence $D_{n}=n^{\alpha}$ with $\alpha\in(1/3,1/2)$. We make use of the following truncated and tail decomposition
$X_{i\ell\ell^{\prime}}=X_{i\ell}X_{i\ell^{\prime}}=X_{i\ell\ell^{\prime},1}^{D_{n}}+X_{i\ell\ell^{\prime},2}^{D_{n}}$,
where $X_{i\ell\ell^{\prime},1}^{D_{n}}=X_{i\ell}X_{i\ell^{\prime}}I\left\{\left\vert E(
X_{i\ell}X_{i\ell^{\prime}}\right\vert >D_{n}\right\} $, $X_{i\ell\ell^{\prime
},2}^{D_{n}}=X_{i\ell}X_{i\ell^{\prime}}I\left\{\left\vert X_{i\ell}X_{i\ell^{\prime
}}\right\vert \leq D_{n}\right\} $. Correspondingly the truncated and tail
parts of $\varsigma_{i,m}$ are $\varsigma
_{i,m,v}\equiv \varsigma
_{i,m,v,\ell,\ell^{\prime}}=\frac{1}{N}\sum_{j=1}^{N}B_{m}^{2}(\bs{z}_{j})X_{i\ell\ell^{\prime},v}^{D_{n}}$, $v=1,2$. According
to Assumption (A3), for\ any $\ell,\ell^{\prime}=0,\ldots,p$,
\begin{equation*}
\sum_{n=1}^{\infty}P\left\{\left\vert X_{n\ell}X_{n\ell^{\prime}}\right\vert
>D_{n}\right\} \leq \sum_{n=1}^{\infty}\frac{ E
\left\vert X_{n\ell}X_{n\ell^{\prime}}\right\vert^{3}}{D_{n}^{3}}%
\leq C_{b}\sum_{n=1}^{\infty}D_{n}^{-3}<\infty .
\end{equation*}%
By Borel-Cantelli Lemma, $\frac{1}{N}\sum_{j=1}^{N}B_{m}^{2}(\bs{z}_{j})X_{i\ell\ell^{\prime},1}^{D_{n}}=0$, almost surely.
So for any $k\geq 1$, $\sup_{m,\ell,\ell^{\prime}}\left\vert n^{-1}\sum_{i=1}^{n}\varsigma
_{i,m,1}\right\vert =O_{a.s.}(n^{-k})$. Since $N^{1/2}|\triangle|\rightarrow \infty$ as $N\rightarrow \infty$, 
\begin{align*}
| E(\varsigma_{i,m,1})| &=| E(X_{i\ell\ell^{\prime},1}^{D_{n}})| \left\{\frac{1}{N}\sum_{j=1}^{N}B_{m}^{2}(\bs{z}_{j})\right\} \\
&\leq D_{n}^{-2} E\left\vert X_{i\ell}X_{i\ell^{\prime}}\right\vert^{3}
\left\{\int_{\Omega}B_{m}^{2}(\bs{z}) d\bs{z}+O(N^{-1/2}|\triangle|)\right\} \leq C D_{n}^{-2}|\triangle|^2.
\end{align*}
Next, we consider the truncated part $\varsigma_{i,m,2}$. Define $\varsigma_{i,m,2}^{\ast}=\varsigma_{i,m,2}- E(\varsigma_{i,m,2}) $, then $ E\varsigma_{i,m,2}^{\ast}=0$, and
\[
 E(\varsigma_{i,m,2}^{\ast})^{2} = E(\varsigma_{i,m,2})^{2}-\left(
 E\varsigma_{i,m,2}\right)^{2}=\left\{\frac{1}{N}\sum_{j=1}^{N}B_{m}^{2}(\bs{z}_{j})\right\}^2
\left\{ E(X_{i\ell\ell^{\prime},2}^{D_{n}})^{2}-(EX_{i\ell\ell^{\prime},2}^{D_{n}})^2\right\}.
\]
Note that $ E(X_{i\ell\ell^{\prime},1}^{D_{n}})^2
\leq D_{n}^{-1} E\left\vert X_{i\ell}X_{i\ell^{\prime}}\right\vert^{3}\leq cD_{n}^{-1}$,
thus, $ E(X_{i\ell\ell^{\prime},2}^{D_{n}})^{2}=  E(X_{i\ell\ell^{\prime}})
^{2}- E(X_{i\ell\ell^{\prime},1}^{D_{n}})^{2} =  E(X_{i\ell\ell^{\prime}})^{2}-o(1)$.
Therefore, there exists $c_{\varsigma}$ such that for large $n$, we have $%
 E(\varsigma_{i,m,2}^{\ast})^{2}\geq c_{\varsigma}%
 E(X_{i\ell\ell^{\prime}})^{2}\times \left\{\frac{1}{N}\sum_{j=1}^{N}B_{m}^{2}(\bs{z}_{j})\right\}^2$.
Next for any $k>2$,
\begin{align*}
 E\left\vert \varsigma_{i,m,2}^{\ast}\right\vert^{k}&=
 E\left\vert \varsigma_{i,m,2}- E\left(
\varsigma_{i,m,2}\right) \right\vert^{k}\leq 2^{k-1}\left(
 E\left\vert \varsigma_{i,m,2}\right\vert^{k}+\left\vert
 E(\varsigma_{i,m,2}) \right\vert^{k}\right) \\
&=2^{k-1}\left\{ E\left\vert X_{i\ell\ell^{\prime},2}^{D_{n}}\right\vert^{k}+O(1)\right\}\left\{\frac{1}{N}
\sum_{j=1}^{N}B_{m}^{2}(\bs{z}_{j})\right\}^{k},
\end{align*}
then there exists $C_{\varsigma}>0$ such that for any $k>2$ and large $n$,
\begin{align*}
 E\left\vert \varsigma_{i,m,2}^{\ast}\right\vert^{k}&\leq
2^{k-1}\left\{D_{n}^{k-2} E\left(X_{i\ell\ell^{\prime
}}\right)^{2}+O(1)\right\}\left\{\frac{1}{N}\sum_{j=1}^{N}B_{m}^{2}(\bs{z}_{j})\right\}^{k}\\
&\leq 2^{k}D_{n}^{k-2}  E\left(\varsigma_{i,m,2}^{\ast
}\right)^{2} \left\{\frac{1}{N}\sum_{j=1}^{N}B_{m}^{2}(\bs{z}_{j})\right\}^{k-2}
\leq \left(C_{\varsigma}D_{n}|\triangle|^{2}\right)^{k-2}k!%
 E\left(\varsigma_{i,m,2}^{\ast}\right)^{2},
\end{align*}
which implies that $\left\{\varsigma_{i,m,2}^{\ast}\right\}_{i=1}^{n}$
satisfies Cram\'{e}r's condition with constant $C_{\varsigma}D_{n}|\triangle|^{2}$. Applying Bernstein's inequality to $\sum_{i=1}^{n}\varsigma_{i,m,2}^{\ast}$, for $k>2$ and any large enough $\delta>0$,
\begin{align*}
P&\left\{\left\vert \frac{1}{n}\sum_{i=1}^{n}\varsigma_{i,m,2}^{\ast}\right\vert \geq
\delta n^{-1/2}|\triangle|^{2}(\log n)^{1/2}\right\}
\leq 2\exp \left\{-\frac{\delta^{2}\log (n)}{4+2C_{\varsigma}D_{n}\delta (\log n)^{1/2}n^{-1/2}}\right\}.
\end{align*}
Assume that $|\triangle|^{-2}\asymp n^{\tau}$ for some $0<\tau<\infty$, we have
\begin{equation*}
\sum_{n=1}^{\infty}P\left\{\max_{\substack{m\in\mathcal{M}\\ 0 \leq \ell,\ell^{\prime}\leq p}} 
\left\vert \frac{1}{n}\sum_{i=1}^{n}\varsigma_{i,m,2}^{\ast
}\right\vert \geq \delta n^{-1/2}|\triangle|^{2}(\log n)^{1/2}\right\}
\leq 2\sum_{n=1}^{\infty}\sum_{m\in \mathcal{M}}\sum_{0\leq
\ell,\ell^{\prime}\leq p} n^{-2-\tau} <\infty.
\end{equation*}%
Thus, $\sup_{m,\ell,\ell^{\prime}}\left\vert n^{-1}\sum_{i=1}^{n}\varsigma
_{i,m,2}^{\ast}\right\vert =O_{a.s.}\left\{n^{-1/2}|\triangle|^{2}(\log n)^{1/2}\right\} $ as $n\rightarrow \infty $, by Borel-Cantelli Lemma. Furthermore,
\begin{align*}
&\max_{m,\ell,\ell^{\prime}}\left\vert n^{-1}\sum_{i=1}^{n}\varsigma_{i,m}- E\varsigma_{i,m}\right\vert
\leq \max_{m,\ell,\ell^{\prime}}\left\vert n^{-1}\sum_{i=1}^{n}\varsigma_{i,m,1}\right\vert
+\max_{m,\ell,\ell^{\prime}}\left\vert
n^{-1}\sum_{i=1}^{n}\varsigma_{i,m,2}^{\ast}\right\vert
+\max_{m,\ell,\ell^{\prime}}\left\vert  E\varsigma_{i,m,1}\right\vert \\
&~=O_{a.s.}(n^{-k}) +O_{a.s.}\left\{n^{-1/2}|\triangle|^{2}(\log n)^{1/2}\right\} +O
\left(D_{n}^{-2}|\triangle|^{2}\right) =O_{a.s.}\left\{n^{-1/2}|\triangle|^{2}(\log n)^{1/2}\right\} .
\end{align*}%
Finally, we notice that
\begin{align*}
\max_{\substack{m\in\mathcal{M}\\ 0 \leq \ell,\ell^{\prime}\leq p}} & \left| \frac{1}{nN}\sum_{i=1}^{n}\sum_{j=1}^{N}B_{m}^{2}(\bs{z}_{j})X_{i\ell}X_{i\ell^{\prime}}-\Phi_{m,\ell,\ell^{\prime}}\right| \\
=&\max_{\substack{m\in\mathcal{M}\\ 0 \leq \ell,\ell^{\prime}\leq p}} \left\vert  n^{-1}\sum_{i=1}^{n}\varsigma_{i,m}- E\varsigma_{i,m}\right\vert
+ \left\vert EX_{i\ell}X_{i\ell^{\prime}}\right\vert \max_{m\in\mathcal{M}} \left\vert  \frac{1}{N}\sum_{j=1}^{N}B_{m}^{2}(\bs{z}_{j})-\int_{\Omega}B_{m}^{2}(\bs{z})d\bs{z}\right\vert  \\
=&O_{a.s.}\left\{n^{-1/2}|\triangle|^{2}(\log n)^{1/2}\right\}+O(N^{-1/2}|\triangle|).
\end{align*}
We obtain the desired result.
\end{proof}

The following lemma provide the uniform convergence rate at which the empirical inner product in (\ref{DEF:empirical_product}) approximates the theoretical inner product in (\ref{DEF:theoretical_product}).
\begin{lemma}
\label{LEM:Rnorder-vec}
Let $g_{\ell}^{(1)}(\bs{z})=\sum_{m\in \mathcal{M}}c_{\ell m}^{(1)}B_{m}(\bs{z})$, $g_{\ell}^{(2)}(\bs{z})=\sum_{m\in \mathcal{M}}c_{\ell m}^{(2)}B_{m}(\bs{z})$ be any spline functions in $\mathcal{S}_{d}^{r}(\triangle)$. Denote  $\bs{g}(\bs{z})=(g_0(\bs{z}), \ldots, g_p(\bs{z}))^{\top}$ with $g_{\ell}\in \mathcal{S}_{d}^{r}(\triangle)$, $\ell=0,\ldots, p$. Suppose Assumptions (A3) and (A5) hold, and  $N^{1/2}|\triangle|\rightarrow \infty$ as $N\rightarrow \infty$, then
\[
R_{n,N}=\sup\limits_{\bs{g}^{(1)},\bs{g}^{(2)}\in \mathcal{S}_{d}^{r}(\triangle)}\left|
\frac{\left\langle \bs{g}^{(1)},\bs{g}^{(2)}\right\rangle_{n,N}-\left\langle \bs{g}^{(1)},\bs{g}^{(2)}\right\rangle}{\left\|
\bs{g}^{(1)}\right\|\left\|\bs{g}^{(2)}\right\|}\right| =O_{P}\{n^{-1/2}(\log n)^{1/2}+N^{-1/2}|\triangle|^{-1}\}.
\]
\end{lemma}
\begin{proof}
It is easy to see
\begin{eqnarray*}
\langle \bs{g}^{(1)},\bs{g}^{(2)}\rangle_{n,N}&=&\frac{1}{nN}\sum_{i=1}^{n}\sum_{j=1}^{N}\left\{\sum_{\ell =0}^{p}\sum_{m\in \mathcal{M}}c_{\ell m}^{(1)}X_{i\ell}B_{m}(\bs{z}_{j})\right\} \left\{\sum_{\ell ^\prime=0}^{p}\sum_{m^{\prime}\in \mathcal{M}}c_{\ell' m^\prime}^{(2)}X_{i\ell^\prime}B_{m^\prime}(\bs{z}_{j})\right\}\\
&=&\sum_{\ell,m}\sum_{\ell ^\prime,m^\prime}c_{\ell m}^{(1)}c_{\ell ^\prime m^\prime}^{(2)}\frac{1}{nN}\sum_{i=1}^{n}\sum_{j=1}^{N} X_{i\ell}X_{i\ell^\prime}B_{m}(\bs{z}_{j})B_{m^\prime}(\bs{z}_{j}).
\end{eqnarray*}
Note that $\|\bs{g}^{(r)}\|^{2}=\sum_{\ell,m}\sum_{\ell ^\prime,m^\prime}c_{\ell m}^{(r)}c_{\ell ^\prime m^\prime}^{(r)} E(X_{\ell}X_{\ell^{\prime}})\int_{\Omega}B_{m}(\bs{z})B_{m^\prime}(\bs{z})d\bs{z}$, $r=1,2$. It follows from Assumptions (A1), (A2), Lemmas \ref{LEM:normequity} and \ref{LEM:norm} that,
\begin{align*}
	c_{v}|\triangle|^{2}\sum_{\ell,m}\{c_{\ell m}^{(v)}\}^2 &\leq \Vert \bs{g}^{(v)}\Vert^{2} \leq C_{v}|\triangle|^{2}\sum_{\ell,m}\{c_{\ell m}^{(v)}\}^2,\\
	C_{1}|\triangle|^{2}\left[\sum_{\ell,m}\{c_{\ell m}^{(1)}\}^2\sum_{\ell ^\prime,m^\prime}\{c_{\ell ^\prime m^\prime}^{(2)}\}^2\right]^{1/2} & \leq \| \bs{g}^{(1)}\| \| \bs{g}^{(2)}\|
	\leq C_{2}|\triangle|^{2}\left[\sum_{\ell,m}\{c_{\ell m}^{(1)}\}^2\sum_{\ell ^\prime,m^\prime}\{c_{\ell ^\prime m^\prime}^{(2)}\}^2\right]^{1/2}.
\end{align*}

With the above preparation, we have
\begin{align}
	R_{n,N}\leq & \frac{\sum_{\ell, \ell', |m-m'| \leq (d+2)(d+1)/2} |c_{\ell m}^{(1)}c_{\ell ^\prime m^\prime}^{(2)}|}{C_{1}|\triangle|^{2}\left[\sum_{\ell,m}\{c_{\ell m}^{(1)}\}^2\sum_{\ell ^\prime,m^\prime}\{c_{\ell ^\prime m^\prime}^{(2)}\}^2\right]^{1/2}}\label{EQ:R_{n,N}} 
	\\ \nonumber
	&\times\max_{\substack{m, m'\in\mathcal{M}\\
	0 \leq \ell,\ell^{\prime}\leq p}} \left| \frac{1}{nN}\sum_{i=1}^{n}\sum_{j=1}^{N}B_{m}(\bs{z}_{j})B_{m'}(\bs{z}_{j})X_{i\ell}X_{i\ell^{\prime}}- E(X_{\ell}X_{\ell^{\prime}})\int_{\Omega}B_{m}(\bs{z})B_{m'}(\bs{z})d\bs{z}\right|\\
	\leq &\frac{C}{|\triangle|^{2}} \max_{\substack{m, m'\in\mathcal{M}\\
	0 \leq \ell,\ell^{\prime}\leq p}} \left| \frac{1}{nN}\sum_{i=1}^{n}\sum_{j=1}^{N}B_{m}(\bs{z}_{j})B_{m'}(\bs{z}_{j})X_{i\ell}X_{i\ell^{\prime}}- E(X_{\ell}X_{\ell^{\prime}})\int_{\Omega}B_{m}(\bs{z})B_{m'}(\bs{z})d\bs{z}\right|.
	\nonumber
	\end{align}
The desired result follows from (\ref{EQ:R_{n,N}}) and Lemma \ref{LEM:inner product}.	
\end{proof}

As a direct result of Lemma \ref{LEM:Rnorder-vec}, we can see that
\begin{eqnarray}
\sup_{\bs{g}\in \mathcal{S}_{d}^{r}(\triangle)}\left| \left. \left\| \bs{g}\right\|_{n,N}^{2}\right/ \Vert \bs{g}\Vert^{2}-1\right| =O_{P}\{n^{-1/2}(\log n)^{1/2}+N^{-1/2}|\triangle|^{-1}\}.  \label{EQ:normratio}
\end{eqnarray}

\vskip .10in \noindent \textbf{A.2. Uniform convergence of the unpenalized spline estimators} \vskip .10in

In this section, we consider the unpenalized spline smoothing approach.
The unpenalized bivariate spline estimator of $\bs{\beta}^{o}=(\beta_{0}^{o},\ldots,\beta_{p}^{o})^{\top}$ is defined as
\begin{equation}
\widetilde{\bs{\beta}}=(\widetilde{\beta}_{0},\ldots,\widetilde{\beta}_{p})^{\top}
=\argmin_{\bs{\beta}\in \mathcal{G}^{(p+1)}} \sum_{i=1}^{n}\sum_{j=1}^{N}\left\{Y_{i}(\bs{z}_{j})-\sum_{\ell =0}^{p}X_{i\ell}\beta_{\ell}(\bs{z}_{j})
\right\}^{2}.
\label{DEF:beta_tilde}
\end{equation}
Denote
\begin{align*}
\widetilde{\bs{\theta}}_{\mu}&=(\widetilde{\bs{\theta}}_{\mu,0}^{\top},\ldots,\widetilde{\bs{\theta}}_{\mu,p}^{\top})^{\top}
=\bs{\Gamma}_{n,0}^{-1}
\frac{1}{nN}\sum_{i=1}^{n}\sum_{j=1}^{N} \left\{\widetilde{\mathbf{X}}_{i}\otimes \widetilde{\mathbf{B}}(\bs{z}_{j})\right\}\widetilde{\mathbf{X}}_{i}^{\top}\bs{\beta}^{o}(\bs{z}_j)\notag\\
\widetilde{\bs{\theta}}_{\eta}&=(\widetilde{\bs{\theta}}_{\eta,0}^{\top},\ldots,\widetilde{\bs{\theta}}_{\eta,p}^{\top})^{\top}
=\bs{\Gamma}_{n,0}^{-1}
\frac{1}{nN}\sum_{i=1}^{n}\sum_{j=1}^{N} \left\{\widetilde{\mathbf{X}}_{i}\otimes \widetilde{\mathbf{B}}(\bs{z}_{j})\right\}\eta_{i}(\bs{z}_j),\notag\\
\widetilde{\bs{\theta}}_{\varepsilon}&=(\widetilde{\bs{\theta}}_{\varepsilon,0}^{\top},\ldots,
\widetilde{\bs{\theta}}_{\varepsilon,p}^{\top})^{\top}
=\bs{\Gamma}_{n,0}^{-1}
\frac{1}{nN}\sum_{i=1}^{n}\sum_{j=1}^{N} \left\{\widetilde{\mathbf{X}}_{i}\otimes \widetilde{\mathbf{B}}(\bs{z}_{j})\right\}\sigma(\bs{z}_j)\varepsilon_{ij},
\end{align*}
where
\begin{equation}
\bs{\Gamma}_{n,0}=\frac{1}{nN}\sum_{i=1}^{n}\sum_{j=1}^{N}(\widetilde{\mathbf{X}}_{i}\widetilde{\mathbf{X}}_{i}^{\top})\otimes
\{\widetilde{\mathbf{B}}(\bs{z}_{j})\widetilde{\mathbf{B}}^{\top}(\bs{z}_{j})\}.
\label{DEF:Gamma_0}
\end{equation}

\begin{lemma}
Under Assumptions (A3) and (A5), if $N^{1/2}|\triangle|\rightarrow \infty$ as $N\rightarrow \infty$, then there exist constants $0 < c_{\Gamma} < C_{\Gamma} < \infty$, such that with probability approaching 1, as $N\rightarrow \infty$, $n\rightarrow \infty$,
$c_{\Gamma}|\triangle|^{2} \leq \lambda_{\min}(\bs{\Gamma}_{n,0}) \leq
\lambda_{\max}(\bs{\Gamma}_{n,0}) \leq  C_{\Gamma}|\triangle|^{2}$,
where $\bs{\Gamma}_{n,0}$ is in (\ref{DEF:Gamma_0}).
\label{LEM:Gamma_0}
\end{lemma}
\begin{proof}
Note that for any vector $\bs{\theta}=(\bs{\theta}_0^{\top},\cdots,\bs{\theta}_p^{\top})^{\top}$ with $\bs{\gamma}_{\ell}=(\gamma_{\ell m},m\in\mathcal{M})^{\top}$,
\begin{align}
\bs{\theta}^{\top}\bs{\Gamma}_{n,0}\bs{\theta} 
&=\frac{1}{nN}\bs{\gamma}^{\top}\sum_{i=1}^{n}\sum_{j=1}^{N}(\widetilde{\mathbf{X}}_{i}
\widetilde{\mathbf{X}}_{i}^{\top})\otimes \{\mathbf{B}(\bs{z}_{j})\mathbf{B}^{\top}(\bs{z}_{j})\}\bs{\gamma}
=\|\bs{g}_{\bs{\gamma}}\|_{n,N}^2, \label{EQ:g_gamma}
\end{align}
where $\bs{\gamma}=\mathbf{Q}_2\bs{\theta} $, and $\bs{g}_{\bs{\gamma}}=(g_{\bs{\gamma}_0}, \ldots, g_{\bs{\gamma}_p})^{\top}$ with $g_{\bs{\gamma}_\ell}=\sum_{m\in \mathcal{M}}\gamma_{\ell m}B_{m}$. By (\ref{EQ:normratio}), we have
\[
c (1-R_{n,N}) |\triangle|^{2}\Vert \bs{\gamma}\Vert^{2}   \leq (1-R_{n,N}) \Vert \bs{g}_{\bs{\gamma}}\Vert^{2}
\leq \Vert \bs{g}_{\bs{\gamma}}\Vert_{n,N}^{2} = (1+R_{n,N}) \Vert \bs{g}_{\bs{\gamma}}\Vert^{2}\leq C (1+R_{n,N})
|\triangle|^{2}\Vert \bs{\gamma}\Vert^{2},
\]
in which we have used the stability conditions in Lemma \ref{LEM:normequity}.
\end{proof}

Next, we consider the following decomposition $\widetilde{\bs{\beta}}(\bs{z})=\widetilde{\bs{\beta}}_{\mu}(\bs{z})
+\widetilde{\bs{\eta}}(\bs{z})+\widetilde{\bs{\varepsilon}}(\bs{z})$, where
\begin{align}
\widetilde{\bs{\beta}}_{\mu}(\bs{z})&=(\widetilde{\beta}_{\mu,0}(\bs{z}),\ldots,\widetilde{\beta}_{\mu,p}(\bs{z}))^{\top}
=\{\mathbf{I} \otimes \widetilde{\mathbf{B}}(\bs{z})\}^{\top}\widetilde{\bs{\theta}}_{\mu},\label{DEF:beta_tilde_mu}\\
\widetilde{\bs{\eta}}(\bs{z})&=(\widetilde{\eta}_0(\bs{z}), \ldots, \widetilde{\eta}_p(\bs{z}))^{\top}=\{\mathbf{I} \otimes \widetilde{\mathbf{B}}(\bs{z})\}^{\top}\widetilde{\bs{\theta}}_{\eta},
\label{DEF:eta_tilde}\\
\widetilde{\bs{\varepsilon}}(\bs{z})&=(\widetilde{\varepsilon}_0(\bs{z}), \ldots, \widetilde{\varepsilon}_p(\bs{z}))^{\top}=\{\mathbf{I} \otimes \widetilde{\mathbf{B}}(\bs{z})\}^{\top}\widetilde{\bs{\theta}}_{\varepsilon}. \label{DEF:eps_tilde}
\end{align}

\begin{lemma}
\label{LEM:error-unif-order}
Under Assumptions (A2)--(A5) and (C1), if $N^{1/2}|\triangle|\rightarrow \infty$ as $N\rightarrow \infty$, $\|\sum_{k=1}^{\infty}\lambda_{k}^{1/2}\psi_{k}\|_{\infty} < \infty$ and  $n^{1/(4+\delta_2)} \ll n^{1/2}N^{-1/2}|\triangle|^{-1}$ for some $\delta_2$, then for $\widetilde{\bs{\eta}}$ and $\widetilde{\bs{\varepsilon}}$ in (\ref{DEF:eta_tilde}) and (\ref{DEF:eps_tilde}), $\|\widetilde{\bs{\eta}}\|_{\infty}=O_{P}\{n^{-1/2}(\log n)^{1/2}\}$ and $
\|\widetilde{\bs{\varepsilon}}\|_{\infty}=O_{P}\{(nN)^{-1/2}(\log n)^{1/2}|\triangle|^{-1}\}$.
\end{lemma}

\begin{proof}
Note that for any $\ell=0,1,\ldots,p$, $\widetilde{\eta}_{\ell}(\bs{z})=\sum_{m \in \mathcal{M}}\widetilde{\theta}_{\eta,\ell,m}\widetilde{B}_{m}(\bs{z})$ for some coefficients $\widetilde{\theta}_{\eta,\ell,m}$, so the order of $\widetilde{\eta}_{\ell}(\bs{z})$ is related to that of $\widetilde{\theta}_{\eta,\ell,m}$. In fact
\begin{align*}
\|\widetilde{\bs{\eta}}\|_{\infty}&=\max_{0\leq\ell\leq p} \|\widetilde{\eta}_{\ell}\|_{\infty}
\leq C_{\eta} \|\widetilde{\bs{\theta}}_{\eta,\ell}\|_{\infty}
=\left\|(\mathbf{e}_{\ell}\otimes \mathbf{1})^{\top} \bs{\Gamma}_{n,0}^{-1}
\left[\frac{1}{nN} \sum_{i=1}^{n}\sum_{j=1}^{N}\left\{\widetilde{\mathbf{X}}_{i}\otimes \widetilde{\mathbf{B}}(\bs{z}_{j})\right\}\eta_{i}(\bs{z}_j)\right] \right\|_{\infty},
\end{align*}
where $\widetilde{\bs{\theta}}_{\eta}=(\widetilde{\theta}_{\eta,\ell,m})_{m \in \widetilde{\mathcal{M}}}$ with $\widetilde{\mathcal{M}}$ being an index set of the transformed Bernstein basis polynomials $\widetilde{B}_{m}(\bs{z})$ and $\bs{\Gamma}_{n,0}$ is the symmetric positive definite matrix defined in (\ref{DEF:Gamma_0}). Thus, by Lemma \ref{LEM:Gamma_0},
\[
\left\|\widetilde{\bs{\eta}}\right\|_{\infty} \leq C|\triangle|^{-2} \max_{0\leq\ell\leq p}\max_{m \in \widetilde{\mathcal{M}}}\left| \frac{1}{nN} \sum_{i=1}^{n}\sum_{j=1}^{N} X_{i\ell}
\eta_{i}(\bs{z}_j)\widetilde{B}_{m}(\bs{z}_j)\right|,
\]
almost surely. Next, we show that with probability $1$, 
\begin{equation}
\max_{0\leq\ell\leq p}\max_{m \in \widetilde{\mathcal{M}}}\left| \frac{1}{nN} \sum_{i=1}^{n}\sum_{j=1}^{N} X_{i\ell}\eta_{i}(\bs{z}_j)\widetilde{B}_{m}(\bs{z}_j)\right| =O\left\{n^{-1/2}|\triangle|^{2}(\log n)^{1/2}\right\}.
\label{EQ:BEsupnorm}
\end{equation}
To prove (\ref{EQ:BEsupnorm}), let
$\varpi_{i}=\varpi_{i,m}=\sum_{k=1}^{\infty}\lambda_k^{1/2}X_{i\ell}\xi_{ik}\frac{1}{N}\sum_{j=1}^{N} \widetilde{B}_{m}(\bs{z}_j)\psi_{k}(\bs{z}_{j})$, where $ E(\varpi_{i})=0$ and
\begin{align*}
 E(\varpi_{i}^2)&=\frac{ E(X_{i\ell}^{2})}{N^2}\sum_{j=1}^{N} \sum_{j^{\prime}=1}^{N} \widetilde{B}_{m}(\bs{z}_j)\widetilde{B}_{m}(\bs{z}_{j^\prime})
G_{\eta}(\bs{z}_j,\bs{z}_{j^\prime}) \\
& \asymp
\int_{\Omega^2}G_{\eta}(\bs{z},\bs{z}^{\prime})B_m(\bs{z})B_{m}(\bs{z}^{\prime})
d\bs{z}d\bs{z}^{\prime} \asymp |\triangle|^4.
\end{align*}

We decompose the random variable $\varpi_{i}$ into a tail part and a truncated part,
\begin{align*}
\varpi_{i,1}^{D_{n}}&=\sum_{k=1}^{\infty}\lambda_k^{1/2}\left\{\frac{1}{N}\sum_{j=1}^{N}
\widetilde{B}_{m}(\bs{z}_j)\psi_{k}(\bs{z}_{j})\right\} X_{i\ell}\xi_{ik}I\left\{ \left|X_{i\ell}\xi_{ik}\right|
>D_{n}\right\},\\
\varpi_{i,2}^{D_{n}}&=\sum_{k=1}^{\infty}\lambda_k^{1/2}\left\{\frac{1}{N}\sum_{j=1}^{N}
\widetilde{B}_{m}(\bs{z}_j)\psi_{k}(\bs{z}_{j})\right\} X_{i\ell}\xi_{ik}I\left\{ \left|X_{i\ell}\xi_{ik}\right| \leq D_{n}\right\} -\mu_{i}^{D_{n}}, \\
\mu_{i}^{D_{n}}&=
\sum_{k=1}^{\infty}\lambda_k^{1/2}\left\{\frac{1}{N}\sum_{j=1}^{N}\widetilde{B}_{m}(\bs{z}_j)
\psi_{k}(\bs{z}_{j})\right\}  E\left[X_{i\ell}\xi_{ik}I\left\{ \left| X_{i\ell}\xi_{ik}\right| \leq D_{n}\right\} \right],
\end{align*}
where $%
D_{n}=n^{\alpha }\left(1/(4+\delta_1) < \alpha < 1/2\right)$. At first, we show that tail part vanishes almost surely. Note that, for any $k \geq 1$,
\begin{equation*}
\sum_{n=1}^{\infty }P\left\{ \left| X_{n\ell}\xi_{nk}\right| >D_{n}\right\} \leq
\sum_{n=1}^{\infty }\frac{ E\left|X_{n\ell}\xi_{nk} \right| ^{4+\delta_{1}}}{%
D_{n}^{4+\delta_1}}\leq \upsilon _{\delta_1 }\sum_{n=1}^{\infty }D_{n}^{-(4+\delta_{1})}<\infty .  
\label{keyH}
\end{equation*}
By the Borel-Cantelli's lemma,  we can show that $ E
\left| \frac{1}{n}\sum_{i=1}^{n}\varpi_{i,1}^{D_{n}} \right| =O\left( n^{-r}\right)$, for any $r>0$. As $ E (\varpi_{i}) =0$, then it is straightforward to verify that $\mu_{i}^{D_{n}}=- E(\varpi_{i,1}^{D_{n}})=O(D_n^{-2}|\triangle|^2)$.

Next, notice that $ E(\varpi_{i,2}^{D_{n}})=0$. Then, 
$\mathrm{Var}(\varpi_{i,2}^{D_{n}})= E(\varpi_{i}^{2}) -  E(\varpi_{i,1}^{D_{n}})^2 -( \mu^{D_{n}})^{2}
\asymp |\triangle|^4$. Also, we have, for any $r\geq 3$,
\begin{align*}
 E|\varpi_{i,2}^{D_{n}}|^r&= E\left|\sum_{k=1}^{\infty}\lambda_k^{1/2}\frac{1}{N}
\sum_{j=1}^{N}\widetilde{B}_{m}(\bs{z}_j)\psi_{k}(\bs{z}_{j})
\left[X_{i\ell}\xi_{ik}I\left\{ \left|X_{i\ell}\xi_{ik}\right|
\leq D_{n}\right\} \right]
-\mu _{i}^{D_{n}}\right|^r\\
&\leq 2^{r-1}\Bigg[ E\left|\sum_{k=1}^{\infty}\lambda_r^{1/2}
\frac{1}{N}\sum_{j=1}^{N}\widetilde{B}_{m}(\bs{z}_j)\psi_{k}(\bs{z}_{j})
X_{i\ell}\xi_{ik}I\left\{ \left|X_{i\ell}\xi_{ik}\right| \leq D_{n}\right\} \right|^{r}
+\left( \mu _{i}^{D_{n}}\right)^{r}\Bigg]  \\
&\leq \left\{2D_n
\frac{1}{N}\sum_{j=1}^{N}\widetilde{B}_{m}(\bs{z}_j)
\sum_{k=1}^{\infty}\lambda_{k}^{1/2}\psi_{k}(\bs{z}_{j})\right\}^{r-2}
 E|\varpi_{i,2}^{D_{n}}|^2 \leq  (CD_n |\triangle|^2)^{r-2} E|\varpi_{i,2}^{D_{n}}|^2.
\end{align*}
Thus, $ E\left|\varpi_{i,2}/n\right|^{r}\leq \{Cn^{-1}D_{n}|\triangle|^{2}\}
^{r-2}r! E(\varpi_{i,2}^{2}/n^2)<\infty $ with the Cramer constant $c^{*}=Cn^{-1}D_{n}|\triangle|^{2}$. By the Bernstein inequality, for any large enough $\delta >0$,
\[
P\left\{\left| \frac{1}{n}\sum_{i=1}^{n}\varpi_{i}^{D_{n}}\right| \geq \delta n^{-1/2}|\triangle|^{2}(\log n)^{1/2}\right\}
\leq 2\exp \left\{\frac{-\delta^{2}\log n}{4c+2\delta CD_n (\log n)^{1/2} n^{-1/2}} \right\} \leq
2n^{-3}.
\]
Hence,
\[
\sum_{n=1}^{\infty}P\left\{ \max_{0\leq\ell\leq p} \max_{m \in \mathcal{M}}\left| \frac{1}{n}\sum_{i=1}^{n}\varpi_{i} \right| \geq \delta n^{-1/2}|\triangle|^{2}(\log n)^{1/2}\right\} \leq C |\triangle|^{-2} \sum_{n=1}^{\infty}n^{-3}<\infty
\]
for such $\delta >0$. Thus, Borel-Cantelli's lemma implies that $\left\|\widetilde{\bs{\eta}} \right\|_{\infty}=O_{P}\{n^{-1/2}(\log n)^{1/2}\}$. The result of $\left\|\widetilde{\bs{\varepsilon}} \right\|_{\infty}=O_{P}\{(nN)^{-1/2}(\log n)^{1/2}|\triangle|^{-1}\}$ can be established similarly, thus omitted.
\end{proof}

For $\widetilde{\bs{\beta}}(\bs{z})$ defined in (\ref{DEF:beta_tilde}), Theorem \ref{THM:bias-rate} below provides its uniform convergence rate to $\bs{\beta}^{o}$.

\begin{theorem}
\label{THM:bias-rate}
Under Assumptions (A1)--(A6), for $\widetilde{\bs{\beta}}(\bs{z})$ defined in (\ref{DEF:beta_tilde}), $\|\widetilde{\bs{\beta}}-\bs{\beta}^{o}\|_{\infty}=O_{P}\{|\triangle|^{d+1}\|\bs{\beta}^{o}\|_{d+1,\infty}+n^{-1/2}(\log n)^{1/2}\}$.
\end{theorem}

\begin{proof}
Note that $\|\widetilde{\bs{\beta}}-\bs{\beta}^{o}\|_{\infty}\leq
\|\widetilde{\bs{\beta}}_{\mu}-\bs{\beta}^{o}\|_{\infty}
+\|\widetilde{\bs{\eta}}\|_{\infty}
+\|\widetilde{\bs{\varepsilon}}\|_{\infty}$, where
\[
\widetilde{\bs{\beta}}_{\mu}=\argmin_{\bs{g}\in \mathcal{G}^{(p+1)}}\sum_{i=1}^{n}\sum_{j=1}^{N}\left\{\sum_{\ell =0}^{p}X_{i\ell}(\beta_{\ell}^{o}-g_{\ell})(\bs{z}_{j})\right\}^{2}.
\]
Let $\bs{\beta}^{\ast}=(\beta_{0}^{\ast},\ldots, \beta_{p}^{\ast})^{\ast}\in \mathcal{G}^{(p+1)}$, where $\beta_{\ell}^{\ast}$'s are the best approximation to $\beta_{\ell}^{o}$'s with the approximation rate $\|\beta_{\ell}^{\ast}-\beta_{\ell}^{o}\|_{\infty}\leq C
|\triangle|^{d+1}\|\bs{\beta}^{o}\|_{d+1,\infty}$ for any $\ell=0,\ldots,p$. By \cite{Lai:Wang:13}, 
\begin{equation}
\|\widetilde{\bs{\beta}}_{\mu}-\bs{\beta}^{o}\|_{\infty}\leq
\|\widetilde{\bs{\beta}}_{\mu}-\bs{\beta}^{\ast}\|_{\infty} + \|\bs{\beta}^{\ast}-\bs{\beta}^{o}\|_{\infty}
\leq C
|\triangle|^{d+1}\|\bs{\beta}^{o}\|_{d+1,\infty}.
\label{EQ:beta_tilde_bias}
\end{equation}
The desired result follows from Lemma \ref{LEM:error-unif-order}.
\end{proof}

\vskip .10in \noindent \textbf{A.3. Asymptotic properties of penalized spline estimators} \vskip .10in

Let $\widetilde{\mathbf{B}}(\bs{z})=\mathbf{Q}_2^{\top}\mathbf{B}(\bs{z})$, then for $\mathbb{U}=\mathbf{X}\otimes (\mathbf{B}\mathbf{Q}_{2})$ defined in Section 2.2, we have
\[
\mathbb{U}^{\top}=(\widetilde{\mathbf{X}}_1\otimes\widetilde{\mathbf{B}}(\bs{z}_{1}),\ldots,\widetilde{\mathbf{X}}_1\otimes\widetilde{\mathbf{B}}(\bs{z}_{N}),\ldots,\widetilde{\mathbf{X}}_n\otimes\widetilde{\mathbf{B}}(\bs{z}_{1}),\ldots,\widetilde{\mathbf{X}}_n\otimes\widetilde{\mathbf{B}}(\bs{z}_{N})),
\]
and 
$\mathbb{U}^{\top}\mathbb{U}=\sum_{i=1}^{n}\sum_{j=1}^{N}(\widetilde{\mathbf{X}}_{i}\widetilde{\mathbf{X}}_{i}^{\top})\otimes
\{\widetilde{\mathbf{B}}(\bs{z}_{j})\widetilde{\mathbf{B}}^{\top}(\bs{z}_{j})\}$, $
\mathbb{U}^{\top}\mathbb{Y}=\sum_{i=1}^{n}\sum_{j=1}^{N} \{\widetilde{\mathbf{X}}_{i}\otimes \widetilde{\mathbf{B}}(\bs{z}_{j})\}Y_{ij}$. Let
\begin{equation}
\bs{\Gamma}_{n,\rho}=\frac{1}{nN}\sum_{i=1}^{n}\sum_{j=1}^{N}(\widetilde{\mathbf{X}}_{i}\widetilde{\mathbf{X}}_{i}^{\top})\otimes
\{\widetilde{\mathbf{B}}(\bs{z}_{j})\widetilde{\mathbf{B}}^{\top}(\bs{z}_{j})\}
+\frac{\rho_{n}}{nN} \mathbf{I}_{p} \otimes \mathbf{Q}_2^{\top} [\langle B_{m},B_{m^{\prime}}\rangle_{\mathcal{E}}]_{m,m^{\prime}\in \mathcal{M}}\mathbf{Q}_2
\label{DEF:Gamma_rho},
\end{equation}
which is a symmetric positive definite matrix.

Next, we define
\begin{align*}
\widehat{\bs{\theta}}_{\mu}&=(\widehat{\bs{\theta}}_{\mu,0}^{\top},\ldots,\widehat{\bs{\theta}}_{\mu,p}^{\top})^{\top}
=\bs{\Gamma}_{n,\rho}^{-1}\frac{1}{nN}\sum_{i=1}^{n}\sum_{j=1}^{N} \left\{\widetilde{\mathbf{X}}_{i}\otimes \widetilde{\mathbf{B}}(\bs{z}_{j})\right\} \widetilde{\mathbf{X}}_{i}^{\top}\bs{\beta}^{o}(\bs{z}_j), \notag\\
\widehat{\bs{\theta}}_{\eta}&=(\widehat{\bs{\theta}}_{\eta,0}^{\top},\ldots,\widehat{\bs{\theta}}_{\eta,p}^{\top})^{\top}
=\bs{\Gamma}_{n,\rho}^{-1}
\frac{1}{nN}\sum_{i=1}^{n}\sum_{j=1}^{N} \left\{\widetilde{\mathbf{X}}_{i}\otimes \widetilde{\mathbf{B}}(\bs{z}_{j})\right\}\sum_{k=1}^{\infty}\lambda_k^{1/2}\xi_{ik}\psi_{k}(\bs{z}_j),\notag\\
\widehat{\bs{\theta}}_{\varepsilon}&=(\widehat{\bs{\theta}}_{\varepsilon,0}^{\top},\ldots,
\widehat{\bs{\theta}}_{\varepsilon,p}^{\top})^{\top}
=\bs{\Gamma}_{n,\rho}^{-1}
\frac{1}{nN}\sum_{i=1}^{n}\sum_{j=1}^{N} \left\{\widetilde{\mathbf{X}}_{i}\otimes \widetilde{\mathbf{B}}(\bs{z}_{j})\right\}\sigma(\bs{z}_j)\varepsilon_{ij}.
\label{DEF:theta-eta-eps}
\end{align*}

Note that, for any $\ell=0,\ldots,p$, the penalized bivariate spline estimator $\widehat{\beta}_{\ell}$ can be written as:
\begin{equation}
\widehat{\beta}_{\ell}(\bs{z})=\widehat{\beta}_{\mu,\ell}(\bs{z})+\widehat{\eta}_{\ell}(\bs{z})
+\widehat{\varepsilon}_{\ell}(\bs{z}),
\label{EQ:decompose2}
\end{equation}
where
\begin{equation*}
\widehat{\beta}_{\mu,\ell}(\bs{z})=\widetilde{\mathbf{B}}(\bs{z})^{\top}\widehat{\bs{\theta}}_{\mu,\ell},~~
\widehat{\eta}_{\ell}(\bs{z})=\widetilde{\mathbf{B}}(\bs{z})^{\top}\widehat{\bs{\theta}}_{\eta,\ell},~~
\widehat{\varepsilon}_{\ell}(\bs{z})=\widetilde{\mathbf{B}}(\bs{z})^{\top}\widehat{\bs{\theta}}_{\varepsilon,\ell},
\label{DEF:pls_estimator}
\end{equation*}
Therefore, we have
\begin{equation}
\widehat{\beta}_{\ell}(\bs{z})-\beta^{o}_{\ell}(\bs{z})=\widehat{\beta}_{\mu,\ell}(\bs{z})-\beta^{o}_{\ell}(\bs{z})
+\widehat{\eta}_{\ell}(\bs{z})+\widehat{\varepsilon}_{\ell}(\bs{z}).
\label{EQ:decompose3}
\end{equation}

\begin{lemma}
\label{LEM:Gamma_rho}
Under Assumptions (A3)--(A5), if $N^{1/2}|\triangle|\rightarrow \infty$ as $N\rightarrow \infty$, then there exist constants $0 < c_{\Gamma} < C_{\Gamma} < \infty$, such that with probability approaching 1 as $N\rightarrow \infty$ and $n\rightarrow \infty$,
$c_{\Gamma}|\triangle|^{2} \leq \lambda_{\min}(\bs{\Gamma}_{n,\rho}) \leq
\lambda_{\max}(\bs{\Gamma}_{n,\rho}) \leq  C_{\Gamma}\left(|\triangle|^{2}+\frac{\rho_{n}}{nN|\triangle|^{2}}\right)$.
\end{lemma}
\begin{proof}
By (\ref{EQ:g_gamma}), it is easy to see that, for any vector $\bs{\theta}=(\bs{\theta}_0^{\top},\cdots,\bs{\theta}_p^{\top})^{\top}$,
\[
\bs{\theta}^{\top}\bs{\Gamma}_{n,\rho}\bs{\theta} =\|\bs{g}_{\bs{\gamma}}\|_{n,N}^2
+\frac{\rho_n}{nN}\sum_{\ell =0}^{p} \bs{\gamma}_{\ell}^{\top}[\langle B_{m},B_{m^{\prime}}\rangle_{\mathcal{E}}]_{m,m^{\prime}\in \mathcal{M}} \bs{\gamma}_{\ell},
\]
where $\bs{\gamma}=(\bs{\gamma}_0,\ldots, \bs{\gamma}_p)^{\top}=\mathbf{Q}_2\bs{\theta} $ with $\bs{\gamma}_{\ell}=(\gamma_{\ell m},m\in\mathcal{M})^{\top}$.
Using the Markov's inequality in the supplement of \cite{Lai:Wang:13} and Lemma \ref{LEM:normequity}, we have
\[
\sum_{\ell =0}^{p}\left\|\sum_{m\in \mathcal{M}}\gamma_{\ell m}B_{m}\right\|_{\mathcal{E}}^2\leq \frac{C}{|\triangle|^{4}}\sum_{\ell =0}^{p}\left\Vert \sum_{m\in \mathcal{M}}\gamma_{\ell m}B_{m}\right\Vert_{L _2} ^{2}\leq \frac{C}{|\triangle|^{2}}\Vert\bs{\gamma}\Vert ^{2}.
\]
Thus, the largest eigenvalue of the matrix $\bs{\Gamma}_{n,\rho}$ in (\ref{DEF:Gamma_rho}) satisfies that
$\lambda_{\max}(\bs{\Gamma}_{n,\rho}) \leq C\left\{(1+R_{n,N})
|\triangle|^{2}+(nN|\triangle|^{2})^{-1}\rho_{n}\right\}$. 
Thus, we have with probability approaching 1,
$\lambda_{\max}(\bs{\Gamma}_{n,\rho})\leq C_{\Gamma} \left\{|\triangle|^{2}+(nN|\triangle|^{2})^{-1}\rho_{n}\right\}$
for some positive constant $C_{\Gamma}$. On the other hand, we use Lemma \ref{LEM:normequity} and equation (\ref{EQ:normratio}) to have
$\Vert \bs{g}_{\bs{\gamma}}\Vert_{n,N}^{2} =(1-R_{n,N})\Vert\bs{g}_{\bs{\gamma}}\Vert^{2} \geq c (1-R_{n,N})
|\triangle|^{2}\Vert \bs{\gamma}\Vert^{2}$.

Therefore, $\lambda_{\min}(\bs{\Gamma}_{n,\rho}) \geq c(1-R_{n,N})
|\triangle|^{2}=c_{\Gamma}|\triangle|^{2}$.
\end{proof}



\begin{lemma}
\label{LEM:uniformbiasrate}
Under Assumptions (A1), (A3) and (A5), if $N^{1/2}|\triangle| \to \infty$, one has
$\|\widehat{\bs{\beta}}_{\mu}-\bs{\beta}^{o}\|_{\infty}=O_{P}\left\{
\frac{\rho_{n}}{nN|\triangle|^{3}}\|\bs{\beta}^{o}\|_{2,\infty}
+\left(1+\frac{\rho_{n}}{nN|\triangle|^{5}}\right)
|\triangle|^{d +1}\|\bs{\beta}^{o}\|_{d+1,\infty}
\right\}$.
\end{lemma}
\begin{proof}
Define
\begin{equation}
A_{n}=\sup_{\bs{g}\in \mathcal{G}^{(p+1)}}\left\{ \frac{\|\bs{g}\|
_{\infty}}{\|\bs{g}\|_{n,N}},\|\bs{g}\|_{n,N}\neq 0\right\},~~ \overline{A}_{n}=\sup_{\bs{g}\in \mathcal{G}^{(p+1)}}\left\{\frac{%
\|\bs{g}\|_{\mathcal{E}}}{\|\bs{g}\|_{n,N}},\|\bs{g}\|_{n,N}\neq 0\right\},
\label{DEF:An}
\end{equation}
where random variables $A_{n}$ and $\overline{A}_{n}$ depend on
the collection of $X_{i\ell}$'s, $i=1,\ldots,n$, $\ell=0,\ldots,p$. It is clear that
$\|\bs{\beta}^{o}-\widehat{\bs{\beta}}_{\mu}\|_{\infty}\leq \|\bs{\beta}^{o}-\widetilde{\bs{\beta}}_{\mu}\|_{\infty}+\| \widetilde{\bs{\beta}}_{\mu}-\widehat{\bs{\beta}}_{\mu}\|_{\infty}$,
where $\widetilde{\bs{\beta}}_{\mu}$ is given in (\ref{DEF:beta_tilde_mu}), and $\|\widetilde{\bs{\beta}}_{\mu}-\bs{\beta}^{o}\|_{\infty}\leq C|\triangle|^{d+1}\|\bs{\beta}^{o}\|_{d+1,\infty}$ according to (\ref{EQ:beta_tilde_bias}).

By the definition of $A_{n}$ in  (\ref{DEF:An}), we have
\begin{equation}
\|\widetilde{\bs{\beta}}_{\mu}-\widehat{\bs{\beta}}_{\mu}\|_{\infty}\leq
A_{n}\|\widetilde{\bs{\beta}}_{\mu}-\widehat{\bs{\beta}}_{\mu}\|_{n,N}.
\label{EQ:pls-ps-sup}
\end{equation}
Note that the penalized spline $\widehat{\bs{\beta}}_{\mu}$ of $\bs{\beta}^{o}$ is
characterized by the orthogonality relations
\begin{eqnarray}
nN\langle\bs{\beta}^{o}-\widehat{\bs{\beta}}_{\mu},\bs{g}\rangle _{n,N}=\rho_{n} \langle \widehat{\bs{\beta}}_{\mu},\bs{g}\rangle _{\mathcal{E}},\quad \textrm{for all } \bs{g}\in \mathcal{G}^{(p+1)},
\label{EQ:pls-bias}
\end{eqnarray}
while $\widetilde{\bs{\beta}}_{\mu}$ is characterized by
\begin{eqnarray}
\langle\bs{\beta}^{o} -\widetilde{\bs{\beta}}_{\mu},\bs{g}\rangle _{n,N}=0,\quad \textrm{for all }\bs{g}\in \mathcal{G}^{(p+1)}.
\label{EQ:ls-bias}
\end{eqnarray}
By (\ref{EQ:pls-bias}) and (\ref{EQ:ls-bias}), we have $nN\langle \widetilde{\bs{\beta}}_{\mu}-\widehat{\bs{\beta}}_{\mu},\bs{g}\rangle _{n,N}=\rho_{n} \langle
\widehat{\bs{\beta}}_{\mu},\bs{g}\rangle _{\mathcal{E}}$, for all
$\bs{g}\in \mathcal{G}^{(p+1)}$. Inserting $\bs{g}=\widetilde{\bs{\beta}}_{\mu}-\widehat{\bs{\beta}}_{\mu}$ yields
that
\begin{eqnarray}
nN\|\widetilde{\bs{\beta}}_{\mu}-\widehat{\bs{\beta}}_{\mu}\|_{n,N}^{2}=\rho_{n} \langle \widehat{\bs{\beta}}_{\mu},\widetilde{\bs{\beta}}_{\mu}-\widehat{\bs{\beta}}_{\mu}\rangle _{\mathcal{E}}.  \label{EQ:pls-ls}
\end{eqnarray}
Thus, by Cauchy-Schwarz inequality and the definition of
$\overline{A}_{n}$.
\[
nN\|\widetilde{\bs{\beta}}_{\mu}-\widehat{\bs{\beta}}_{\mu}\|_{n,N}^{2}\leq
\rho_n \|\widehat{\bs{\beta}}_{\mu}\|_{\mathcal{E}}\|\widetilde{\bs{\beta}}_{\mu}
-\widehat{\bs{\beta}}_{\mu}\|_{\mathcal{E}}\leq \rho_{n} \overline{A}_{n} \|\widehat{\bs{\beta}}_{\mu}\|_{\mathcal{E}}\| \widetilde{\bs{\beta}}_{\mu}-\widehat{\bs{\beta}}_{\mu}\|_{n,N}.
\]
Similarly, using (\ref{EQ:pls-ls}), $nN\|\widetilde{\bs{\beta}}_{\mu}-\widehat{\bs{\beta}}_{\mu}\|_{n,N}^{2}=\rho_n \{
\langle \widehat{\bs{\beta}}_{\mu},\widetilde{\bs{\beta}}_{\mu}\rangle
_{\mathcal{E}}-\langle \widehat{\bs{\beta}}_{\mu},\widehat{\bs{\beta}}_{\mu}\rangle _{\mathcal{E}}\} \geq 0$. Thus, by Cauchy-Schwarz inequality, $\|\widehat{\bs{\beta}}_{\mu}\|_{\mathcal{E}}^{2}\leq \langle \widehat{\bs{\beta}}_{\mu},\widetilde{\bs{\beta}}_{\mu}\rangle _{\mathcal{E}}\leq \|\widehat{\bs{\beta}}_{\mu}\|_{\mathcal{E}}\|\widetilde{\bs{\beta}}_{\mu}\|_{\mathcal{E}}$,
which implies that $\|\widehat{\bs{\beta}}_{\mu}\|_{\mathcal{E}}\leq \|\widetilde{\bs{\beta}}_{\mu}\|_{\mathcal{E}}$. Therefore,
\begin{eqnarray}
\|\widetilde{\bs{\beta}}_{\mu}-\widehat{\bs{\beta}}_{\mu}\|_{n,N}\leq
\rho_n (nN)^{-1}\overline{A}_{n} \|\widetilde{\bs{\beta}}_{\mu}\|_{\mathcal{E}}. \label{EQ:pls-ps2}
\end{eqnarray}
Combining (\ref{EQ:pls-ps-sup}) and (\ref{EQ:pls-ps2}) yields that
\[
\|\widetilde{\bs{\beta}}_{\mu}-\widehat{\bs{\beta}}_{\mu}\|_{\infty}\leq A_{n}
\|\widetilde{\bs{\beta}}_{\mu}-\widehat{\bs{\beta}}_{\mu}\|_{n,N}\leq
\rho_n (nN)^{-1}A_{n} \overline{A}_{n} \|\widetilde{\bs{\beta}}_{\mu}\|_{\mathcal{E}}.
\]
By Lemma \ref{LEM:appord}, we have
\[
\|\widetilde{\bs{\beta}}_{\mu}\|_{\mathcal{E}} =C_{1}
\{\|\bs{\beta}^{o}\|_{2,\infty}+\sum_{a_{1}+a_{2}=2 }\|\nabla_{z_{1}}^{a_{1}}\nabla_{z_{2}}^{a_{2}}(
\bs{\beta}^{o}- \widetilde{\bs{\beta}}_{\mu})\|_{\infty}\}
\leq C_{2}(\|\bs{\beta}^{o}\|_{2,\infty}+
|\triangle|^{d -1}\|\bs{\beta}^{o}\| _{d +1,\infty}).
\]
It follows
\begin{eqnarray}
\label{EQ:pls-ps-sup-2}
\|\widetilde{\bs{\beta}}_{\mu}-\widehat{\bs{\beta}}_{\mu}\|_{\infty}=\rho_n (nN)^{-1}A_{n}\overline{A}_{n}
C_{2}(\left\|\bs{\beta}^{o}\right\|_{2,\infty}+|\triangle|^{d -1}\|\bs{\beta}^{o}\|
_{d+1,\infty}) .
\end{eqnarray}
Next we derive the order of $A_{n}$ and $\overline{A}_{n}$. By Markov's inequality, for any $\bs{g}\in \mathcal{G}^{(p+1)}$, $\|\bs{g}\|_{\infty}\leq C|\triangle| ^{-1}\|\bs{g}\|$, $\|\bs{g}\|_{\mathcal{E}}\leq C|\triangle|^{-2}\|\bs{g}\|$. Equation (\ref{EQ:normratio}) implies that 
\[
\sup_{\bs{g}\in\mathcal{G}(\triangle)}\left\{ \left. \|\bs{g}\|_{n,N}\right/
\|\bs{g}\|\right\} \geq \left[1-O_{P}\left\{(\log n)^{1/2}{n}^{-1/2}+ N^{-1/2}|\triangle|^{-1}\right\}\right]^{1/2}.
\]
Thus, we have
\begin{align*}
A_{n} &\leq C|\triangle| ^{-1}\left[1-O_{P}\left\{(\log n)^{1/2}{n}^{-1/2}+ N^{-1/2}|\triangle|^{-1}\right\} \right]
^{-1/2}=O_{P}\left( |\triangle| ^{-1}\right) , \\
\overline{A}_{n} &\leq C|\triangle|
^{-2}\left[1-O_{P}\left\{(\log n)^{1/2}{n}^{-1/2}+ N^{-1/2}|\triangle|^{-1}\right\}
\right]^{-1/2}=O_{P}\left( |\triangle| ^{-2}\right).
\end{align*}
Plugging the order of $A_{n}$ and $\overline{A}_{n}$ into (\ref{EQ:pls-ps-sup-2}) yields that
\[
\|\widetilde{\bs{\beta}}_{\mu}-\widehat{\bs{\beta}}_{\mu}\|_{\infty}
=O_{P}\left\{ \frac{C_{2}\rho_n}{nN|\triangle| ^{3}}(\|\bs{\beta}^{o}\|_{2,\infty}
+|\triangle|^{d -1}\|\bs{\beta}^{o}\|_{d +1,\infty}) \right\}.
\]
Hence,
\[
\|\widehat{\bs{\beta}}_{\mu}-\bs{\beta}^{o}\|_{\infty}\leq C_{1}|\triangle|^{d +1}\|\bs{\beta}^{o}\|_{d+1,\infty}+O_{P}\left\{ \frac{C_{2}\rho_n}{nN|\triangle|^{3}}
\left(\|\bs{\beta}^{o}\|_{2,\infty}+|\triangle|^{d -1}\|
\bs{\beta}^{o}\|_{d +1,\infty}\right) \right\} .
\]
Therefore, Lemma \ref{LEM:uniformbiasrate} is established.
\end{proof}

\begin{lemma}
\label{LEM:thetatilde-eta}
Suppose Assumptions (A2)--(A5) hold and  $N^{1/2}|\triangle|\rightarrow \infty$ as $N\rightarrow \infty$, then $\|\widehat{\bs{\theta}}_{\eta}\|^2=O_P(n^{-1}|\triangle|^{-2})$.
\end{lemma}
\begin{proof}
Note that $\widehat{\bs{\theta}}_{\eta}=\bs{\Gamma}_{n,\rho}^{-1}
\frac{1}{nN}\sum_{i=1}^{n}\sum_{j=1}^{N} \left\{\widetilde{\mathbf{X}}_{i}\otimes \widetilde{\mathbf{B}}(\bs{z}_{j})\right\}\sum_{k=1}^{\infty}\lambda_k^{1/2}\xi_{ik}\psi_{k}(\bs{z}_j)$.
According to Lemma \ref{LEM:Gamma_rho},
\begin{align*}
\|\widehat{\bs{\theta}}_{\eta}\|^2 \asymp& \frac{1}{n^2N^2|\triangle|^4} \sum_{i,i^{\prime}=1}^{n}\sum_{j,j^{\prime}=1}^{N}
\left\{\widetilde{\mathbf{X}}_{i}\otimes\widetilde{\mathbf{B}}(\bs{z}_j)\right\}^{\top}\\
&\times\sum_{k=1}^{\infty}\lambda_k^{1/2}\xi_{ik}\psi_{k}(\bs{z}_j) \left\{\mathbf{X}_{i^{\prime}}\otimes\widetilde{\mathbf{B}}(\bs{z}_{j^{\prime}})\right\}\sum_{k=1}^{\infty}\lambda_k^{1/2}\xi_{i^\prime k}\psi_{k}(\bs{z}_{j^\prime}).
\end{align*}
Note that
\begin{align*}
\widetilde{\mathbf{X}}_{i}&\otimes\widetilde{\mathbf{B}}(\bs{z}_j)\sum_{k=1}^{\infty}\lambda_k^{1/2}\xi_{ik}\psi_{k}(\bs{z}_j) \\
&=\left(X_{i0}\widetilde{\mathbf{B}}(\bs{z}_j)^{\top}\sum_{k=1}^{\infty}\lambda_k^{1/2}\xi_{ik}\psi_{k}(\bs{z}_j),\ldots, X_{ip}\widetilde{\mathbf{B}}(\bs{z}_j)^{\top}\sum_{k=1}^{\infty}\lambda_k^{1/2}\xi_{ik}\psi_{k}(\bs{z}_j)\right)^{\top},
\end{align*}
so one has
\[
\|\widehat{\bs{\theta}}_{\eta}\|^2
\asymp \frac{1}{n^2N^2|\triangle|^4} \sum_{\ell =0}^{p} \sum_{i,i^{\prime}=1}^{n}\sum_{j,j^{\prime}=1}^{N} X_{i\ell}X_{i^{\prime}\ell}\widetilde{\mathbf{B}}(\bs{z}_j)^{\top}
\widetilde{\mathbf{B}}(\bs{z}_{j^{\prime}}) \sum_{k, k^{\prime}=1}^{\infty}(\lambda_{k}\lambda_{k'})^{1/2}\xi_{ik}\psi_{k}(\bs{z}_j)
\xi_{i^{\prime}k^{\prime}}\psi_{k^{\prime}}(\bs{z}_{j^{\prime}}).
\]
Because the eigenvalues of $\mathbf{Q}_2\mathbf{Q}_2^{\top}$ are either 0 or 1, under Assumptions (A2) and (A3), 
for any $\ell,i$, 
\begin{align*}
\frac{1}{N^2}&\sum_{j=1}^N\sum_{j^{\prime}=1}^{N}
 E\left\{X_{i\ell}^{2}\widetilde{\mathbf{B}}(\bs{z}_j)^{\top}\widetilde{\mathbf{B}}(\bs{z}_{j^{\prime}})\sum_{k, k^{\prime}=1}^{\infty}(\lambda_{k}\lambda_{k'})^{1/2}\xi_{ik}\psi_{k}(\bs{z}_j)
\xi_{ik^{\prime}}\psi_{k^{\prime}}(\bs{z}_{j^{\prime}})\right\}\\
&\leq C \sum_{m\in \mathcal{M}}
\frac{1}{N^2}\sum_{j=1}^N\sum_{j^{\prime}=1}^{N} B_{m}(\bs{z}_j)B_{m}(\bs{z}_{j^{\prime}})
G_{\eta}(\bs{z}_j,\bs{z}_{j^{\prime}}).
\end{align*}
Assumption (A4) and (\ref{EQ:G_integration}) imply that
\begin{align*}
\frac{1}{N^2}&\sum_{j\neq j^{\prime}} B_{m}(\bs{z}_j)B_{m}(\bs{z}_{j^{\prime}})
G_{\eta}(\bs{z}_j,\bs{z}_{j^{\prime}})=\int_{T_{m}\times T_{m}}G_{\eta}(\bs{z},\bs{z}^{\prime})B_m(\bs{z})B_m(\bs{z}^{\prime})
d\bs{z}d\bs{z}^{\prime}\\
& \times \{1+O(N^{-1/2}|\triangle|^{3})\}=O(|\triangle|^{4}).
\end{align*}
Thus,
\[
\frac{1}{N^2}\sum_{j=1}^N\sum_{j^{\prime}=1}^{N} EX_{i\ell}^{2}\mathbf{B}(\bs{z}_j)^{\top}\mathbf{B}(\bs{z}_{j^{\prime}})\sum_{k, k^{\prime}=1}^{\infty}(\lambda_{k}\lambda_{k'})^{1/2}\xi_{ik}\psi_{k}(\bs{z}_j)
\xi_{ik^{\prime}}\psi_{k^{\prime}}(\bs{z}_{j^{\prime}}) \leq C|\triangle|^{2}.
\]
Next for any $\ell$, $i\neq i^{\prime}$, $j$, $j^{\prime}$, we have
\begin{align*}
E& \left\{X_{i\ell}X_{i^{\prime}\ell} \mathbf{B}(\bs{z}_j)^{\top}
\mathbf{B}(\bs{z}_{j^{\prime}}) \sum_{k, k^{\prime}=1}^{\infty}(\lambda_{k}\lambda_{k'})^{1/2}\xi_{ik}\psi_{k}(\bs{z}_j)
\xi_{ik^{\prime}}\psi_{k^{\prime}}(\bs{z}_{j^{\prime}})\right\}\\
&=E (X_{i\ell}X_{i^{\prime}\ell})\sum_{m\in \mathcal{M}}B_{m}^{2}(\bs{z}_j)B_{m}^{2}(\bs{z}_{j^{\prime}}) \sum_{k,k^{\prime}} E \left\{(\lambda_{k}\lambda_{k'})^{1/2}\xi_{ik}\xi_{i^{\prime}k^{\prime}}\psi_{k}(\bs{z}_{j})\psi_{k^{\prime}}(\bs{z}_{j^{\prime}})\right\}=0.
\end{align*}
Therefore, $ E\|\widehat{\bs{\theta}}_{\eta}\|^2
\leq C p (n^{-1}|\triangle|^{-2})$. The conclusion of the lemma follows.
\end{proof}

\begin{lemma}
	\label{LEM:thetatilde-eps}
	Suppose Assumptions (A2)--(A5) hold and  $N^{1/2}|\triangle|\rightarrow \infty$ as $N\rightarrow \infty$, then $\|\widehat{\bs{\theta}}_{\varepsilon}\|^2=O_P(n^{-1}N^{-1}|\triangle|^{-4})$.
\end{lemma}
\begin{proof}
By the definition of $\widehat{\bs{\theta}}_{\varepsilon}$ in (\ref{DEF:theta-beta-eps}), we have
\begin{align*}
	\|\widehat{\bs{\theta}}_{\varepsilon}\|^2&=\frac{1}{n^2N^2|\triangle|^4}\left(|\triangle|^{-2}\bs{\Gamma}_{n,\rho}\right)^{-1}
\sum_{i=1}^{n}\sum_{j=1}^{N} \left\{\widetilde{\mathbf{X}}_{i}\otimes \widetilde{\mathbf{B}}(\bs{z}_{j})\right\}^{\top}\sigma(\bs{z}_j)\varepsilon_{ij}\\
	&\quad \times\left(|\triangle|^{-2}\bs{\Gamma}_{n,\rho}\right)^{-1}
	\sum_{i=1}^{n} \sum_{j^{\prime}=1}^{N} \left\{\widetilde{\mathbf{X}}_{i}\otimes \widetilde{\mathbf{B}}(\bs{z}_{j^{\prime}})\right\}\sigma(\bs{z}_{j^{\prime}})\varepsilon_{ij^{\prime}}.
	\end{align*}
	By Lemma \ref{LEM:Gamma_rho},
	\[
\|\widehat{\bs{\theta}}_{\varepsilon}\|^2 \asymp \frac{1}{n^2N^2|\triangle|^4} \sum_{i,i^{\prime}=1}^{n}\sum_{j,j^{\prime}=1}^{N}\left\{\widetilde{\mathbf{X}}_{i}
\otimes\widetilde{\mathbf{B}}(\bs{z}_j)\right\}^{\top}\sigma(\bs{z}_j)\varepsilon_{ij} \left\{\mathbf{X}_{i^{\prime}}\otimes\widetilde{\mathbf{B}}(\bs{z}_{j^{\prime}})\right\}
\sigma(\bs{z}_{j^{\prime}})\varepsilon_{i^{\prime}j^{\prime}}.
\]
 Note that
\[
    \widetilde{\mathbf{X}}_{i}\otimes\widetilde{\mathbf{B}}(\bs{z}_j)\sigma(\bs{z}_j)\varepsilon_{ij}
	=\left(X_{i0}\widetilde{\mathbf{B}}(\bs{z}_j)^{\top}\sigma(\bs{z}_j)\varepsilon_{ij}, X_{i1}\widetilde{\mathbf{B}}(\bs{z}_j)^{\top}\sigma(\bs{z}_j)\varepsilon_{ij},\ldots, X_{ip}\widetilde{\mathbf{B}}(\bs{z}_j)^{\top}\sigma(\bs{z}_j)\varepsilon_{ij}\right)^{\top},
\]
so one has
\[
\|\widetilde{\bs{\theta}}_{\varepsilon}\|^2  \asymp \frac{1}{n^2N^2|\triangle|^4} \sum_{\ell =0}^{p} \sum_{i,i^{\prime}=1}^{n}\sum_{j,j^{\prime}=1}^{N} X_{i\ell}X_{i^{\prime}\ell}\widetilde{\mathbf{B}}(\bs{z}_j)^{\top}
\widetilde{\mathbf{B}}(\bs{z}_{j^{\prime}}) \sigma(\bs{z}_j) \sigma(\bs{z}_{j^{\prime}})\varepsilon_{ij} \varepsilon_{i^{\prime}j}.
 \]
Because the eigenvalues of $\mathbf{Q}_2\mathbf{Q}_2^{\top}$ are either 0 or 1, under Assumption (A2), for any $\ell,i$, by (\ref{EQ:sigma_integration}),
\begin{align*}
\frac{1}{N}\sum_{j=1}^{N}&\widetilde{\mathbf{B}}(\bs{z}_j)^{\top}\widetilde{\mathbf{B}}(\bs{z}_{j})
\sigma^{2}(\bs{z}_j)
=\mathbf{B}(\bs{z}_j)^{\top}\mathbf{Q}_2\mathbf{Q}_2^{\top} \mathbf{B}(\bs{z}_{j})\sigma^{2}(\bs{z}_j)
\leq C \sum_{m\in \mathcal{M}}\frac{1}{N}\sum_{j=1}^{N}B_{m}^{2}(\bs{z}_j) \sigma^{2}(\bs{z}_j)\\
&\leq C \sum_{m\in \mathcal{M}} \int_{{T_{\lceil m /d^{\ast}\rceil}}}\sigma^{2}(\bs{z})B_{m}^{2}(\bs{z})d\bs{z}\{1+O(N^{-1/2}|\triangle|^{-1})\}\leq C.
\end{align*}
Next note that for any $\ell,i,j\neq j^{\prime}$,
$ E\{X_{i\ell}^{2}\widetilde{\mathbf{B}}(\bs{z}_j)^{\top}
    \widetilde{\mathbf{B}}(\bs{z}_{j^{\prime}})\varepsilon_{ij} \varepsilon_{ij^{\prime}}\}=0$,
and for any $\ell$, $i\neq i^{\prime}$, $j$, $j^{\prime}$,
	$ E\{X_{i\ell}X_{i^{\prime}\ell}\widetilde{\mathbf{B}}(\bs{z}_j)^{\top} \widetilde{\mathbf{B}}(\bs{z}_{j^{\prime}})\sigma(\bs{z}_j)\sigma(\bs{z}_{j^{\prime}})
\varepsilon_{ij}\varepsilon_{i^{\prime}j^{\prime}}\}=0$.
	Therefore,
\[
 E\|\widehat{\bs{\theta}}_{\varepsilon}\|^2 \asymp \frac{1}{nN|\triangle|^4} \sum_{\ell =0}^{p}  E(X_{i\ell}^{2})\frac{1}{N}\sum_{j=1}^{N} \widetilde{\mathbf{B}}(\bs{z}_j)^{\top}
\widetilde{\mathbf{B}}(\bs{z}_{j}) \sigma^{2}(\bs{z}_j)
\leq Cp (nN)^{-1} |\triangle|^{-4}.
\]
The conclusion of the lemma follows.
\end{proof}

\textit{Proof of Theorem \ref{THM:beta-convergence}.} 
By Lemma \ref{LEM:thetatilde-eta}, Lemma \ref{LEM:thetatilde-eps}, and the properties of the bivariate spline basis functions in Lemma \ref{LEM:normequity}, $\Vert \widehat{\eta}_{\ell}\Vert_{L_2}^2 \asymp  |\triangle|^2 \Vert \widehat{\bs{\theta}}_{\eta,\ell}\Vert^2 =O_P(n^{-1})$ and $\Vert \widehat{\varepsilon}_{\ell}\Vert_{L_2}^2 \asymp  |\triangle|^2 \Vert \widehat{\bs{\theta}}_{\varepsilon,\ell}\Vert^2 =O_P(n^{-1}N^{-1}|\triangle|^{-2})$, for any $\ell=0,1,\ldots,p$. It is clear that
$\Vert \widehat{\beta}_{\ell}-\beta^{o}_{\ell}\Vert_{L_2}^2 \leq \Vert \widehat{\beta}_{\mu,\ell}-\beta^{o}_{\ell}\Vert_{L_2}^2 + \Vert \widehat{\eta}_{\ell}\Vert_{L_2}^2+ \Vert \widehat{\varepsilon}_{\ell}\Vert_{L_2}^2$, where the  asymptotic order of $\Vert \widehat{\beta}_{\mu,\ell}-\beta^{o}_{\ell}\Vert_{L_2}$ is the same as $\Vert \widehat{\beta}_{\mu,\ell}-\beta^{o}_{\ell}\Vert_{\infty}$. The desired result follows from Lemma \ref{LEM:uniformbiasrate}. $\square$


\begin{lemma}
\label{LEM:normality}
Under Assumptions (A1)--(A6), if for any $\ell=0,1,\ldots, p$, $|X_{i\ell}| \leq C_{\ell} < \infty$, then as $N\rightarrow \infty$ and $n\rightarrow \infty$, one has for any vector $\mathbf{a}=(\mathbf{a}_0^{\top},\ldots,\mathbf{a}_p^{\top})^{\top}$ with $\mathbf{a}^{\top}\mathbf{a}=1$,
$[\mathrm{Var}\{\mathbf{a}^{\top}(\widehat{\bs{\theta}}_{\eta}+\widehat{\bs{\theta}}_{\varepsilon})\}]^{-1/2}
	\{\mathbf{a}^{\top}(\widehat{\bs{\theta}}_{\eta}+\widehat{\bs{\theta}}_{\varepsilon})\}
	\overset{\mathcal{L}}{\longrightarrow}N(0,1)$,
where $\widehat{\bs{\theta}}_{\eta}$ and $\widehat{\bs{\theta}}_{\varepsilon}$ are given in (\ref{DEF:theta-beta-eps}).
\end{lemma}
\begin{proof}
For coefficient vectors $\widehat{\bs{\theta}}_{\eta}$, $\widehat{\bs{\theta}}_{\varepsilon}$ and the matrix $%
\bs{\Gamma}_{n,\rho}$ defined in (\ref{DEF:Gamma_rho}), $\mathrm{Var}\{\mathbf{a}^{\top}(\widehat{\bs{\theta}}_{\eta}+\widehat{\bs{\theta}}_{\varepsilon})\}=\mathbf{a}^{\top}\{ E(\widehat{\bs{\theta}}_{\eta}\widehat{\bs{\theta}}_{\eta}^{\top})
	+ E(\widehat{\bs{\theta}}_{\varepsilon}\widehat{\bs{\theta}}_{\varepsilon}^{\top})\}\mathbf{a}$. 
Denote $\bs{\Psi}_{\eta}=(\bs{\Psi}_{\eta,\ell,\ell^{\prime}})_{\ell,\ell^{\prime}}$ and $\bs{\Psi}_{\varepsilon}=(\bs{\Psi}_{\epsilon,\ell,\ell^{\prime}})_{\ell,\ell^{\prime}}$, with
\begin{align*}
\bs{\Psi}_{\eta,\ell,\ell^{\prime}}&=\frac{1}{n^2N^2}\sum_{i=1}^{n}\sum_{j=1}^N\sum_{j^{\prime}=1}^N X_{i\ell}X_{i\ell^{\prime}}\widetilde{\mathbf{B}}(\bs{z}_{j})
\widetilde{\mathbf{B}}^{\top}(\bs{z}_{j^{\prime}})G_{\eta}(\bs{z}_j,\bs{z}_{j^{\prime}}),\\
\bs{\Psi}_{\varepsilon,\ell,\ell^{\prime}}&=\frac{1}{n^2N^2}\sum_{i=1}^{n}\sum_{j=1}^{N}
X_{i\ell}X_{i\ell^{\prime}}\widetilde{\mathbf{B}}(\bs{z}_{j})
\widetilde{\mathbf{B}}^{\top}(\bs{z}_{j})\sigma^2(\bs{z}_j),
\end{align*}
then, we have
\begin{align*}
\mathbf{a}^{\top} E(\widehat{\bs{\theta}}_{\eta}\widehat{\bs{\theta}}_{\eta}^{\top})\mathbf{a}
=& E\mathbf{a}^{\top}\bs{\Gamma}_{n,\rho}^{-1}\frac{1}{n^2N^2}
\sum_{i=1}^{n}\sum_{j,j^{\prime}=1}^{N}
\left\{\widetilde{\mathbf{X}}_{i}\otimes \widetilde{\mathbf{B}}(\bs{z}_{j})\right\}\left\{\widetilde{\mathbf{X}}_{i}\otimes \widetilde{\mathbf{B}}(\bs{z}_{j^{\prime}})\right\}^{\top}
G_{\eta}(\bs{z}_j,\bs{z}_{j^{\prime}})\bs{\Gamma}_{n,\rho}^{-1}\mathbf{a}\\
=& E\mathbf{a}^{\top}\bs{\Gamma}_{n,\rho}^{-1}\bs{\Psi}_{\eta}\bs{\Gamma}_{n,\rho}^{-1}\mathbf{a},\\
	\mathbf{a}^{\top} E(\widehat{\bs{\theta}}_{\varepsilon}\widehat{\bs{\theta}}_{\varepsilon}^{\top})\mathbf{a}
	=& E\mathbf{a}^{\top}\bs{\Gamma}_{n,\rho}^{-1}\frac{1}{n^2N^2}
	\sum_{i=1}^{n}\sum_{j=1}^{N}
	\left\{\widetilde{\mathbf{X}}_{i}\otimes \widetilde{\mathbf{B}}(\bs{z}_{j})\right\}\left\{\widetilde{\mathbf{X}}_{i}\otimes \widetilde{\mathbf{B}}(\bs{z}_{j})\right\}^{\top}
	\sigma^2(\bs{z}_j)\bs{\Gamma}_{n,\rho}^{-1}\mathbf{a} \\
	=& E\mathbf{a}^{\top}\bs{\Gamma}_{n,\rho}^{-1}\bs{\Psi}_{\varepsilon}\bs{\Gamma}_{n,\rho}^{-1}\mathbf{a}.
\end{align*}

Note that for any vector $\mathbf{a}$ with $\mathbf{a}^{\top}\mathbf{a}=1$, we can rewrite as $\mathbf{a}^{\top}(\widehat{\bs{\theta}}_{\eta}+\widehat{\bs{\theta}}_{\varepsilon})=\sum_{i=1}^n a_i^{\eta+\varepsilon} \mathfrak{z}_i$, where
\begin{align*}
(a_i^{\eta+\varepsilon})^2=&\mathbf{a}^{\top}\bs{\Gamma}_{n,\rho}^{-1}\frac{1}{n^2N^2}
\sum_{j=1}^{N}\sum_{j^{\prime}=1}^{N}
\left\{\widetilde{\mathbf{X}}_{i}\otimes \widetilde{\mathbf{B}}(\bs{z}_{j})\right\}\left\{\widetilde{\mathbf{X}}_{i}\otimes \widetilde{\mathbf{B}}(\bs{z}_{j^{\prime}})\right\}^{\top}
G_{\eta}(\bs{z}_j,\bs{z}_{j^{\prime}})\bs{\Gamma}_{n,\rho}^{-1}\mathbf{a}\\
&+\mathbf{a}^{\top}\bs{\Gamma}_{n,\rho}^{-1}\frac{1}{n^2N^2}\sum_{j=1}^{N}
\left\{\widetilde{\mathbf{X}}_{i}\otimes \widetilde{\mathbf{B}}(\bs{z}_{j})\right\}\left\{\widetilde{\mathbf{X}}_{i}\otimes \widetilde{\mathbf{B}}(\bs{z}_{j})\right\}^{\top}
\sigma^2(\bs{z}_j)\bs{\Gamma}_{n,\rho}^{-1}\mathbf{a}
=(a^{\eta}_i)^2+(a_i^{\varepsilon})^2,
\end{align*}
and conditional on $\{\widetilde{\mathbf{X}}_{i},i=1,\ldots,n\}$, $\mathfrak{z}_i$ are independent with mean zero and variance one. Thus, $\sum_{i=1}^n (a_i^{\eta})^2=\mathbf{a}^{\top}\bs{\Gamma}_{n,\rho}^{-1}\bs{\Psi}_{\eta}\bs{\Gamma}_{n,\rho}^{-1}\mathbf{a} $ and $\sum_{i=1}^n (a_i^{\varepsilon})^2=\mathbf{a}^{\top}\bs{\Gamma}_{n,\rho}^{-1}\bs{\Psi}_{\varepsilon}\bs{\Gamma}_{n,\rho}^{-1}\mathbf{a}$.

According to Lemma \ref{LEM:Gamma_rho}, Assumptions (A2) and (A4),
\[
 E\mathbf{a}^{\top}\bs{\Gamma}_{n,\rho}^{-1}\bs{\Psi}_{\eta}\bs{\Gamma}_{n,\rho}^{-1}\mathbf{a}
\geq c_{\Gamma}^{-2}\left(|\triangle|^{2}+\frac{\rho_{n}}{nN|\triangle|^{2}}\right)^{-2}   E\mathbf{a}^{\top}\bs{\Psi}_{\eta}\mathbf{a},
\]
where
\begin{align*}
\mathbf{a}^{\top}\bs{\Psi}_{\eta}\mathbf{a} &=
\frac{1}{n^2N^2} \sum_{\ell,\ell^{\prime}=0}^p\sum_{i=1}^n \sum_{j=1}^N\sum_{j^{\prime}=1}^N X_{i\ell}X_{i\ell^{\prime}}\bs{a}_{\ell}^{\top} \widetilde{\mathbf{B}}(\bs{z}_j)\widetilde{\mathbf{B}}(\bs{z}_{j^{\prime}})^{\top} \bs{a}_{\ell^{\prime}} G_{\eta}(\bs{z}_j,\bs{z}_{j^{\prime}})\\
&=
\frac{1}{n^2} \sum_{k=1}^\infty \sum_{i=1}^n \left\{\frac{1}{N}\sum_{\ell=0}^p \sum_{j=1}^N \lambda_k^{1/2} X_{i\ell} g_{\ell}(\bs{z}_j)\psi_{k}(\bs{z}_j)\right\}^2
\end{align*}
with $g_{\ell}(\bs{z})=\bs{a}_{\ell}^{\top} \widetilde{\mathbf{B}}(\bs{z})$.
Therefore, by Assumption (A3), we have
\begin{align*}
 E\mathbf{a}^{\top}\bs{\Psi}_{\eta}\mathbf{a}
&=\frac{1}{n}
\sum_{k=1}^{\infty}\sum_{\ell=0}^p \left\{\frac{1}{N} \sum_{j=1}^N \lambda_k^{1/2}g_{\ell}(\bs{z}_j)\psi_{k}(\bs{z}_j)\right\}^2\\
&\geq \frac{c}{nN^2}
\sum_{\ell=0}^p \sum_{j=1}^N \sum_{j^{\prime}=1}^N
g_{\ell}(\bs{z}_j) g_{\ell}(\bs{z}_{j^{\prime}})G_{\eta}(\bs{z}_j,\bs{z}_{j^{\prime}})\\
& \asymp \frac{1}{n} \sum_{\ell=0}^p \int_{\Omega^2} g_{\ell}(\bs{z})g_{\ell}(\bs{z}')
G_{\eta}(\bs{z}, \bs{z}') d\bs{z}d\bs{z}'.
\end{align*}
Noting that the eigenvalues of $G_{\eta}$ are strictly positive,  we have
\[
E\mathbf{a}^{\top}\bs{\Psi}_{\eta}\mathbf{a} \geq c_{1}n^{-1}\sum_{\ell=0}^p
 \int_{\Omega} g_{\ell}^2(\bs{z}) d\bs{z}  \geq c_2 n^{-1}|\triangle|^{2}\|\mathbf{a}\|^2.
\] 
Therefore, we have
$ E\mathbf{a}^{\top}\bs{\Gamma}_{n,\rho}^{-1}\bs{\Psi}_{\eta}\bs{\Gamma}_{n,\rho}^{-1}\mathbf{a}
\geq cn^{-1} \left(1+\frac{\rho_{n}}{nN|\triangle|^{4}}\right)^{-2}|\triangle|^{-2}$.
Similarly, one can show that
$
 E\mathbf{a}^{\top}\bs{\Gamma}_{n,\rho}^{-1}\bs{\Psi}_{\varepsilon}\bs{\Gamma}_{n,\rho}^{-1}\mathbf{a}
\geq c(nN)^{-1} \left(1+\frac{\rho_{n}}{nN|\triangle|^{4}}\right)^{-2} |\triangle|^{-2}$.
In addition,
\begin{align*}
\max (a_i^{\eta})^2 &\leq
\frac{C}{|\triangle|^4}\mathbf{a}^{\top}\frac{1}{n^2N^2}\sum_{j=1}^{N}\sum_{j^{\prime}=1}^{N}
\left\{(\widetilde{\mathbf{X}}_{i}\widetilde{\mathbf{X}}_{i}^{\top})\otimes \widetilde{\mathbf{B}}(\bs{z}_{j}) \widetilde{\mathbf{B}}(\bs{z}_{j^{\prime}})^{\top}\right\}G_{\eta}(\bs{z}_j,\bs{z}_{j^{\prime}})\mathbf{a} \\
& \leq \frac{C}{|\triangle|^4} \sum_{\ell,\ell^{\prime}=1}^p\frac{1}{n^2N^2} \max_i |X_{i\ell}X_{i\ell^{\prime}}| \sum_{j=1}^{N}\sum_{j^{\prime}=1}^{N} g_{\ell}(\bs{z}_j)g_{\ell'}(\bs{z}_{j^{\prime}})  G_{\eta}(\bs{z}_j,\bs{z}_{j^{\prime}}) \leq Cn^{-2}|\triangle|^{-2},\\
\max (a_i^{\varepsilon})^2 &\leq
\frac{C}{|\triangle|^4}\mathbf{a}^{\top}\frac{1}{n^2N^2}\sum_{j=1}^{N}
\left\{(\widetilde{\mathbf{X}}_{i}\widetilde{\mathbf{X}}_{i}^{\top})\otimes \widetilde{\mathbf{B}}(\bs{z}_{j}) \widetilde{\mathbf{B}}(\bs{z}_{j})^{\top}\right\}\sigma^2(\bs{z}_j)\mathbf{a} \\
& \leq \frac{C}{|\triangle|^4} \sum_{\ell,\ell^{\prime}=1}^p\frac{1}{n^2N^2} \max_i|X_{i\ell}X_{i\ell^{\prime}}|\sum_{j=1}^{N} g_{\ell}(\bs{z}_j)g_{\ell'}(\bs{z}_{j^{\prime}}) \sigma^2(\bs{z}_j)  \leq Cn^{-2}N^{-1}|\triangle|^{-2}.
\end{align*}

Thus, if $\rho_n n^{-1}N^{-1}|\triangle|^{-4}\rightarrow 0$, we have
\[
\frac{\max_{1\leq i\leq n}(a_i^{\eta}+a_i^{\varepsilon})^2}{\sum_{i=1}^n (a_i^{\eta}+a_i^{\varepsilon})^2} \leq Cn^{-1}\left(1+\frac{\rho_n}{nN|\triangle|^4}\right)^{2} \to 0,
\]
which satisfies the Lindeberg condition.
\end{proof}

\begin{theorem}
\label{THM:variance-bias}
Under Assumptions (A1)--(A6), if for any $\ell=0,1,\ldots, p$, $|X_{i\ell}| \leq C_{\ell} < \infty$,
$\sup_{\bs{z}\in \Omega} [\mathrm{Var}\{\widehat{\beta}_{\ell}(\bs{z})\}]^{-1/2}(\widehat{\beta}_{\mu,\ell}(\bs{z})
-\beta_{\ell}^{o}(\bs{z}))=o_{P}(1)$, for $\ell=0,\ldots, p$.
\end{theorem}
\begin{proof}
Using similar arguments as in the proof of Lemma \ref{LEM:normality} and the result of Lemma \ref{LEM:Gamma_rho},  we have for any $\|\mathbf{a}\|=1$, $ E\mathbf{a}^{\top}\bs{\Gamma}_{n,\rho}^{-1}\bs{\Psi}_{\eta}\bs{\Gamma}_{n,\rho}^{-1}\mathbf{a}
\leq C_{\Gamma}^{-2}|\triangle|^{-4}  E\mathbf{a}^{\top}\bs{\Psi}_{\eta}\mathbf{a}
\leq Cn^{-1} |\triangle|^{-2}$, and $
 E\mathbf{a}^{\top}\bs{\Gamma}_{n,\rho}^{-1}\bs{\Psi}_{\varepsilon}\bs{\Gamma}_{n,\rho}^{-1}\mathbf{a}
\leq C_{\Gamma}^{-2}|\triangle|^{-4}  E\mathbf{a}^{\top}\bs{\Psi}_{\varepsilon}\mathbf{a}
\leq C(nN)^{-1} |\triangle|^{-2}$. Therefore, based on the proof of Lemma \ref{LEM:normality}, for any $\|\mathbf{a}\|=1$,
\begin{align*}
cn^{-1} |\triangle|^{-4} \left(1+\frac{\rho_{n}}{nN|\triangle|^{4}}\right)^{-2} \leq  E\mathbf{a}^{\top}\bs{\Gamma}_{n,\rho}^{-1}\bs{\Psi}_{\eta}\bs{\Gamma}_{n,\rho}^{-1}\mathbf{a}
&\leq  Cn^{-1}|\triangle|^{-2},\\
c(nN)^{-1} \left(1+\frac{\rho_{n}}{nN|\triangle|^{4}}\right)^{-2} |\triangle|^{-2} \leq  E\mathbf{a}^{\top}\bs{\Gamma}_{n,\rho}^{-1}\bs{\Psi}_{\varepsilon}\bs{\Gamma}_{n,\rho}^{-1}\mathbf{a}
&\leq  C(nN)^{-1} |\triangle|^{-2}.
\end{align*}
Thus,
\begin{align*}
&\mathrm{Var}(\widehat{\beta}_{\ell})=\{\mathbf{e}_{\ell} \otimes \widetilde{\mathbf{B}}(\bs{z})\}^{\top}
 E\{\bs{\Gamma}_{n,\rho}^{-1}(\bs{\Psi}_{\eta}+\bs{\Psi}_{\varepsilon})\bs{\Gamma}_{n,\rho}^{-1}\} \{\mathbf{e}_{\ell} \otimes \widetilde{\mathbf{B}}(\bs{z})\} \label{EQ:var_eta_order}\\
&\asymp \{\mathbf{e}_{\ell} \otimes \widetilde{\mathbf{B}}(\bs{z})\}^{\top}
 E(\bs{\Gamma}_{n,\rho}^{-1}\bs{\Psi}_{\eta}\bs{\Gamma}_{n,\rho}^{-1}) \{\mathbf{e}_{\ell} \otimes \widetilde{\mathbf{B}}(\bs{z})\}\asymp \{\mathbf{e}_{\ell} \otimes \widetilde{\mathbf{B}}(\bs{z})\}^{\top} \notag\\
&\quad \times
 E\left[\bs{\Gamma}_{n,\rho}^{-1}\frac{1}{n^2N^2}\sum_{i=1}^{n}\sum_{j,j^{\prime}=1}^{N}
\left\{\widetilde{\mathbf{X}}_{i}\otimes \widetilde{\mathbf{B}}(\bs{z}_{j})\right\}\left\{\widetilde{\mathbf{X}}_{i}\otimes \widetilde{\mathbf{B}}(\bs{z}_{j^{\prime}})\right\}^{\top}
G_{\eta}(\bs{z}_j,\bs{z}_{j^{\prime}})\bs{\Gamma}_{n,\rho}^{-1}\right] 
 \{\mathbf{e}_{\ell} \otimes \widetilde{\mathbf{B}}(\bs{z})\}.\notag
\end{align*}
By Lemma \ref{LEM:Gamma_rho}, we have
\begin{align*}
\mathrm{Var}(\widehat{\beta}_{\ell}) \lesssim & \frac{1}{nN^2|\triangle|^{2}}  \sum_{j=1}^N\sum_{j^{\prime}=1}^N \{\mathbf{e}_{\ell} \otimes \widetilde{\mathbf{B}}(\bs{z})\}^{\top} \widetilde{\mathbf{B}}(\bs{z}_j)\widetilde{\mathbf{B}}(\bs{z}_{j^{\prime}})^{\top}\{\mathbf{e}_{\ell} \otimes \widetilde{\mathbf{B}}(\bs{z})\} G_{\eta}(\bs{z}_j,\bs{z}_{j^{\prime}}),\\
\mathrm{Var}(\widehat{\beta}_{\ell}) \gtrsim & \frac{1}{nN^2|\triangle|^{2}}  \sum_{j=1}^N\sum_{j^{\prime}=1}^N \{\mathbf{e}_{\ell} \otimes \widetilde{\mathbf{B}}(\bs{z})\}^{\top} \widetilde{\mathbf{B}}(\bs{z}_j)\widetilde{\mathbf{B}}(\bs{z}_{j^{\prime}})^{\top}\{\mathbf{e}_{\ell} \otimes \widetilde{\mathbf{B}}(\bs{z})\} G_{\eta}(\bs{z}_j,\bs{z}_{j^{\prime}}) \\
& \times \left(1+\frac{\rho_{n}}{nN|\triangle|^{4}}\right)^{-2},
\end{align*}
and according to Lemmas \ref{LEM:normequity} and \ref{LEM:integration}, we have
$c n^{-1}\left(1+\frac{\rho_{n}}{nN|\triangle|^{4}}\right)^{-2}
\leq \mathrm{Var}(\widehat{\beta}_{\ell})\leq Cn^{-1}$.
According to Lemma \ref{LEM:uniformbiasrate}, if $\rho_n n^{-1/2}N^{-1}|\triangle|^{-3}\rightarrow 0$ and $n^{1/2}|\triangle|^{d +1}\rightarrow 0$, the bias term in (\ref{EQ:decompose3})
is negligible compared to the order of $[\mathrm{Var}\{\widehat{\beta}_{\ell}(\bs{z})\}]^{1/2}$.
\end{proof}

\textit{Proof of Theorem \ref{THM:beta-normality}.} 
Theorem \ref{THM:beta-normality} follows from (\ref{EQ:decompose3}), Lemma
\ref{LEM:normality} and Theorem \ref{THM:variance-bias}. $\square$

\vskip .10in \noindent \textbf{A.4. Asymptotic properties of piecewise constant spline estimators} \vskip .10in

In this section, we study the asymptotic properties of the piecewise constant spline estimators defined in the spline space $\mathcal{PC}(\triangle)$. Define piecewise constant bivariate spline functions
\begin{equation}
\widehat{\bs{\beta}}_{\mu}^{\mathrm{c}}(\bs{z})=(\widehat{\beta}^{\mathrm{c}}_{\mu,0}(\bs{z}),\ldots,\widehat{\beta}^{\mathrm{c}}_{\mu, p}(\bs{z}))^{\top}=\widehat{\mathbf{V}}_{m(\bs{z})}^{-1}\left\{\frac{1}{nN}
\sum_{i=1}^{n}\sum_{j=1}^{N}B_{m(\bs{z})}(\bs{z}_{j})X_{i\ell}
\sum_{\ell^{\prime}=0}^{p}\beta_{\ell^{\prime}}^{o}(\bs{z}_{j})X_{i\ell^{\prime}}\right\}_{\ell =0}^{p},
\label{EQ:beta-mu-pc}
\end{equation}
\begin{align}
\widehat{\bs{\eta}}(\bs{z})&=(\widehat{\eta}_{0}(\bs{z}),\ldots,\widehat{\eta}_{p}(\bs{z}))^{\top}
=\widehat{\mathbf{V}}_{m(\bs{z})}^{-1}\left\{\frac{1}{nN}\sum_{i=1}^{n}
\sum_{j=1}^{N}B_{m(\bs{z})}(\bs{z}_{j})X_{i\ell}\sum_{k=1}^{\infty}\xi_{ik}\psi_{k}(\bs{z}_{j})\right\}_{\ell =0}^{p},
\label{EQ:eta-hat-pc}\\
\widehat{\bs{\varepsilon}}(\bs{z})&=(\widehat{\varepsilon}_{0}(\bs{z}),\ldots,\widehat{\varepsilon}_{p}(\bs{z}))^{\top}=\widehat{\mathbf{V}}_{m(\bs{z})}^{-1}\left\{\frac{1}{nN}\sum_{i=1}^{n}\sum_{j=1}^{N}B_{m(\bs{z})}(\bs{z}_{j})X_{i\ell}\varepsilon_{ij}\right\}_{\ell =0}^{p},
\label{EQ:eps-hat-pc}
\end{align}
where $\widehat{\mathbf{V}}_{m(\bs{z})}$ is defined in (\ref{DEF:V-m}).

The next two theorems concern the functions $\widehat{\beta}_{\mu,\ell}^{\mathrm{c}}(\bs{z})$, $\widehat{\eta}_{\ell}(\bs{z}),$ $\widehat{\varepsilon}_{\ell}(\bs{z})$, $\ell=0,\ldots,p$, given in (\ref{EQ:beta-mu-pc}), (\ref{EQ:eta-hat-pc}) and (\ref{EQ:eps-hat-pc}). Theorem \ref{THM:uniformbiasrate2} gives the uniform convergence rate of $\widehat{\beta}_{\mu,\ell}(\bs{z})$ to $\beta^{o}_{\ell}(\bs{z})$. 

\begin{theorem}
\label{THM:uniformbiasrate2}
Under Assumptions  (A1$'$), (A2)--(A6), the constant spline functions $\widehat{\beta}^{\mathrm{c}}_{\mu,\ell}(\bs{z})$, $\ell=0,\ldots,p$, satisfy
$\sup_{\bs{z}\in \Omega}\sup_{0\leq \ell\leq p}\left\vert \widehat{\beta}^{\mathrm{c}}_{\mu,\ell}(\bs{z})-\beta_{\ell}^{o}(\bs{z})\right\vert =O_{P}(|\triangle|)$.
\end{theorem}

In the following, we provide detailed proofs of Theorems \ref{THM:uniformbiasrate2}. For the random matrix $\widehat{\mathbf{V}}_{m}$ defined in (\ref{DEF:V-m}), the lemma below shows that its inverse can be approximated by the inverse of a deterministic matrix $A_{m}^{-1}\bs{\Sigma}_{X}^{-1}$, where $A_m=\int_{\Omega} B_{m}(\bs{z})d\bs{z}$.

\begin{lemma}
\label{LEM:matrix appro.} Under Assumptions (A3) and (A5), for any $m\in \mathcal{M}$, we have
\begin{equation}
	\widehat{\mathbf{V}}_{m}^{-1}=A_{m}^{-1}\bs{\Sigma}_{X}^{-1}+O_{P}\left\{n^{-1/2}|\triangle|^{2}(\log n)^{1/2}+N^{-1/2}|\triangle|)\right\}.
\label{EQ:matrix approx}
\end{equation}
\end{lemma}
\begin{proof}
By Lemma \ref{LEM:inner product}, 
$\left\Vert \widehat{\mathbf{V}}_{m}-A_{m}\bs{\Sigma}_{X}\right\Vert_{\infty}=O_{P}\left\{n^{-1/2}|\triangle|^{2}(\log n)^{1/2}+N^{-1/2}|\triangle|)\right\}$.
Using the fact that for any matrices $\mathbf{A}$ and $\mathbf{B}$,
$\left(\mathbf{A}+\delta\mathbf{B}\right)^{-1}=\mathbf{A}^{-1}-\delta\mathbf{A}^{-1}\mathbf{B}\mathbf{A}^{-1}+O(\delta^{2})$, we obtain (\ref{EQ:matrix approx}).
\end{proof}

\textit{Proof of Theorem \ref{THM:uniformbiasrate2}.} According to Lemma \ref{LEM:appord}, there exist functions $\beta_{\ell}^{\ast}\in\mathcal{PC}(\triangle)$ that satisfies $\left\Vert
\beta_{\ell}^{\ast}-\beta_{\ell}^{o}\right\Vert_{\infty}=O(|\triangle|) $
for $\ell=0,1,\ldots,p$. By the definition of $\widehat{\beta}_{\mu,\ell}(\bs{z}) $ in (\ref{EQ:beta-mu-pc}),
$\widehat{\bs{\beta}}^{\mathrm{c}}_{\mu}(\bs{z})=\left( \widehat{\beta}_{\mu,0}^{\mathrm{c}}(\bs{z}),\widehat{\beta}^{\mathrm{c}}_{\mu,1}(\bs{z}),\ldots,\widehat{\beta}%
^{\mathrm{c}}_{\mu,p}(\bs{z})\right)^{\top}=\left( \widetilde{\gamma}_{m(\bs{z}),0},\ldots,\widetilde{\gamma}_{m(\bs{z}),p}\right) ^{\top}=\widetilde{\bs{\gamma}}_{m(\bs{z})}$,
where $\widetilde{\bs{\gamma}}_{m}=\widehat{\mathbf{V}}_{m}^{-1}\left\{(nN)^{-1}
\sum_{i=1}^{n}\sum_{j=1}^{N}B_{m}(\bs{z}_{j})X_{i\ell}\sum
_{\ell^{\prime}=0}^{p}\beta_{\ell^{\prime}}^{o}(\bs{z}_{j})X_{i\ell^{\prime}}\right\}
_{\ell =0}^{p}$ for $\widehat{\mathbf{V}}_{m}$ defined in (\ref{DEF:V-m}).

Let
\[
\widetilde{\bs{\beta}}(\bs{z})=(\widetilde{\beta}_{0}(\bs{z}),\widetilde{\beta}_{1}(\bs{z}),\ldots,\widetilde{\beta}
_{p}(\bs{z}))^{\top}=\widehat{\mathbf{V}}_{m(\bs{z})}^{-1}\left[ \frac{1}{nN}%
\sum_{i=1}^{n}\sum_{j=1}^{N}B_{m(\bs{z})}(\bs{z}_{j})X_{i\ell}\sum_{\ell^{\prime}=0}^{p} \beta_{\ell^{\prime}}^{\ast}(\bs{z}_{j}) X_{i\ell^{\prime}}\right]_{\ell =0}^{p},
\]
 then
\begin{equation*}
\widehat{\bs{\beta}}^{\mathrm{c}}_{\mu}(\bs{z}) -\widetilde{\bs{\beta}}(\bs{z})=\widehat{\mathbf{V}}_{m(\bs{z})}^{-1}
\left[\frac{1}{nN}\sum_{i=1}^{n}\sum_{j=1}^{N}B_{m(\bs{z})}(\bs{z}_{j})X_{i\ell}\sum_{\ell^{\prime}=0}^{p}\left\{ \beta_{\ell^{\prime}}^{o}(\bs{z}_{j})-\beta_{\ell^{\prime}}^{\ast}(\bs{z}_{j})\right\}
X_{i\ell^{\prime}}\right]_{\ell =0}^{p}.
\end{equation*}
Observing that $\widetilde{\beta}_{\ell}\equiv \beta_{\ell}^{\ast}$ as $\beta_{\ell}^{\ast}\in \mathcal{PC}(\triangle)$, $\widehat{\beta}_{\mu,\ell}^{\mathrm{c}}(\bs{z}) =\widehat{\beta}_{\mu,\ell}^{\mathrm{c}}(\bs{z}) -\widetilde{\beta}_{\ell}(\bs{z}) +\beta_{\ell}^{\ast}(\bs{z})$, $\ell=0,1,\ldots,p$.

It is easy to see $\Vert \widehat{\beta}_{\mu,\ell}^{\mathrm{c}}-\widetilde{\beta}_{\ell}\Vert
_{\infty}=O_{P}(|\triangle|) $. Hence, for $\ell=0,1,\ldots ,p$,
$\Vert \widehat{\beta}_{\mu,\ell}^{\mathrm{c}}-\beta_{\ell}^{o}\Vert_{\infty}\leq \Vert \widehat{\beta}_{\mu,\ell}^{\mathrm{c}}-\widetilde{\beta}_{\ell}\Vert_{\infty}+\Vert \beta_{\ell}^{o}-\beta_{\ell}^{\ast}\Vert
_{\infty}=O_{P}(|\triangle|)$, which completes the proof. \hspace{5.5cm}$\square$

By Lemma \ref{LEM:matrix appro.}, the inverse of the random matrix $\widehat{\mathbf{V}}_{m}$ can be approximated by that of a deterministic matrix $A_{m}\bs{\Sigma}_{X}$.
Substituting $\widehat{\mathbf{V}}_{m}$ with $A_{m}\bs{\Sigma}_{X}$ in (\ref{EQ:eta-hat-pc}) and (\ref{EQ:eps-hat-pc}), we define the random vectors
\begin{align}
\widehat{\bs{\eta}}^{\ast}(\bs{z}) &=(\widehat{\eta}^{\ast}_{0}(\bs{z}),\ldots,\widehat{\eta}^{\ast}_{p}(\bs{z}))^{\top}=A_{m(\bs{z})}^{-1}\bs{\Sigma}_{X}^{-1}\left\{\frac{1}{nN}\sum_{i=1}^{n}\sum_{j=1}^{N}B_{m(\bs{z})}(\bs{z}_{j})X_{i\ell}%
\sum_{k=1}^{\infty}\xi_{ik}\psi_{k}(\bs{z}_{j}) \right\}_{\ell =0}^{p},  \label{EQ:eta-star} \\
\widehat{\bs{\varepsilon}}^{\ast}(\bs{z})
&=(\widehat{\varepsilon}^{\ast}_{0}(\bs{z}),\ldots,\widehat{\varepsilon}^{\ast}_{p}(\bs{z}))^{\top}=A_{m(\bs{z})}^{-1}\bs{\Sigma}_{X}^{-1}\left\{\frac{1}{nN}\sum_{i=1}^{n}\sum
\limits_{j=1}^{N}B_{m(\bs{z})}(\bs{z}_{j})X_{i\ell}
\varepsilon_{ij}\right\}_{\ell =0}^{p}.
\label{EQ:eps-star}
\end{align}

The next lemma implies that the difference between $\widehat{\bs{\eta}}^{\ast}(\bs{z})$ and $\widehat{\bs{\eta}}(\bs{z})$ and the difference between $\widehat{\bs{\varepsilon}}^{\ast}(\bs{z})$ and $\widehat{\bs{\varepsilon}}(\bs{z})$ are both negligible uniformly over $\bs{z}\in \Omega$.

\begin{lemma}
\label{LEM:xietilda-hat}
Under Assumptions (A2)--(A5) and (C1), if $N^{1/2}|\triangle|\rightarrow \infty$ as $N\rightarrow \infty$, $\|\sum_{k=1}^{\infty}\lambda_{k}^{1/2}\psi_{k}\|_{\infty} < \infty$ and $n^{1/(4+\delta_2)} \ll n^{1/2}N^{-1/2}|\triangle|^{-1}$ for some $\delta_2$, then for $\widehat{\bs{\eta}}(\bs{z})$, $\widehat{\bs{\varepsilon}}(\bs{z})$ given in (\ref{EQ:eta-hat-pc}), (\ref{EQ:eps-hat-pc}) and $\widehat{\bs{\eta}}^{\ast}(\bs{z})$, $\widehat{\bs{\varepsilon}}^{\ast}(\bs{z})$ given in (\ref{EQ:eta-star}), (\ref{EQ:eps-star}), as $N\rightarrow\infty$ and $n\rightarrow\infty$, we have
\begin{align}
\sup_{\bs{z}\in \Omega}\left\Vert \widehat{\bs{\eta}}(\bs{z}) -\widehat{\bs{\eta}}^{\ast}(\bs{z}) \right\Vert_{\infty}
&=O_{P}\left\{n^{-1}|\triangle|^{4}\log (n)+n^{-1/2}N^{-1/2}|\triangle|^{3} (\log n)^{1/2}\right\},
\label{EQ:xi-diff} \\
\sup_{\bs{z}\in \Omega}\left\Vert \widehat{\bs{\varepsilon}}(\bs{z}) -\widehat{\bs{\varepsilon}}^{\ast}(\bs{z})
\right\Vert_{\infty} &=O_{P}\left\{n^{-1}N^{-1/2}|\triangle|^{3}\log (n)+n^{-1/2}N^{-1}|\triangle|^{2} (\log n)^{1/2}\right\}.  
\label{EQ:e-diff}
\end{align}
\end{lemma}

\begin{proof}
Comparing $\widehat{\bs{\eta}}(\bs{z})$
and $\widehat{\bs{\eta}}^{\ast}(\bs{z})$ given in (\ref{EQ:eta-hat-pc}) and (\ref{EQ:eta-star}), we have
\[
\widehat{\bs{\eta}}(\bs{z}) -\widehat{\bs{\eta}}^{\ast}(\bs{z}) =\left\{\widehat{\mathbf{V}}_{m}^{-1}-A_{m(\bs{z})}^{-1}\bs{\Sigma}_{X}^{-1}\right\}\left\{\frac{1%
}{nN}\sum_{i=1}^{n}\sum_{j=1}^{N}B_{m(\bs{z})}(\bs{z}_{j})X_{i\ell}%
\sum_{k=1}^{\infty}\xi_{ik}\psi_{k}(\bs{z}_{j}) \right\}_{\ell =0}^{p}.
\]
Now let $\zeta_{i,m,\ell}\equiv \zeta_{i}=n^{-1}\left[
X_{i\ell}\sum_{k=1}^{\infty}\left\{\frac{1}{N}\sum_{j=1}^{N}B_{m}(\bs{z}_{j})\psi_{k}(\bs{z}_{j}) \right\} \xi_{ik}\right] $, then it is easy to see that $
\frac{1}{nN}\sum_{i=1}^{n}\sum_{j=1}^{N}B_{m}(\bs{z}_{j})X_{i\ell}%
\sum_{k=1}^{\infty}\xi_{ik}\psi_{k}(\bs{z}_{j})=%
\frac{1}{N}\sum_{i=1}^{n}\zeta_{i,m,\ell}.
$
It is easy to see that $E (\zeta_{i})=0$, and
\begin{align*}
\sigma_{\zeta_{i},n}^{2} &=E\left(\zeta_{i}^{2}\right) 
=n^{-2}E(X_{\ell}^{2})\int_{T_m\times T_m}G_{\eta}\left(\bs{u},\bs{v}\right)d\bs{u}d\bs{v}
\{1+O(N^{-1/2}|\triangle|^{-1})\}.
\end{align*}

Note that $\left\{\sigma_{\zeta_{i},n}^{-1}\zeta_{i}\right\}_{i=1}^{n}$ are uncorrelated random variables with mean 0. Assume that $|\triangle|^{-2}\asymp n^{\tau}$ for some $0<\tau<\infty$, we can show that for any large enough $\delta >0$,
$P\left[ \left\vert \sum_{i=1}^{n}\zeta_{i}\right\vert \geq \delta
\left\{C\log (n)n^{-1} |\triangle|^{4}E(X_{i\ell}^{2})\right\}^{1/2}\right] \leq 2n^{-2-\tau}$. Therefore,
\begin{equation*}
\sum_{n=1}^{\infty}P\left\{\sup_{m\in \mathcal{M},0\leq \ell\leq p}
\left\vert \sum_{i=1}^{n}\zeta_{i,m,\ell}\right\vert \geq \delta  n^{-1/2}|\triangle|^{2}(\log n)^{1/2}\right\} <\infty .
\end{equation*}%
Thus, $\sup_{m,\ell}\left\vert \sum_{i=1}^{n}\zeta_{i,m,\ell}\right\vert =O_{P}\left\{n^{-1/2}|\triangle|^{2}(\log n)^{1/2}\right\} $ as $n\rightarrow\infty$ by Borel-Cantelli Lemma. It follows that $\sup_{m,\ell}\left\vert n^{-1}\sum_{i=1}^{n}\zeta_{i,m,\ell}\right\vert =O_{P}\left\{n^{-1/2}|\triangle|^{2}(\log n)^{1/2}\right\} $. Finally, according to (\ref{EQ:matrix approx}), we obtain (\ref{EQ:xi-diff}). The result in (\ref{EQ:e-diff}) can be proved similarly.
\end{proof}

\begin{lemma}
\label{LEM:approx-cov}
For any $\bs{z}\in \Omega$, the covariance matrices
of $\widehat{\bs{\eta}}^{\ast}(\bs{z})$ and $\widehat{\bs{\varepsilon}}^{\ast}(\bs{z}) $ are
\begin{align*}
\bs{\Sigma}_{\eta}(\bs{z})&=E\left\{\widehat{\bs{\eta}}^{\ast}(\bs{z}) \widehat{\bs{\eta}}^{\ast\top}(\bs{z}) \right\} =
A_{m(\bs{z})}^{-2}\bs{\Sigma}_{X}^{-1}\frac{1}{nN^2}\sum_{k=1}^{\infty}
\lambda_{k}\left\{\sum_{j=1}^{N}B_{m(\bs{z})}(\bs{z}_{j})\psi_{k}(\bs{z}_{j}) \right\}^{2},\\
\bs{\Sigma}_{\varepsilon}(\bs{z})&=E\left\{\widehat{\bs{\varepsilon}}^{\ast}(\bs{z})
\widehat{\bs{\varepsilon}}^{\ast \top}(\bs{z}) \right\} =A_{m(\bs{z})}^{-2}\bs{\Sigma}_{X}^{-1}
\frac{1}{nN^2}\sum_{j=1}^{N}B_{m(\bs{z})}^{2}(\bs{z}_{j})\sigma^{2}(\bs{z}_{j}),
\end{align*}
in addition,
\begin{equation}
\sup_{\bs{z}\in \Omega}\left\Vert \bs{\Sigma}_{\eta}(\bs{z})+%
\bs{\Sigma}_{\varepsilon}(\bs{z})-\bs{\Sigma}_{n}(\bs{z})\right\Vert_{\infty}=O(n^{-1}N^{-1/2}|\triangle|^{-1}),
\label{EQ:app. cov}
\end{equation}
where $\bs{\Sigma}_{n}(\bs{z})$ is given in (\ref{DEF:Sigma(z)}).
\end{lemma}
\begin{proof}
Note that $A_{m(\bs{z})}^{2}\widehat{\bs{\eta}}^{\ast}(\bs{z}) \widehat{\bs{\eta}}^{\ast\top}(\bs{z})$ is equal to 
\begin{align*}
	\bs{\Sigma}_{X}^{-1}\left\{\frac{1%
	}{n^2N^2}\sum_{i=1}^{n}\sum_{j=1}^{N}B_{m(\bs{z})}(\bs{z}_{j})X_{i\ell}%
	\sum_{k=1}^{\infty}\xi_{ik}\psi_{k}(\bs{z}_{j})
	\sum_{i^{\prime}=1}^{n}\sum_{j^{\prime}=1}^{N}B_{m(\bs{z})}(\bs{z}_{j^{\prime}})X_{i^{\prime}\ell^{\prime}}
	\sum_{k^{\prime}=1}^{\infty}\xi_{i^{\prime}k^{\prime}}\psi_{k^{\prime}}(\bs{z}_{j^{\prime}})\right\}_{\ell,\ell^{\prime}=0}^{p}\bs{\Sigma}_{X}^{-1}.
\end{align*}%
Thus,%
\begin{align*}
\bs{\Sigma}_{\eta}(\bs{z})= E\left\{\widehat{\bs{\eta}}^{\ast}(\bs{z}) \widetilde{\bs{\eta}}^{\top}(\bs{z}) \right\} 
&=A_{m(\bs{z})}^{-2}\bs{\Sigma}_{X}^{-1}\frac{1}{nN^2}\sum_{k=1}^{\infty}
\lambda_{k}\left\{\sum_{j=1}^{N}B_{m(\bs{z})}(\bs{z}_{j})\psi_{k}(\bs{z}_{j}) \right\}^{2}.
\end{align*}%
Similarly, we can derive the covariance of $\widehat{\bs{\varepsilon}}^{\ast}(\bs{z})$: 
$\bs{\Sigma}_{\varepsilon}(\bs{z})
=A_{m(\bs{z})}^{-2}\bs{\Sigma}_{X}^{-1}\frac{1}{nN^2}\sum_{j=1}^{N}B_{m(\bs{z})}^{2}(\bs{z}_{j})\sigma^{2}(\bs{z}_{j})$. 
Observe that
\[
\sum_{k=1}^{\infty}\lambda_{k}\left\{\frac{1}{N}\sum_{j=1}^{N}B_{m(\bs{z})}(\bs{z}_{j})\psi_{k}(\bs{z}_{j}) \right\}^{2} =\frac{1}{N^2}\sum_{j=1}^{N}\sum_{j^{\prime}=1}^{N}G_{\eta}(\bs{z}_{j},\bs{z}_{j^{\prime}})
B_{m(\bs{z})}(\bs{z}_{j})B_{m(\bs{z})}(\bs{z}_{j^{\prime}}).
\]
Hence, by (\ref{EQ:G_integration}) and (\ref{EQ:sigma_integration}) in Lemma \ref{LEM:integration}, (\ref{EQ:app. cov}) holds.
Therefore,
\begin{align*}
\bs{\Sigma}_{\eta}(\bs{z})+ \bs{\Sigma}_{\varepsilon}(\bs{z})
=&(nA_{m(\bs{z})}^2)^{-1}\bs{\Sigma}_{X}^{-1}
 \int_{T_{m(\bs{z})}\times T_{m(\bs{z})}}G_{\eta}\left(\bs{u},\bs{v}\right)d\bs{u}d\bs{v} \{1+O(N^{-1/2}|\triangle|^{-1})\}\\
&+ (nNA_{m(\bs{z})}^2)^{-1} \bs{\Sigma}_{X}^{-1}\int_{T_{m(\bs{z})}}\sigma^{2}(\bs{u}) d\bs{u}\{1+O(N^{-1/2}|\triangle|^{-1})\}\\
=& n^{-1}\bs{\Sigma}_{X}^{-1}G_{\eta}\left(\bs{z},\bs{z}\right)\{1+O(N^{-1/2}|\triangle|^{-1})\}.
\end{align*}%
Therefore,
$\sup_{\bs{z}\in \Omega}\left\Vert \bs{\Sigma}_{\eta}(\bs{z})+ \bs{\Sigma}_{\varepsilon}(\bs{z})-n^{-1}\bs{\Sigma}_{X}^{-1}G_{\eta}
\left(\bs{z},\bs{z}\right)\right\Vert_{\infty}=O(n^{-1}N^{-1/2}|\triangle|^{-1})$.
The desired result in (\ref{EQ:app. cov}) follows.
\end{proof}

\textit{Proof of Theorem \protect\ref{THM:multinormal}.} 
Note that, for any vector $\mathbf{a}=\left(a_{0},\ldots,a_{p}\right)^{\top}\in \mathcal{R}^{(p+1)}$, we have $E\left[\sum_{\ell =0}^{p}a_{\ell}\left\{\widehat{\eta}_{\ell}^{\ast}(\bs{z}) +\widehat{\varepsilon}_{\ell}^{\ast}(\bs{z}) \right\} \right] =0$, and
\begin{align*}
\sum_{\ell=0}^{p}a_{\ell}\widehat{\eta}_{\ell}^{\ast}(\bs{z})&=\bs{a}^{\top}\frac{A_{m(\bs{z})}^{-1}
\bs{\Sigma}_{X}^{-1}}{nN}\sum_{i=1}^{n}\sum_{j=1}^{N}B_{m(\bs{z})}(\bs{z}_j)\sum_{k=1}^{\infty}
\xi_{ik}\psi_{k}(\bs{z}_j)\mathbf{X}_i
=\sum_{i=1}^{n} \bs{a}^{\top}A_{m(\bs{z})}^{-1}\bs{\Sigma}_{X}^{-1}\bs{\mathfrak{z}}_i^{\eta},\\
\sum_{\ell=0}^{p}a_{\ell}\widehat{\varepsilon}_{\ell}^{\ast}(\bs{z})
&=\bs{a}^{\top}\frac{A_{m(\bs{z})}^{-1}\bs{\Sigma}_{X}^{-1}}{nN}\sum_{i=1}^{n}\sum_{j=1}^{N}
B_{m(\bs{z})}(\bs{z}_j)\varepsilon_{ij}\mathbf{X}_i
=\sum_{i=1}^{n}\bs{a}^{\top}A_{m(\bs{z})}^{-1}\bs{\Sigma}_{X}^{-1} \bs{\mathfrak{z}}_i^{\varepsilon},
\end{align*}
where $\bs{\mathfrak{z}}_i^{\eta}=\frac{1}{nN}\sum_{j=1}^{N}B_{m(\bs{z})}(\bs{z}_j)\sum_{k=1}^{\infty}\xi_{ik}\psi_{k}(\bs{z}_j)\mathbf{X}_i$ and $\bs{\mathfrak{z}}_i^{\varepsilon}=\frac{1}{nN}\sum_{j=1}^{N}B_{m(\bs{z})}(\bs{z}_j)\varepsilon_{ij}\mathbf{X}_i$  are independent sequences with variances
$\mathrm{Var}(\bs{\mathfrak{z}}_i^{\eta})=\frac{1}{n^2N^2}\sum_{j,j^{\prime}}^{}B_{m(\bs{z})}(\bs{z}_j)B_{m(\bs{z})}(\bs{z}_j^{\prime})G_{\eta}(\bs{z}_j, \bs{z}_{j^{\prime}})\bs{\Sigma}_{X}$ and $
\mathrm{Var}(\bs{\mathfrak{z}}_i^{\varepsilon})=\frac{1}{n^2N^2}\sum_{j=1}^{N}B_{m(\bs{z})}(\bs{z}_j)B_{m(\bs{z})}(\bs{z}_j)\sigma^2(\bs{z}_j)\bs{\Sigma}_{X}$, respectively. Therefore, we have
\begin{align*}
\mathrm{Var}\left(\bs{a}^{\top}A_{m(\bs{z})}^{-1}\bs{\Sigma}_{X}^{-1}\bs{\mathfrak{z}}_i^{\eta}\right)&
=\frac{1}{n}\bs{a}^{\top}\bs{\Sigma}_{\eta}(\bs{z})\bs{a}
=\frac{A_{m(\bs{z})}^{-2}}{n^2N^2}\sum_{j,j^{\prime}=1}^{N}B_{m(\bs{z})}(\bs{z}_j)
B_{m(\bs{z})}(\bs{z}_{j^{\prime}})G_{\eta}(\bs{z}_j, \bs{z}_{j^{\prime}})\bs{a}^{\top}\bs{\Sigma}^{-1}_{X}\bs{a},\\
\mathrm{Var}\left(\bs{a}^{\top}A_{m(\bs{z})}^{-1}\bs{\Sigma}_{X}^{-1} \bs{\mathfrak{z}}_i^{\varepsilon}\right)&=\frac{1}{n}\bs{a}^{\top}\bs{\Sigma}_{\varepsilon}(\bs{z})\bs{a}
=\frac{A_{m(\bs{z})}^{-2}}{n^2N^2}\sum_{j=1}^{N}B_{m(\bs{z})}(\bs{z}_j)B_{m(\bs{z})}(\bs{z}_j)
\sigma^2(\bs{z}_j)\bs{a}^{\top}\bs{\Sigma}^{-1}_{X}\bs{a}.
\end{align*}
Using central limit theorem, we have
\[
\left[\bs{a}^{\top}\left\{\bs{\Sigma}_{\eta}(\bs{z})+\bs{\Sigma}_{\varepsilon}(\bs{z}) \right\}\bs{a} \right]^{-1/2} \sum_{\ell =0}^{p}a_{\ell}\left\{\widehat{\eta}_{\ell}^{\ast}(\bs{z})
+\widehat{\varepsilon}_{\ell}^{\ast}(\bs{z}) \right\} \overset{\mathcal{L}}{%
\longrightarrow}N(0,1).
\]
By (\ref{EQ:app. cov}), as $N\rightarrow \infty $ and $n\rightarrow\infty$,
$\{\mathbf{a}^{\top}\bs{\Sigma}_{n}(\bs{z})\mathbf{a}\}^{-1/2}\sum_{\ell =0}^{p}a_{\ell}\left\{\widehat{\eta}_{\ell}^{\ast}(\bs{z})
+\widehat{\varepsilon}_{\ell}^{\ast}(\bs{z}) \right\} \overset{\mathcal{L}}{%
\longrightarrow}N(0,1)$. Therefore,
$\{\mathbf{a}^{\top}\bs{\Sigma}_{n}(\bs{z})\mathbf{a}\}^{-1/2}\sum_{\ell =0}^{p}a_{\ell}\{\widehat{\beta}_{\ell}^{\mathrm{c}}(\bs{z}) -\beta_{\ell}^{o}(\bs{z})\} \overset{\mathcal{L}}{\longrightarrow}N(0,1)$ follows from (\ref{EQ:decompose2}), Theorem \ref{THM:uniformbiasrate2}, Lemma \ref{LEM:xietilda-hat} and Slutsky's Theorem. Applying Cram\'{e}r-Wold's device, we obtain
$\bs{\Sigma}_{n}^{-1/2}(\bs{z})\{\widehat{\beta}_{\ell}^{\mathrm{c}}(\bs{z}) -\beta_{\ell}^{o}(\bs{z})\}_{\ell =0}^{p}\overset{\mathcal{L}}{\longrightarrow}N\left(\mathbf{0},\mathbf{I}_{(p+1)\times (p+1)}\right)$, as $N\rightarrow \infty $ and $n\rightarrow\infty$,
and consequently, $\sigma_{n,\ell\ell}^{-1}(\bs{z})\{\widehat{\beta}_{\ell}^{\mathrm{c}}(\bs{z})-\beta_{\ell}^{o}(\bs{z})\} \overset{\mathcal{L}}{\longrightarrow} N(0,1)$, for any $\bs{z}\in \Omega$ and $\ell=0,\ldots,p$. \hfill $\square$

\vskip .10in \noindent \textbf{A.5. Convergence of the covariance estimator} \vskip .10in

For any $i=1,\ldots,n$, and estimated residuals $\widehat{R}_{ij}=Y_{ij}-\sum_{\ell=0}^{p}X_{i\ell}\widehat{\beta}_{\ell}(\bs{z}_j)$, denote
$
\widehat{\bs{\vartheta}}_{i}=\argmin_{\bs{\theta}}\sum_{j=1}^{N}
\left\{\widehat{R}_{ij}-\mathbf{B}_{\eta}^{\top}(\bs{z}_{j})\mathbf{Q}_{\eta,2}\bs{\theta} \right\}^{2},
$
where $\mathbf{B}_{\eta}(\bs{z})$ is the set of bivariate spline basis functions used to estimate $\eta_i(\bs{z})$, and $\mathbf{Q}_{\eta,2}$ is given in the following QR decomposition of the transpose of the smoothness matrix $\mathbf{H}_{\eta}$:
$\mathbf{H}_{\eta}^{\top}=\mathbf{Q}_{\eta}\mathbf{R}_{\eta}=(\mathbf{Q}_{\eta,1}~\mathbf{Q}_{\eta,2})
 \binom{\mathbf{R}_{\eta,1}}{\mathbf{R}_{\eta,2}}$.
Then, the bivariate spline estimator of $\eta_{i}(\bs{z})$  can be written as
$\widehat{\eta}_{i}(\bs{z})=\mathbf{B}_{\eta}(\bs{z})^{\top}\mathbf{Q}_{\eta,2}
\widehat{\bs{\vartheta}}_{i}=\widetilde{\mathbf{B}}_{\eta}(\bs{z})^{\top}
\widehat{\bs{\vartheta}}_{i}$. Let
\begin{equation*}
\bs{\Upsilon}_{n}=\frac{1}{N}\sum_{j=1}^{N}\widetilde{\mathbf{B}}_{\eta}(\bs{z}_{j})
\widetilde{\mathbf{B}}_{\eta}^{\top}(\bs{z}_{j}),
\label{DEF:Upsilon}
\end{equation*}
then we have
\begin{align}
\notag
\widehat{\bs{\vartheta}}_{i}&=\bs{\Upsilon}_{n}^{-1}\frac{1}{N}\sum_{j=1}^{N}\widetilde{\mathbf{B}}_{\eta}(\bs{z}_{j}) \widehat{R}_{ij}\\
&=\bs{\Upsilon}_{n}^{-1}\frac{1}{N}\sum_{j=1}^{N}\widetilde{\mathbf{B}}_{\eta}(\bs{z}_{j}) \left[
\sum_{\ell=0}^{p}X_{i\ell}\{\beta_{\ell}^{o}(\bs{z}_j)-\widehat{\beta}_{\ell}(\bs{z}_j)\}
+\eta_{i}(\bs{z}_j)+\sigma(\bs{z}_j)\varepsilon_{ij}\right].
\label{EQ:theta_eta}
\end{align}

\begin{lemma}
\label{LEM:Upsilon_0}
Under Assumptions (A3)--(A5), if $(N^{1/2}|\triangle_{\eta}|)/\log(|\triangle_{\eta}|^{-1})\rightarrow \infty$ as $N\rightarrow \infty$, then there exist constants $0 < c_{\Upsilon} < C_{\Upsilon} < \infty$, such that with probability approaching 1 as $N\rightarrow \infty$, $n\rightarrow \infty$, $c_{\Upsilon}|\triangle_{\eta}|^{2} \leq \lambda_{\min}(\bs{\Upsilon}_{n}) \leq
\lambda_{\max}(\bs{\Upsilon}_{n}) \leq  C_{\Upsilon}|\triangle_{\eta}|^{2}$.
\end{lemma}
The proof is similar to  the proof of \ref{LEM:Gamma_0}, thus omitted.

Next we define
\begin{align}
\widetilde{b}_{i}(\bs{z})&=\widetilde{\mathbf{B}}_{\eta}(\bs{z})^{\top}\bs{\Upsilon}_{n}^{-1}\frac{1}{N}\sum_{j=1}^{N} \widetilde{\mathbf{B}}_{\eta}(\bs{z}_{j})
\sum_{\ell=0}^{p}X_{i\ell}\{\beta_{\ell}^{o}(\bs{z}_j)-\widehat{\beta}_{\ell}(\bs{z}_j)\}, \label{DEF:theta-beta-eps}\\
\widetilde{\eta}_{i}(\bs{z})&=\widetilde{\mathbf{B}}_{\eta}(\bs{z})^{\top}
\bs{\Upsilon}_{n}^{-1}\frac{1}{N}\sum_{j=1}^{N} \widetilde{\mathbf{B}}_{\eta}(\bs{z}_{j})\eta_{i}(\bs{z}_j), ~
\widetilde{\varepsilon}_{i}(\bs{z})=\widetilde{\mathbf{B}}_{\eta}(\bs{z})^{\top}
\bs{\Upsilon}_{n}^{-1}\frac{1}{N}\sum_{j=1}^{N} \widetilde{\mathbf{B}}_{\eta}(\bs{z}_{j})\sigma(\bs{z}_j)\varepsilon_{ij}.
\notag
\end{align}
Then, the estimation error $D_i(\bs{z})=\widehat{\eta}_{i}(\bs{z})-\eta_{i}(\bs{z})$ in (\ref{DEF:eta_i_hat})
can be decomposed as the following:
\begin{equation*}
D_i(\bs{z})=\widetilde{b}_{i}(\bs{z})+\nabla\eta_{i}(\bs{z})
+\widetilde{\varepsilon}_{i}(\bs{z}).
\label{EQ:decompose4}
\end{equation*}

For any $\bs{z}$, $\bs{z}^{\prime}\in \Omega$, denote
\begin{equation*}
\widetilde{G}_{\eta}(\bs{z},\bs{z}^{\prime})=n^{-1}\sum_{i=1}^{n}\eta_i(\bs{z})\eta_i(\bs{z}^{\prime}).
\label{DEF:G_tilde}
\end{equation*}
The following lemma shows the uniform convergence of $\widetilde{G}_{\eta}(\bs{z},\bs{z}^{\prime})$ to $G_{\eta}(\bs{z},\bs{z}^{\prime})$ in probability over all $(\bs{z},\bs{z}^{\prime})\in \Omega^2$.

\begin{lemma}
\label{LEM:Gtilde-G}
Under Assumptions (A1)--(A5) and (C1)--(C3), $\sup_{(\bs{z},\bs{z}^{\prime})\in \Omega^2}|\widetilde{G}_{\eta}(\bs{z},\bs{z}^{\prime})-G_{\eta}(\bs{z},\bs{z}^{\prime})|=O_{P}\{n^{-1/2}(\log n)^{1/2}\}$.
\end{lemma}
\begin{proof}
Let $\bar{\xi}_{\cdot k k'}=n^{-1}\sum_{i=1}^{n}\xi_{ik}\xi_{ik^{\prime}}$, then
\[
\widetilde{G}_{\eta}(\bs{z},\bs{z}^{\prime})-G_{\eta}(\bs{z},\bs{z}^{\prime})
=\sum_{k=1}^{\infty }\lambda_{k}\psi_{k}(\bs{z})\psi_{k}(\bs{z}^{\prime})
\left(\bar{\xi}_{\cdot kk}-1\right)+ \sum_{k \neq k'}\bar{\xi}_{\cdot k k'}(\lambda_{k}\lambda_{k'})^{1/2}\psi_{k}(\bs{z})\psi_{k'}(\bs{z}^{\prime}).
\]
As $E\left[\sum_{k=1}^{\infty }\lambda_{k}\psi_{k}(\bs{z})\psi_{k}(\bs{z}^{\prime})
\left( \bar{\xi}_{\cdot kk}-1\right) \right]=0$, then $E\{\widetilde{G}_{\eta}(\bs{z},\bs{z}^{\prime})-G_{\eta}(\bs{z},\bs{z}^{\prime})\}=0$. Note that $E\{\eta^2(\bs{z})\eta^2(\bs{z}^{\prime})\}
= G_{\eta}(\bs{z},\bs{z})G_{\eta}(\bs{z}^{\prime},\bs{z}^{\prime})
+2G_{\eta}^2(\bs{z},\bs{z}^{\prime})
+\sum_{k=1}^{\infty}\lambda_{k}^{2}E(\xi_{1k}^4-3)\psi^2_{k}(\bs{z})\psi^2_{k}(\bs{z}^{\prime})$. Next,
\begin{align*}
E&\left\{\widetilde{G}_{\eta}(\bs{z},\bs{z}^{\prime})-G_{\eta}(\bs{z},\bs{z}^{\prime})\right\}^2= E\left\{\frac{1}{n}\sum_{i=1}^{n}\eta_i(\bs{z})\eta_i(\bs{z}^{\prime})-G_{\eta}(\bs{z},\bs{z}^{\prime})\right\}^2\\
&= \frac{1}{n} \left\{G_{\eta}(\bs{z},\bs{z})G_{\eta}(\bs{z}^{\prime},\bs{z}^{\prime})+G_{\eta}^2(\bs{z},\bs{z}^{\prime})
+\sum_{k=1}^{\infty}\lambda_{k}^{2}E(\xi_{1k}^4-3)\psi^2_{k}(\bs{z})\psi^2_{k}(\bs{z}^{\prime})\right\}.
\end{align*}
Therefore, $E\left\{\widetilde{G}_{\eta}(\bs{z},\bs{z}^{\prime})-G_{\eta}(\bs{z},\bs{z}^{\prime})\right\}^2\asymp n^{-1}$. Hence, following from Bernstein inequality,
$\sup_{(\bs{z},\bs{z}^{\prime})\in \Omega^2}\left\vert \widetilde{G}_{\eta}(\bs{z},\bs{z}^{\prime})-G_{\eta}(\bs{z},\bs{z}^{\prime})\right\vert =O_{P}\{n^{-1/2}(\log n)^{1/2}\}$, and the desired result follows.
\end{proof}

\textit{Proof of Theorem \ref{THM:Ghat-G}.} 
Note that
\[
\sup_{(\bs{z},\bs{z}^{\prime})\in \Omega^2}|\widehat{G}_{\eta}(\bs{z},\bs{z}^{\prime})
-G_{\eta}(\bs{z},\bs{z}^{\prime})|
\leq \sup_{(\bs{z},\bs{z}^{\prime})\in \Omega^2}\{|\widehat{G}_{\eta}(\bs{z},\bs{z}^{\prime})
-\widetilde{G}_{\eta}(\bs{z},\bs{z}^{\prime})|+|\widetilde{G}_{\eta}(\bs{z},\bs{z}^{\prime})
-G_{\eta}(\bs{z},\bs{z}^{\prime})|\},
\]
where $\sup_{(\bs{z},\bs{z}^{\prime})\in \Omega^2}|\widetilde{G}_{\eta}(\bs{z},\bs{z}^{\prime})
-G_{\eta}(\bs{z},\bs{z}^{\prime})|=o_{P}(1)$ according to Lemma \ref{LEM:Gtilde-G}, and
\begin{align*}
\sup_{(\bs{z},\bs{z}^{\prime})\in \Omega^2} & |\widehat{G}_{\eta}(\bs{z},\bs{z}^{\prime})
-\widetilde{G}_{\eta}(\bs{z},\bs{z}^{\prime})|
\leq \sup_{(\bs{z},\bs{z}^{\prime})\in \Omega^2}\left\vert n^{-1}\sum_{i=1}^{n}\eta_i(\bs{z})D_i(\bs{z}^{\prime})\right\vert \notag\\
&+\sup_{(\bs{z},\bs{z}^{\prime})\in \Omega^2}\left\vert n^{-1}\sum_{i=1}^{n}\eta_i(\bs{z}^{\prime})D_i(\bs{z})\right\vert 
+\sup_{(\bs{z},\bs{z}^{\prime})\in \Omega^2}
\left\vert n^{-1}\sum_{i=1}^{n}D_i(\bs{z})D_i(\bs{z}^{\prime})\right\vert.
\end{align*}
With some simple calculations, we have
\[
\sum_{i=1}^{n}\eta_i(\bs{z})D_i(\bs{z}^{\prime})=\sum_{i=1}^{n}\eta_i(\bs{z})\widetilde{b}_{i}(\bs{z}^{\prime})
+\sum_{i=1}^{n}\eta_i(\bs{z})\nabla\eta_{i}(\bs{z}^{\prime})
+\sum_{i=1}^{n}\eta_i(\bs{z})\widetilde{\varepsilon}_{i}(\bs{z}^{\prime}),
\]
where $\nabla\eta_{i}=\widetilde{\eta}_{i}-\eta_{i}$. According to (\ref{EQ:eta_b}), (\ref{EQ:eta_etatilde-eta}) and (\ref{EQ:eta_eps}), we have 
\[
\sup_{(\bs{z},\bs{z}^{\prime})\in \Omega^2}\left\vert n^{-1}\sum_{i=1}^{n}\eta_i(\bs{z})D_i(\bs{z}^{\prime})
+n^{-1}\sum_{i=1}^{n}\eta_i(\bs{z}^{\prime})D_i(\bs{z})\right\vert
=o_{P}(1).
\]
Note that
\begin{align*}
\sum_{i=1}^{n}D_i(\bs{z})D_i(\bs{z}^{\prime})=&\sum_{i=1}^{n}\widetilde{b}_{i}(\bs{z})\widetilde{b}_{i}(\bs{z}^{\prime})
+\sum_{i=1}^{n}\nabla\eta_{i}(\bs{z})\nabla\eta_{i}(\bs{z}^{\prime})+\sum_{i=1}^{n}\widetilde{b}_{i}(\bs{z})\nabla\eta_{i}(\bs{z}^{\prime})
+\sum_{i=1}^{n}\widetilde{\varepsilon}_{i}(\bs{z})\widetilde{\varepsilon}_{i}(\bs{z}^{\prime})\\
&+\sum_{i=1}^{n}\nabla\eta_{i}(\bs{z})\widetilde{\varepsilon}_{i}(\bs{z}^{\prime})
+\sum_{i=1}^{n}\widetilde{b}_{i}(\bs{z})\widetilde{\varepsilon}_{i}(\bs{z}^{\prime}).
\end{align*}
It follows from (\ref{EQ:b-b}), (\ref{EQ:etatilde-eta_etatilde-eta}), (\ref{EQ:b_etatilde-eta})--(\ref{EQ:b-eps}) that
$\sup_{(\bs{z},\bs{z}^{\prime})\in \Omega^2}
\left\vert n^{-1}\sum_{i=1}^{n}D_i(\bs{z})D_i(\bs{z}^{\prime})\right\vert=o_{P}(1)$.
The desired result is established. \hfill $\square$

\begin{lemma}
\label{LEM:b-others}
Under Assumptions (A1)--(A5), (C1)--(C3), we have
\begin{align}
\sup_{(\bs{z},\bs{z}^{\prime})\in \Omega^2}\left \vert n^{-1}\sum_{i=1}^{n}\widetilde{b}_{i}(\bs{z})\widetilde{b}_{i}(\bs{z}^{\prime})\right\vert
&=O_{P}\{n^{-1}|\triangle_{\eta}|^{-2}(\log n)^{1/2}\},\label{EQ:b-b}\\
\sup_{(\bs{z},\bs{z}^{\prime})\in \Omega^2}\left\vert \sum_{i=1}^{n}\eta_i(\bs{z})\widetilde{b}_{i}(\bs{z}^{\prime})\right\vert
&=O_{P}\{n^{-1}(\log n)^{1/2}\}. 
\label{EQ:eta_b}
\end{align}
\end{lemma}
\begin{proof}
According to (\ref{EQ:decompose3}) and (\ref{DEF:theta-beta-eps}), we have
\begin{align}
\widetilde{b}_{i}(\bs{z})&=\widetilde{\mathbf{B}}_{\eta}(\bs{z})^{\top}\bs{\Upsilon}_{n}^{-1}\frac{1}{N}\sum_{j=1}^{N} \widetilde{\mathbf{B}}(\bs{z}_{j})
\sum_{\ell=0}^{p}X_{i\ell}\{\beta_{\ell}^{o}(\bs{z}_j)-\widehat{\beta}_{\ell}(\bs{z}_j)\} \notag\\
&=\widetilde{\mathbf{B}}_{\eta}(\bs{z})^{\top}\bs{\Upsilon}_{n}^{-1}\frac{1}{N}\sum_{j=1}^{N} \widetilde{\mathbf{B}}(\bs{z}_{j})
\sum_{\ell=0}^{p}X_{i\ell}\{\beta_{\ell}^{o}(\bs{z}_j)-\widehat{\beta}_{\mu,\ell}(\bs{z}_j)
-\widehat{\eta}_{\ell}(\bs{z}_j)-\widehat{\varepsilon}_{\ell}(\bs{z}_j)\}.
\label{EQ:b_tilde}
\end{align}
Thus,
\begin{align*}
\frac{1}{n}\sum_{i=1}^{n}&\widetilde{b}_{i}(\bs{z})\widetilde{b}_{i}(\bs{z}^{\prime})
=\frac{1}{n} \sum_{i=1}^{n}\widetilde{\mathbf{B}}_{\eta}(\bs{z})^{\top}\bs{\Upsilon}_{n}^{-1}
\left[\frac{1}{N}\sum_{j=1}^{N} \widetilde{\mathbf{B}}(\bs{z}_{j})
\sum_{\ell=0}^{p}X_{i\ell}\{\beta_{\ell}^{o}(\bs{z}_j)-\widehat{\beta}_{\ell}(\bs{z}_j)\} \right]\\
&\times \Bigg[\frac{1}{N}\sum_{j^{\prime}=1}^{N} \widetilde{\mathbf{B}}(\bs{z}_{j^{\prime}})^{\top}		\sum_{\ell^{\prime}=0}^{p}X_{i\ell^{\prime}}\{\beta_{\ell^{\prime}}^{o}(\bs{z}_{j^{\prime}})
-\widehat{\beta}_{\ell^{\prime}}(\bs{z}_{j^{\prime}})\} \Bigg]
\bs{\Upsilon}_{n}^{-1}  \widetilde{\mathbf{B}}_{\eta}(\bs{z}^{\prime}) \\
\asymp & \frac{1}{n|\triangle_{\eta}|^{4}} \sum_{i=1}^{n}\widetilde{\mathbf{B}}_{\eta}(\bs{z})^{\top} \Bigg[\frac{1}{N^2}\sum_{j,j^{\prime}=1}^{N} \widetilde{\mathbf{B}}(\bs{z}_{j})\widetilde{\mathbf{B}}(\bs{z}_{j^{\prime}})^{\top}
\sum_{\ell,\ell^{\prime}=0}^{p}X_{i\ell}X_{i\ell^{\prime}}
\{\beta_{\ell}^{o}(\bs{z}_j)-\widehat{\beta}_{\ell}(\bs{z}_j)\} \\
& \times \{\beta_{\ell^{\prime}}^{o}(\bs{z}_{j^{\prime}})-\widehat{\beta}_{\ell^{\prime}}(\bs{z}_{j^{\prime}})\} \Bigg]\widetilde{\mathbf{B}}_{\eta}(\bs{z}').
\end{align*}
Therefore, by Theorem \ref{THM:beta-convergence}, we have
\[
E\left\{\frac{1}{n}\sum_{i=1}^{n}\widetilde{b}_{i}(\bs{z})\widetilde{b}_{i}(\bs{z}^{\prime})\right\}\\
\asymp \sum_{\ell=0}^{p}\sum_{\ell^{\prime}=0}^{p}|\triangle_{\eta}|^{-2} \|\beta_{\ell}^{o}-\widehat{\beta}_{\ell}\|\|\beta_{\ell^{\prime}}^{o}-\widehat{\beta}_{\ell^{\prime}}\|
\asymp n^{-1}|\triangle_{\eta}|^{-2}.
\]
We have $E\left\{n^{-1}\sum_{i=1}^{n}\widetilde{b}_{i}(\bs{z})\widetilde{b}_{i}(\bs{z}^{\prime})\right\}^2
=\frac{1}{n^2}\sum_{i,i^{\prime}=1}^{n}E\left\{\widetilde{b}_{i}(\bs{z})\widetilde{b}_{i}(\bs{z}^{\prime})
\widetilde{b}_{i^{\prime}}(\bs{z})\widetilde{b}_{i^{\prime}}(\bs{z}^{\prime})\right\}$,
where
\begin{align*}
&E\left\{\widetilde{b}_{i}(\bs{z})\widetilde{b}_{i}(\bs{z}^{\prime})\widetilde{b}_{i^{\prime}}(\bs{z})
\widetilde{b}_{i^{\prime}}(\bs{z}^{\prime})\right\} \asymp |\triangle_{\eta}|^{-8}\\
&~ \times  E\widetilde{\mathbf{B}}_{\eta}(\bs{z})^{\top}
\Bigg[\frac{1}{N^2}\sum_{j,j^{\prime}=1}^{N} \widetilde{\mathbf{B}}(\bs{z}_{j})\widetilde{\mathbf{B}}(\bs{z}_{j^{\prime}})^{\top}
\sum_{\ell,\ell^{\prime}=0}^{p}X_{i\ell}X_{i\ell^{\prime}}
\{\beta_{\ell}^{o}(\bs{z}_j)-\widehat{\beta}_{\ell}(\bs{z}_j)\} \{\beta_{\ell^{\prime}}^{o}(\bs{z}_{j^{\prime}})-\widehat{\beta}_{\ell^{\prime}}(\bs{z}_{j^{\prime}})\} \Bigg] \widetilde{\mathbf{B}}_{\eta}(\bs{z}^{\prime}) \\
&~ \times \widetilde{\mathbf{B}}_{\eta}(\bs{z})^{\top}
\Bigg[\frac{1}{N^2}\sum_{j,j^{\prime}=1}^{N} \widetilde{\mathbf{B}}(\bs{z}_{j})\widetilde{\mathbf{B}}(\bs{z}_{j^{\prime}})^{\top}
\sum_{\ell,\ell^{\prime}=0}^{p}X_{i^{\prime}\ell}X_{i^{\prime}\ell^{\prime}}
\{\beta_{\ell}^{o}(\bs{z}_j)-\widehat{\beta}_{\ell}(\bs{z}_j)\} \{\beta_{\ell^{\prime}}^{o}(\bs{z}_{j^{\prime}})-\widehat{\beta}_{\ell^{\prime}}(\bs{z}_{j^{\prime}})\} \Bigg] \widetilde{\mathbf{B}}_{\eta}(\bs{z}^{\prime})\\
&\asymp n^{-2}|\triangle_{\eta}|^{-4}.
\end{align*}
Thus, (\ref{EQ:b-b}) follows from the Bernstein inequality after the discretization.

Following from (\ref{EQ:b_tilde}), we have, for any  $i, i'= 1,\ldots, n$,
\begin{align*}
& E\left\{\widetilde{b}_{i}(\bs{z}')
\widetilde{b}_{i'}(\bs{z}^{\prime})\eta_i(\bs{z})\eta_{i^{\prime}}(\bs{z})\right\}\\
&\asymp |\triangle_{\eta}|^{-4} \widetilde{\mathbf{B}}_{\eta}(\bs{z}')^{\top}\frac{1}{N^2}\sum_{j,j'=1}^{N} \widetilde{\mathbf{B}}(\bs{z}_{j})\widetilde{\mathbf{B}}(\bs{z}_{j'})^{\top} E\left\{\sum_{\ell, \ell'=0}^{p}X_{i\ell}X_{i\ell'}\widehat{\eta}_{\ell}(\bs{z}_{j})\widehat{\eta}_{\ell'}(\bs{z}_{j'})
\eta_i(\bs{z})\eta_{i^{\prime}}(\bs{z})\right\}\widetilde{\mathbf{B}}_{\eta}(\bs{z}'),\\
& E\left\{\sum_{\ell, \ell'=0}^{p}X_{i\ell}X_{i\ell'}\widehat{\eta}_{\ell}(\bs{z}_{j})\widehat{\eta}_{\ell'}(\bs{z}_{j'})\eta_i(\bs{z})\eta_{i^{\prime}}(\bs{z})\right\}=\frac{1}{n^2N^2}
\sum_{i^{\prime\prime}, i^{\prime\prime\prime}=1}^{n} E\Bigg[\left\{\mathbf{X}_i \otimes \widetilde{\mathbf{B}}(\bs{z}_j)\right\}^{\top} \mathbf{\Gamma}_{n,\rho}^{-1}\\
&\quad  \times \sum_{j^{\prime\prime}, j^{\prime\prime\prime}=1}^{N} \mathbf{X}_{i^{\prime\prime}}\mathbf{X}_{i^{\prime\prime\prime}}^{\top} \otimes \widetilde{\mathbf{B}}(\bs{z}_{j^{\prime\prime}}) \widetilde{\mathbf{B}}(\bs{z}_{j^{\prime\prime\prime}})^{\top}
\mathbf{\Gamma}_{n,\rho}^{-1}\mathbf{X}_{i^{\prime}} \otimes \widetilde{\mathbf{B}}(\bs{z}_{j^{\prime}})\Bigg] E\left\{\eta_{i}(\bs{z})\eta_{i^{\prime}}(\bs{z})\eta_{i^{\prime\prime}}(\bs{z}_{j^{\prime\prime}})
\eta_{i^{\prime\prime\prime}}(\bs{z}_{j^{\prime\prime\prime}})\right\}\\
&=\frac{1}{n^2N^2} E\Bigg[\left\{\mathbf{X}_i \otimes \widetilde{\mathbf{B}}(\bs{z}_j)\right\}^{\top} \mathbf{\Gamma}_{n,\rho}^{-1}\sum_{j^{\prime\prime}, j^{\prime\prime\prime}=1}^{N} \mathbf{X}_{i^{\prime\prime}}\mathbf{X}_{i^{\prime\prime\prime}}^{\top} \otimes \widetilde{\mathbf{B}}(\bs{z}_{j^{\prime\prime}}) \widetilde{\mathbf{B}}(\bs{z}_{j^{\prime\prime\prime}})^{\top}
\mathbf{\Gamma}_{n,\rho}^{-1} \left\{\mathbf{X}_{i^{\prime}} \otimes \widetilde{\mathbf{B}}(\bs{z}_{j^{\prime}})\right\}\Bigg]\\
&\quad \times  E\left\{\eta_{i}(\bs{z})\eta_{i^{\prime}}(\bs{z})\eta_{i}(\bs{z}_{j^{\prime\prime}})
\eta_{i^{\prime}}(\bs{z}_{j^{\prime\prime\prime}})+\eta_{i}(\bs{z})\eta_{i^{\prime}}(\bs{z})\eta_{i^{\prime}}(\bs{z}_{j^{\prime\prime}})
\eta_{i}(\bs{z}_{j^{\prime\prime\prime}})\right\} \asymp n^{-2}.
\end{align*}
Therefore, $ E\left\{\frac{1}{n}\sum_{i=1}^{n}\eta_i(\bs{z})\widetilde{b}_{i}(\bs{z}^{\prime})\right\}^2=\frac{1}{n^2}\sum_{i, i'=1}^{n}E\{\widetilde{b}_{i}(\bs{z}')\widetilde{b}_{i'}(\bs{z}^{\prime})\eta_i(\bs{z})\eta_{i^{\prime}}(\bs{z})\} = O(n^{-2})$.
\end{proof}

\begin{lemma}
\label{LEM:eta-eta}
Under Assumptions (A1)--(A5), (C1)--(C3), we have
\begin{align}
\sup_{(\bs{z},\bs{z}^{\prime})\in \Omega^2}\left\vert n^{-1}\sum_{i=1}^{n}
\nabla\eta_{i}(\bs{z})\nabla\eta_{i}(\bs{z}^{\prime})\right\vert &=O_{P}\left\{|\triangle_{\eta}|^{2(s+1)}\sum_{k=1}^{K_n}\lambda_k \|\psi_{k}\|_{s+1,\infty}^2+\sum_{k=K_n+1}^{\infty} \lambda_{k} \|\psi_{k}\|_{\infty}^2 \right\}, \label{EQ:etatilde-eta_etatilde-eta}\\
\sup_{(\bs{z},\bs{z}^{\prime})\in \Omega^2}\left\vert n^{-1}
\sum_{i=1}^{n}\eta_i(\bs{z})\nabla\eta_{i}(\bs{z}^{\prime})\right\vert
&=O_{P}\left\{|\triangle_{\eta}|^{s+1}\sum_{k=1}^{K_n} \lambda_k \|\psi_{k}\|_{s+1,\infty} \|\psi_{k}\|_{\infty}
+\sum_{k=K_n+1}^{\infty} \lambda_k \|\psi_{k}\|_{\infty}^2\right\},\label{EQ:eta_etatilde-eta}
\end{align}
\begin{align}
\sup_{(\bs{z},\bs{z}^{\prime})\in \Omega^2}\left\vert \sum_{i=1}^{n}
\nabla\eta_{i}(\bs{z})\widetilde{b}_{i}(\bs{z}')\right\vert
&=O_{P}\left\{(\log n)^{1/2}n^{-1}|\triangle_{\eta}|^{s+1}\sum_{k=1}^{K_n} \lambda_k \|\psi_{k}\|_{s+1,\infty} \|\psi_{k}\|_{\infty}\right\}\notag\\
&+O_{P}\left\{(\log n)^{1/2}n^{-1}\sum_{k=K_n+1}^{\infty} \lambda_k \|\psi_{k}\|_{\infty}^2\right\}. \label{EQ:b_etatilde-eta}
\end{align}
\end{lemma}
\begin{proof}
For any $k\geq 1$, denote
$\widetilde\psi_{k}(\bs{z})=\widetilde{\mathbf{B}}(\bs{z})^{\top}\bs{\Upsilon}_{n}^{-1}
\frac{1}{N}\sum_{j=1}^{N} \widetilde{\mathbf{B}}(\bs{z}_{j})\psi_{k}(\bs{z}_j)$, and $\nabla \psi_{k}=\widetilde\psi_{k}-\psi_{k}$.
According to Assumption (C2), we hav(C3)e, for any $k\geq 1$, $\|\nabla\psi_{k}\|_{\infty}\leq C|\triangle_{\eta}|^{s+1} \|\psi_{k}\|_{s+1,\infty}$ and $\|\widetilde{\psi}_{k}\|_{\infty}\leq \|\psi_{k}\|_{\infty}+\|\nabla\psi_{k}\|_{\infty}\leq 2\|\psi_{k}\|_{\infty}$, as $n \to \infty$. It is easy to see that $\nabla\eta_{i}(\bs{z}^{\prime})=\sum_{k=1}^{\infty} \lambda_k^{1/2}\xi_{ik} \nabla {\psi}_{k}(\bs{z}^{\prime})$.

We first show (\ref{EQ:etatilde-eta_etatilde-eta}). Let $\bar{\xi}_{\cdot k k^{\prime}}=n^{-1}\sum_{i=1}^{n} \xi_{ik}\xi_{ik^{\prime}}$, where $ E(\bar{\xi}_{\cdot k k^{\prime}})=I(k=k^{\prime})$ and $ E(\bar{\xi}_{\cdot k k^{\prime}})^2\leq (E\xi_{ik}^4E\xi_{ik^{\prime}}^4)^{1/2}\leq C$.
Simple calculation yields that
$
\frac{1}{n}\sum_{i=1}^{n}\nabla\eta_{i}(\bs{z})
\nabla\eta_{i}(\bs{z}^{\prime})
=\sum_{k,k'=1}^{\infty}\bar{\xi}_{\cdot k k^{\prime}}(\lambda_{k}\lambda_{k'})^{1/2}\nabla\psi_{k}(\bs{z})
\nabla\psi_{k^{\prime}}(\bs{z}^{\prime})$.
Thus, by Assumption (C2), we have
\begin{align*}
\sup_{(\bs{z},\bs{z}') \in \Omega^2}\Bigg| E&\left\{\frac{1}{n}\sum_{i=1}^{n}\nabla\eta_{i}(\bs{z})
\nabla\eta_{i}(\bs{z}^{\prime})\right\}\Bigg| =\sup_{(\bs{z},\bs{z}') \in \Omega^2}\Bigg|\sum_{k=1}^{\infty} \lambda_{k}\nabla\psi_{k}(\bs{z})
\nabla\psi_{k}(\bs{z}^{\prime})\Bigg| \\
&\leq |\triangle_{\eta}|^{2(s+1)} \sum_{k=1}^{K_n}\lambda_k \|\psi_{k}\|_{s+1,\infty}^2
+C_{\psi}\sum_{k=K_n+1}^{\infty} \lambda_{k} \|\psi_{k}\|_{\infty}^2.
\end{align*}
In addition, we have
\begin{align*}
&\sup_{\bs{z},\bs{z}^{\prime}\in \Omega}  E\Bigg\{\nabla\eta_{i}(\bs{z})
\nabla\eta_{i}(\bs{z}^{\prime})\Bigg\}^2=\sup_{\bs{z},\bs{z}^{\prime}\in \Omega}  E\left[\sum_{k=1}^{\infty} \xi_{ik}^2\lambda_k(\nabla\psi_{k})^2(\bs{z})	\sum_{k^{\prime}=1}^{\infty}\xi_{ik^{\prime}}^2\lambda_{k^{\prime}}
(\nabla\psi_{k^{\prime}})^2(\bs{z}^{\prime})\right]\\
&\asymp  n^{-1}\Bigg\{|\triangle_{\eta}|^{2(s+1)}\sum_{k=1}^{K_n}\lambda_k \|\psi_{k}\|_{s+1,\infty}^2
+\sum_{k=K_n+1}^{\infty} \lambda_{k} \|\psi_{k}\|_{\infty}^2\Bigg\}^{2}.
\end{align*}
Thus,
\[
\sup_{\bs{z},\bs{z}^{\prime}\in \Omega} \mathrm{Var}\left\{\frac{1}{n}\sum_{i=1}^{n}\nabla\eta_{i}(\bs{z})\nabla\eta_{i}(\bs{z}^{\prime})\right\} \asymp \Bigg\{|\triangle_{\eta}|^{2(s+1)}\sum_{k=1}^{K_n}\lambda_k \|\psi_{k}\|_{s+1,\infty}^2
+\sum_{k=K_n+1}^{\infty} \lambda_{k} \|\psi_{k}\|_{\infty}^2\Bigg\}^{2}.
\]
Therefore, using the discretization method and Bernstein inequality
\begin{align*}
&\sup_{(\bs{z},\bs{z}') \in \Omega^2}\Bigg| \frac{1}{n}\sum_{i=1}^{n}\nabla\eta_{i}(\bs{z})\nabla\eta_{i}(\bs{z}^{\prime}) -  E \{\nabla\eta_{i}(\bs{z})\nabla\eta_{i}(\bs{z}^{\prime}) \}  \Bigg|\\
&=O_P\Bigg\{(\log n)^{1/2} n^{-1/2}|\triangle_{\eta}|^{2(s+1)}\sum_{k=1}^{K_n}\lambda_k \|\psi_{k}\|_{s+1,\infty}^2
(\log n)^{1/2} n^{-1/2}\sum_{k=K_n+1}^{\infty} \lambda_{k} \|\psi_{k}\|_{\infty}^2\Bigg\}.
\end{align*}
Next we derive (\ref{EQ:eta_etatilde-eta}). Noting that
$n^{-1}\sum_{i=1}^{n}\eta_i(\bs{z})\nabla\eta_{i}(\bs{z}^{\prime})
=\xi_{\cdot kk^{\prime}}(\lambda_{k}\lambda_{k'})^{1/2}\psi_{k}(\bs{z}^{\prime})
(\nabla\psi_{k^{\prime}})(\bs{z}^{\prime})$,
we have
\begin{align*}
\sup_{(\bs{z},\bs{z}^{\prime})\in \Omega^2}\left\vert  E\left\{\frac{1}{n}\sum_{i=1}^{n}\eta_i(\bs{z})\nabla\eta_{i}(\bs{z}^{\prime})\right\}\right\vert
&\leq \sum_{k=1}^{\infty}\lambda_k\|\psi_{k}\|_{\infty}
\|\nabla\psi_{k^{\prime}}\|_{\infty}\\
&\leq C|\triangle_{\eta}|^{s+1}\sum_{k=1}^{K_n} \lambda_k \|\psi_k\|_{s+1,\infty}\|\psi_{k}\|_{\infty} + \sum_{k=K_n+1}^{\infty}\lambda_k\|\psi_k\|^2_{\infty},\\
\textrm{var}\Bigg\{n^{-1}\sum_{i=1}^{n}\eta_i(\bs{z})\nabla\eta_{i}(\bs{z}^{\prime})\Bigg\}
&=n^{-1}\left[ E\left\{\eta_i^2(\bs{z})\nabla\eta_{i}(\bs{z}^{\prime})^2\right\}- \left\{E\eta_i(\bs{z})\nabla\eta_{i}(\bs{z}^{\prime})\right\}^2\right],
\end{align*}
\begin{align*}
\sup_{\bs{z},\bs{z}^{\prime}\in \Omega} E\left\{\eta_i^2(\bs{z})\nabla\eta_{i}(\bs{z}^{\prime})^2\right\}&=\sup_{\bs{z},\bs{z}^{\prime}\in \Omega}\left\{\sum_{k=1}^{\infty} E\xi_{ik}^{4} \lambda_k^2\psi_{k}^2(\bs{z})(\nabla\psi_k)^2(\bs{z}^{\prime})
+\sum_{k \neq k^{\prime}} \lambda_k\lambda_{k^{\prime}} \psi_{k}^2(\bs{z})(\nabla\psi_k)^2(\bs{z}^{\prime})\right\}\\
&\leq C \left\{|\triangle_{\eta}|^{2(s+1)}\sum_{k=1}^{K_n} \lambda_k \|\psi_k\|_{s+1,\infty}^2 +\sum_{k=K_n+1}^{\infty}\lambda_k\|\psi_k\|^2_{\infty}\right\},
\end{align*}
and 
\[
\sup_{\bs{z},\bs{z}^{\prime}\in \Omega} \left|  E\left\{\eta_i(\bs{z})\nabla\eta_{i}(\bs{z}^{\prime})\right\} \right|
\leq C \left\{|\triangle_{\eta}|^{s+1}\sum_{k=1}^{K_n} \lambda_k \|\psi_{k}\|_{s+1,\infty} \|\psi_{k}\|_{\infty}+\sum_{k=K_n+1}^{\infty} \lambda_k \|\psi_{k}\|_{\infty}^2\right\}.
\] 
Therefore,
\[
\sup_{\bs{z},\bs{z}^{\prime}\in \Omega}  E\Bigg\{n^{-1}\sum_{i=1}^{n}\eta_i(\bs{z})\nabla\eta_{i}(\bs{z}^{\prime})\Bigg\}^2
\!=\!O\left[\left\{|\triangle_{\eta}|^{s+1}\sum_{k=1}^{K_n} \lambda_k \|\psi_{k}\|_{s+1,\infty} \|\psi_{k}\|_{\infty}+\sum_{k=K_n+1}^{\infty} \lambda_k \|\psi_{k}\|_{\infty}^2\right\}^2\right].
\]
Hence,
\begin{align*}
&\sup_{(\bs{z},\bs{z}^{\prime})\in \Omega^2}\left\vert n^{-1}\sum_{i=1}^{n}\nabla\eta_{i}(\bs{z})
\widetilde{\varepsilon}_{i}(\bs{z}^{\prime})
- E\left\{\nabla\eta_{i}(\bs{z})
\widetilde{\varepsilon}_{i}(\bs{z}^{\prime})\right\}\right\vert \\
&=O_{P}\left\{(\log n)^{1/2}n^{-1/2} |\triangle_{\eta}|^{2(s+1)}\sum_{k=1}^{K_n} \lambda_k \|\psi_k\|_{s+1,\infty}^2 +(\log n)^{1/2}n^{-1/2}\sum_{k=K_n+1}^{\infty}\lambda_k\|\psi_k\|^2_{\infty}\right\}
\end{align*}
using the discretization method and Bernstein inequality.

Finally, we provide the proof of (\ref{EQ:b_etatilde-eta}). Note that
\begin{align*}
& E\left\{\widetilde{b}_{i}(\bs{z}')\widetilde{b}_{i^{\prime}}(\bs{z}')
\nabla\eta_{i}(\bs{z})\nabla\eta_{i'}(\bs{z})\right\}\asymp |\triangle_{\eta}|^{-4}\\
&\times \widetilde{\mathbf{B}}_{\eta}(\bs{z}')^{\top}\frac{1}{N^2}\sum_{j,j'=1}^{N} \widetilde{\mathbf{B}}_{\eta}(\bs{z}_{j})\widetilde{\mathbf{B}}_{\eta}(\bs{z}_{j'})^{\top} E\left\{\sum_{\ell, \ell'=0}^{p}X_{i\ell}X_{i\ell'}\widehat{\eta}_{\ell}(\bs{z}_{j})
\widehat{\eta}_{\ell'}(\bs{z}_{j'})\nabla\eta_i(\bs{z})\nabla\eta_{i^{\prime}}(\bs{z})\right\}
\widetilde{\mathbf{B}}_{\eta}(\bs{z}'),
\end{align*}
and by (\ref{EQ:eta_etatilde-eta}),
\begin{align*}
& E\left\{\sum_{\ell, \ell'=0}^{p}X_{i\ell}X_{i\ell'}\widehat{\eta}_{\ell}(\bs{z}_{j})\widehat{\eta}_{\ell'}(\bs{z}_{j'})
\nabla\eta_i(\bs{z})\nabla\eta_{i^{\prime}}(\bs{z})\right\}=\frac{1}{n^2N^2}
\sum_{i^{\prime\prime}, i^{\prime\prime\prime}=1}^{n} E\Bigg[\left\{\mathbf{X}_i \otimes \widetilde{\mathbf{B}}(\bs{z}_j)\right\}^{\top} \mathbf{\Gamma}_{n,\rho}^{-1}\\
&\quad \times \sum_{j^{\prime\prime}, j^{\prime\prime\prime}=1}^{N} \mathbf{X}_{i^{\prime\prime}}\mathbf{X}_{i^{\prime\prime\prime}}^{\top} \otimes \widetilde{\mathbf{B}}(\bs{z}_{j^{\prime\prime}}) \widetilde{\mathbf{B}}(\bs{z}_{j^{\prime\prime\prime}})^{\top}
\mathbf{\Gamma}_{n,\rho}^{-1}\mathbf{X}_{i^{\prime}} \otimes \widetilde{\mathbf{B}}(\bs{z}_{j^{\prime}})\Bigg] E\left\{\eta_{i^{\prime\prime}}(\bs{z}_{j^{\prime\prime}})
\eta_{i^{\prime\prime\prime}}(\bs{z}_{j^{\prime\prime\prime}})\nabla\eta_{i}(\bs{z})\nabla \eta_{i^{\prime}}(\bs{z})\right\}\\
&=\frac{1}{n^2N^2} E\Bigg[\left\{\mathbf{X}_i \otimes \widetilde{\mathbf{B}}(\bs{z}_j)\right\}^{\top} \mathbf{\Gamma}_{n,\rho}^{-1}\sum_{j^{\prime\prime}, j^{\prime\prime\prime}=1}^{N} \mathbf{X}_{i^{\prime\prime}}\mathbf{X}_{i^{\prime\prime\prime}}^{\top} \otimes \widetilde{\mathbf{B}}(\bs{z}_{j^{\prime\prime}}) \widetilde{\mathbf{B}}(\bs{z}_{j^{\prime\prime\prime}})^{\top}
\mathbf{\Gamma}_{n,\rho}^{-1} \left\{\mathbf{X}_{i^{\prime}} \otimes \widetilde{\mathbf{B}}(\bs{z}_{j^{\prime}})\right\}\Bigg]\\
&\quad \times  E\left\{\eta_{i}(\bs{z}_{j^{\prime\prime}})
\eta_{i^{\prime}}(\bs{z}_{j^{\prime\prime\prime}})\nabla\eta_{i}(\bs{z})\nabla\eta_{i^{\prime}}(\bs{z})+\eta_{i^{\prime}}(\bs{z}_{j^{\prime\prime}})
\eta_{i}(\bs{z}_{j^{\prime\prime\prime}})\nabla\eta_{i}(\bs{z})\nabla\eta_{i^{\prime}}(\bs{z})\right\}.
\end{align*}
If $i \neq i'$, we have
\begin{align*}
E \Bigg\{\eta_{i}&(\bs{z}_{j^{\prime\prime}})
\eta_{i^{\prime}}(\bs{z}_{j^{\prime\prime\prime}})\nabla\eta_{i}(\bs{z})\nabla\eta_{i^{\prime}}(\bs{z})
+\eta_{i^{\prime}}(\bs{z}_{j^{\prime\prime}})
\eta_{i}(\bs{z}_{j^{\prime\prime\prime}})\nabla\eta_{i}(\bs{z})\nabla\eta_{i^{\prime}}(\bs{z})\Bigg\}\\
&\asymp \left\{\sum_{k=1}^{K_n}\lambda_k|\triangle|_{\eta}^{s+1}\|\psi_k\|_{s+1,\infty}\|\psi_k\|_{\infty}+\sum_{k=K_n+1}^{\infty}\lambda_k\|\psi_k\|_{\infty}^2\right\}^2.
\end{align*}
If $i = i'$, then we have
\begin{align*}
 E  \Bigg\{\eta_{i}&(\bs{z}_{j^{\prime\prime}})
\eta_{i}(\bs{z}_{j^{\prime\prime\prime}})\nabla\eta_{i}(\bs{z})\nabla\eta_{i}(\bs{z})\Bigg\}
=\sum_{k=1}^{\infty} \lambda_k^2 E\xi_{ik}^4\psi_{k}(\bs{z}'')\psi_{k}(\bs{z}''')\nabla \psi_{k}(\bs{z})\nabla \psi_{k}(\bs{z})\\
& \leq \sum_{k=1}^{K_n}\lambda_k^2|\triangle|_{\eta}^{2s+2}\|\psi_k\|_{s+1,\infty}^2\|\psi_k\|_{\infty}^2+\sum_{k=K_n+1}^{\infty}\lambda_k^2\|\psi_k\|_{\infty}^4.
\end{align*}
Thus,
\begin{align*}
 E &\Bigg\{\!\sum_{\ell, \ell'=0}^{p}  X_{i\ell}X_{i\ell'}\widehat{\eta}_{\ell}(\bs{z}_{j})\widehat{\eta}_{\ell'}(\bs{z}_{j'})
\nabla\eta_i(\bs{z})\nabla\eta_{i^{\prime}}(\bs{z})\Bigg\} 
\!\! \\
&\asymp\!\! \left\{\sum_{k=1}^{K_n}\lambda_k|\triangle|_{\eta}^{s+1}\|\psi_k\|_{s+1,\infty}\|\psi_k\|_{\infty}+\!\!\!\sum_{k=K_n+1}^{\infty}\lambda_k\|\psi_k\|_{\infty}^2\right\}^2.
\end{align*}
Therefore,
\begin{align*}
 E \left\{\frac{1}{n}\sum_{i=1}^{n} \nabla \eta_i(\bs{z})\widetilde{b}_{i}(\bs{z}^{\prime})\right\}^2&=\frac{1}{n^2}\sum_{i, i'=1}^{n} E \widetilde{b}_{i}(\bs{z}')\widetilde{b}_{i^{\prime}}(\bs{z}')\nabla\eta_{i}(\bs{z})\nabla\eta_{i'}(\bs{z})\\
&=O\left[n^{-2} \sum_{k=1}^{K_n}\lambda_k^2|\triangle|_{\eta}^{2s+2}\|\psi_k\|_{s+1,\infty}^2\|\psi_k\|_{\infty}^2+
n^{-2}\sum_{k=K_n+1}^{\infty}\lambda_k^2\|\psi_k\|_{\infty}^4\right].
\end{align*}
Thus, (\ref{EQ:b_etatilde-eta}) is obtained.
\end{proof}

\begin{lemma}
\label{LEM:eps-others}
Under Assumptions (A1)--(A5), (C1)--(C3), we have
\begin{align}
\sup_{(\bs{z},\bs{z}^{\prime})\in \Omega^2}\left\vert n^{-1}\sum_{i=1}^{n}\widetilde{\varepsilon}_{i}(\bs{z})\widetilde{\varepsilon}_{i}(\bs{z}^{\prime})\right\vert &=O_{P}(N^{-1}|\triangle_{\eta}|^{-2}),\label{EQ:eps_eps}\\
\sup_{(\bs{z},\bs{z}^{\prime})\in \Omega^2}\left\vert n^{-1}\sum_{i=1}^{n}\nabla\eta_{i}(\bs{z})
\widetilde{\varepsilon}_{i}(\bs{z}^{\prime})\right\vert &=O_{P}\left\{n^{-1/2}N^{-1/2}(\log n)^{1/2} |\triangle_{\eta}|^{s}\sum_{k=1}^{K_n} \lambda_k^{1/2}\|\psi_{k}\|_{s+1,\infty}\right\}\notag\\
 &\!\!\!\!\!\!\!\!\!\!\!\!\!\!\!\!\!\!\!\!\!\!\!\!+O_P\left\{n^{-1/2}N^{-1/2}|\triangle_{\eta}|^{-1}(\log n)^{1/2}\sum_{k=K_n+1}^{\infty} \lambda_k^{1/2} \|\psi_{k}\|_{\infty}\right\},\label{EQ:etatilde-eta_eps}\\
\sup_{(\bs{z},\bs{z}^{\prime})\in \Omega^2}\left\vert n^{-1}\sum_{i=1}^{n}\widetilde{b}_{i}(\bs{z})\widetilde{\varepsilon}_{i}(\bs{z}^{\prime})\right\vert &=O_{P}\{n^{-1}N^{-1}|\triangle|^{-2}(\log n)^{1/2}\},
\label{EQ:b-eps}\\
\sup_{(\bs{z},\bs{z}^{\prime})\in \Omega^2}\left \vert n^{-1}
\sum_{i=1}^{n}\eta_i(\bs{z})\widetilde{\varepsilon}_{i}(\bs{z}^{\prime})\right\vert &=O_{P}\{n^{-1/2}N^{-1/2}|\triangle_{\eta}|^{-1}(\log n)^{1/2}\}.
\label{EQ:eta_eps}
\end{align}
\end{lemma}
\begin{proof}
We first show (\ref{EQ:eps_eps}). Let $\bar{\varepsilon}_{\cdot jj^{\prime}}=n^{-1}\sum_{i=1}^{n}\varepsilon _{ij}\varepsilon _{ij^{\prime }}$, where $ E(\bar{\varepsilon}_{\cdot jj^{\prime}})=I(j=j')$. Note that
\[
	\frac{1}{n}\sum_{i=1}^{n}\widetilde{\varepsilon}_{i}(\bs{z})\widetilde{\varepsilon}_{i}(\bs{z}^{\prime})
	=\widetilde{\mathbf{B}}(\bs{z})^{\top}\bs{\Upsilon}_{n}^{-1}
		\left\{\frac{1}{N^2}\sum_{j=1}^{N}\sum_{j^{\prime}=1}^{N}
		\widetilde{\mathbf{B}}(\bs{z}_{j})\widetilde{\mathbf{B}}(\bs{z}_{j^{\prime}})^{\top}\sigma(\bs{z}_j) \sigma(\bs{z}_{j^{\prime}}) \bar{\varepsilon}_{\cdot jj^{\prime}}\right\}
	\bs{\Upsilon}_{n}^{-1}\widetilde{\mathbf{B}}(\bs{z}^{\prime}).
\]
It is easy to see that,
\[
 E\left\{\frac{1}{n}\sum_{i=1}^{n}\widetilde{\varepsilon}_{i}(\bs{z})
\widetilde{\varepsilon}_{i}(\bs{z}^{\prime})\right\} =
\widetilde{\mathbf{B}}(\bs{z})^{\top}\bs{\Upsilon}_{n}^{-1}
		\left\{\frac{1}{N^2}\sum_{j=1}^{N}
		\widetilde{\mathbf{B}}(\bs{z}_{j})\widetilde{\mathbf{B}}(\bs{z}_{j^{\prime}})^{\top}\sigma^2(\bs{z}_j) \right\}
	\bs{\Upsilon}_{n}^{-1}\widetilde{\mathbf{B}}(\bs{z}^{\prime}).
\]
Therefore, 
$\sup_{(\bs{z},\bs{z}^{\prime})\in \Omega^2}\left\vert  E\left\{\frac{1}{n}\sum_{i=1}^{n}\widetilde{\varepsilon}_{i}(\bs{z})
\widetilde{\varepsilon}_{i}(\bs{z}^{\prime})\right\} \right\vert
=O(N^{-1}|\triangle_{\eta}|^{-2})$.
In addition, note that
\begin{align*}
 E\left\{\widetilde{\varepsilon}_{i}(\bs{z})\widetilde{\varepsilon}_{i}(\bs{z}^{\prime})\right\}=&\widetilde{\mathbf{B}}(\bs{z})^{\top}\bs{\Upsilon}_{n}^{-1}
		\left\{\frac{1}{N^2}\sum_{j=1}^{N}
		\widetilde{\mathbf{B}}(\bs{z}_{j})\widetilde{\mathbf{B}}(\bs{z}_{j^{\prime}})^{\top}\sigma^2(\bs{z}_j) \right\}
	\bs{\Upsilon}_{n}^{-1}\widetilde{\mathbf{B}}(\bs{z}^{\prime})=O(N^{-1}|\triangle_{\eta}|^{-2}),\\
 E\left\{\widetilde{\varepsilon}_{i}(\bs{z})\widetilde{\varepsilon}_{i}(\bs{z}^{\prime})\right\}^2
=& E\left[\widetilde{\mathbf{B}}_{\eta}(\bs{z})^{\top}\bs{\Upsilon}_{n}^{-1}
\left\{\frac{1}{N^2}\sum_{j=1}^{N}\sum_{j^{\prime}=1}^{N}
\widetilde{\mathbf{B}}_{\eta}(\bs{z}_{j})\widetilde{\mathbf{B}}_{\eta}(\bs{z}_{j^{\prime}})^{\top}\sigma(\bs{z}_j) \sigma(\bs{z}_{j^{\prime}}) \varepsilon_{ij}\varepsilon_{ij^{\prime}}\right\}
\bs{\Upsilon}_{n}^{-1}\widetilde{\mathbf{B}}_{\eta}(\bs{z}')\right]^2\\
\asymp & \frac{|\triangle_{\eta}|^{-8}}{N^4}\sum_{j,j',j'',j'''=1}^{N}\widetilde{\mathbf{B}}_{\eta}(\bs{z}_{j})\widetilde{\mathbf{B}}_{\eta}(\bs{z}_{j^{\prime}})^{\top}\widetilde{\mathbf{B}}_{\eta}(\bs{z}_{j''})\widetilde{\mathbf{B}}_{\eta}(\bs{z}_{j'''})^{\top} \\
&\times \sigma(\bs{z}_j) \sigma(\bs{z}_{j^{\prime}})\sigma(\bs{z}_{j''}) \sigma(\bs{z}_{j'''}) \varepsilon_{ij}\varepsilon_{ij^{\prime}}\varepsilon_{ij''}\varepsilon_{ij'''} \asymp  N^{-2}|\triangle_{\eta}|^{-4}.
\end{align*}
Thus,
$\textrm{var}\left\{\frac{1}{n}\sum_{i=1}^{n}\widetilde{\varepsilon}_{i}(\bs{z})
\widetilde{\varepsilon}_{i}(\bs{z}^{\prime})\right\}
=\frac{1}{n^2}\sum_{i=1}^n\textrm{var}\left\{\widetilde{\varepsilon}_{i}(\bs{z})
\widetilde{\varepsilon}_{i}(\bs{z}^{\prime})\right\}
\asymp n^{-1}N^{-2}|\triangle_{\eta}|^{-4}$. 
Therefore,
\begin{align*}
\sup_{(\bs{z},\bs{z}^{\prime})\in \Omega^2}\left\vert n^{-1}\sum_{i=1}^{n}\widetilde{\varepsilon}_{i}(\bs{z})\widetilde{\varepsilon}_{i}(\bs{z}^{\prime})
- E\left\{\widetilde{\varepsilon}_{i}(\bs{z})\widetilde{\varepsilon}_{i}(\bs{z}^{\prime})\right\}\right\vert &=O_{P}\{n^{-1/2}N^{-1}(\log n)^{1/2}|\triangle_{\eta}|^{-2}\}
\end{align*}
using the discretization method and Bernstein inequality.

Next we derive (\ref{EQ:etatilde-eta_eps}). Note that
\begin{align*} 	
\frac{1}{n}\sum_{i=1}^{n}\nabla\eta_{i}(\bs{z})\widetilde{\varepsilon}_{i}(\bs{z}^{\prime})
	&=\frac{1}{n}\sum_{i=1}^{n}\sum_{k=1}^{\infty}\xi_{ik} \lambda_k^{1/2} \nabla\psi_k(\bs{z})
	\widetilde{\mathbf{B}}_{\eta}(\bs{z}^{\prime})^{\top}\bs{\Upsilon}_{n}^{-1}
	\left\{\frac{1}{N}\sum_{j=1}^{N}
	\widetilde{\mathbf{B}}_{\eta}(\bs{z}_{j})\sigma(\bs{z}_j) \varepsilon_{ij}\right\},\\
\left\{\frac{1}{n}\sum_{i=1}^{n}\nabla\eta_{i}(\bs{z})\widetilde{\varepsilon}_{i}(\bs{z}^{\prime})\right\}^2
&=\frac{1}{n^2}\sum_{i=1}^{n}\sum_{i^{\prime}=1}^{n}\sum_{k=1}^{\infty}\sum_{k^{\prime}=1}^{\infty}
\xi_{ik}\xi_{i^{\prime}k^{\prime}}(\lambda_{k}\lambda_{k'})^{1/2}
\nabla\psi_k(\bs{z}) \\
&\hspace{-4cm} \times \nabla\psi_{k^{\prime}}(\bs{z})
\widetilde{\mathbf{B}}_{\eta}(\bs{z}^{\prime})^{\top}\bs{\Upsilon}_{n}^{-1}
\left\{\frac{1}{N^2}\sum_{j=1}^{N}\sum_{j^{\prime}=1}^{N}
\widetilde{\mathbf{B}}_{\eta}(\bs{z}_{j})\widetilde{\mathbf{B}}_{\eta}(\bs{z}_{j^{\prime}})^{\top}\sigma(\bs{z}_j) \sigma(\bs{z}_{j^{\prime}})
\varepsilon_{ij}\varepsilon_{i^{\prime}j^{\prime}}\right\}
\bs{\Upsilon}_{n}^{-1}\widetilde{\mathbf{B}}_{\eta}(\bs{z}^{\prime}).
\end{align*}
Next observe that $ E\left[\frac{1}{n}\sum_{i=1}^{n}\nabla\eta_{i}(\bs{z})
\widetilde{\varepsilon}_{i}(\bs{z}^{\prime})\right]=0$ and
\begin{align*}
 E&\left\{\frac{1}{n}\sum_{i=1}^{n}\nabla\eta_{i}(\bs{z})
(\nabla\psi_k)^2(\bs{z})\right\}
=\frac{1}{n^2}\sum_{i=1}^{n}\sum_{k=1}^{\infty}
\lambda_{k}(\nabla\psi_k)^2(\bs{z}) \\
& \qquad \qquad \times \widetilde{\mathbf{B}}_{\eta}(\bs{z}^{\prime})^{\top}\bs{\Upsilon}_{n}^{-1}
\left\{\frac{1}{N^2}\sum_{j=1}^{N}
\widetilde{\mathbf{B}}_{\eta}(\bs{z}_{j})\widetilde{\mathbf{B}}_{\eta}(\bs{z}_{j})^{\top}\sigma^2(\bs{z}_j) \right\}
\bs{\Upsilon}_{n}^{-1}\widetilde{\mathbf{B}}_{\eta}(\bs{z}^{\prime}) 
\end{align*}
So, 
\begin{align*}
E\left\{\frac{1}{n}\sum_{i=1}^{n}\nabla\eta_{i}(\bs{z})
(\nabla\psi_k)^2(\bs{z})\right\}
\leq \frac{C_1|\triangle_{\eta}|^{-2}}{nN}\left\{|\triangle_{\eta}|^{2(s+1)}\sum_{k=1}^{K_n} \lambda_k\|\psi_{k}\|_{s+1,\infty}^2 +\sum_{k=K_n+1}^{\infty} \lambda_k \|\psi_{k}\|_{\infty}^2\right\}.
\end{align*}

Thirdly, we prove (\ref{EQ:b-eps}). Note that for any $i$, $i^{\prime}$, $j$, $j^{\prime}$, we have
\begin{align*}
& E\left\{\widetilde{b}_{i}(\bs{z})\varepsilon_{ij}\widetilde{b}_{i^{\prime}}(\bs{z})\varepsilon_{i^{\prime}j'}\right\}\\
&= E\left[\mathbf{B}_{\eta}(\bs{z})^{\top}\bs{\Upsilon}_{n}^{-1}\frac{1}{N^2}\sum_{j'',j'''=1}^{N} \widetilde{\mathbf{B}}_{\eta}(\bs{z}_{j''})\widetilde{\mathbf{B}}_{\eta}(\bs{z}_{j'''})^{\top}\sum_{\ell, \ell'=0}^{p}X_{i\ell}\widehat{\varepsilon}_{\ell}(\bs{z}_{j''})X_{i\ell'}
\widehat{\varepsilon}_{\ell'}(\bs{z}_{j'''})\varepsilon_{ij}\varepsilon_{i^{\prime}j'}
\bs{\Upsilon}_{n}^{-1}\mathbf{B}_{\eta}(\bs{z})\right]\\
&=O(n^{-2}N^{-2}|\triangle|^{-4}).
\end{align*}
Therefore,
\begin{align*}
 E\left\{\widetilde{b}_{i}(\bs{z})\widetilde{\varepsilon}_{i}(\bs{z}^{\prime})\widetilde{b}_{i^{\prime}}(\bs{z})
 \widetilde{\varepsilon}_{i^{\prime}}(\bs{z}^{\prime})\right\}
&=\widetilde{\mathbf{B}}_{\eta}(\bs{z}^{\prime})^{\top}
\bs{\Upsilon}_{n}^{-1}\frac{1}{N^2}\sum_{j,j^{\prime}=1}^{N}
 E\left\{\widetilde{b}_{i}(\bs{z})\varepsilon_{ij}\widetilde{b}_{i^{\prime}}(\bs{z})
\varepsilon_{i^{\prime}j'}\right\}
\bs{\Upsilon}_{n}^{-1}\widetilde{\mathbf{B}}_{\eta}(\bs{z}^{\prime})\\
&=O(n^{-2}N^{-2}|\triangle|^{-4}),\\
 E\left[n^{-1}\sum_{i=1}^{n}\widetilde{b}_{i}(\bs{z})\widetilde{\varepsilon}_{i}(\bs{z}^{\prime})\right]^2
&=\frac{1}{n^2}\sum_{i, i^{\prime}=1}^{n} E\left\{\widetilde{b}_{i}(\bs{z})\widetilde{\varepsilon}_{i}(\bs{z}^{\prime})
\widetilde{b}_{i^{\prime}}(\bs{z})\widetilde{\varepsilon}_{i^{\prime}}(\bs{z}^{\prime})\right\}
=O(n^{-2}N^{-2}|\triangle|^{-4}).
\end{align*}
Finally, we show (\ref{EQ:eta_eps}). Note that
\[
\sum_{i=1}^{n}\eta_i(\bs{z})\widetilde{\varepsilon}_{i}(\bs{z}^{\prime})
=\sum_{i=1}^{n}\widetilde{\mathbf{B}}_{\eta}(\bs{z}^{\prime})^{\top}\bs{\Upsilon}_{n}^{-1}\frac{1}{N}\sum_{j=1}^{N} \widetilde{\mathbf{B}}_{\eta}(\bs{z}_{j})\sigma(\bs{z}_j)\varepsilon_{ij}\sum_{k=1}^{\infty}\xi_{ik}\lambda_k^{1/2}\psi_{k}(\bs{z}),
\]
where $ E\left\{n^{-1}\sum_{i=1}^{n}\eta_i(\bs{z})\widetilde{\varepsilon}_{i}(\bs{z}^{\prime})\right\}=0$, and
\begin{align*}
& E\Bigg\{n^{-1}\sum_{i=1}^{n}\eta_i(\bs{z})\widetilde{\varepsilon}_{i}(\bs{z}^{\prime})\Bigg\}^2
=n^{-1} E\{\eta_i(\bs{z})^2\} E\{\widetilde{\varepsilon}_{i}(\bs{z}^{\prime})^2\}
=n^{-1}G_{\eta}(\bs{z},\bs{z}^{\prime})\widetilde{\mathbf{B}}_{\eta}(\bs{z}^{\prime})^{\top}\\
&\quad \times 
\bs{\Upsilon}_{n}^{-1}\frac{1}{N^2}
\sum_{j=1}^{N} \widetilde{\mathbf{B}}_{\eta}(\bs{z}_{j}) \widetilde{\mathbf{B}}_{\eta}(\bs{z}_{j})^{\top}\sigma^2(\bs{z}_j)\bs{\Upsilon}_{n}^{-1}
\widetilde{\mathbf{B}}_{\eta}(\bs{z}^{\prime})=O(n^{-1}N^{-1}|\triangle_{\eta}|^{-2}).
\end{align*}
Thus, (\ref{EQ:eta_eps}) is obtained.
\end{proof}

\fontsize{12}{14pt plus.8pt minus .6pt}\selectfont
\vskip 0.1in  \noindent \textbf{Appendix B} \vskip 0.1in
\label{SEC:Supp02}
\renewcommand{\thesubsection}{B.\arabic{subsection}} 
\renewcommand{\thetable}{B.\arabic{table}} 
\renewcommand{\thefigure}{B.\arabic{figure}} 
\setcounter{table}{0}
\setcounter{figure}{0}

In this section, we provide some additional results from simulation studies and real application analysis.

\vskip .10in \noindent \textbf{B.1. More results of simulation studies} \vskip .10in

In Section 5.1 of the main paper, we illustrated the advantage of the proposed method over the complex horseshoe domain in \cite{Sangalli:Ramsay:Ramsay:13}. Figure \ref{FIG:triangulations_simu1} shows the two triangulations used for the horseshoe domain in this example. For implementation, the BPST method is conducted over triangulation, $\triangle_1$, while triangulation, $\triangle_2$, is used for PCST method. To visually compare different methods, we display the estimated coefficient functions for Case I (jump function) and Case II (smooth function) in Figures \ref{FIG:EST_Simu1_Jump} and \ref{FIG:EST_Simu1_Smooth}, respectively. The plots are obtained based on the setting: $n=50$, $\lambda_1=0.2$, $\lambda_2=0.05$, $\sigma=1.0$. Table \ref{TAB:eg1_01_2} summarizes the estimation results based on the noise level $\sigma=1.0$.

From these figures, one sees that the BPST and PCST estimates are both very close to the true coefficient functions. When the true coefficient functions are smooth, BPST provides the best estimation, while when the true coefficient function contains jumps, PCST provides a better estimation. The performance of the Tensor method will be affected by the design of the coefficient function. Moreover, from Figure \ref{FIG:EST_Simu1_Jump} and \ref{FIG:EST_Simu1_Smooth}, one can see that even when the coefficient function is smooth across the boundary, the estimation accuracy is also affected by the domain of the true signal, especially the pixels which are closed to the boundary. The performance of the Kernel method is not affected by the design of the coefficient functions, instead, it heavily depends on the noise level due to the three-stage structure. As the noise level increases, the Kernel estimates are getting more blurred.

\begin{figure}[ht]
	\begin{center}
		\begin{tabular}{ccccc}
			\includegraphics[width=4cm,height=3cm]{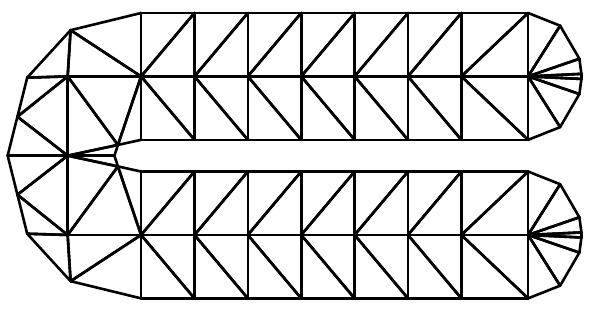} & \includegraphics[width=4cm,height=3cm]{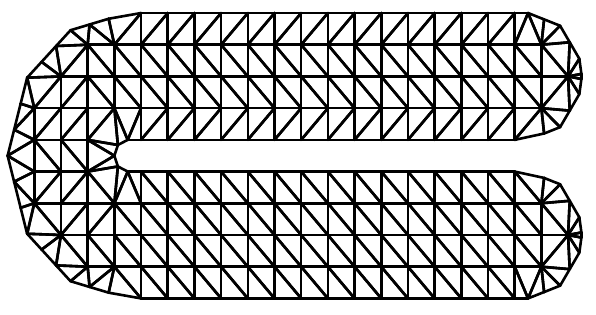} \\[-5pt]
			$\triangle_1$ & $\triangle_2$
		\end{tabular}
		\caption{Triangulations for the horseshoe domain. }
		\label{FIG:triangulations_simu1}
	\end{center}
\end{figure}

\begin{figure}[ht]
	\begin{center}
		\begin{tabular}{ccccccc}
			TRUE & Tensor & Kernel & BPST ($\triangle_1$) & PCST ($\triangle_2$) \\
			\includegraphics[scale=0.27]{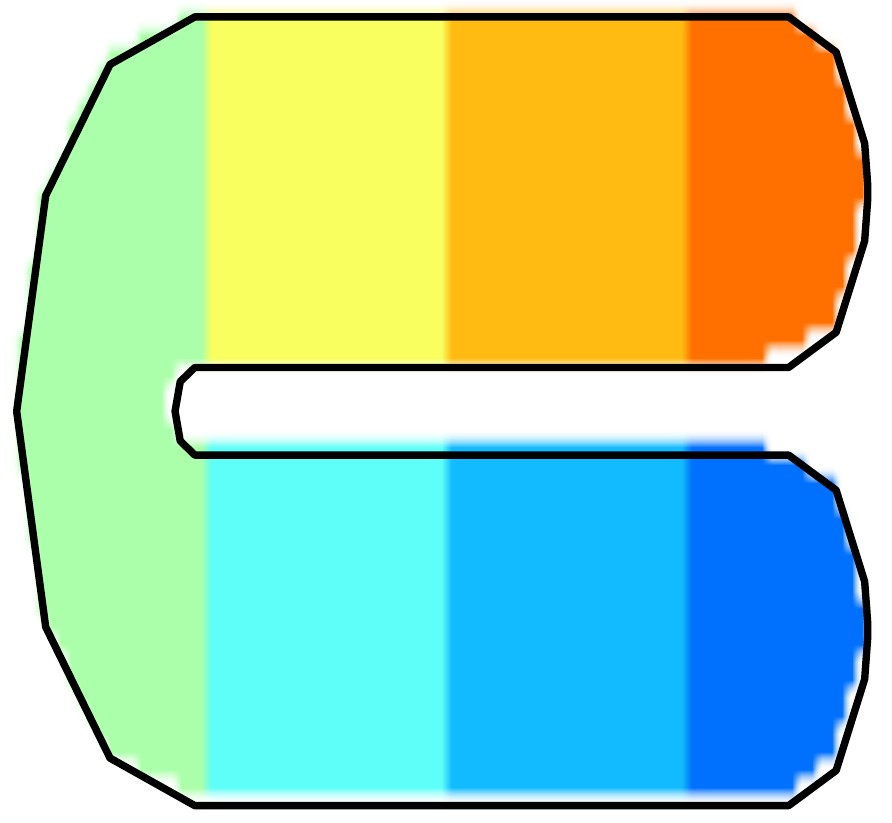} \!\!\!&\!\!\!
			\includegraphics[scale=0.27]{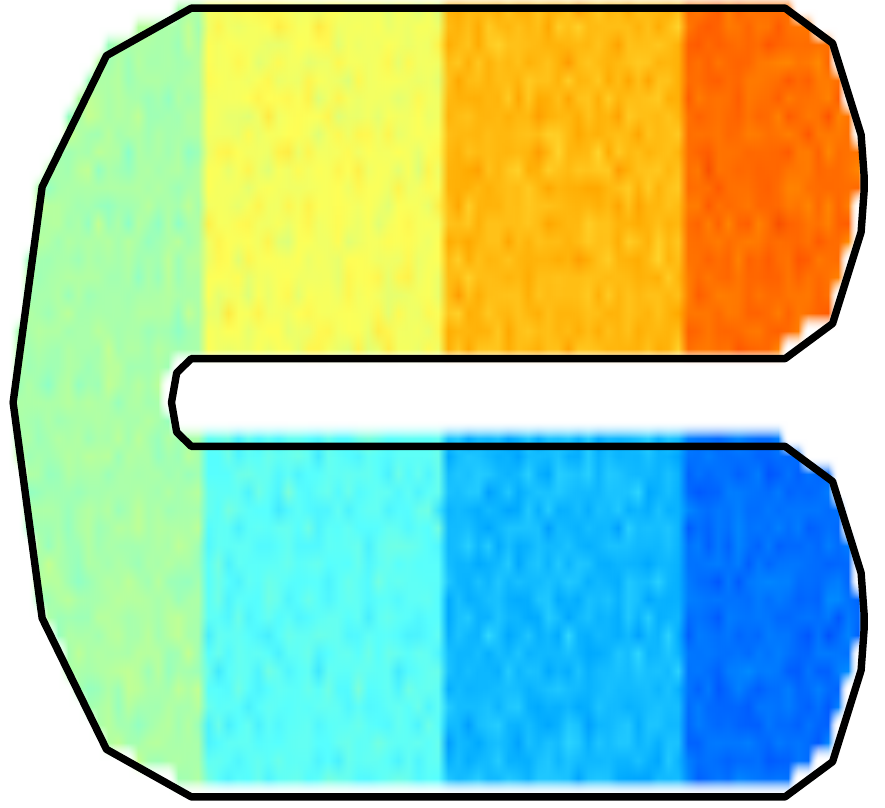} \!\!\!&\!\!\!
			\includegraphics[scale=0.27]{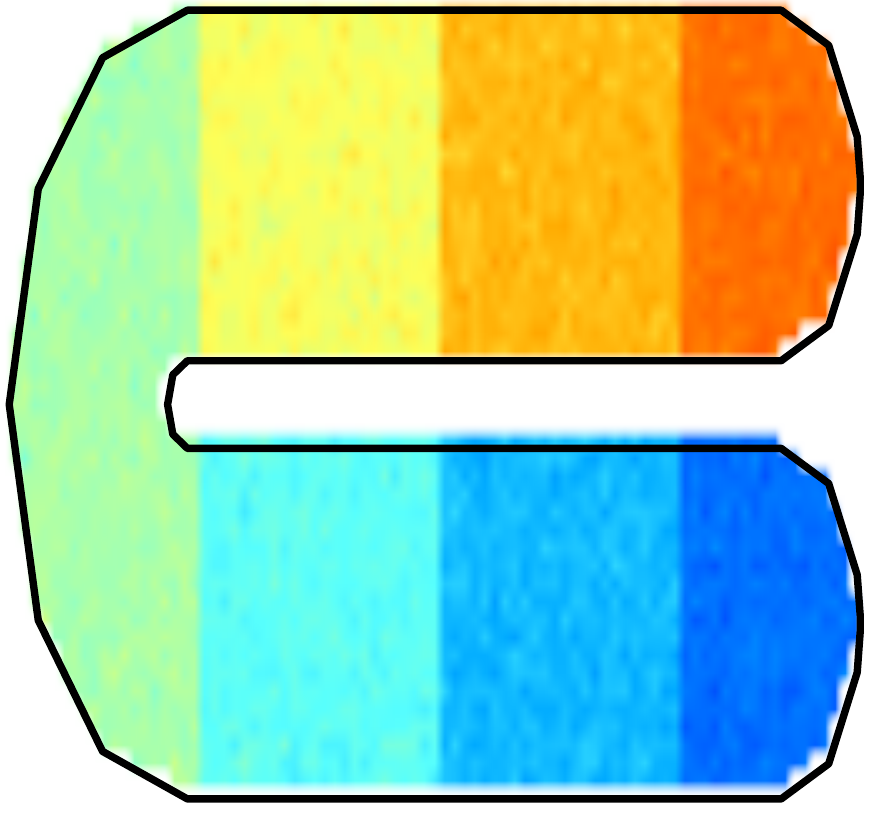} \!\!\!&\!\!\!
			\includegraphics[scale=0.27]{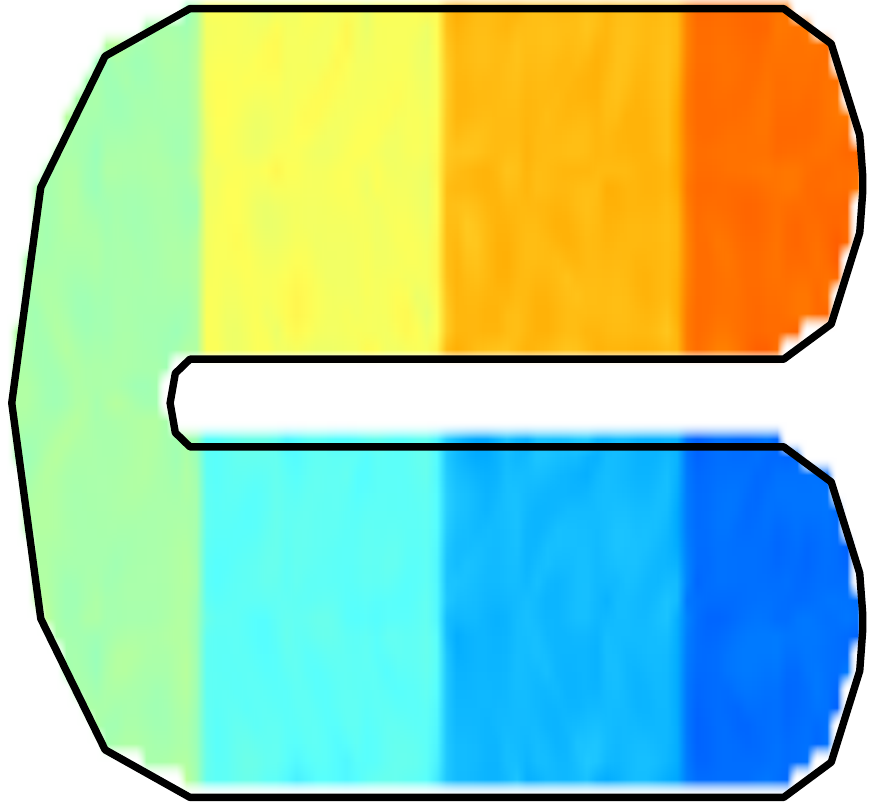} \!\!\!&\!\!\!
			\includegraphics[scale=0.27]{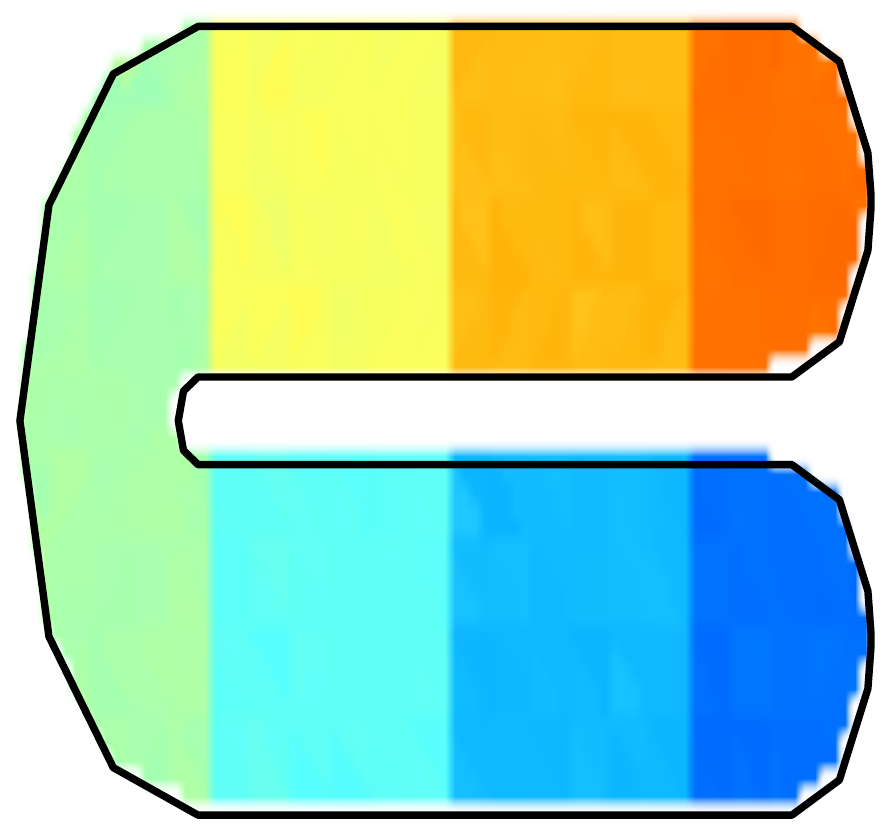} \!\!\!&\!\!\!&\includegraphics[scale=0.27]{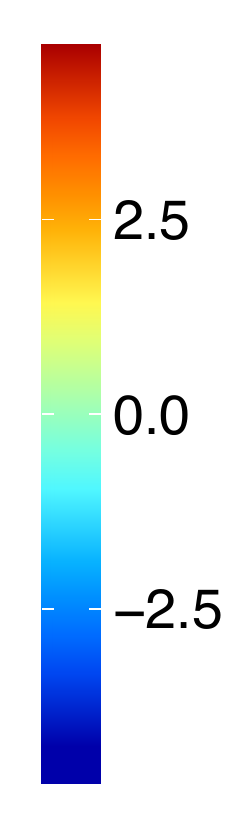}\\[-5pt]
			\multicolumn{6}{c}{$\beta_0$}\\
			\includegraphics[scale=0.27]{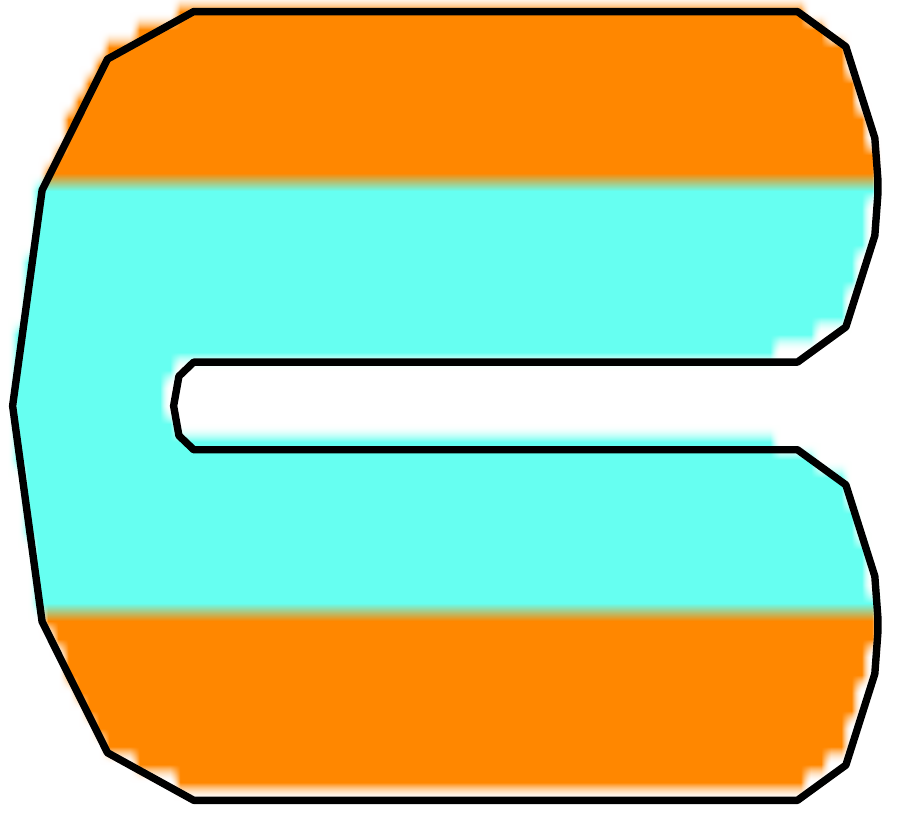} \!\!\!&\!\!\!
			\includegraphics[scale=0.27]{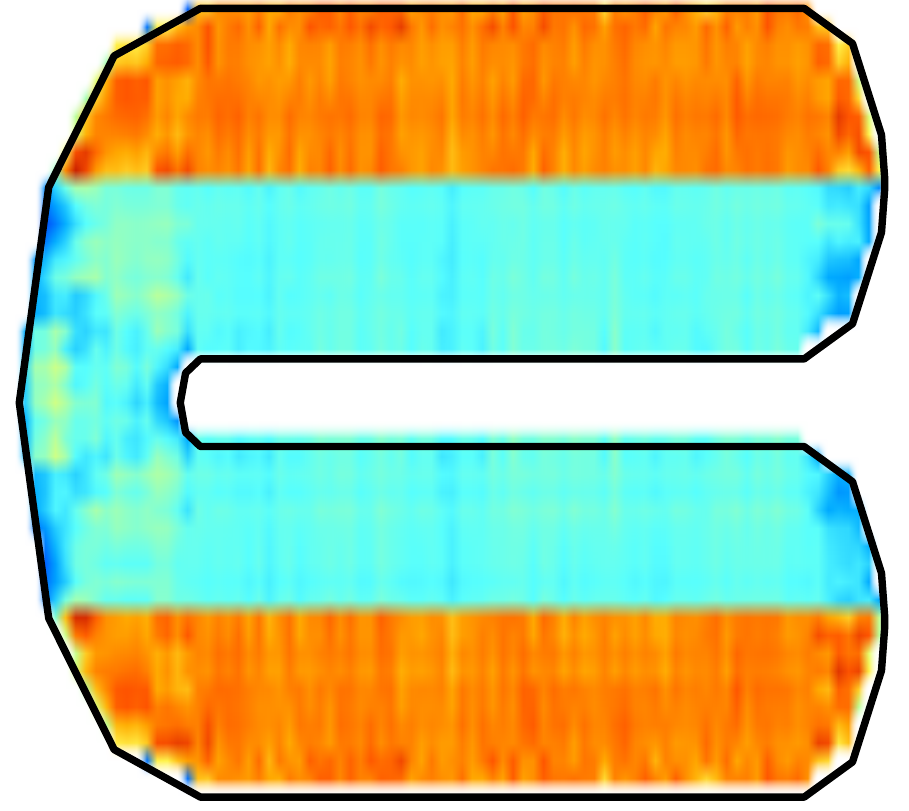} \!\!\!&\!\!\!
			\includegraphics[scale=0.27]{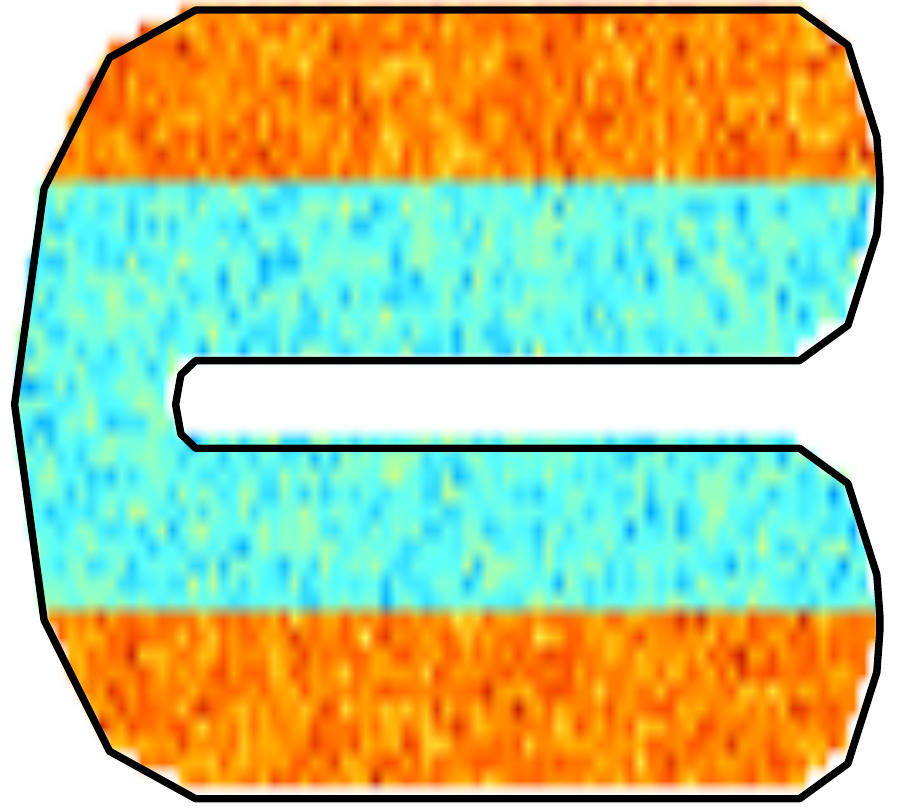} \!\!\!&\!\!\!
			\includegraphics[scale=0.27]{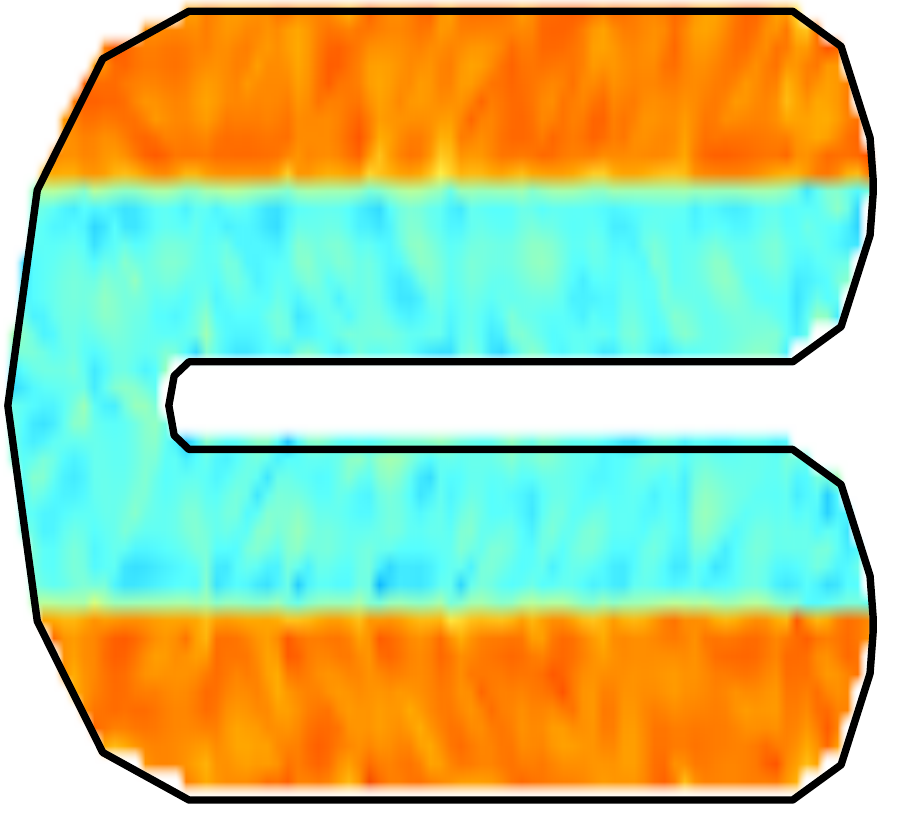} \!\!\!&\!\!\!
			\includegraphics[scale=0.27]{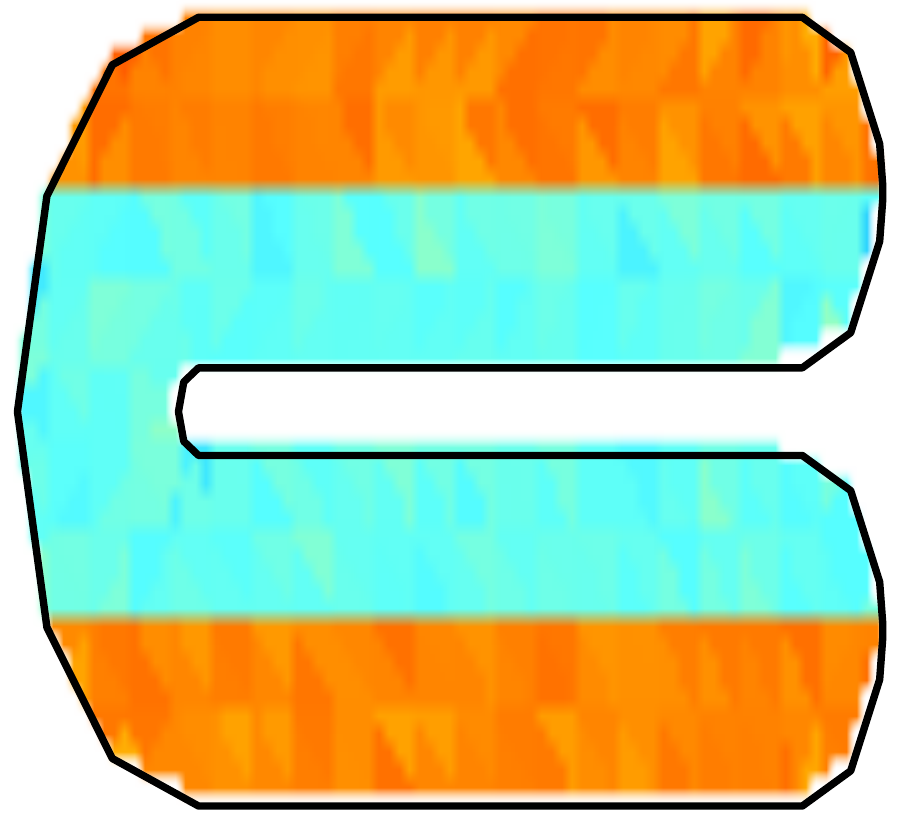} \!\!\!&\!\!\!&\includegraphics[scale=0.27]{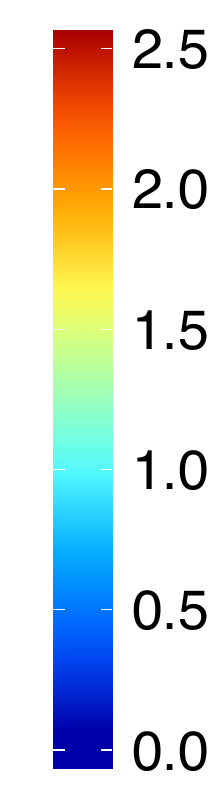}\\[-5pt]
			\multicolumn{6}{c}{$\beta_1$}\\
		\end{tabular}
		\caption{True coefficient functions and their different estimators for Case I in Example 1.}
		\label{FIG:EST_Simu1_Jump}
	\end{center}
\end{figure}

\begin{figure}[ht]
	\begin{center}
		\begin{tabular}{ccccccc}
			TRUE & Tensor & Kernel & BPST ($\triangle_1$) & PCST ($\triangle_2$) \\
			\includegraphics[scale=0.27]{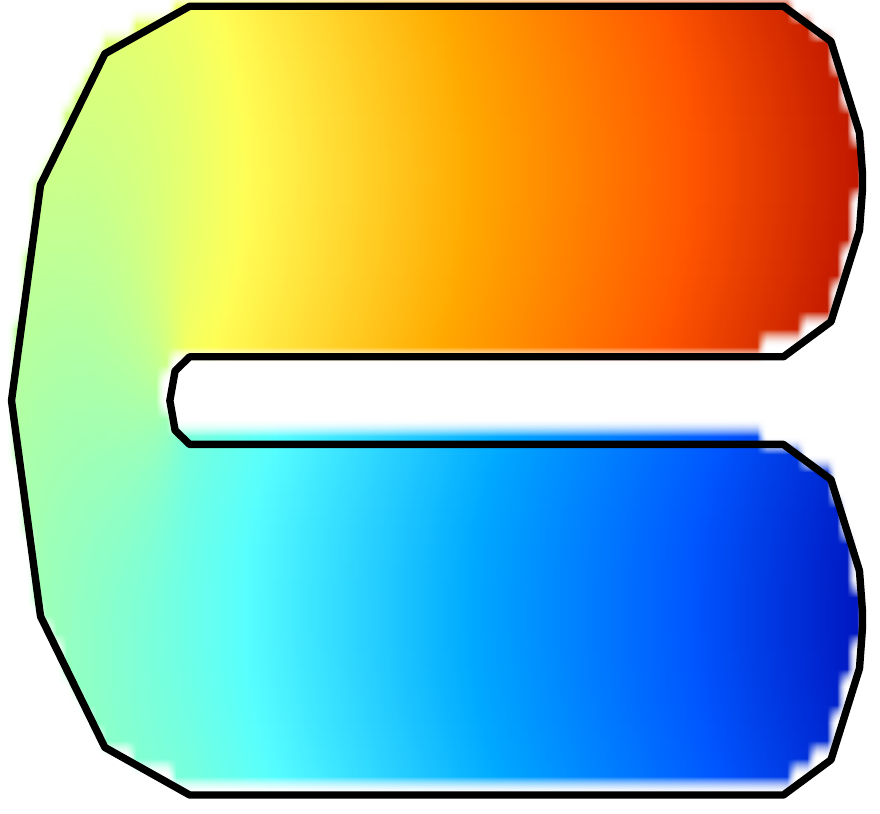} \!\!\!&\!\!\!
			\includegraphics[scale=0.27]{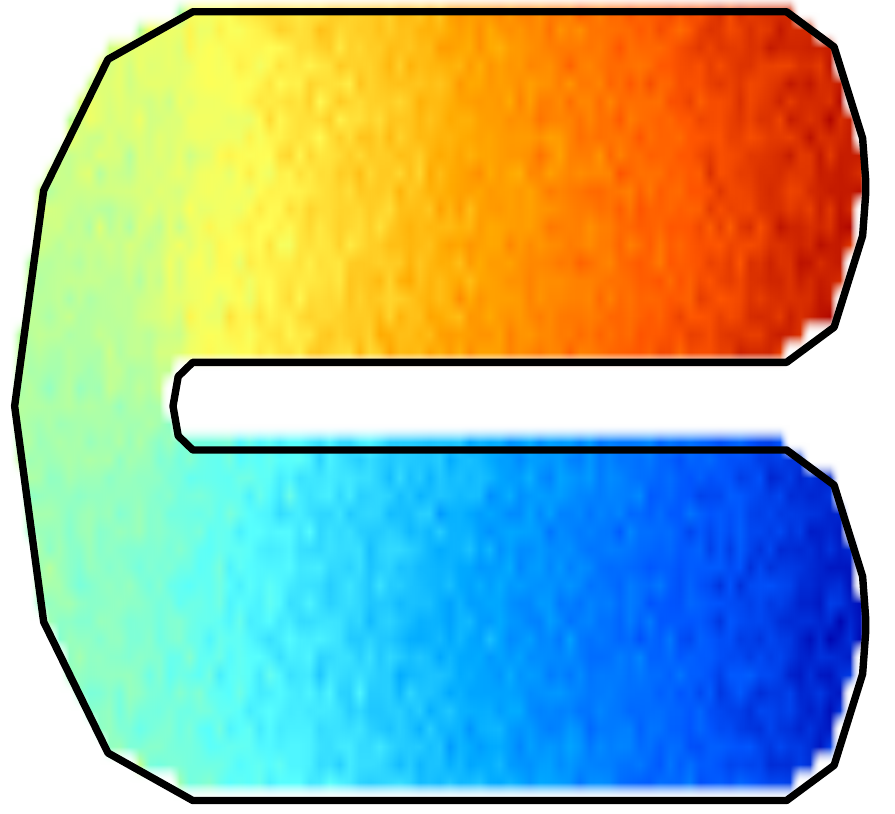} \!\!\!&\!\!\!
			\includegraphics[scale=0.27]{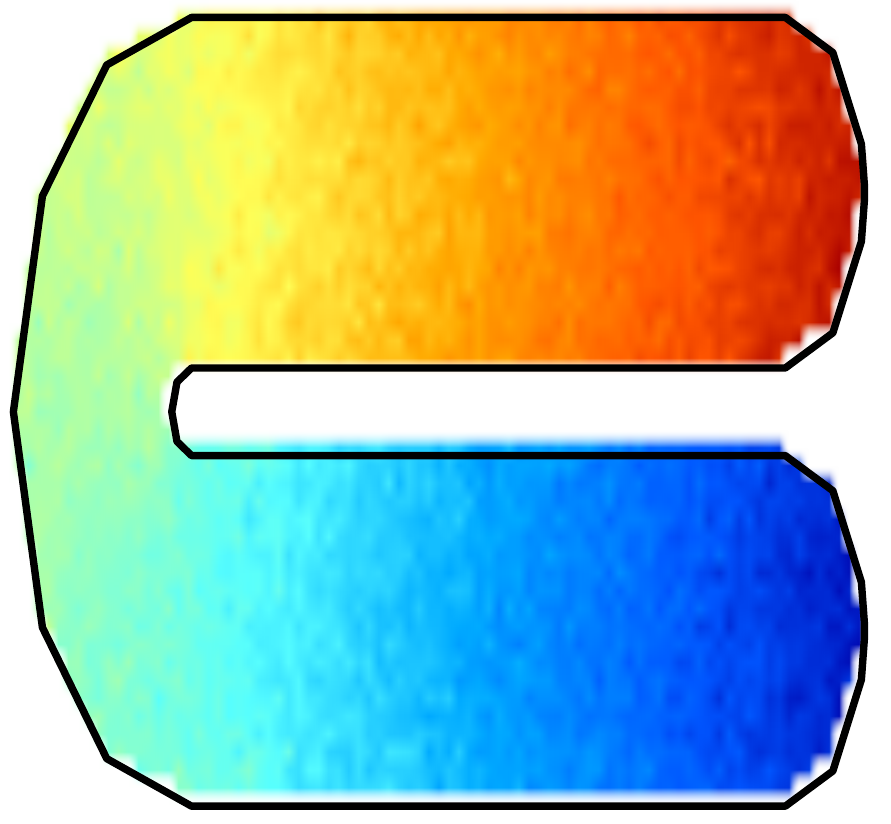} \!\!\!&\!\!\!
			\includegraphics[scale=0.27]{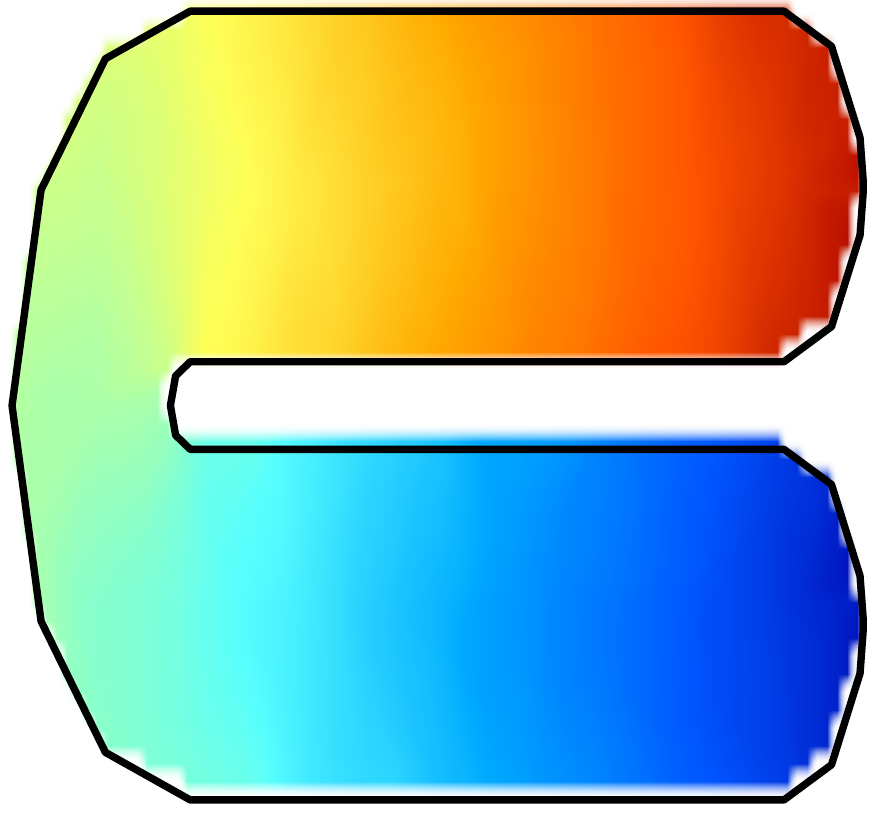} \!\!\!&\!\!\!
			\includegraphics[scale=0.27]{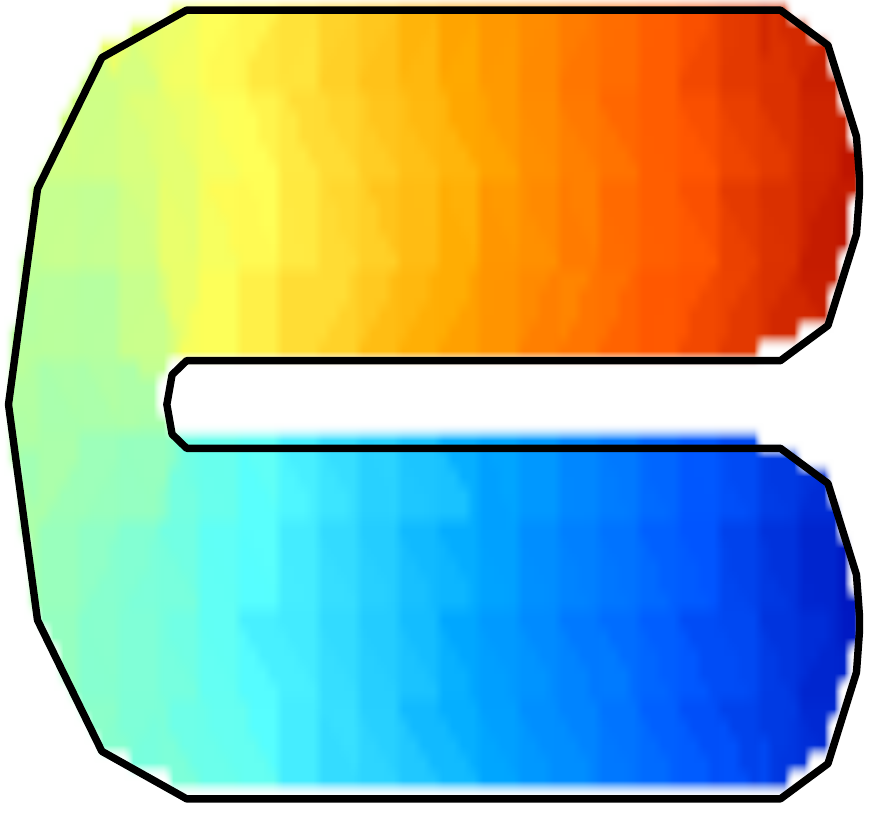} \!\!\!&\!\!\!&\includegraphics[scale=0.27]{legend1}\\[-5pt]
			\multicolumn{6}{c}{$\beta_0$}\\
			\includegraphics[scale=0.27]{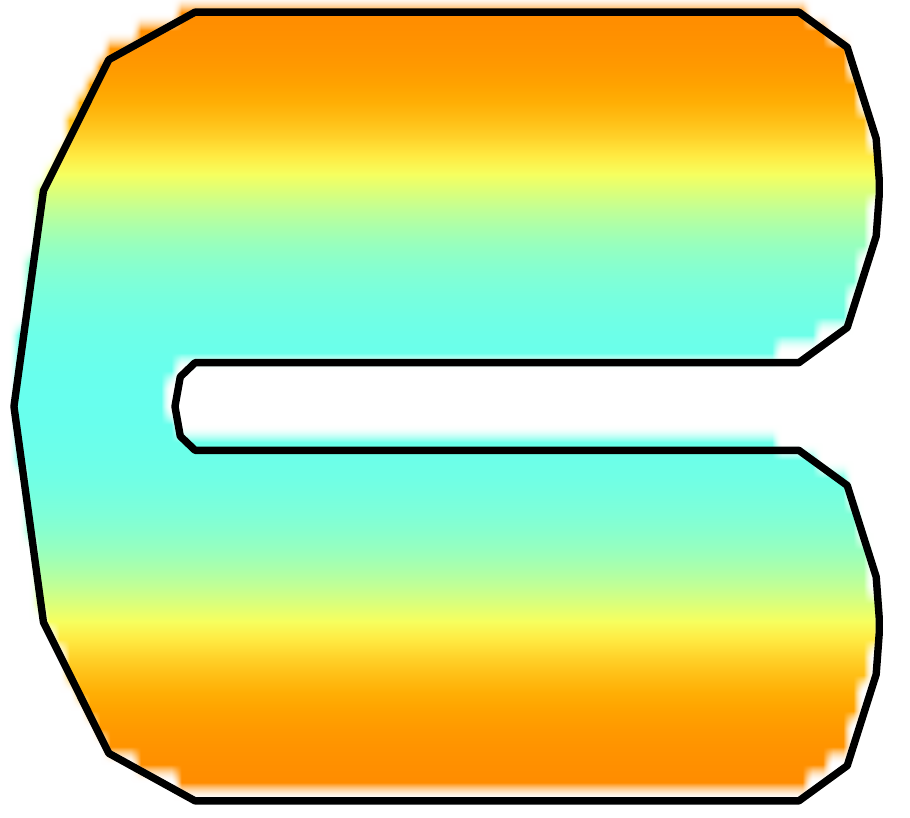} \!\!\!&\!\!\!
			\includegraphics[scale=0.27]{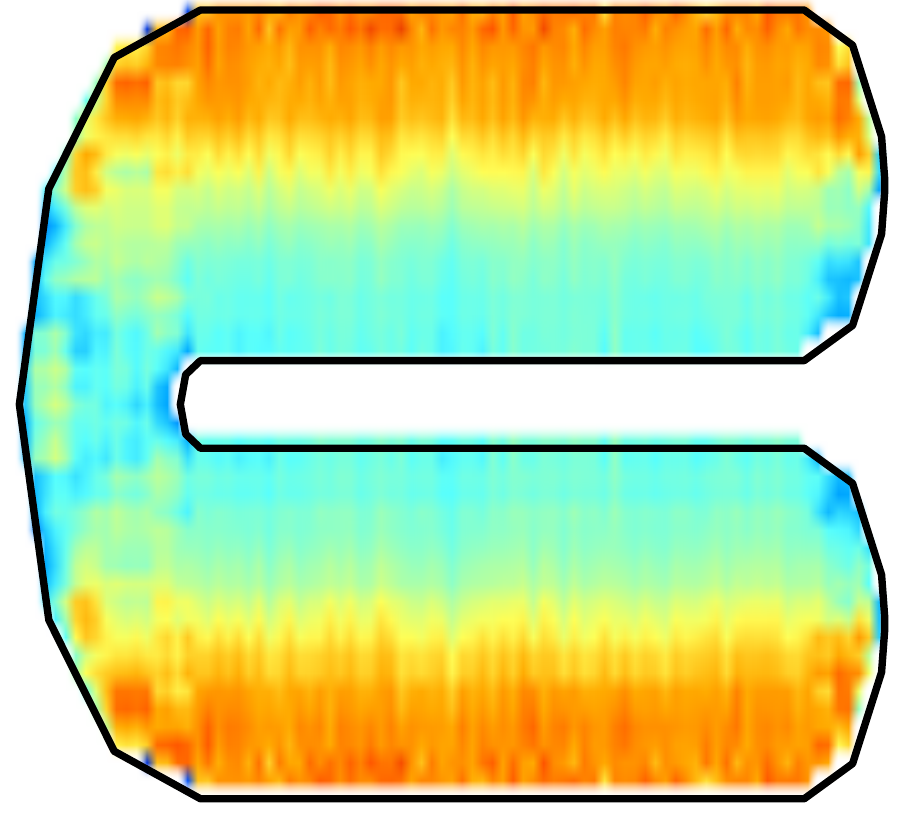} \!\!\!&\!\!\!
			\includegraphics[scale=0.27]{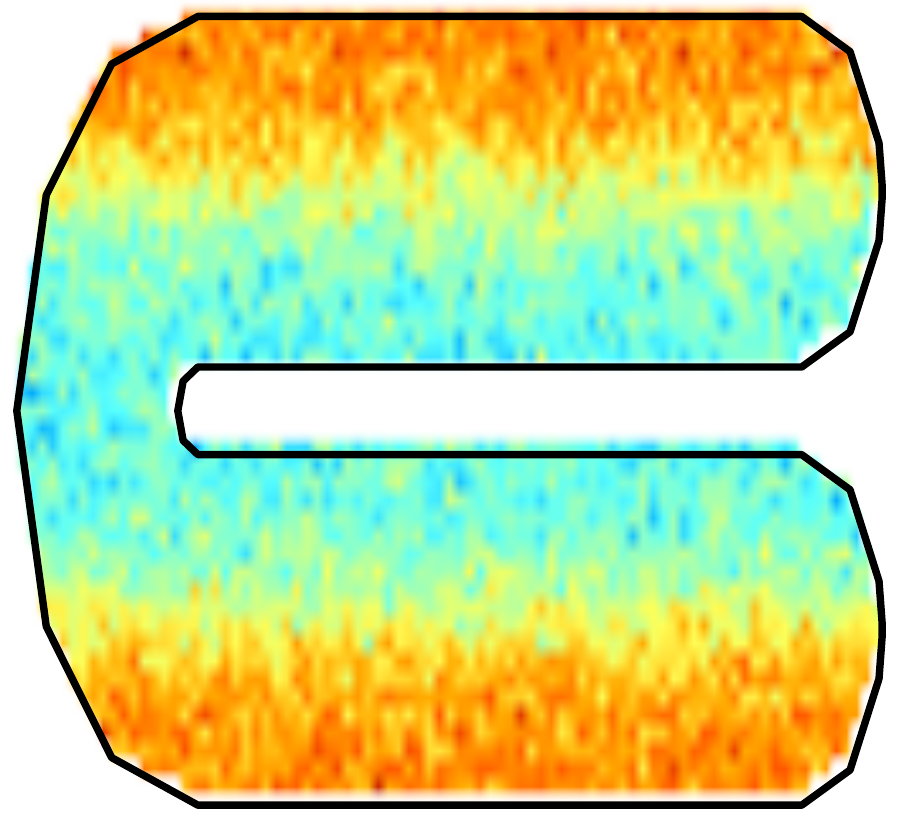} \!\!\!&\!\!\!
			\includegraphics[scale=0.27]{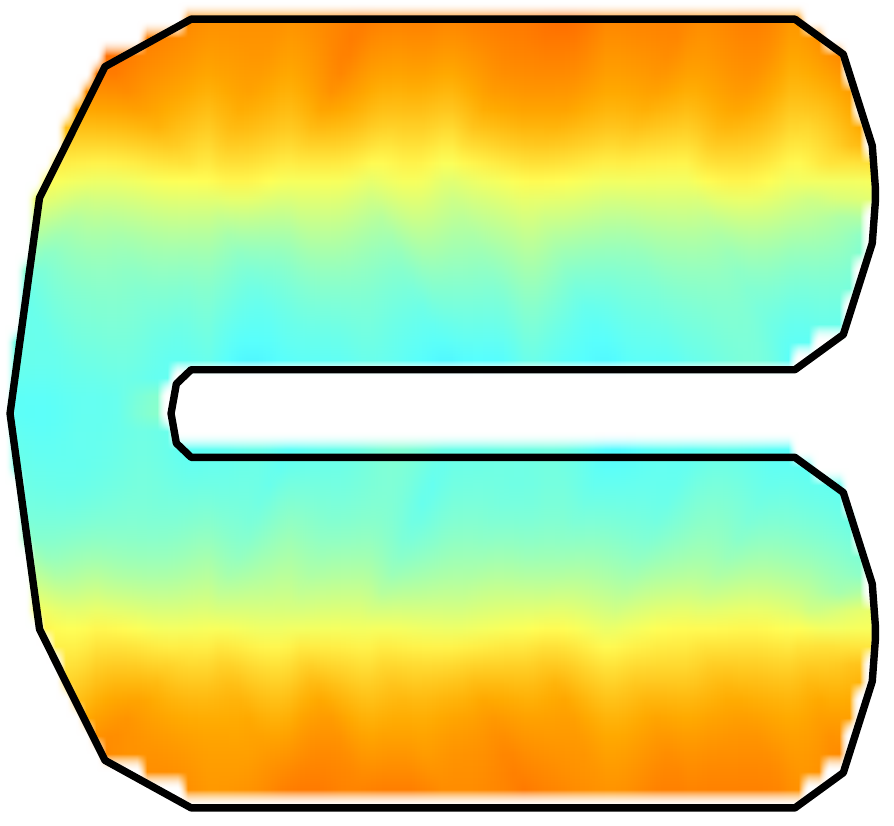} \!\!\!&\!\!\!
			\includegraphics[scale=0.27]{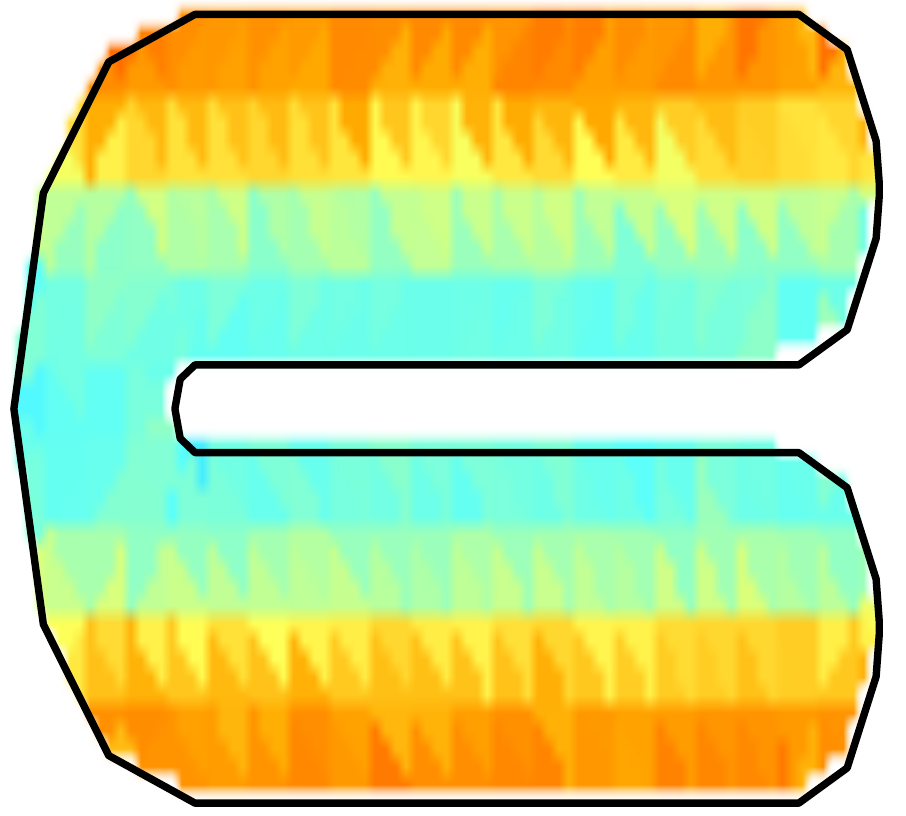} \!\!\!&\!\!\!&\includegraphics[scale=0.27]{legend2}\\[-5pt]
			\multicolumn{6}{c}{$\beta_1$}\\
		\end{tabular}
		\caption{True coefficient functions and their different estimators for Case II in Example 1.}
		\label{FIG:EST_Simu1_Smooth}
	\end{center}
\end{figure}

\begin{table}[t]
\begin{center}
\caption{Estimation errors of the coefficient estimators, $\sigma=1.0$. \label{TAB:eg1_01_2}}
\renewcommand\arraystretch{0.75}
\scalebox{0.9}{\begin{tabular}{cccccccccr}\\ \hline\hline
Function &\multirow{2}{*}{$n$} &\multirow{2}{*}{Method} &\multicolumn{2}{c}{$\lambda_1=0.03,~\lambda_2=0.006$}& &\multicolumn{2}{c}{$\lambda_1=0.2,~\lambda_2=0.05$}  \\ \cline{4-5} \cline{7-8}
Type &&&$\beta_0$ &$\beta_1$ & &$\beta_0$ &$\beta_1$\\ \hline

\multirow{8}{*}{Jump} &\multirow{4}{*}{50} &BPST &0.0059 &0.0075 & &0.0066 &0.0082 \\
& & PCST &0.0023 &0.0023 & &0.0028 &0.0030 \\
& & Kernel &0.0201 &0.0206 & &0.0207 &0.0213 \\ 
& & Tensor &0.0201 &0.0132 & &0.0206 &0.0142 \\ \cline{2-8}
	
&\multirow{4}{*}{100} &BPST &0.0038 &0.0050 & &0.0042 &0.0054 \\
& & PCST &0.0011 &0.0011 & &0.0014 &0.0015 \\
& & Kernel &0.0100 &0.0102 & &0.0104 &0.0105 \\ 
& & Tensor &0.0099 &0.0112 & &0.0103 &0.0120 \\ \cline{1-8}
	
\multirow{8}{*}{Smooth} &\multirow{4}{*}{50} &BPST &0.0010 &0.0012 & &0.0016 &0.0019 \\
& &PCST &0.0049 &0.0065 & &0.0054 &0.0072 \\
& &Kernel &0.0201 &0.0206 & &0.0207 &0.0213 \\ 
& &Tensor &0.0189 &0.0132 & &0.0207 &0.0153\\ \cline{2-8}
		
&\multirow{4}{*}{100} &BPST &0.0006 &0.0007 & &0.0009 &0.0010 \\
& & PCST &0.0037 &0.0054 & &0.0040 &0.0057 \\
& & Kernel &0.0100 &0.0102 & &0.0104 &0.0105 \\ 
& & Tensor &0.0100 &0.0113 & &0.0103 &0.0128\\ \hline\hline
\end{tabular}}
\end{center}
\end{table}

In Section 5.2 of the main paper, we conduct a simulation study based on the domain of the 5th slice of the brain images illustrated in Section 6. Table \ref{TAB:estimation2} demonstrates the estimation results for $\sigma=0.5$. In this example, we focus on the domain of the 35th slices of the brain image. Based on this domain, we consider two types of triangulations: $\triangle_5$ and $\triangle_6$; see Figure \ref{FIG:triangulations_simu2}. Table \ref{TAB:estimation_s35} summarizes the MSE results of the BPST, kernel and tensor methods. The findings are similar to those described in Section 5.2. Tables \ref{TAB:coverage} and \ref{TAB:coverage_s35} summarize the ECRs of the 95\% SCCs for the 5th and 35th slices, respectively, and they are all close to 95\%. As the sample size increases, the ECRs are getting closer to 95\%. Figures \ref{FIG:EST_SCC} and \ref{FIG:EST_SCC_s35} show the true coefficient functions and an example of their estimators and 95\% SCCs based on the 5th and 35th slices, respectively. The plots are generated based on the setting: $n=50$, $\lambda_1=0.1$, $\lambda_2=0.02$ and $\sigma=0.5$.

\begin{figure}[ht]
	\begin{center}
		\begin{tabular}{ccccc}
			\includegraphics[scale=0.6]{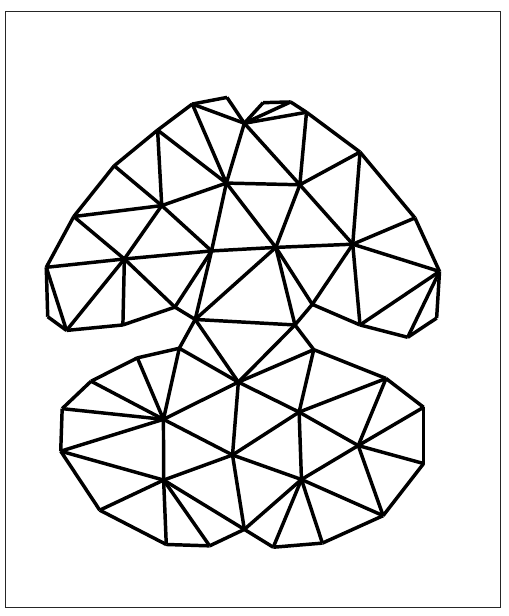} & \includegraphics[scale=0.6]{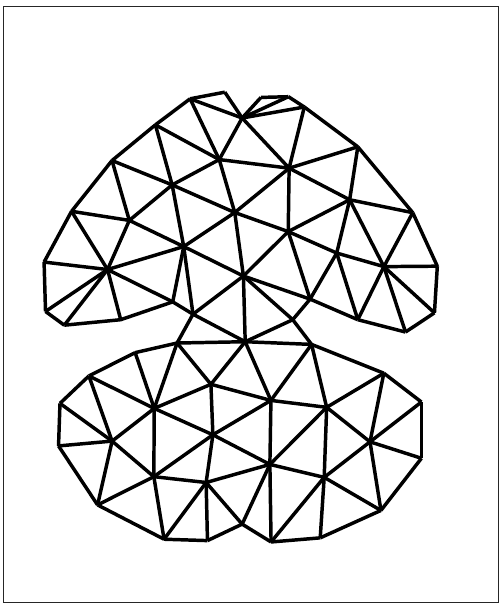} ~~&~~\includegraphics[scale=0.6]{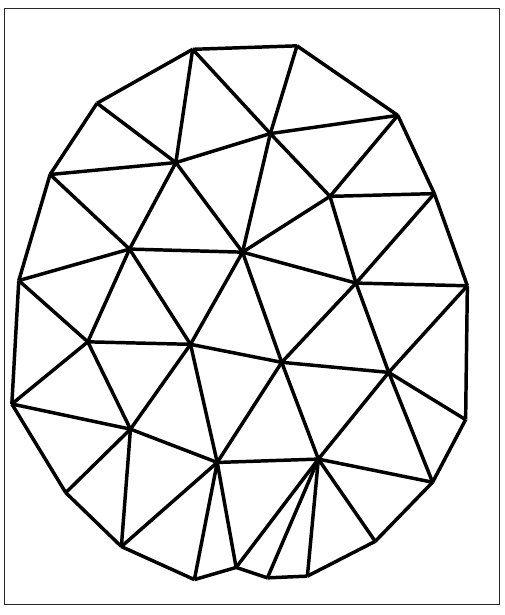} & \includegraphics[scale=0.6]{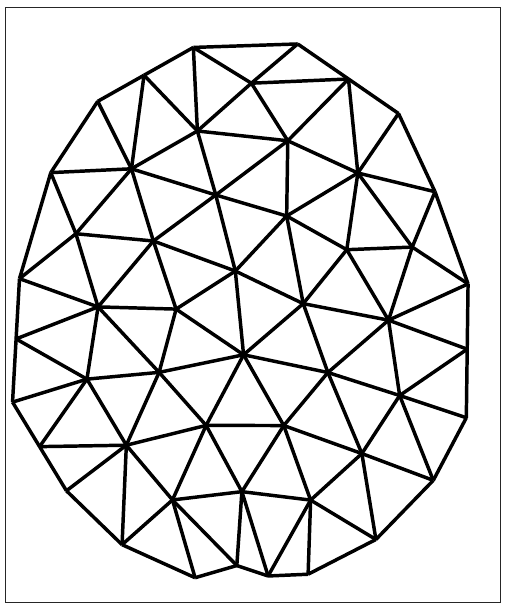} \\[-5pt]
			$\triangle_3$ & $\triangle_4$ & $\triangle_5$&$\triangle_6$
		\end{tabular}
		\caption{Triangulations for the fifth slice ($\triangle_3$, $\triangle_4$) and 35th slice ($\triangle_5$, $\triangle_6$) of the brain image in Simulation Example 2.}
		\label{FIG:triangulations_simu2}
	\end{center}
\end{figure}

\begin{table}[t]
\begin{center}
\caption{Estimation errors of the coefficient function estimators, $\sigma=0.5$. \label{TAB:estimation2}}%
\renewcommand\arraystretch{0.75}
\scalebox{0.9}{\begin{tabular}{lccccccccccr}\\ \hline\hline
\multirow{2}{*}{$n$}  &\multirow{2}{*}{Method} &\multicolumn{3}{c}{$\lambda_1=0.1,~\lambda_2=0.02$}& \multicolumn{3}{c}{$\lambda_1=0.2,~\lambda_2=0.05$} \\ \cline{3-6} \cline{7-8}
	&&$\beta_0$ &$\beta_1$ &$\beta_2$ & $\beta_0$ &$\beta_3$ &$\beta_2$\\ \hline
	\multirow{4}{*}{50} &BPST($\triangle_{3}$) &0.003&0.005&0.005&0.007&0.011&0.010 \\
	&BPST($\triangle_{4}$) & 0.003&0.005&0.005&0.006& 0.009& 0.009\\
	&Kernel &0.008&   0.011  & 0.011&0.011 & 0.016 & 0.016\\
	&Tensor & 0.008&    0.007 &  0.010 & 0.011 & 0.012& 0.014\\ \cline{1-8}
	
	\multirow{4}{*}{100} &BPST($\triangle_{3}$)& 0.002&0.002&0.002 &0.003&0.005&0.005 \\
	&BPST($\triangle_{4}$)& 0.002&0.002&0.002 &0.003&0.004&0.004 \\
	&Kernel & 0.004 & 0.005& 0.005 & 0.005 & 0.008 & 0.007\\
	&Tensor &0.004 &   0.005 &  0.005 &0.005 & 0.007& 0.009\\ \hline\hline
	\end{tabular}}
\end{center}
\end{table}

\begin{table}[t]
	\begin{center}
		\caption{Estimation errors of the coefficient function estimators in the 35th slice. \label{TAB:estimation_s35}}%
		\renewcommand\arraystretch{0.75}
		\begin{tabular}{lccccccccccr}\\ \hline\hline
			\multirow{2}{*}{$n$} &\multirow{2}{*}{$\sigma$} &\multirow{2}{*}{Method} &\multicolumn{3}{c}{$\lambda_1=0.1,~\lambda_2=0.02$}& \multicolumn{3}{c}{$\lambda_1=0.2,~\lambda_2=0.05$} \\ \cline{4-6} \cline{7-9}
			&&&$\beta_0$ &$\beta_1$ &$\beta_2$ & $\beta_0$ &$\beta_1$ &$\beta_2$\\ \hline
			
			\multirow{8}{*}{50} &\multirow{4}{*}{0.5}&BPST($\triangle_5$)&0.003&0.005&0.005&0.007&0.011&0.011\\
			&&BPST($\triangle_6$)&0.003&0.005&0.005&0.007&0.011&0.010\\
			&&Kernel &0.008& 0.012&0.012 & 0.018& 0.018& 0.017\\
			&&Tensor & 0.008&   0.009  & 0.011& 0.012 &    0.015&    0.015 \\ \cline{2-9}
			
			&\multirow{4}{*}{1.0}&BPST($\triangle_5$) &0.003&0.005&0.005&0.007&0.011&0.011\\
			&&BPST($\triangle_6$) &0.003&0.005&0.005&0.007&0.011&0.011\\
			&&Kernel &0.023&0.033&0.033 & 0.027& 0.039& 0.038\\
			&&Tensor &0.023&0.012&0.019& 0.027 &  0.017& 0.023\\ \hline
			
			\multirow{8}{*}{100} &\multirow{4}{*}{0.5}&BPST($\triangle_5$) &0.002&0.002&0.002 &0.003&0.005&0.005\\
			&&BPST($\triangle_6$)&0.002&0.002&0.002 &0.003&0.005&0.005\\
			&&Kernel &0.004 & 0.006 & 0.006 & 0.006 & 0.008 & 0.008\\
			&&Tensor &0.004&0.006&0.007&0.006&  0.010&    0.009
			\\ \cline{2-9}
			
			&\multirow{4}{*}{1.0}&BPST($\triangle_5$) &0.002&0.002&0.002 &0.003&0.005&0.005\\
			&&BPST($\triangle_6$) &0.002&0.002&0.002 &0.003&0.005&0.005\\
			&&Kernel & 0.012 &  0.016 &  0.016& 0.013 & 0.018& 0.018\\
			&&Tensor & 0.013&0.010&    0.013& 0.011& 0.007& 0.011\\
			\hline\hline
		\end{tabular}
	\end{center}
\end{table}

\begin{table}[t]
	\begin{center}
		\renewcommand*{\arraystretch}{0.75}
		\caption{The coverage rate of the $95\%$ SCCs for the coefficient functions defined over the 35th slice. \label{TAB:coverage_s35}}
		\begin{tabular}{lcccccccccr}\\\hline\hline
			\multirow{2}{*}{$n$}&\multirow{2}{*}{$\lambda$}&\multirow{2}{*}{$\sigma$}&\multicolumn{3}{c}{Coverage}&\multicolumn{3}{c}{Width}\\\cline{4-9}
			&&&$\beta_0$&$\beta_1$&$\beta_2$&$\beta_0$&$\beta_1$&$\beta_2$\\\hline
			\multirow{4}{*}{50}&\multirow{2}{*}{(0.1,0.02)}&0.5&0.962&0.916&0.934&0.307&0.344&0.347\\
			&&1.0&0.964&0.926&0.940&0.331&0.368&0.371\\
			\cline{2-9}
			&\multirow{2}{*}{(0.2,0.05)}&0.5&0.952&0.920&0.930&0.426&0.490&0.492\\
			&&1.0&0.96&0.920&0.934&0.449&0.512&0.512\\
			\hline
			\multirow{4}{*}{100}&\multirow{2}{*}{(0.1,0.02)}&0.5&0.956&0.952&0.940&0.214&0.240&0.244\\
			&&1.0&0.962&0.952&0.948&0.239&0.262&0.265\\
			\cline{2-9}
			&\multirow{2}{*}{(0.2,0.05)}&0.5&0.946&0.954&0.932&0.298&0.340&0.346\\
			&&1.0&0.952&0.954&0.938&0.317&0.359&0.365\\
			\hline\hline
		\end{tabular}
	\end{center}
\end{table}

\begin{figure}[ht]
	\begin{center}
		\begin{tabular}{ccccccccc}
			TRUE & Kernel& Tensor & BPST($\triangle_1$) & BPST($\triangle_2 $) & Lower SCC &Upper SCC\\
			\includegraphics[scale=0.19]{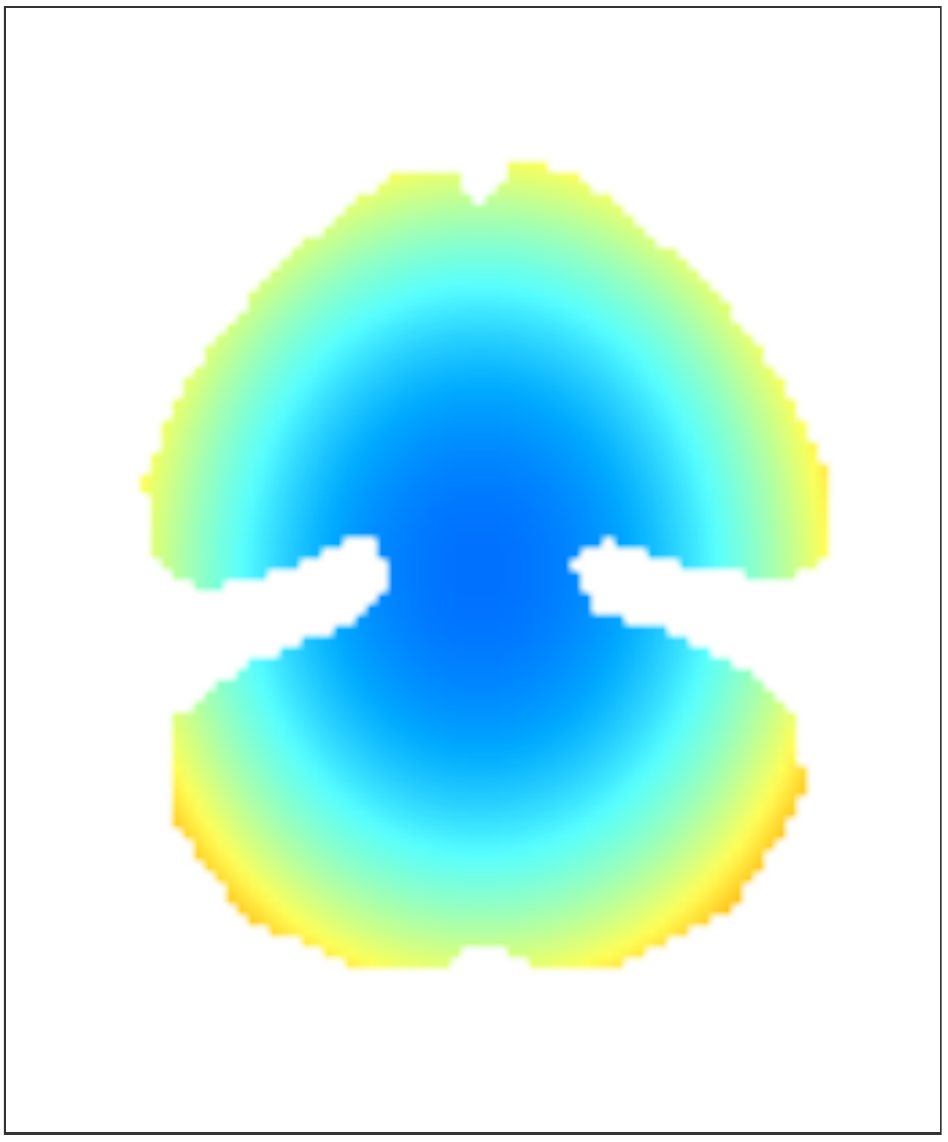} \!\!\!&\!\!\!
			\includegraphics[scale=0.19]{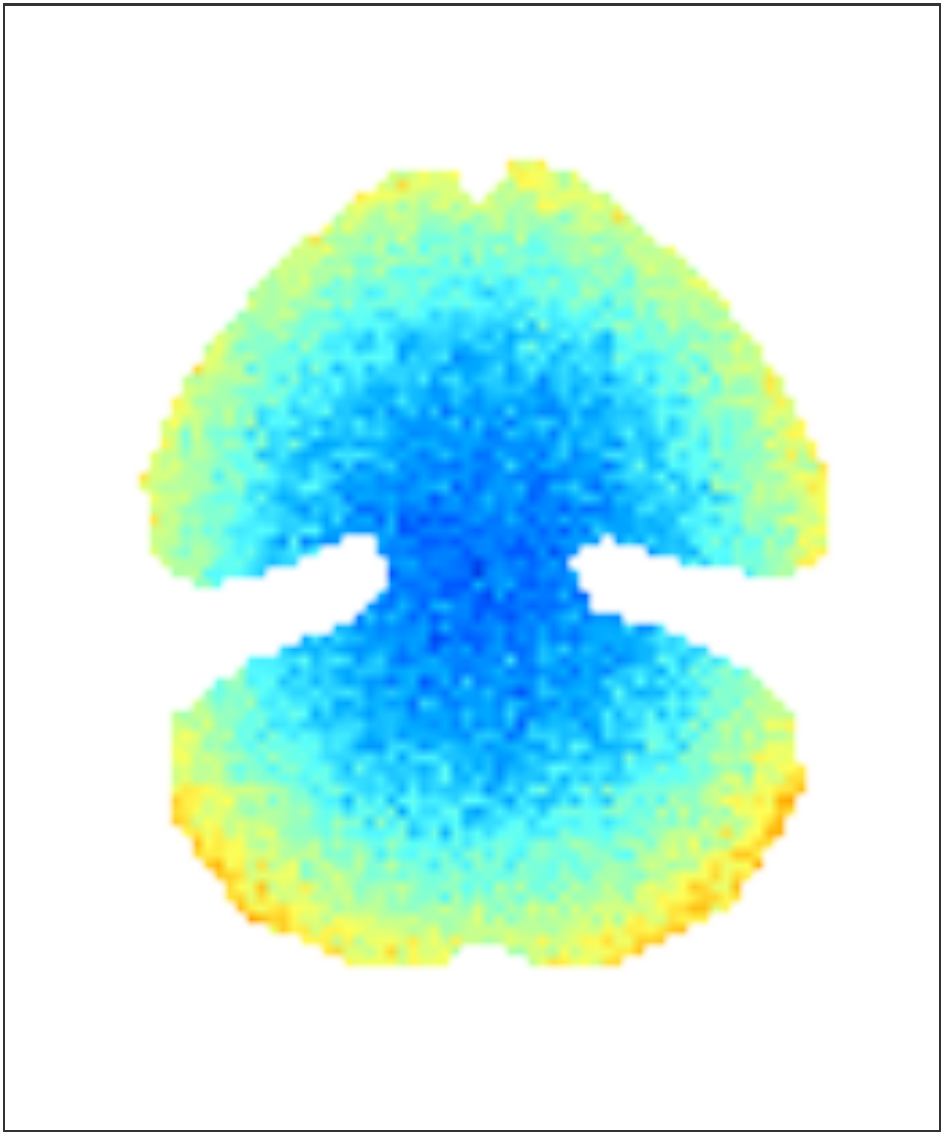} \!\!\!&\!\!\!
			\includegraphics[scale=0.19]{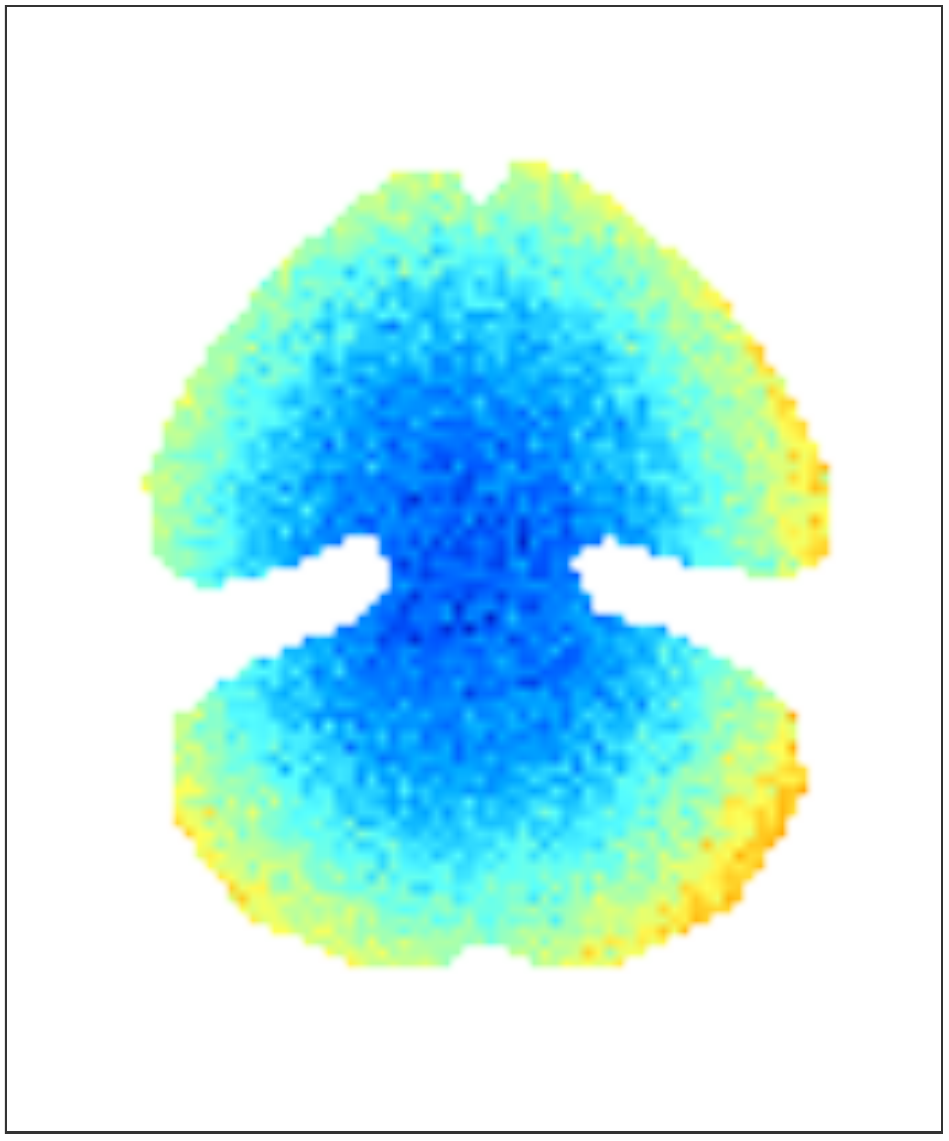} \!\!\!&\!\!\!
			\includegraphics[scale=0.19]{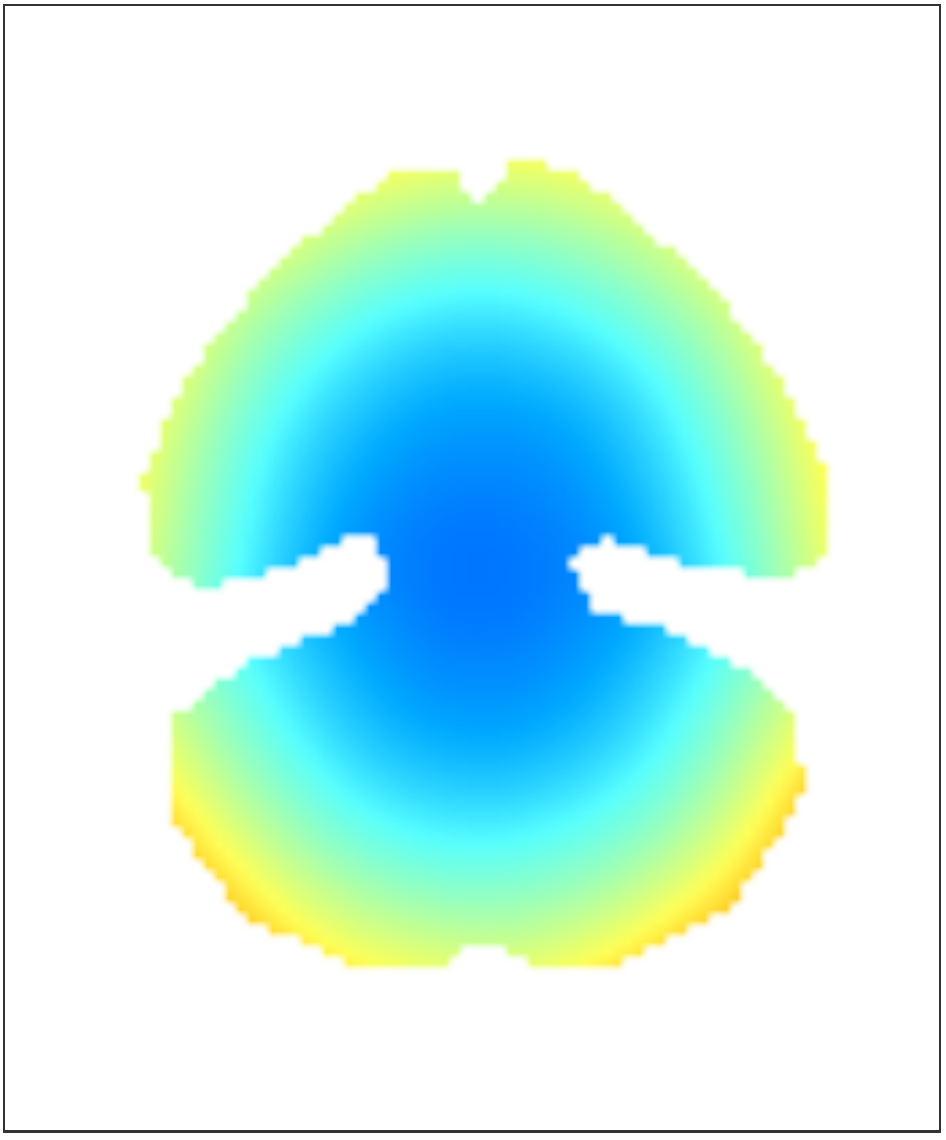} \!\!\!&\!\!\!
			\includegraphics[scale=0.19]{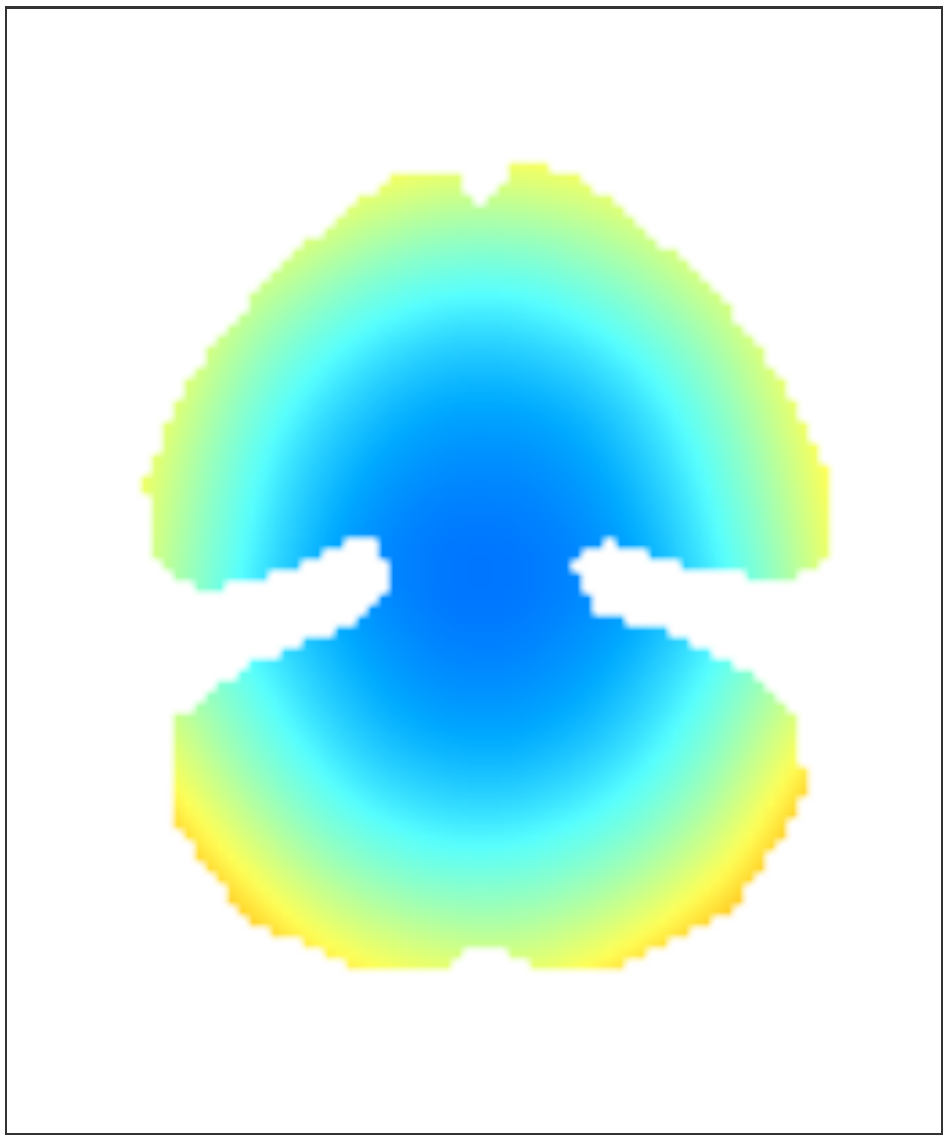} \!\!\!&\!\!\!
			\includegraphics[scale=0.19]{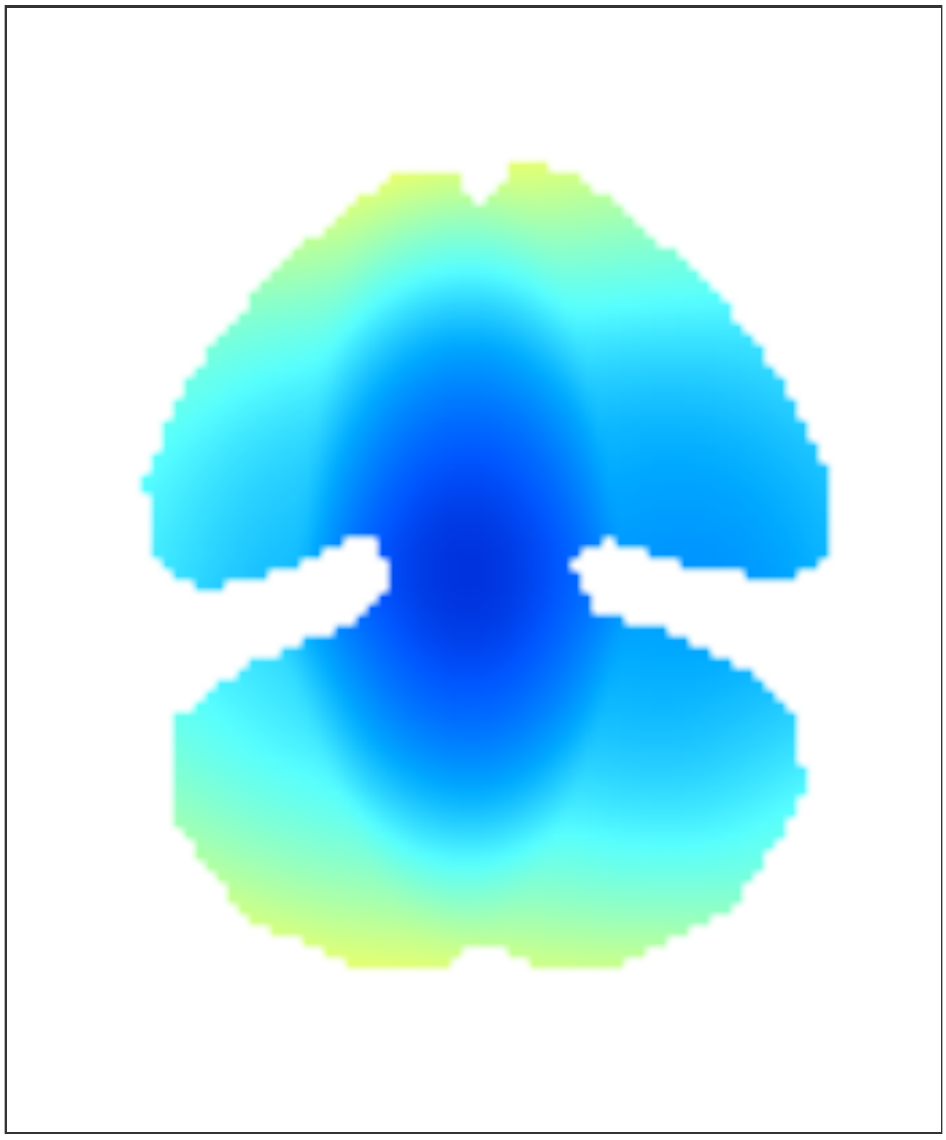} \!\!\!&\!\!\!
			\includegraphics[scale=0.19]{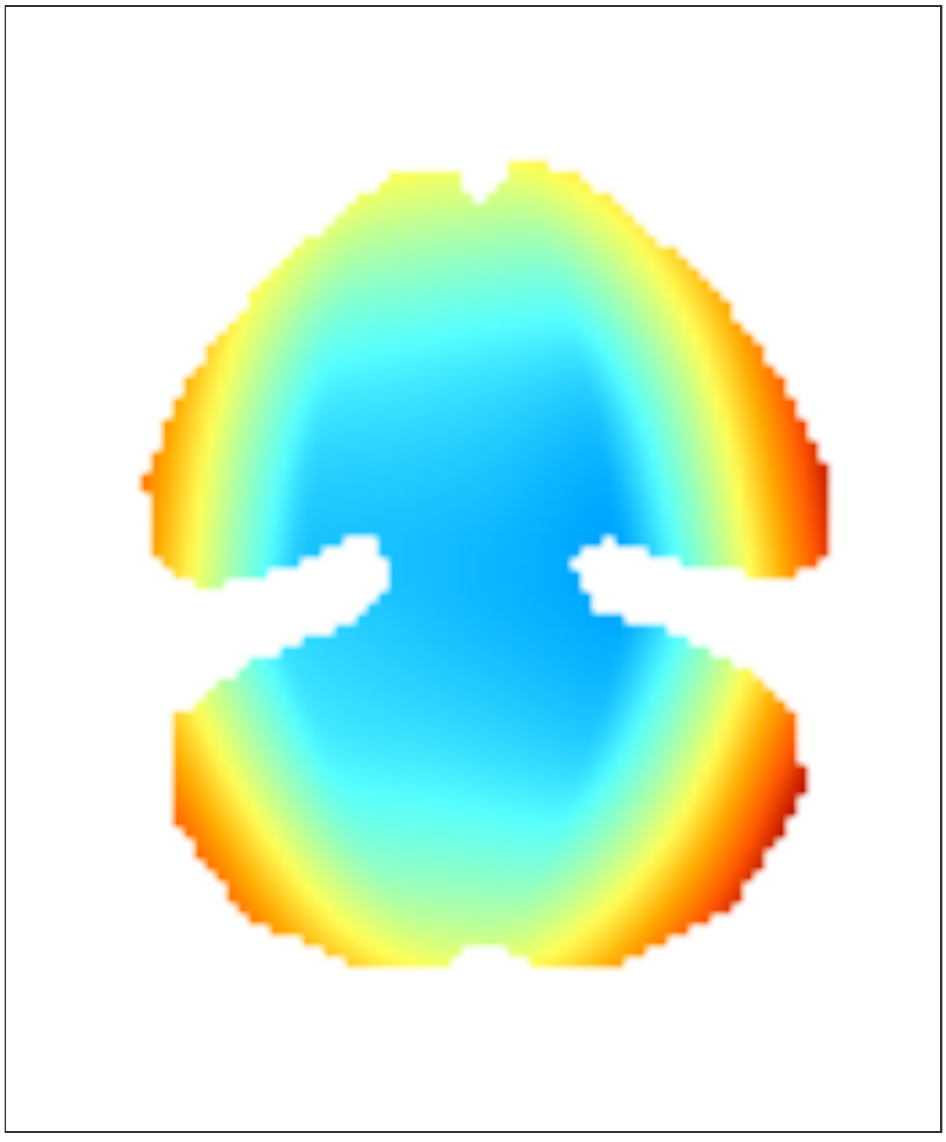} \!\!\!&\!\!\!
			\includegraphics[scale=0.19]{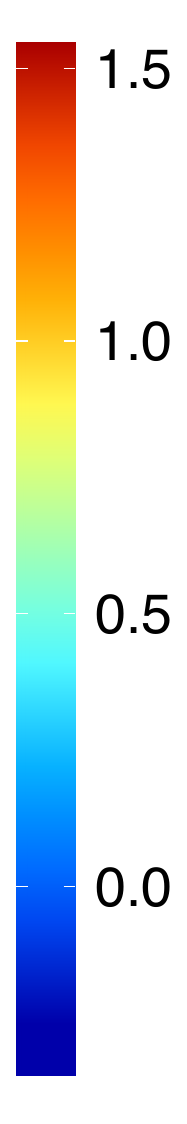}\\[-5pt]
			\multicolumn{7}{c}{$\beta_0$}\\
			\includegraphics[scale=0.19]{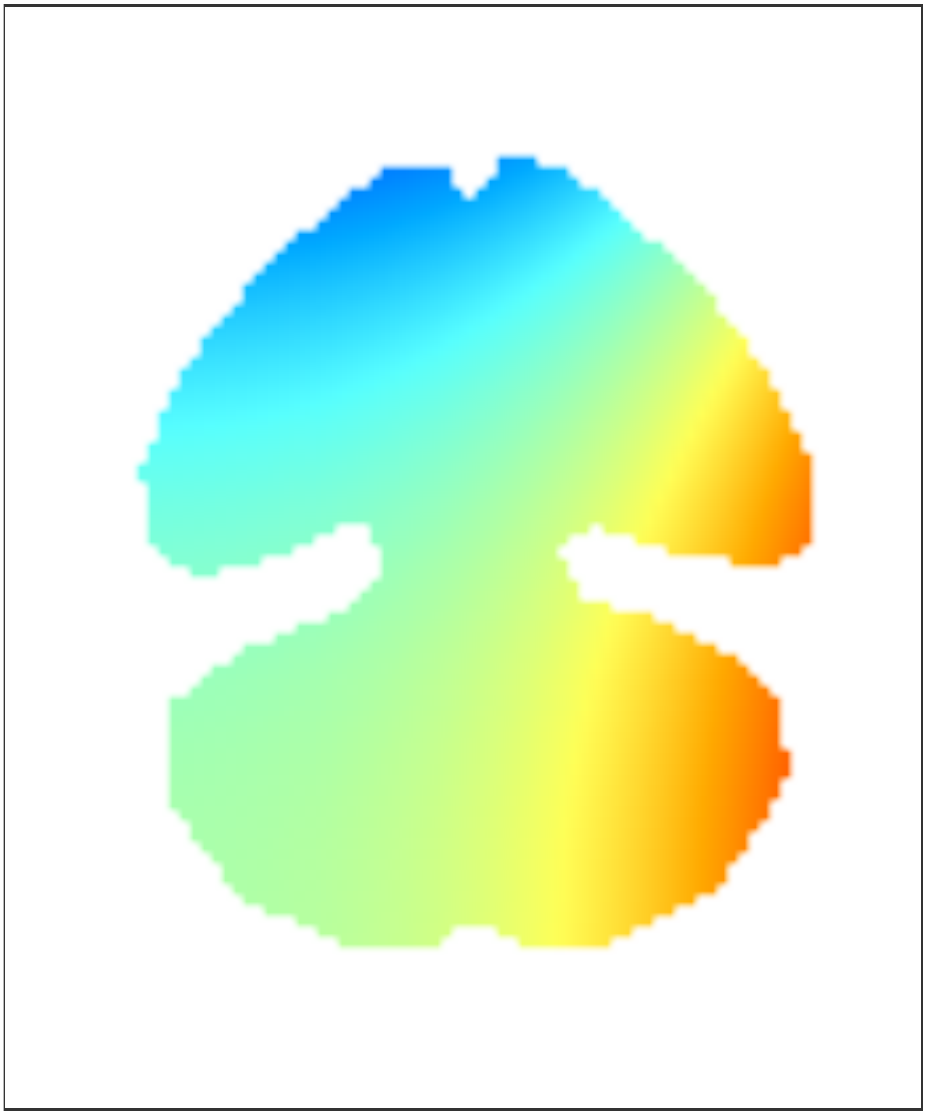} \!\!\!&\!\!\!
			\includegraphics[scale=0.19]{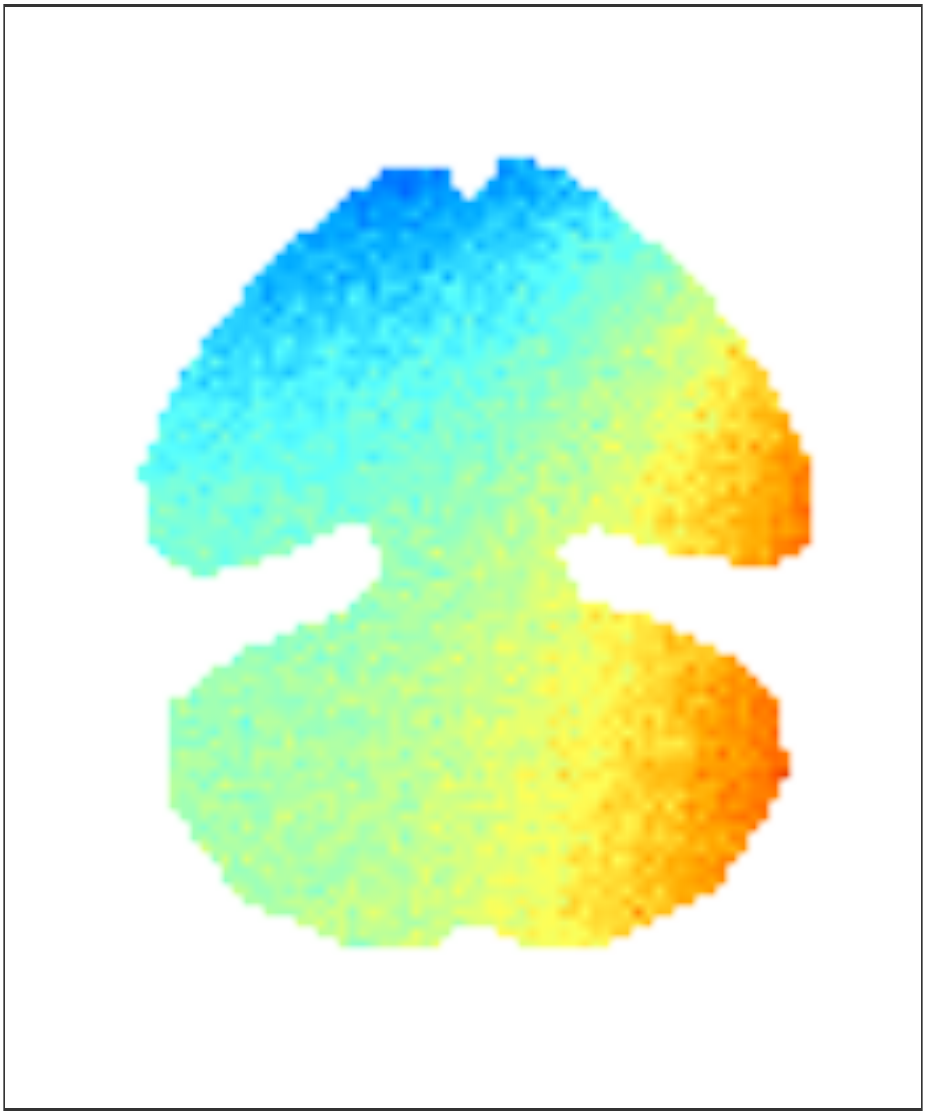} \!\!\!&\!\!\!
			\includegraphics[scale=0.19]{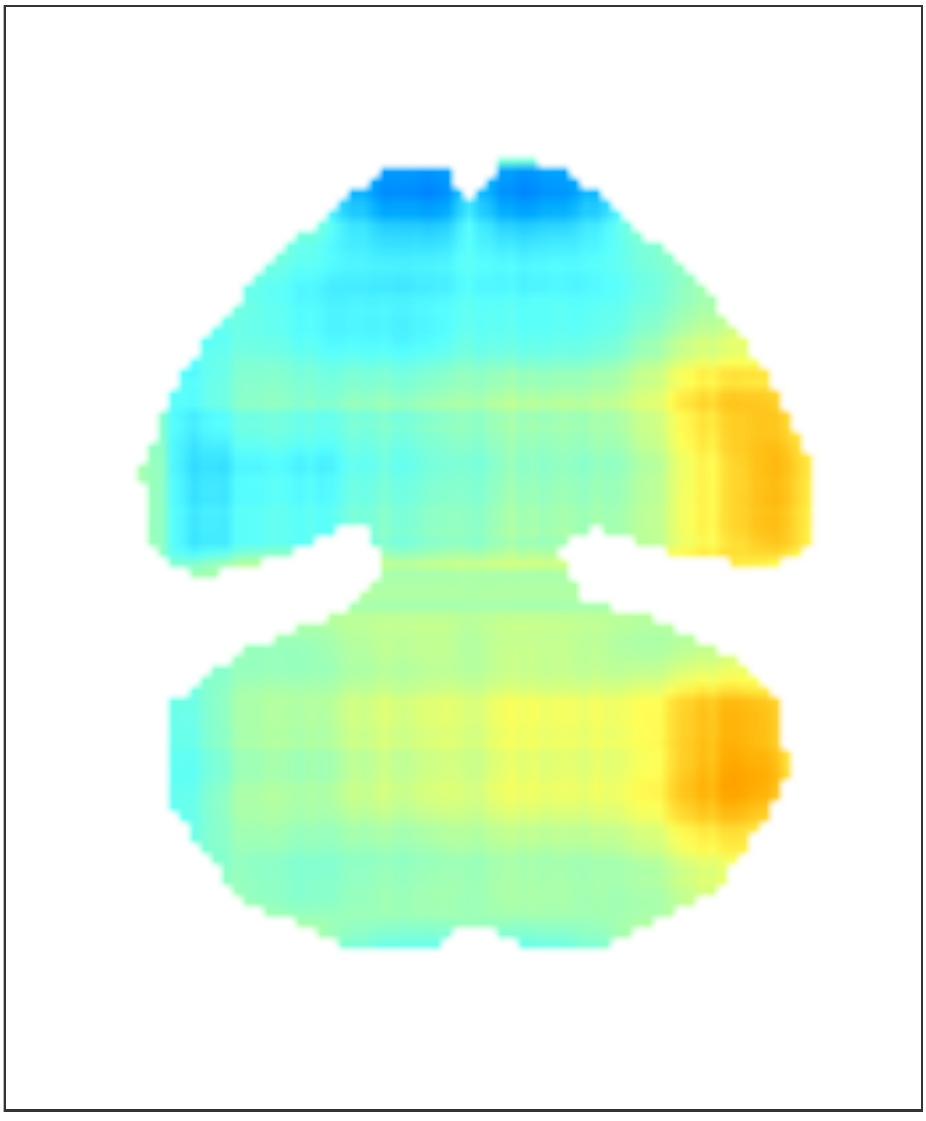} \!\!\!&\!\!\!
			\includegraphics[scale=0.19]{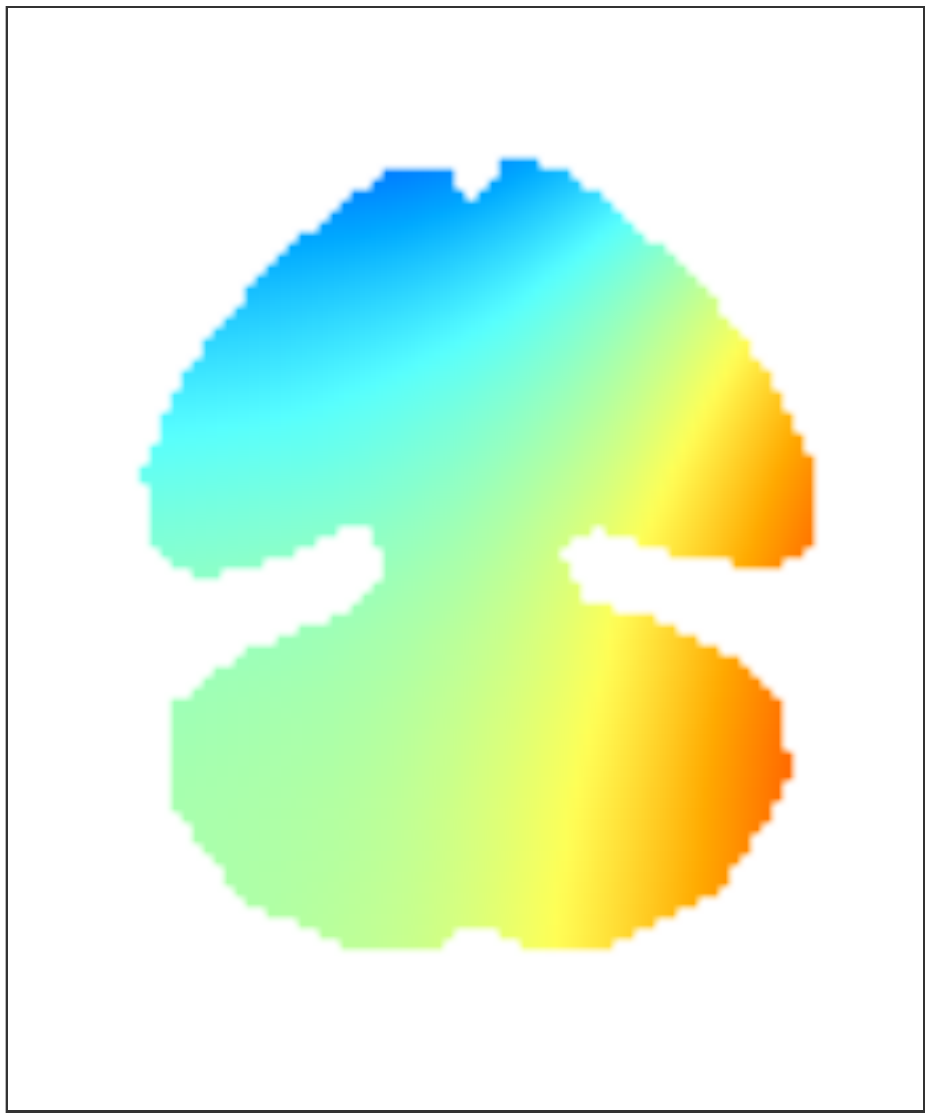} \!\!\!&\!\!\!
			\includegraphics[scale=0.19]{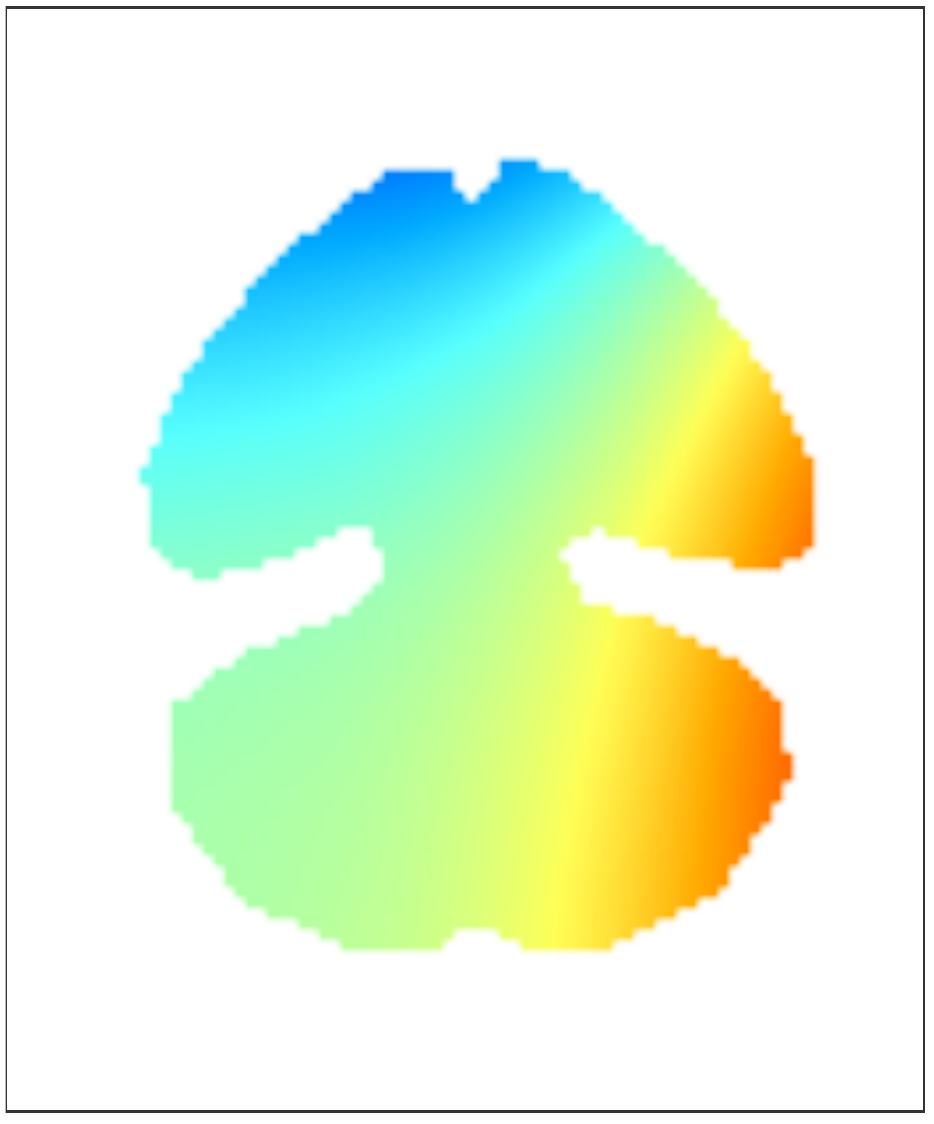} \!\!\!&\!\!\!
			\includegraphics[scale=0.19]{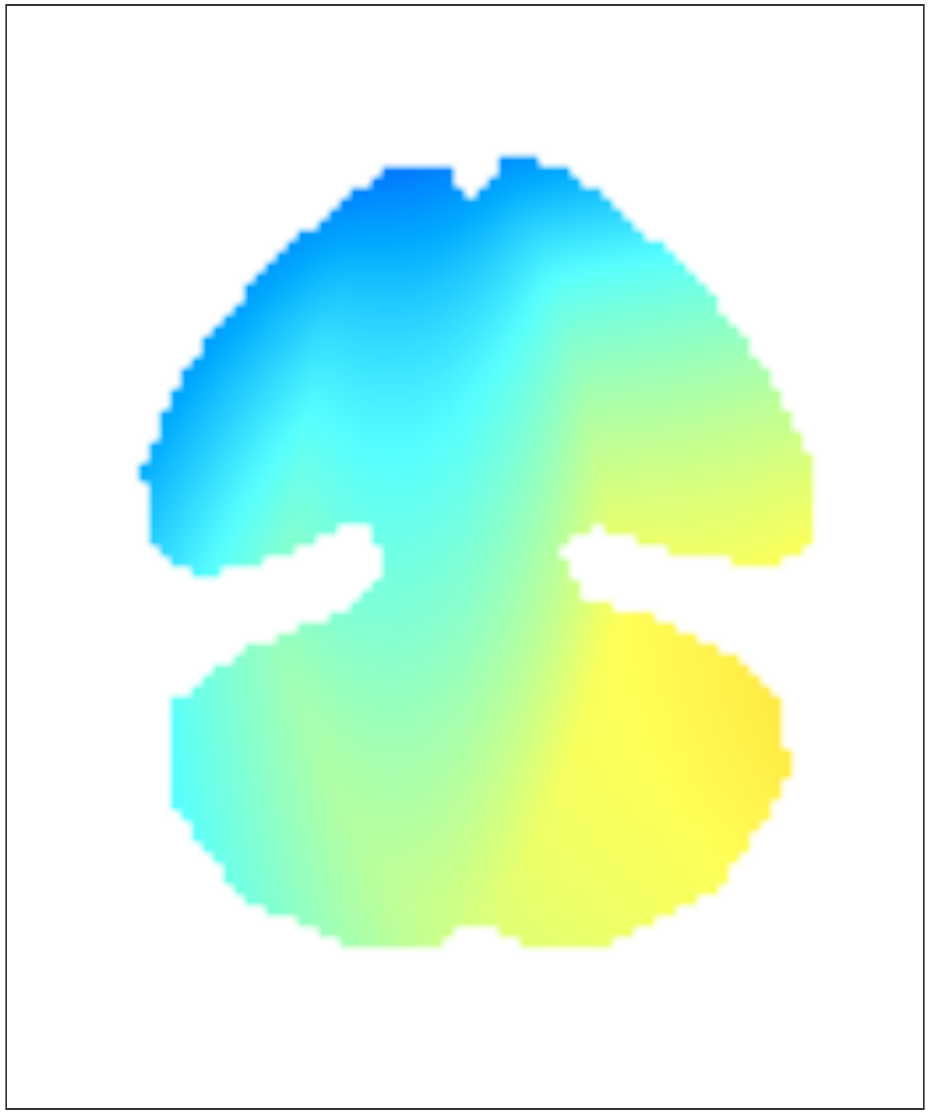} \!\!\!&\!\!\!
			\includegraphics[scale=0.19]{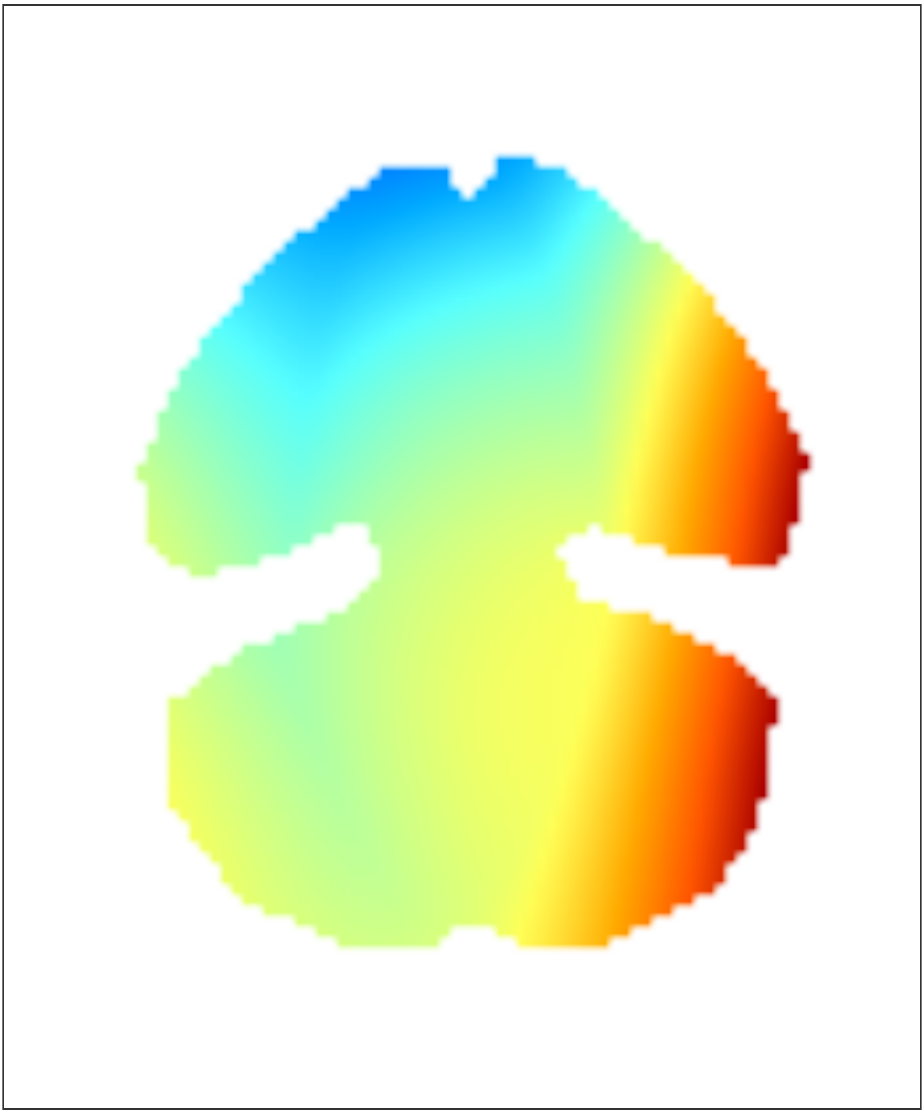} \!\!\!&\!\!\!
			\includegraphics[scale=0.19]{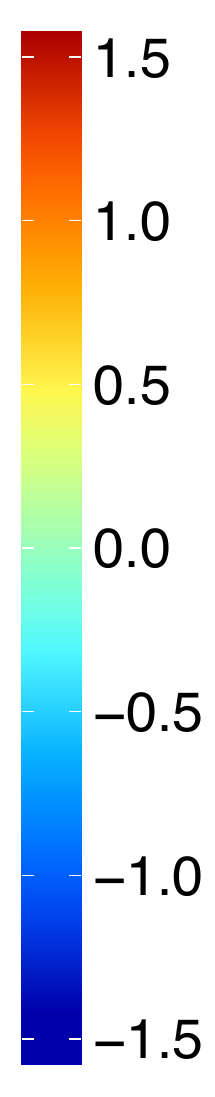}\\[-5pt]
			\multicolumn{7}{c}{$\beta_1$}\\
			\includegraphics[scale=0.19]{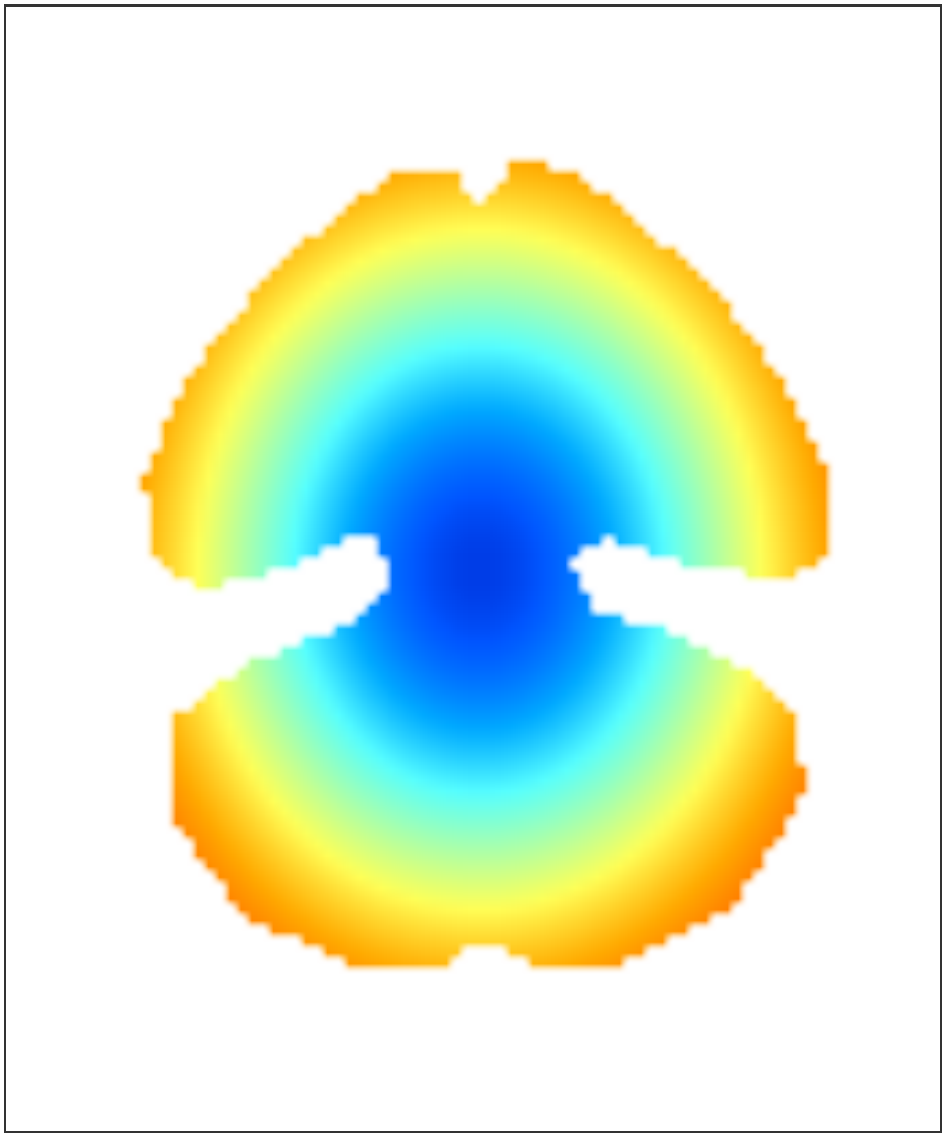} \!\!\!&\!\!\!
			\includegraphics[scale=0.19]{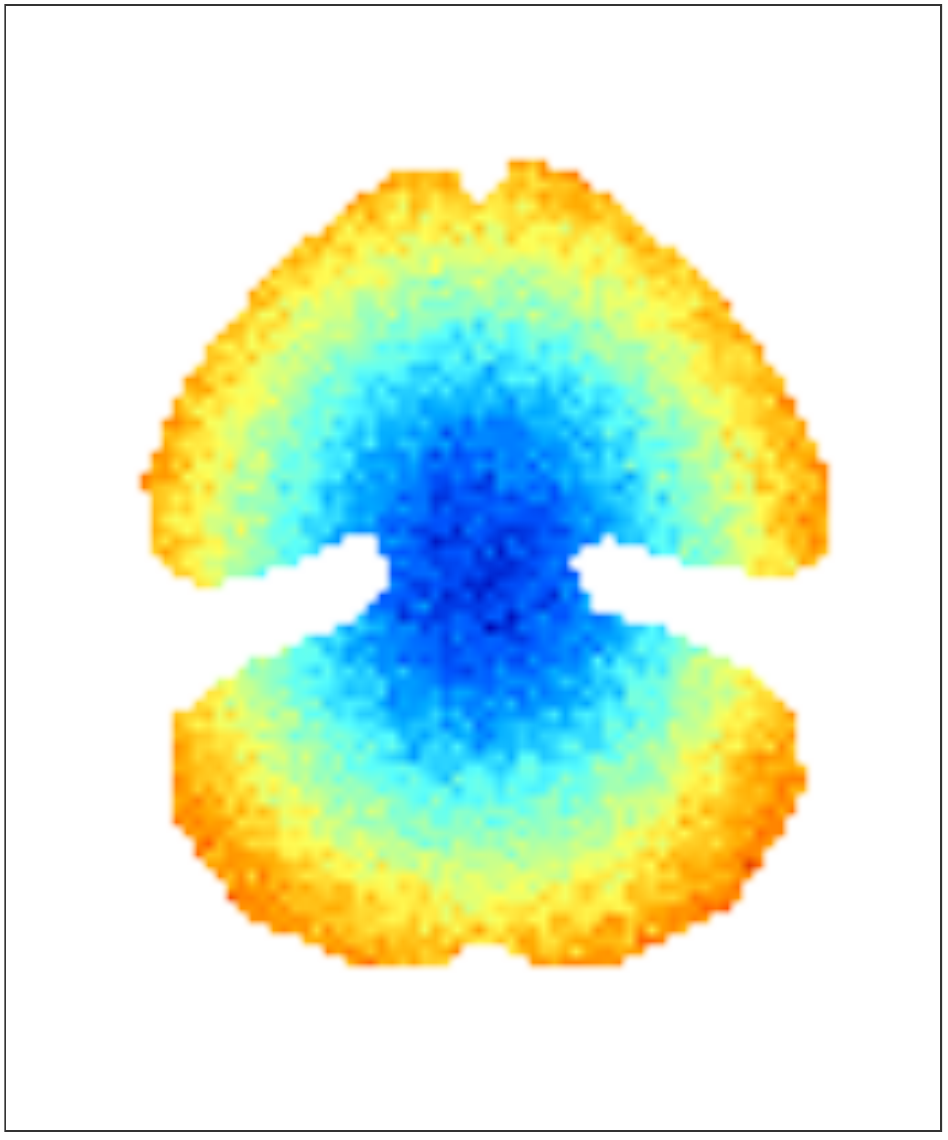} \!\!\!&\!\!\!
			\includegraphics[scale=0.19]{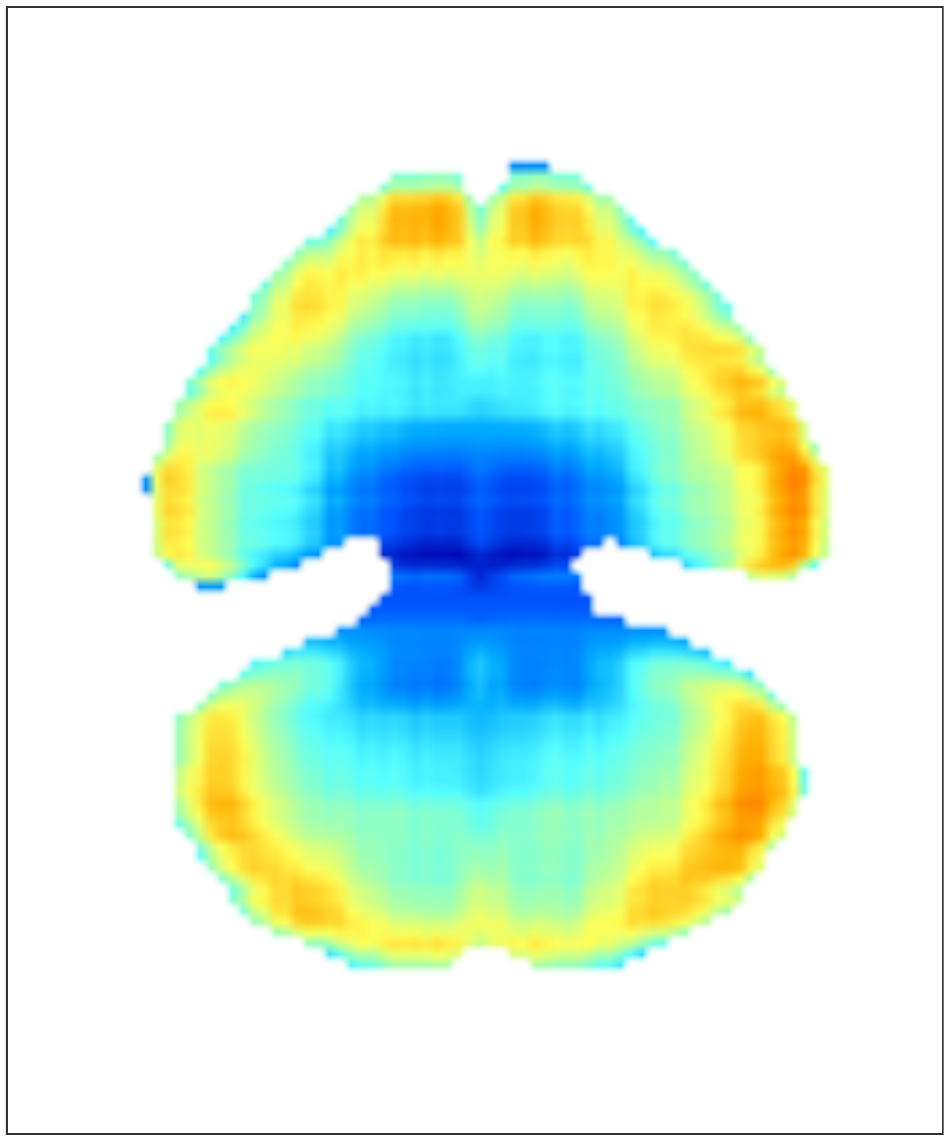} \!\!\!&\!\!\!
			\includegraphics[scale=0.19]{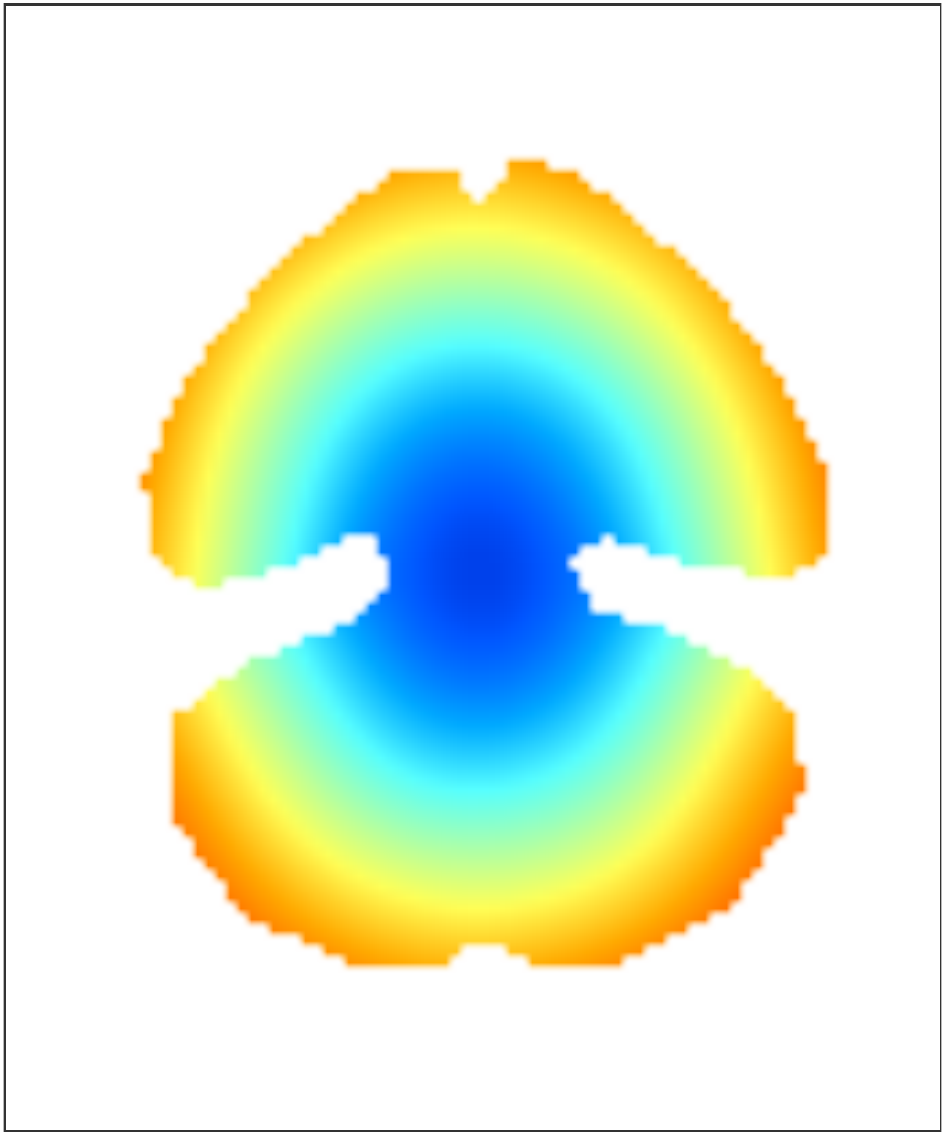} \!\!\!&\!\!\!
			\includegraphics[scale=0.19]{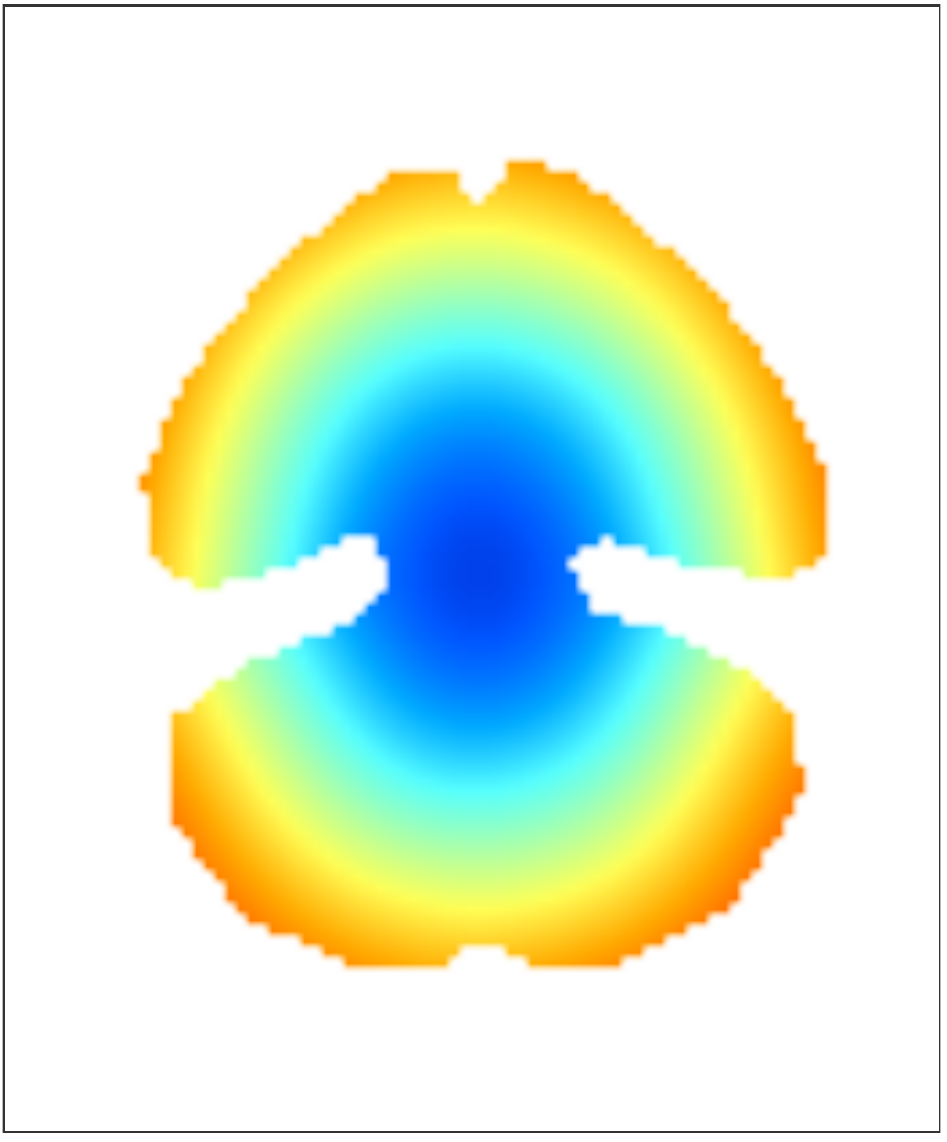} \!\!\!&\!\!\!
			\includegraphics[scale=0.19]{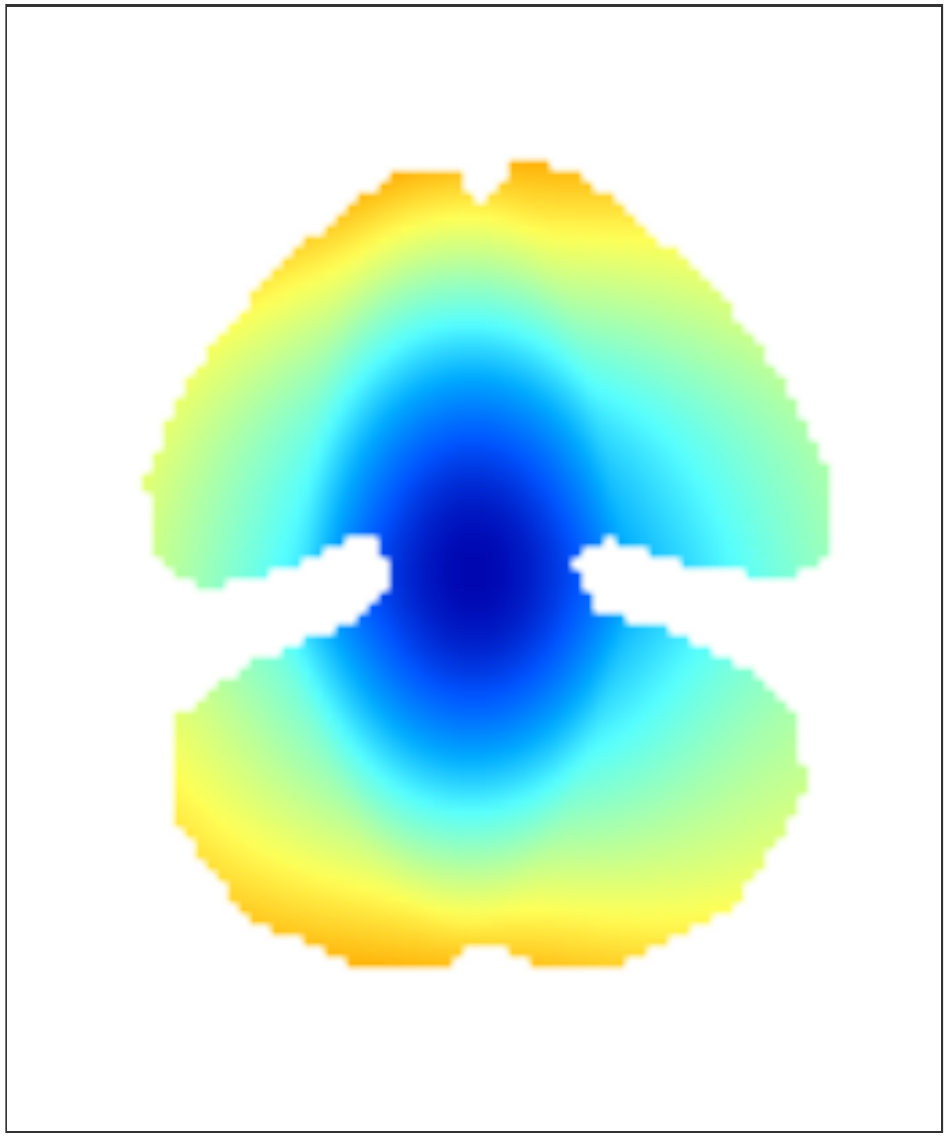} \!\!\!&\!\!\!
			\includegraphics[scale=0.19]{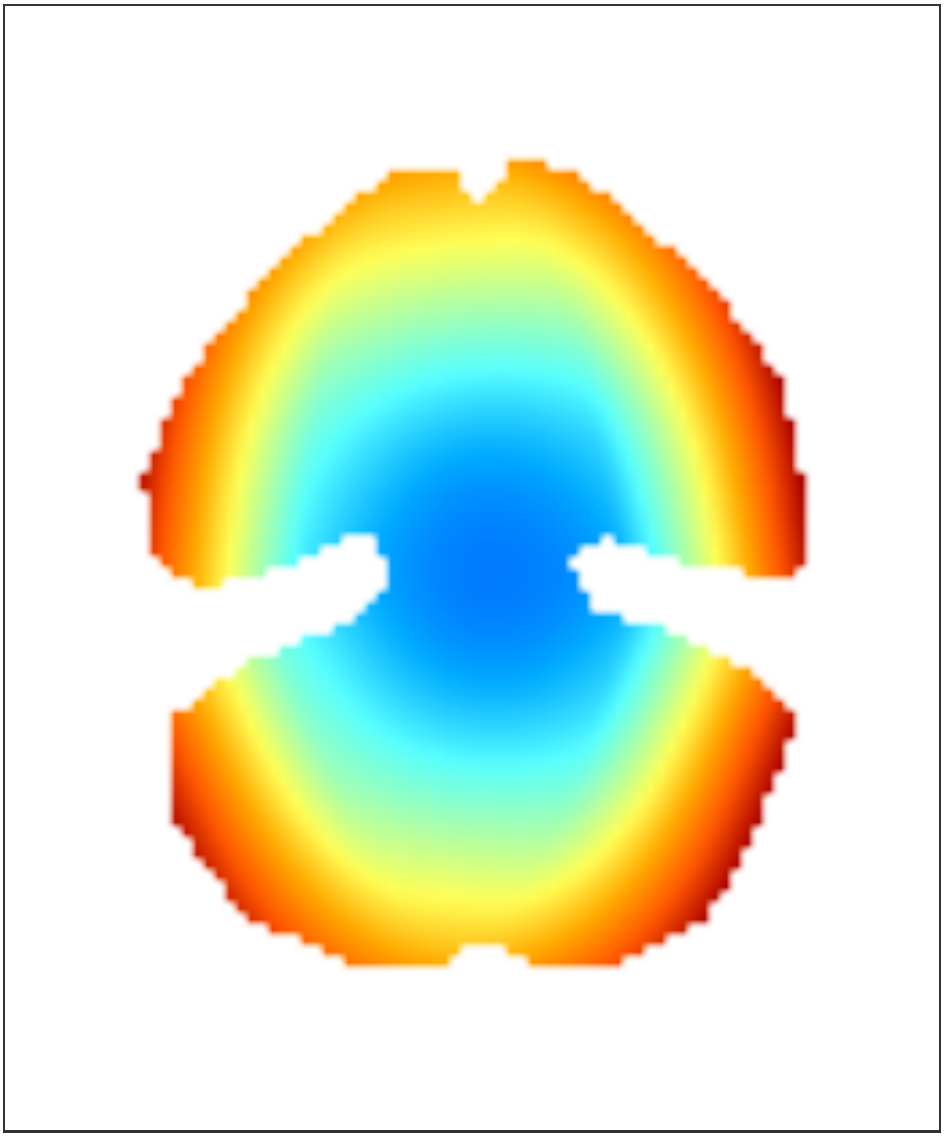} \!\!\!&\!\!\!
			\includegraphics[scale=0.19]{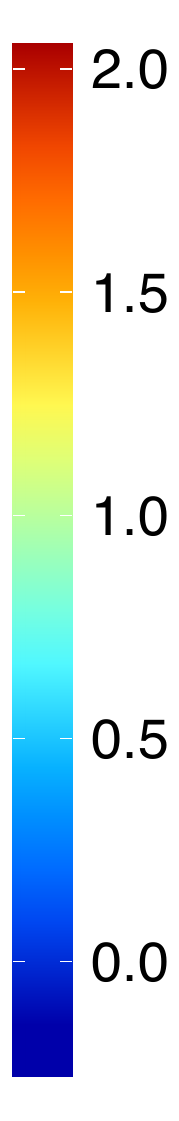}\\
			\multicolumn{7}{c}{$\beta_2$}\\[-5pt]
		\end{tabular}
		\caption{True coefficient functions and their estimators and 95\% SCCs based on the fifth slice.}
		\label{FIG:EST_SCC}
	\end{center}
\end{figure}

\begin{figure}[ht]
	\begin{center}
		\begin{tabular}{ccccccccc}
			TRUE & Kernel& Tensor & BPST($\triangle_1$) & BPST($\triangle_2 $) & Lower SCC &Upper SCC\\
			\includegraphics[width=1.8cm,height=2cm]{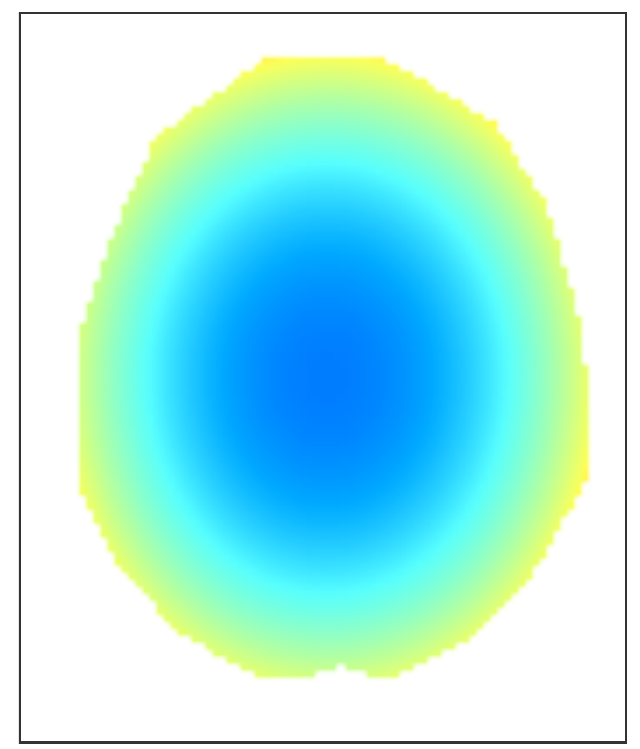} \!\!\!&\!\!\!
			\includegraphics[width=1.8cm,height=2cm]{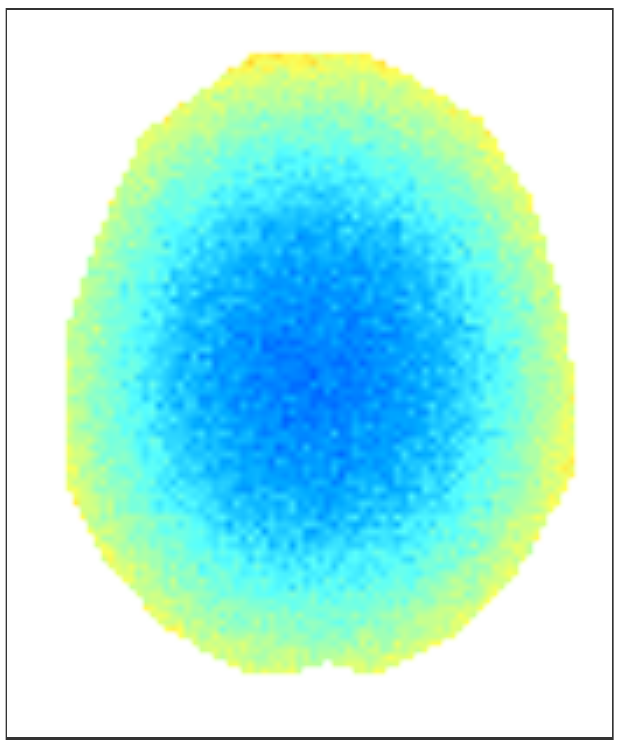} \!\!\!&\!\!\!
			\includegraphics[width=1.8cm,height=2cm]{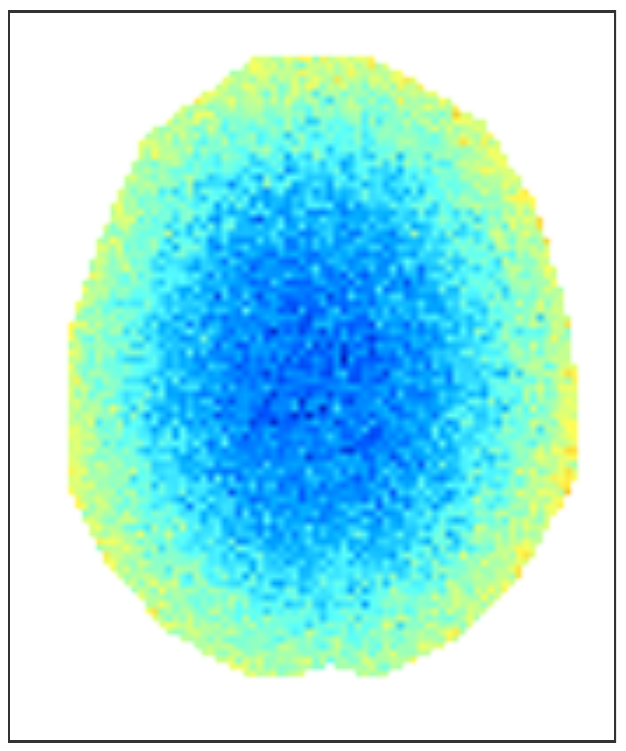} \!\!\!&\!\!\!
			\includegraphics[width=1.8cm,height=2cm]{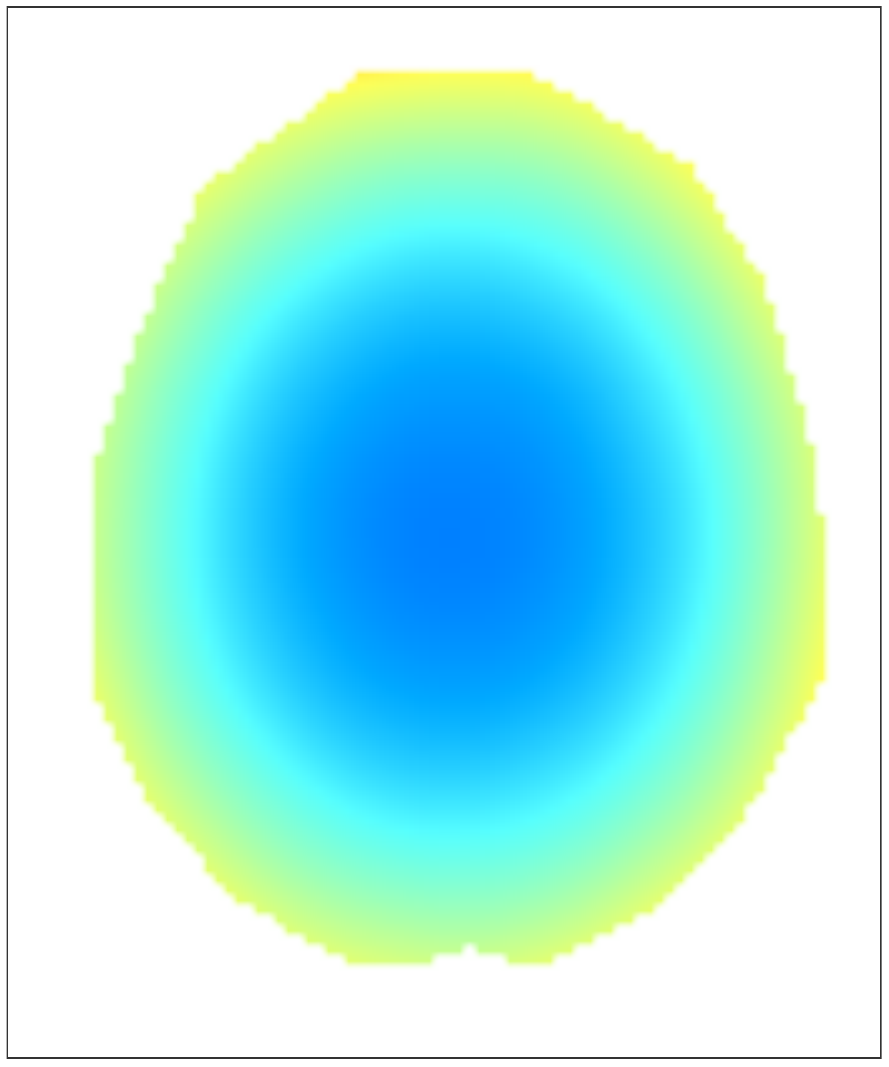} \!\!\!&\!\!\!
			\includegraphics[width=1.8cm,height=2cm]{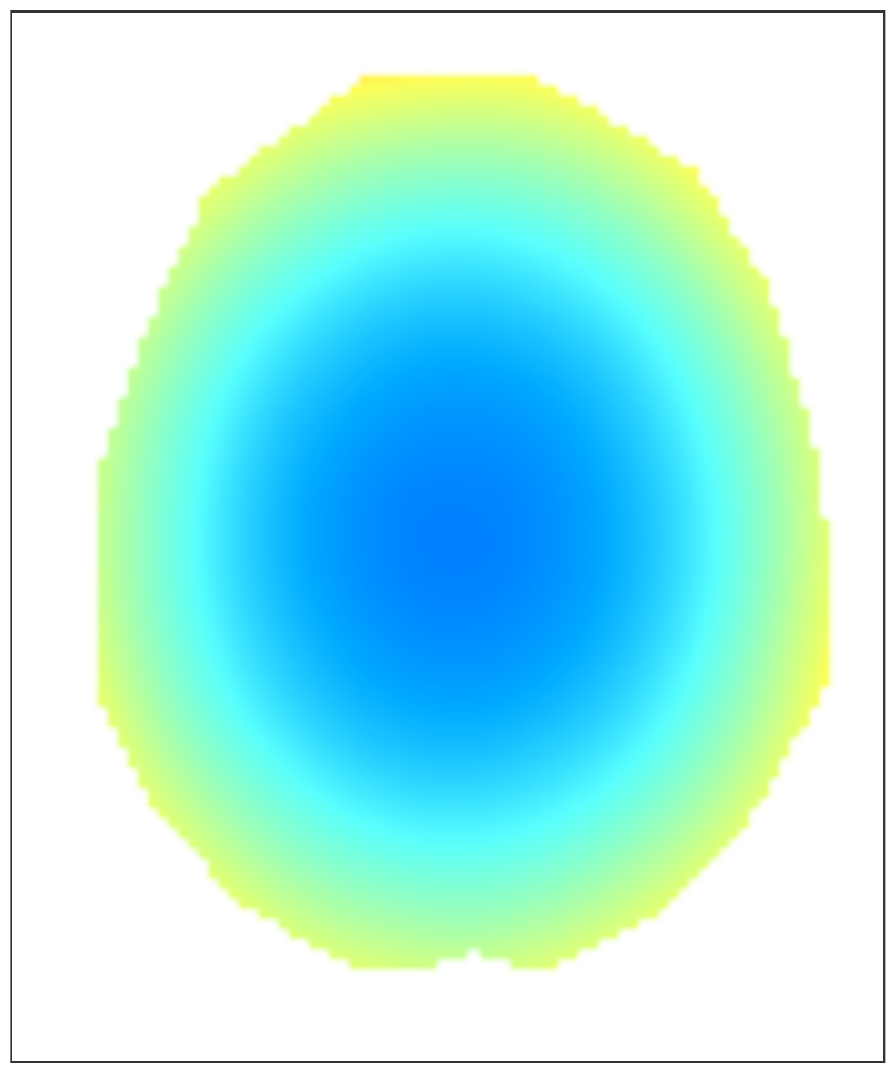} \!\!\!&\!\!\!
			\includegraphics[width=1.8cm,height=2cm]{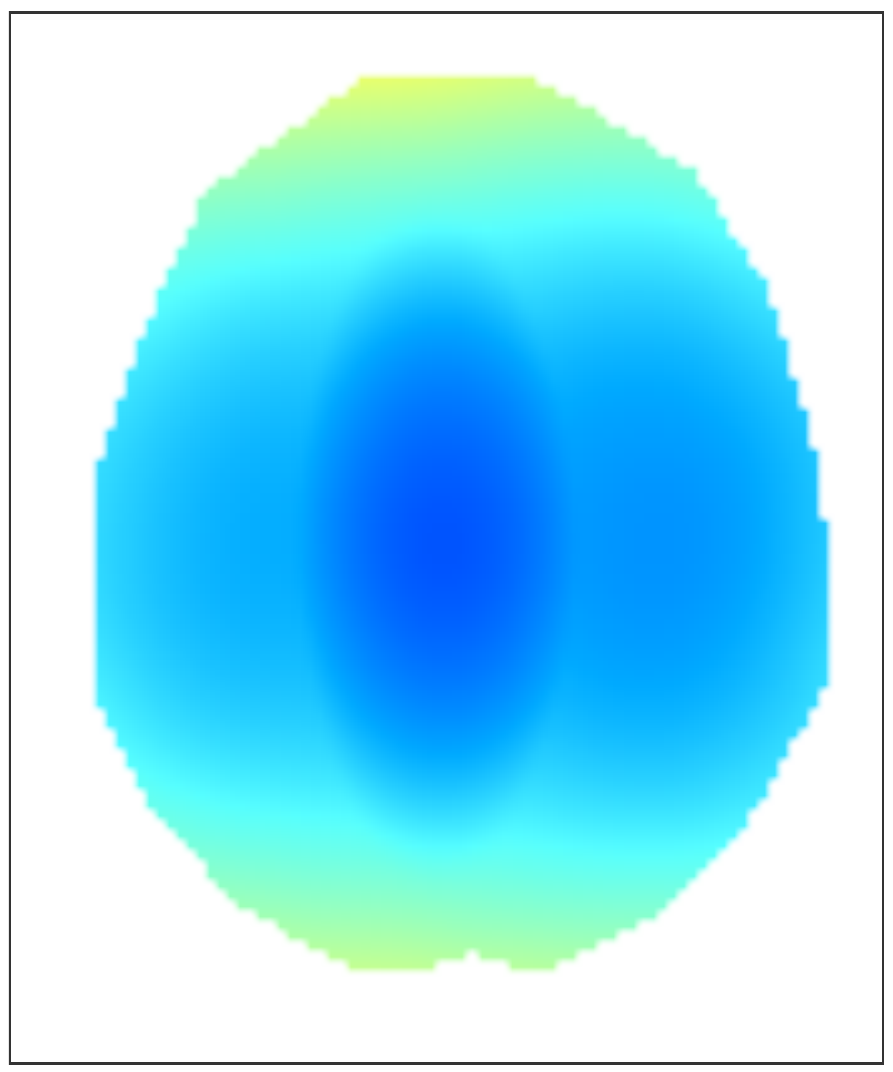} \!\!\!&\!\!\!
			\includegraphics[width=1.8cm,height=2cm]{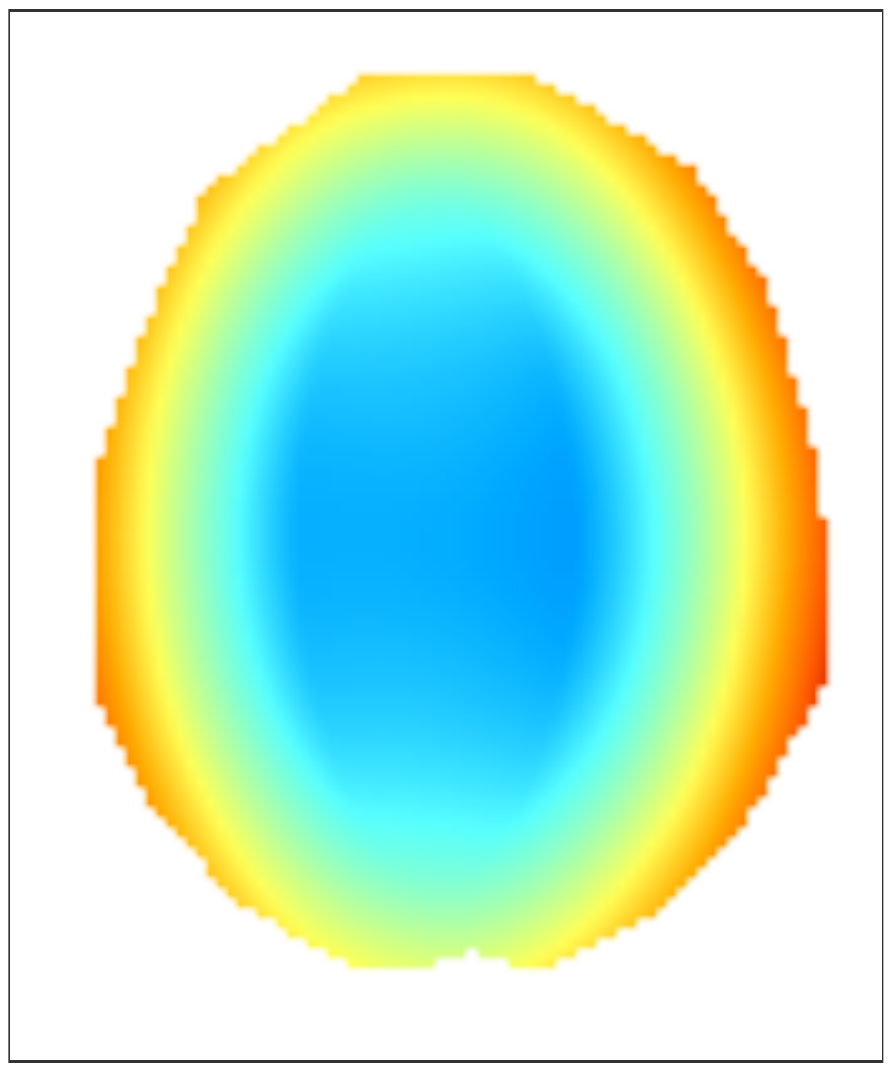} \!\!\!&\!\!\!
			\includegraphics[width=0.5cm,height=2cm]{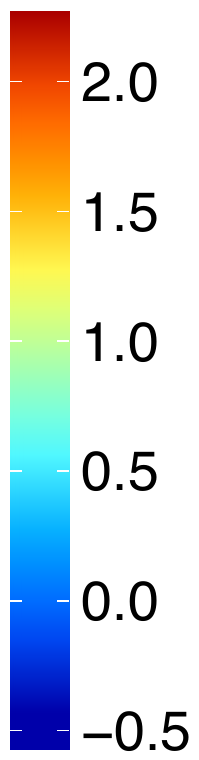}\\[-5pt]
			\multicolumn{7}{c}{$\beta_0$}\\
			\includegraphics[width=1.8cm,height=2cm]{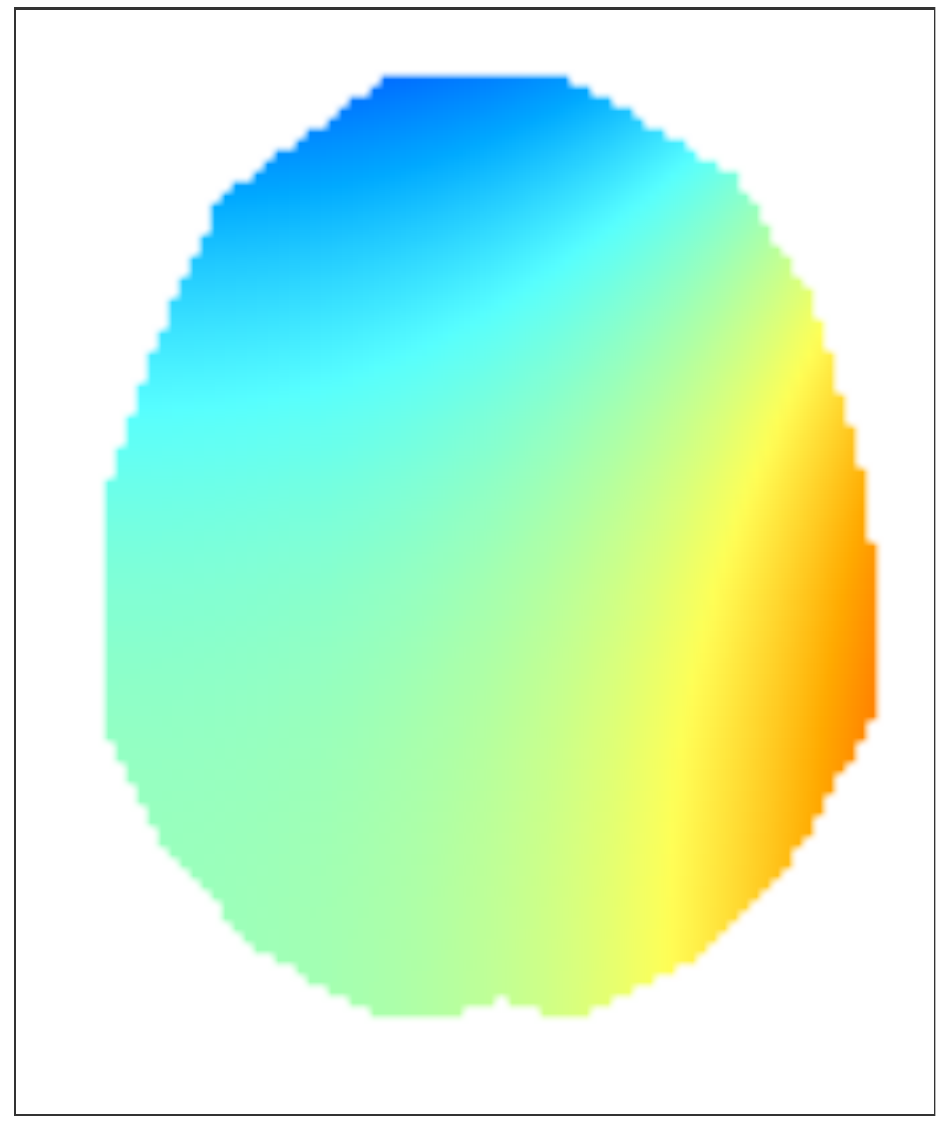} \!\!\!&\!\!\!
			\includegraphics[width=1.8cm,height=2cm]{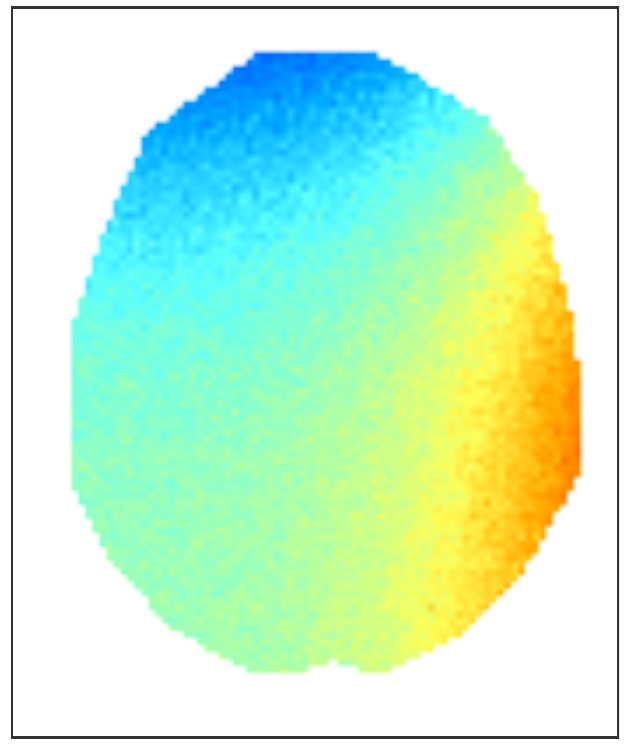} \!\!\!&\!\!\!
			\includegraphics[width=1.8cm,height=2cm]{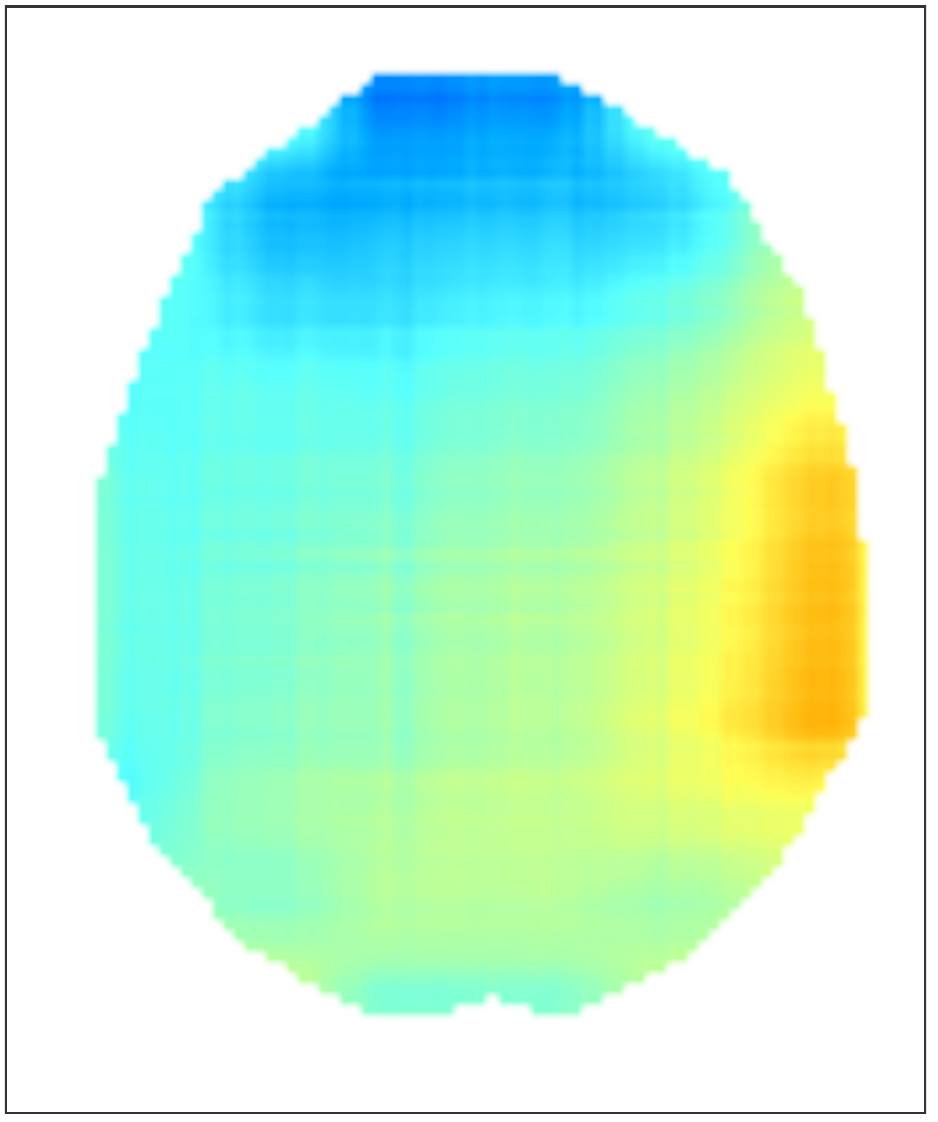} \!\!\!&\!\!\!
			\includegraphics[width=1.8cm,height=2cm]{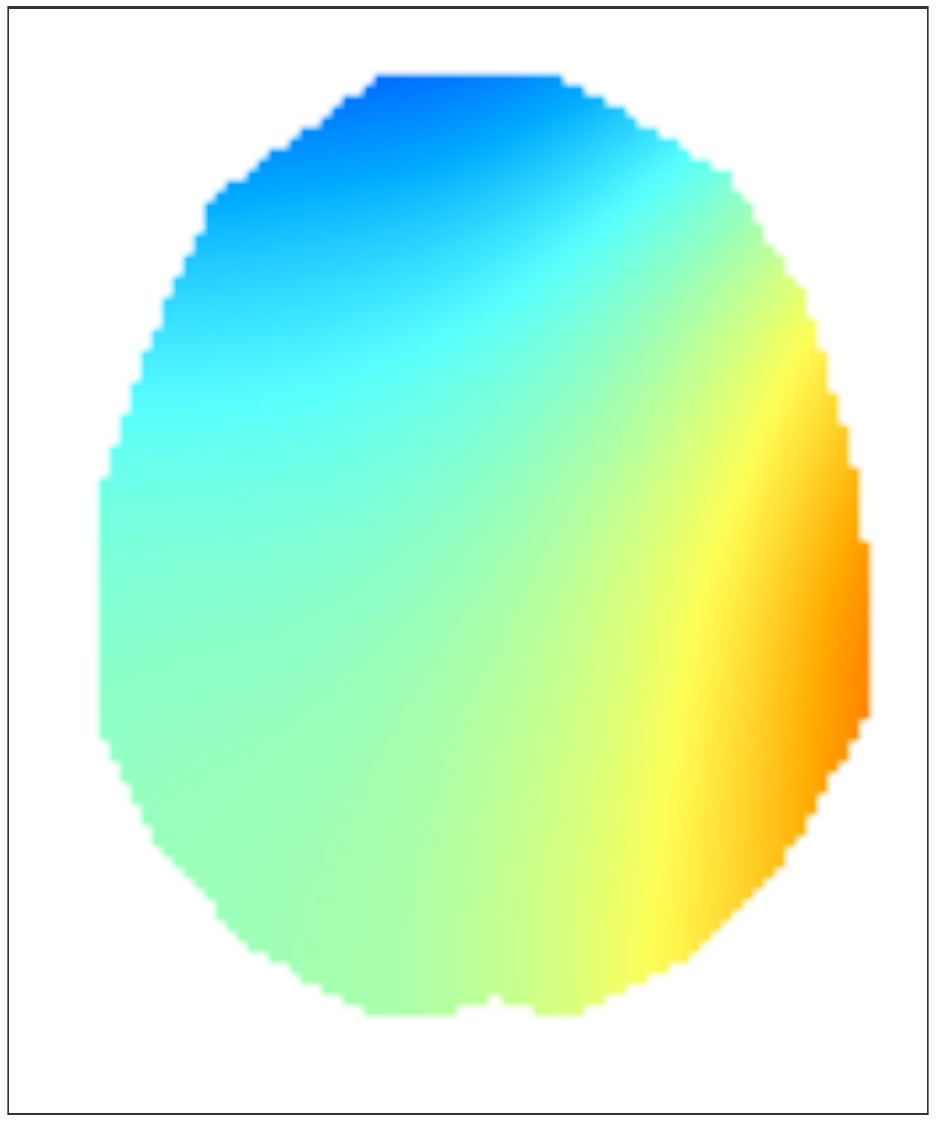} \!\!\!&\!\!\!
			\includegraphics[width=1.8cm,height=2cm]{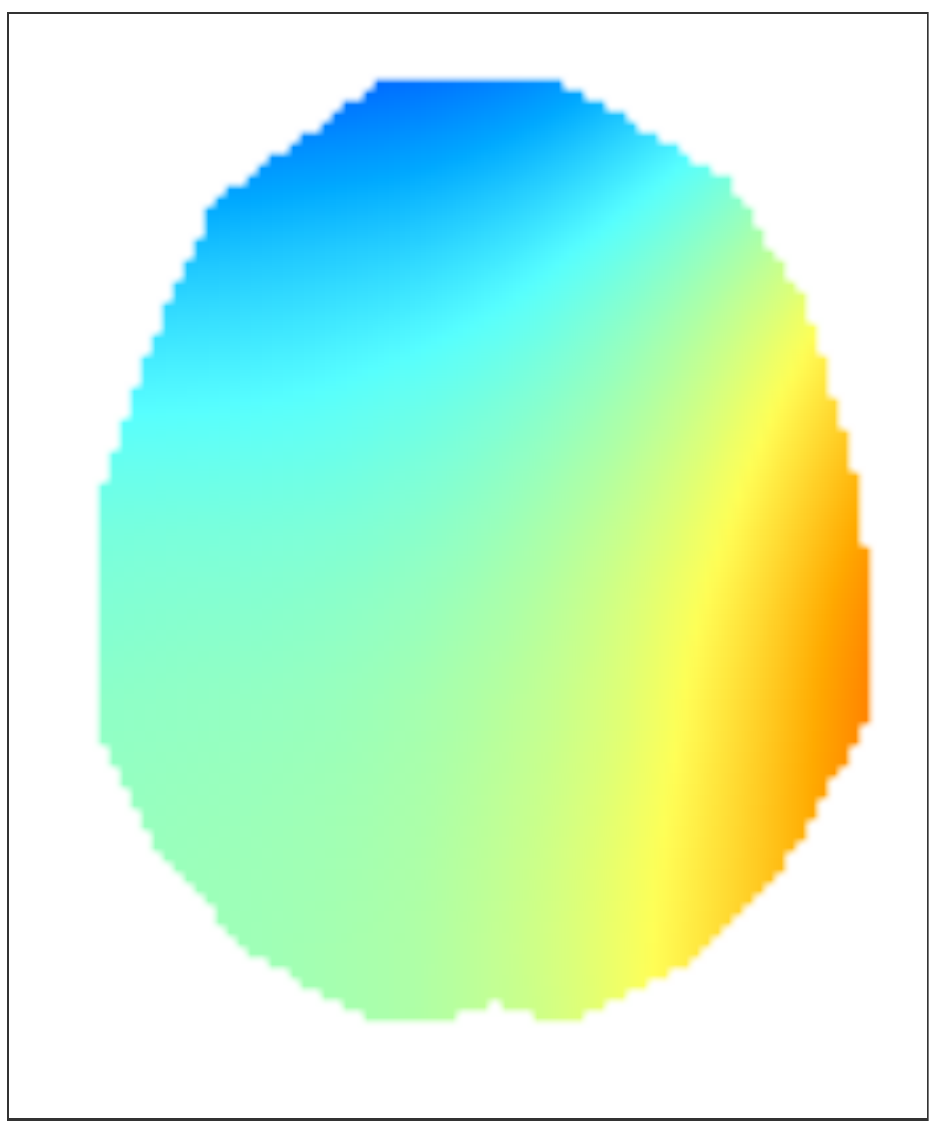} \!\!\!&\!\!\!
			\includegraphics[width=1.8cm,height=2cm]{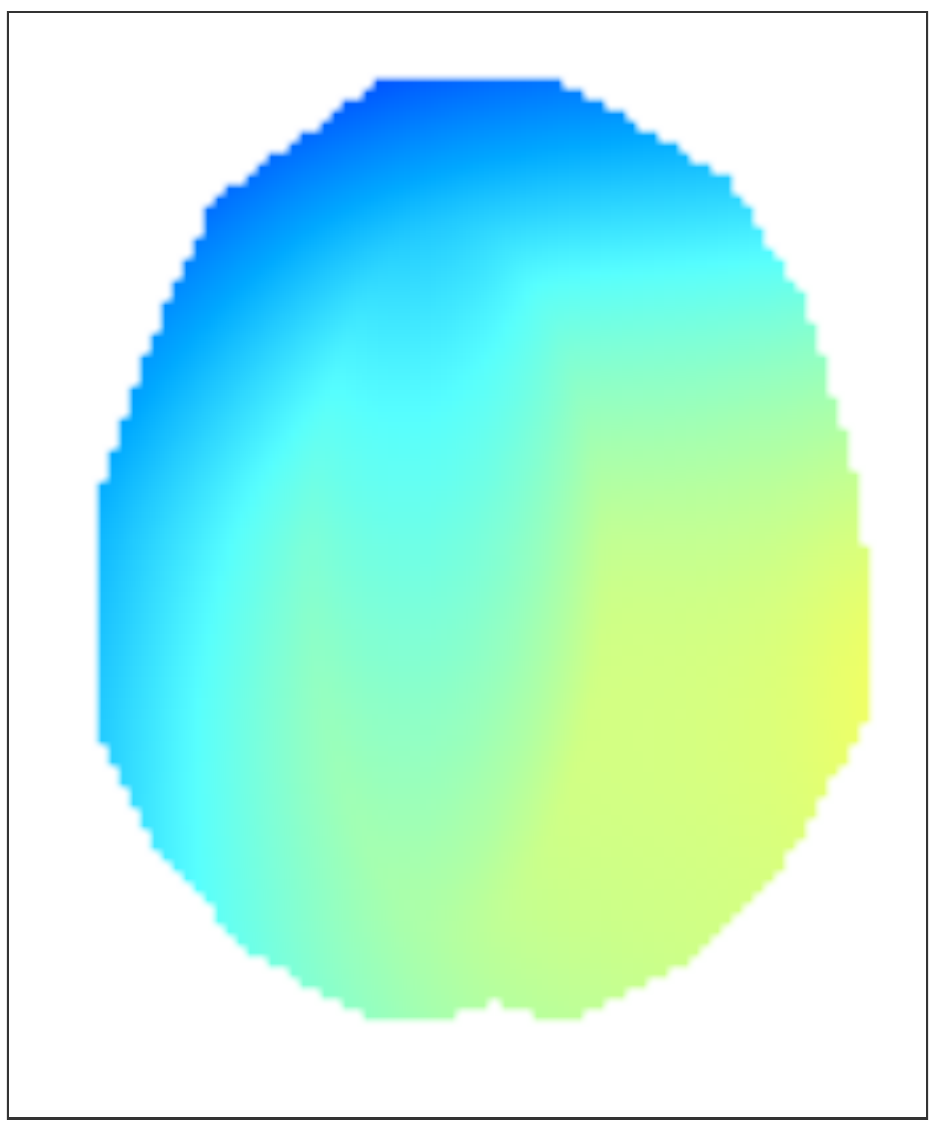} \!\!\!&\!\!\!
			\includegraphics[width=1.8cm,height=2cm]{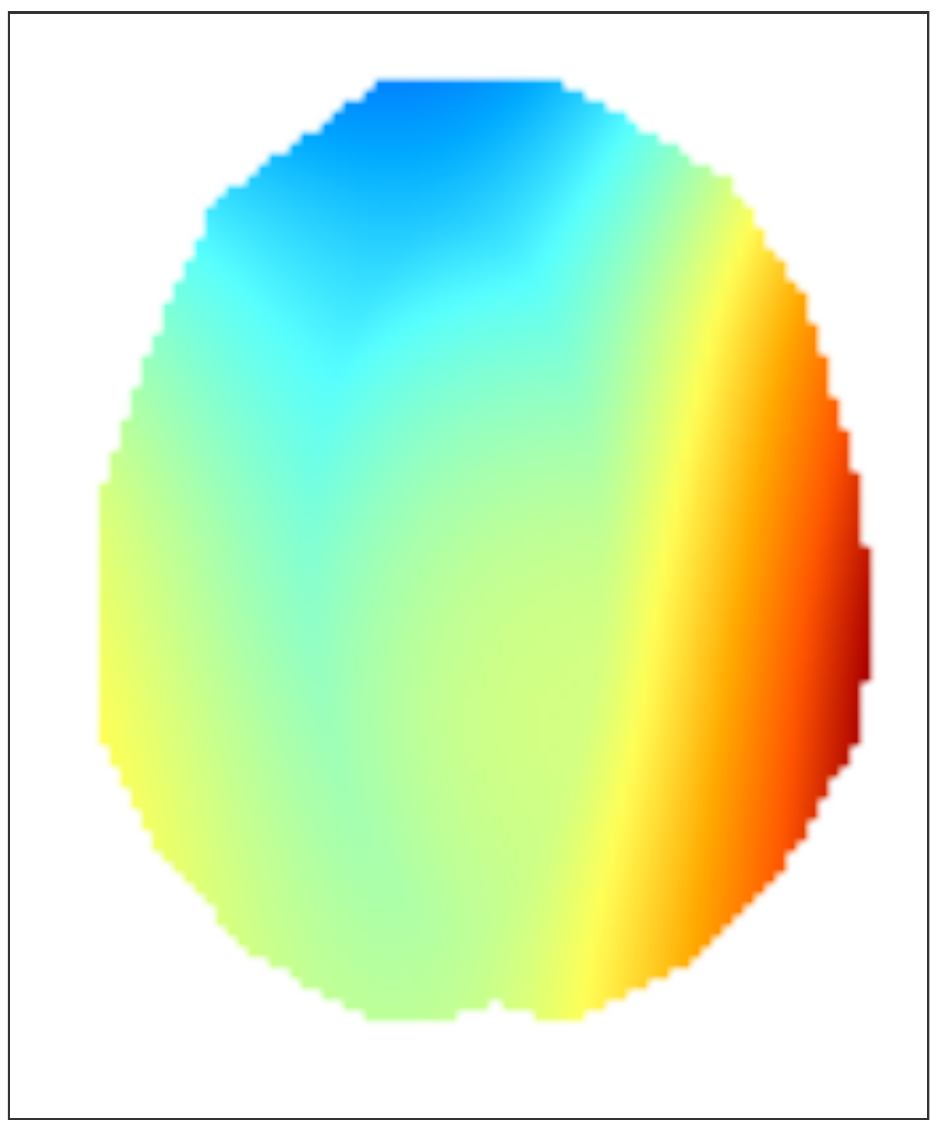} \!\!\!&\!\!\!
			\includegraphics[width=0.5cm,height=2cm]{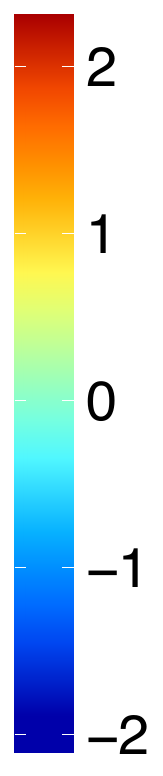}\\[-5pt]
			\multicolumn{7}{c}{$\beta_1$}\\
			\includegraphics[width=1.8cm,height=2cm]{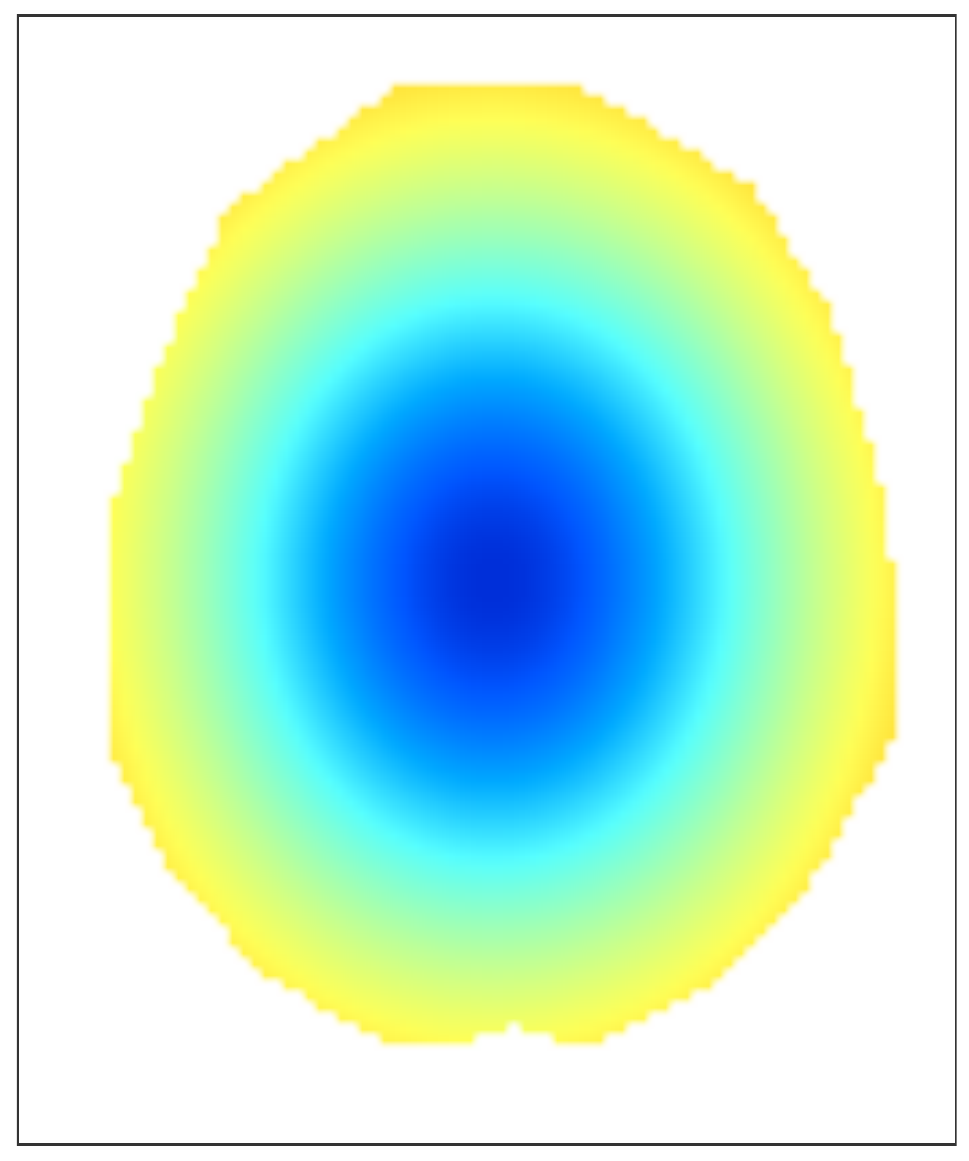} \!\!\!&\!\!\!
			\includegraphics[width=1.8cm,height=2cm]{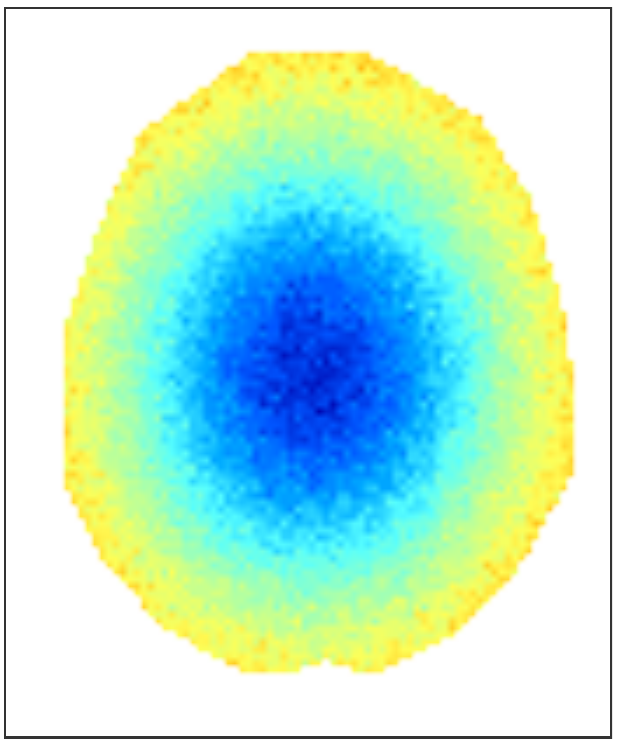} \!\!\!&\!\!\!
			\includegraphics[width=1.8cm,height=2cm]{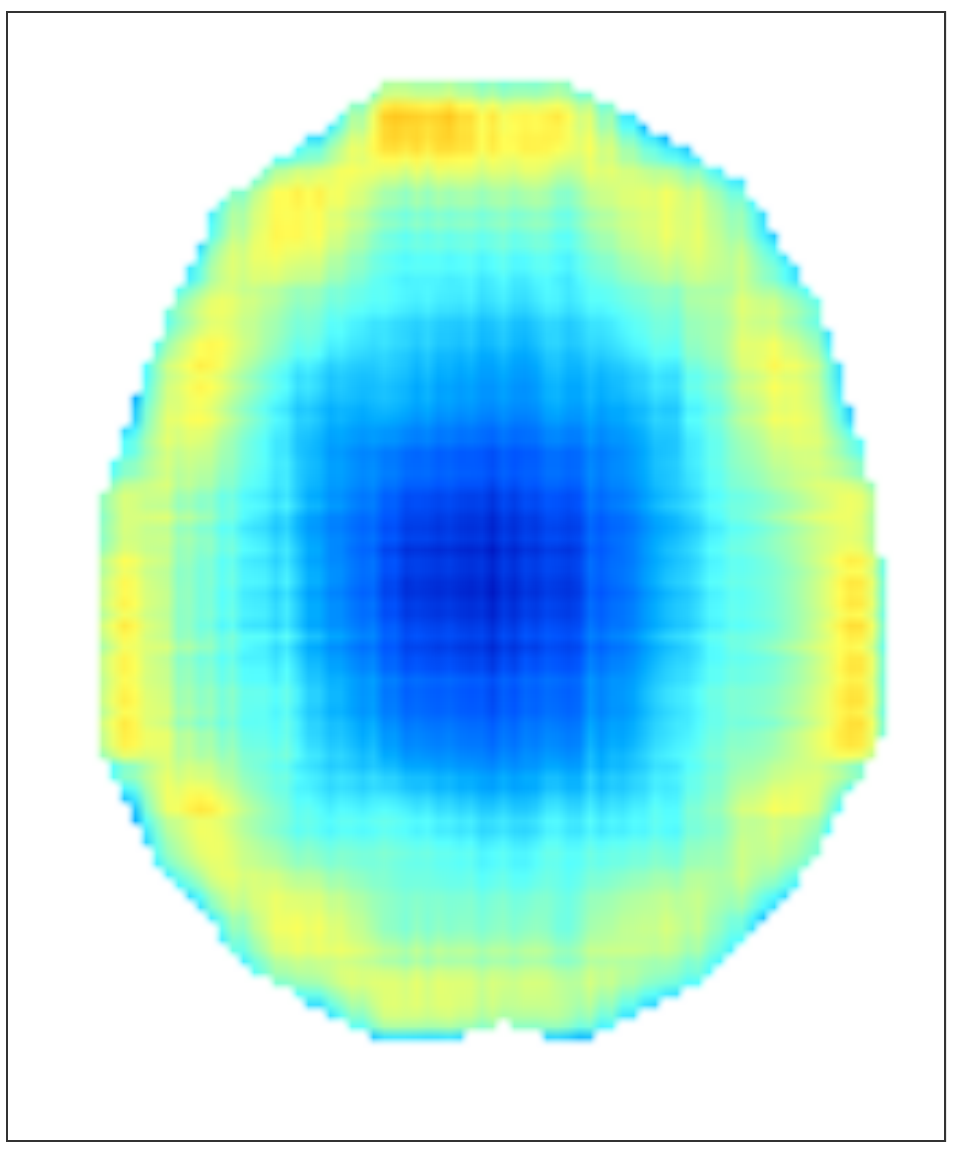} \!\!\!&\!\!\!
			\includegraphics[width=1.8cm,height=2cm]{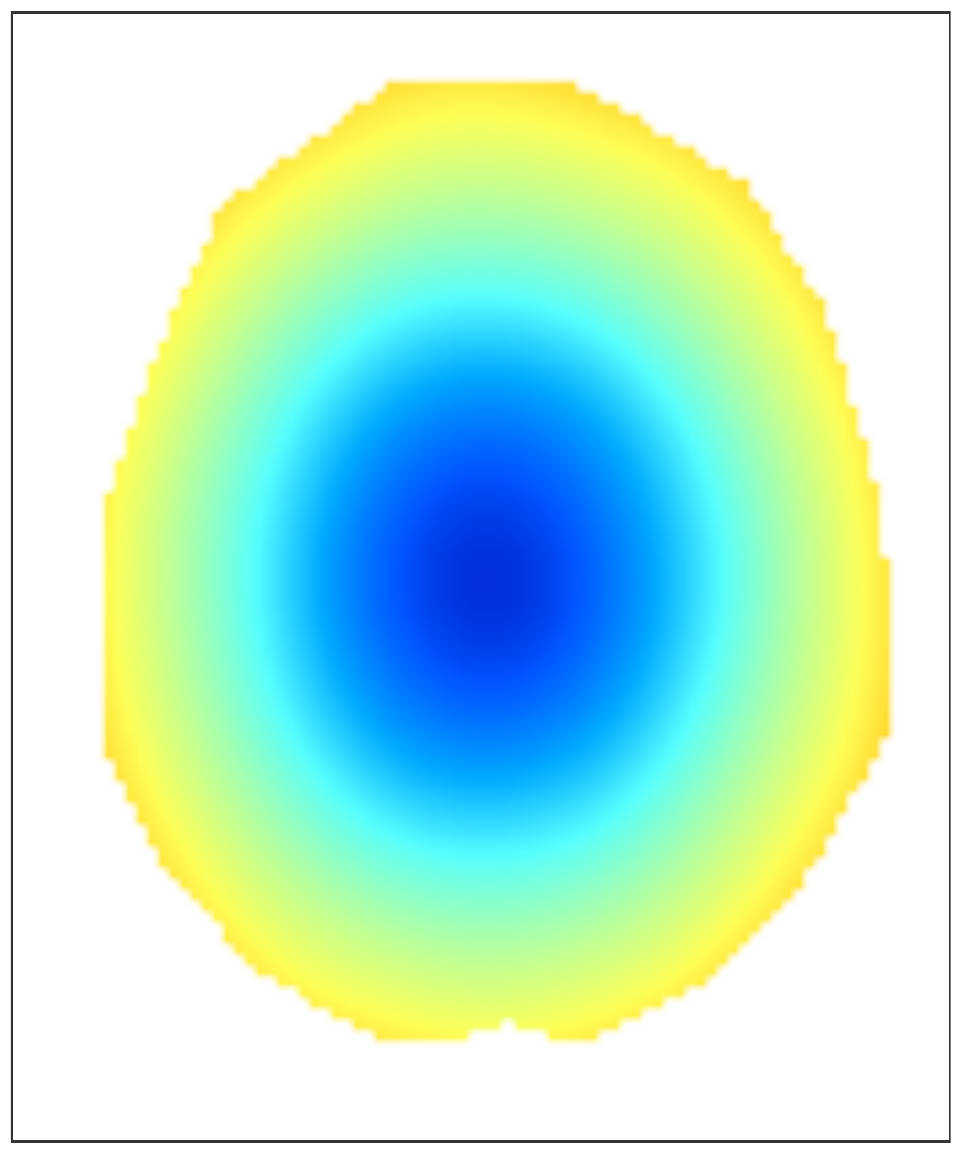} \!\!\!&\!\!\!
			\includegraphics[width=1.8cm,height=2cm]{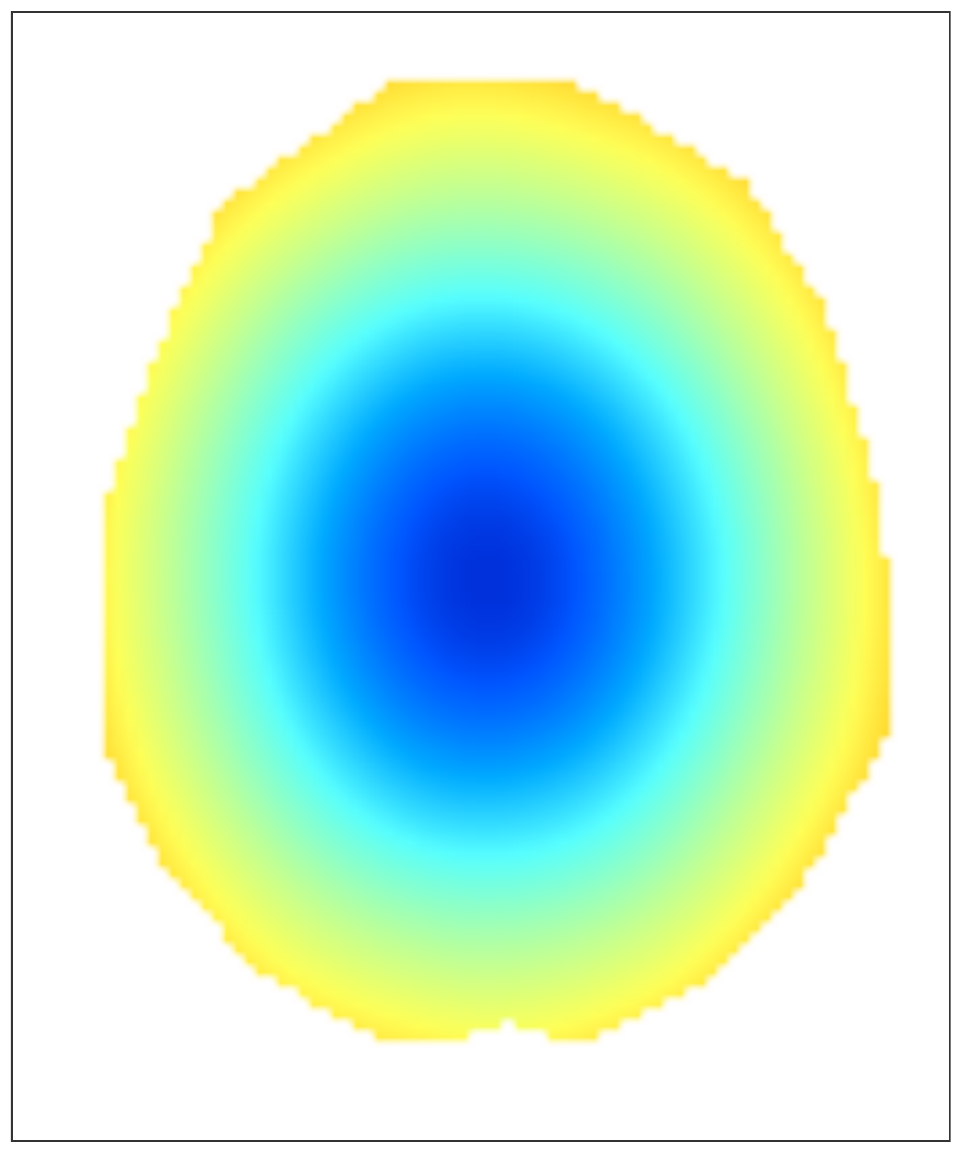} \!\!\!&\!\!\!
			\includegraphics[width=1.8cm,height=2cm]{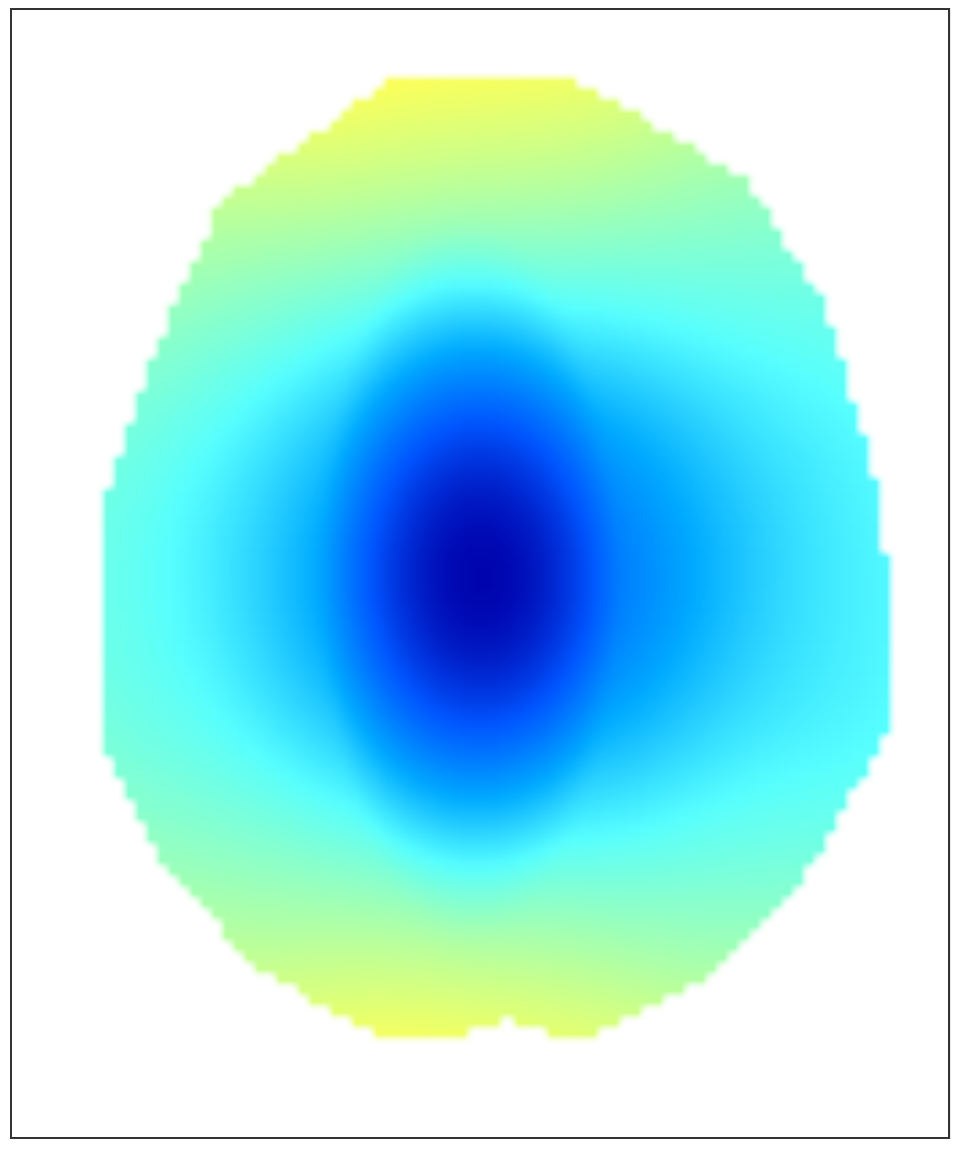} \!\!\!&\!\!\!
			\includegraphics[width=1.8cm,height=2cm]{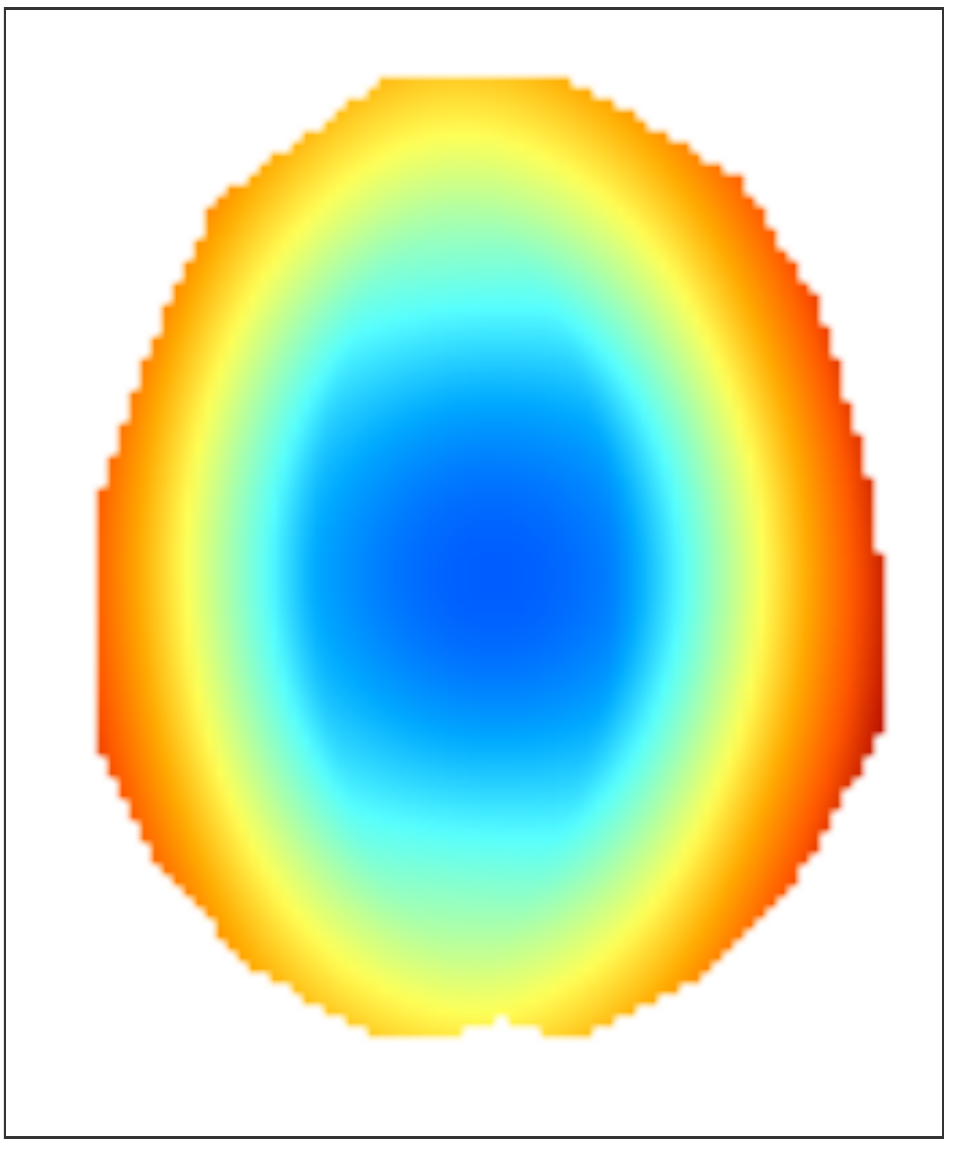} \!\!\!&\!\!\!
			\includegraphics[width=0.5cm,height=2cm]{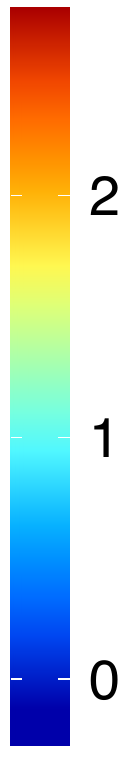}\\
			\multicolumn{7}{c}{$\beta_2$}\\[-5pt]
		\end{tabular}
		\caption{True coefficient functions and their estimators and 95\% SCCs based on the 35th slice.}
		\label{FIG:EST_SCC_s35}
	\end{center}
\end{figure}

\vskip .10in \noindent \textbf{B.2. Additional ADNI data analysis results} \vskip .10in

For the ADNI data described in Section 6 in the main paper, Table \ref{TAB:diagnosis} below summarizes the distribution of patients by diagnosis status and sex. Next, Figure \ref{FIG:APP-TRI} displays the triangulations of slices used for the BPST method in the model fitting and constructing the SCCs. Finally, Figures \ref{FIG:APP-EST2} and \ref{FIG:APP-EST3} provide the image maps of the estimated coefficient functions for eighth, 15th, 35th, 55th, 62nd, and 65th slices, and Figures \ref{FIG:APP-SCC2} and \ref{FIG:APP-SCC3} show the corresponding significance maps. The significance maps in the eighth and 15th slice show that the increase of age increases the brain activities in the cerebellum and temporal lobe, and people with the Alzheimer's disease are more active in the cerebellum, while less active in  the temporal lobe. The significance maps of the 35th slide display that the age has a negative effect on the brain activities in the anterior cingulate gyrus, corpus callosum, and part of the cerebral white matter, while the female has a higher level of activities in these regions. These regions connect the left and right cerebral hemispheres and enabling communication between them. From the significance maps of the 55th, 62nd, and 65th slices, we could see an increase of brain activities in the frontal gyrus, precentral gyrus and postcentral gyrus for people with Alzheimer's disease. Our findings are consistent with the findings in the literature, see \cite{andersen2012stereological}, \cite{bernard2014moving}, and \cite{dubb2003characterization}.

\begin{table}
	\caption{Distribution of patients by diagnosis status and gender. \label{TAB:diagnosis}}
	\centering
	\fbox{
		\begin{tabular}{l|ccc|c}
			& CN & MCI & AD & All\\ \hline
			Male   & 70 & 136 & 72 & 278\\
			Female & 42 & 77  & 50 & 169\\ \hline
			All & 112 & 213 & 122 & 447\\
	\end{tabular}}
\end{table}

\begin{figure}[ht]
	\begin{center}
		\begin{tabular}{cccccccc}
			Slice 5 & Slice 8& Slice 15 & Slice 35  \\ 
			\includegraphics[scale=0.6]{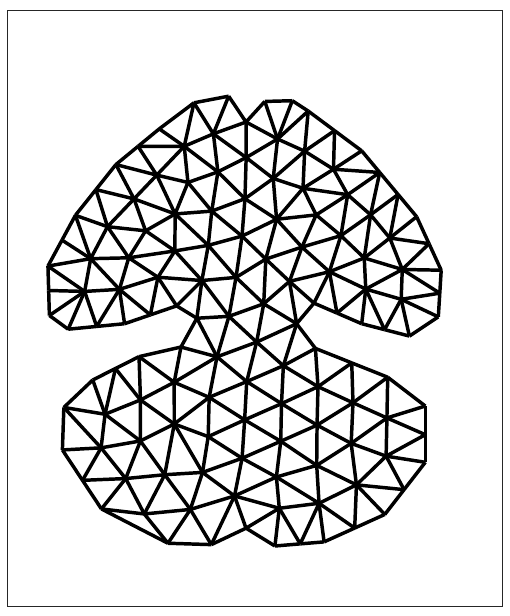} &  \includegraphics[scale=0.6]{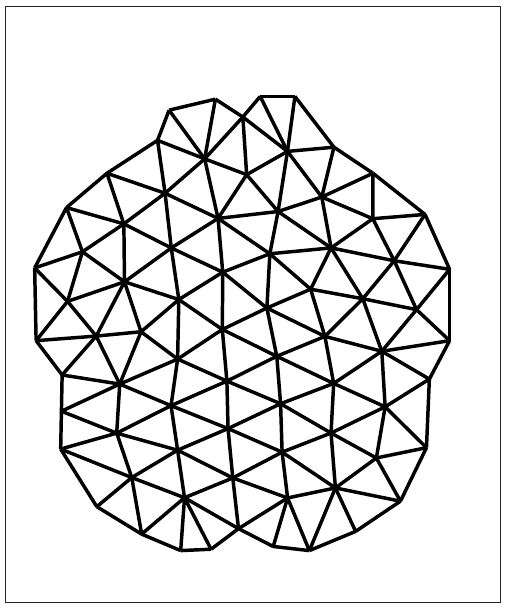} &
			\includegraphics[scale=0.6]{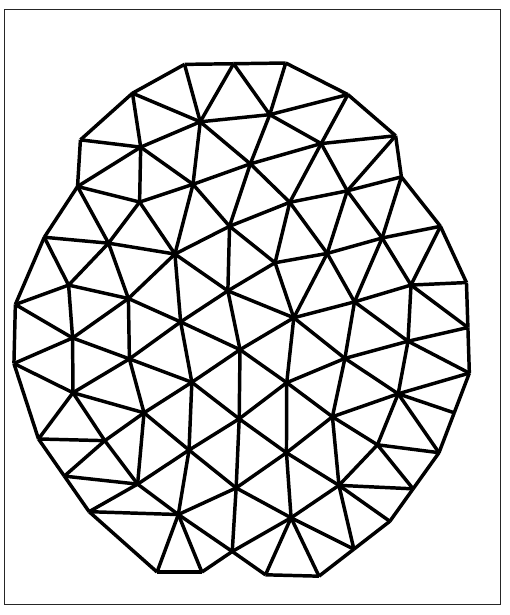} &  \includegraphics[scale=0.6]{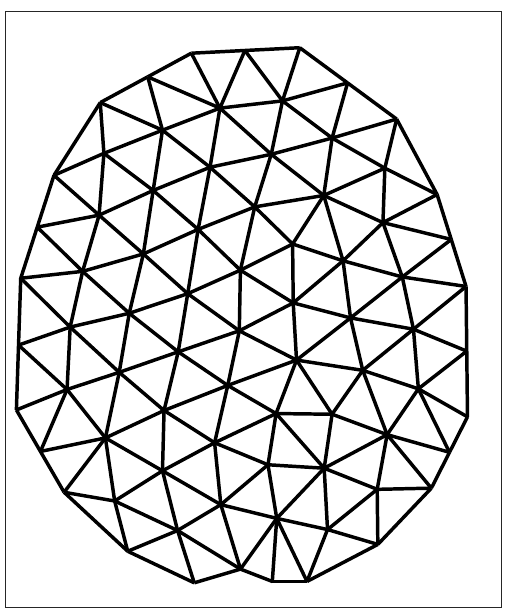} &\\
			Slice 55 & Slice 62 & Slice 65\\
			\includegraphics[scale=0.6]{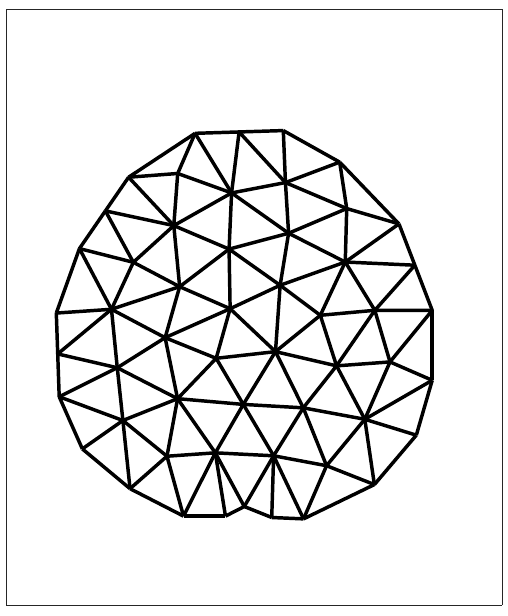}&
			\includegraphics[scale=0.6]{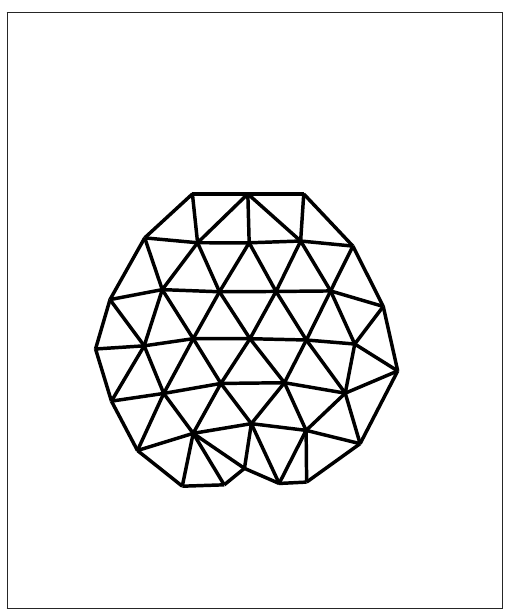} & \includegraphics[scale=0.6]{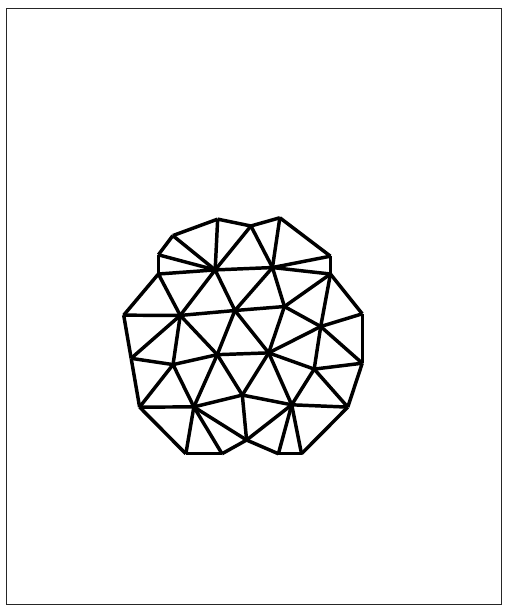} &
		\end{tabular}
		\caption{Triangulation sets used in the ADNI data analysis.}
		\label{FIG:APP-TRI}
	\end{center}
\end{figure}

\begin{sidewaysfigure}[htbp]
	\begin{center}
		\begin{tabular}{cccccccc} 
			Intercept & MA & AD&Age&Sex&$\textrm{APOE}_1$&$\textrm{APOE}_2$ \\ 
			\includegraphics[scale=0.33]{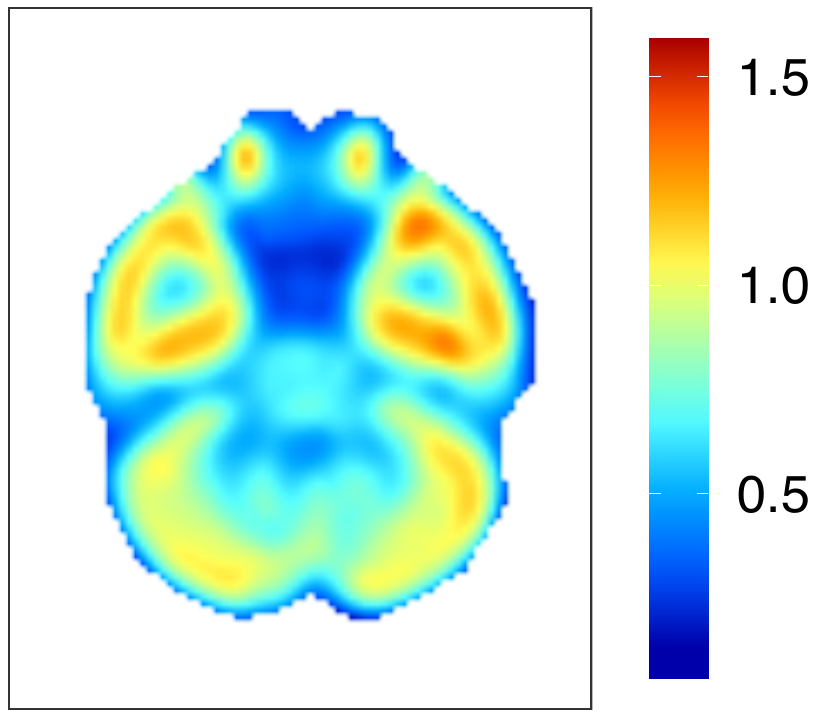} \!\!\!\!\! & \!\!\!\!
			\includegraphics[scale=0.33]{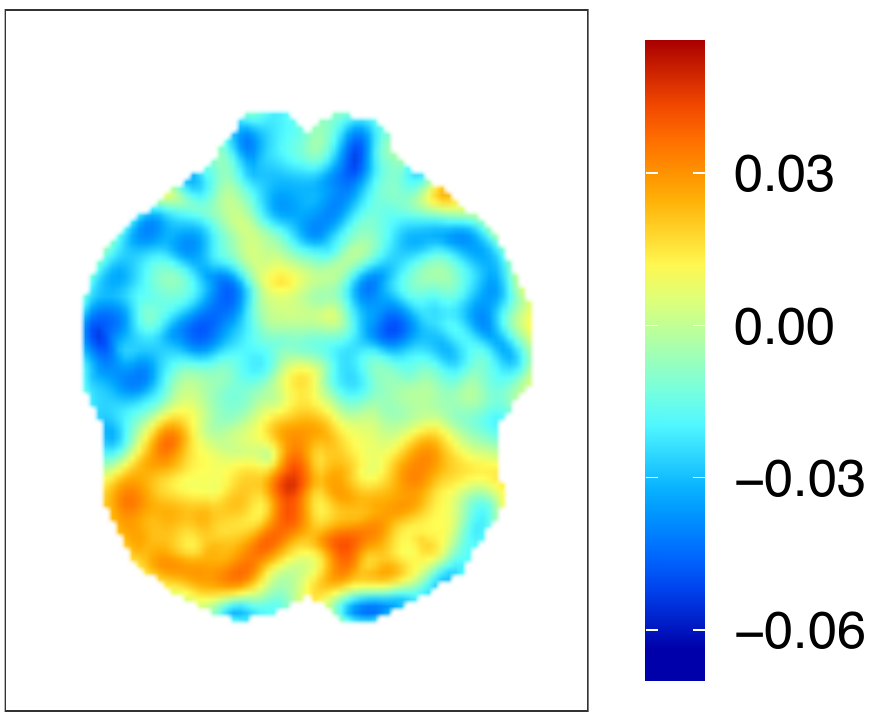} \!\!\!\!\! & \!\!\!\!
			\includegraphics[scale=0.33]{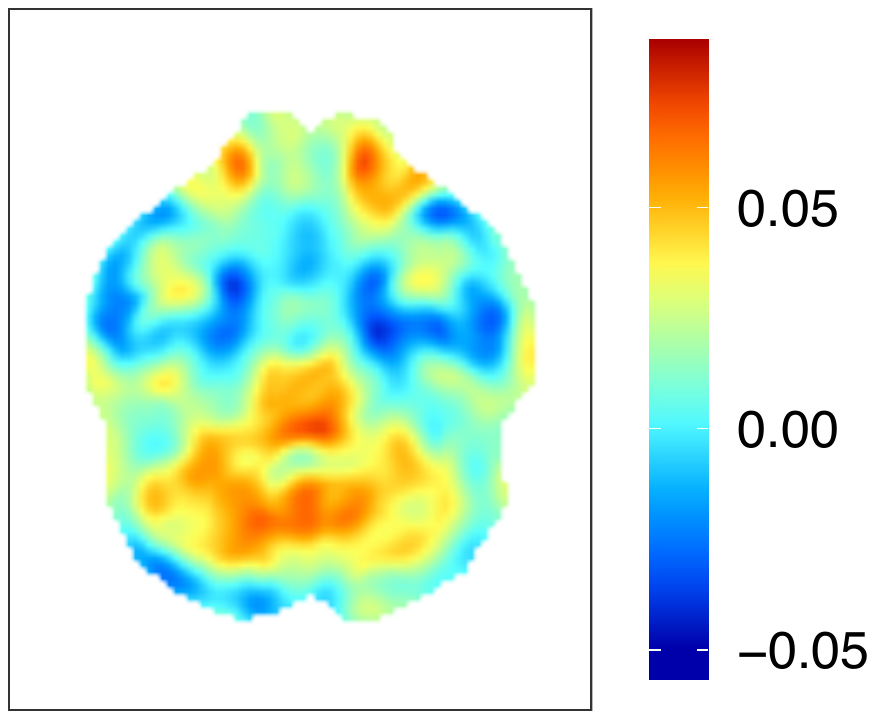} \!\!\!\!\! & \!\!\!\!
			\includegraphics[scale=0.33]{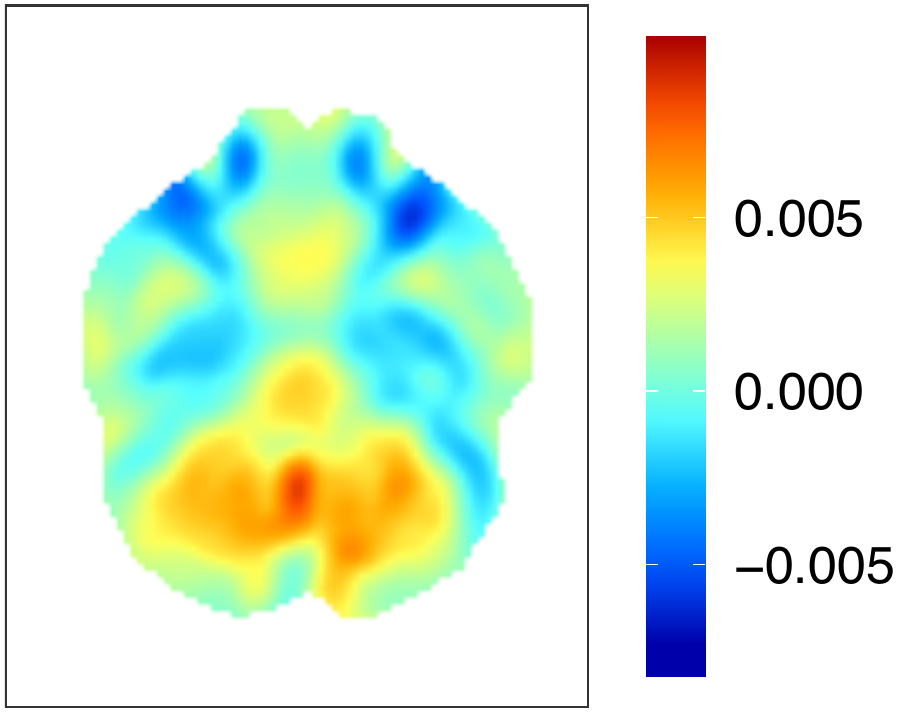}\!\!\!\!\! & \!\!\!\!
			\includegraphics[scale=0.33]{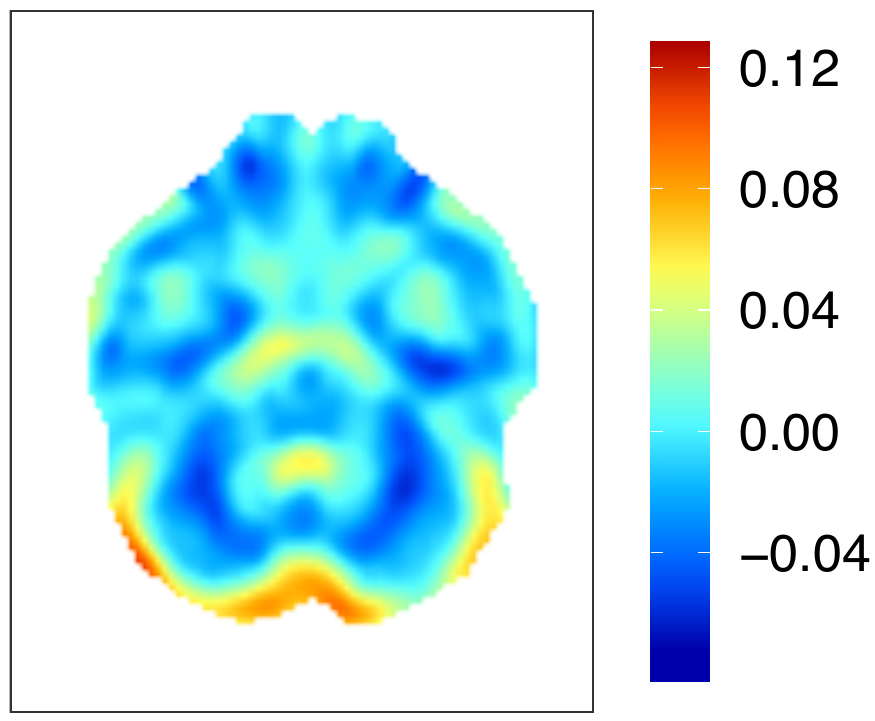} \!\!\!\!\! & \!\!\!\!
			\includegraphics[scale=0.33]{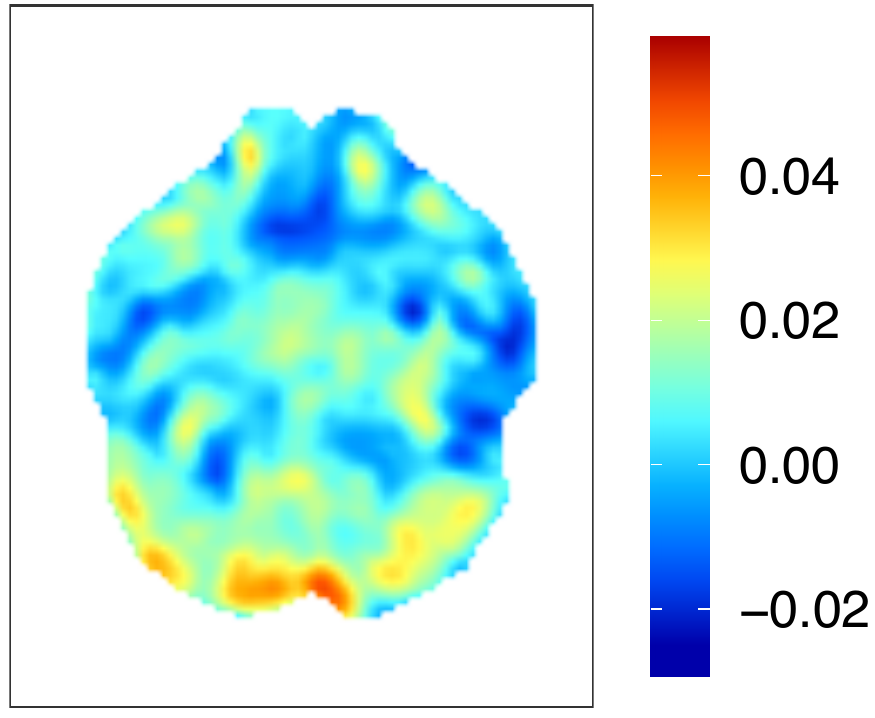} \!\!\!\!\! & \!\!\!\!
			\includegraphics[scale=0.33]{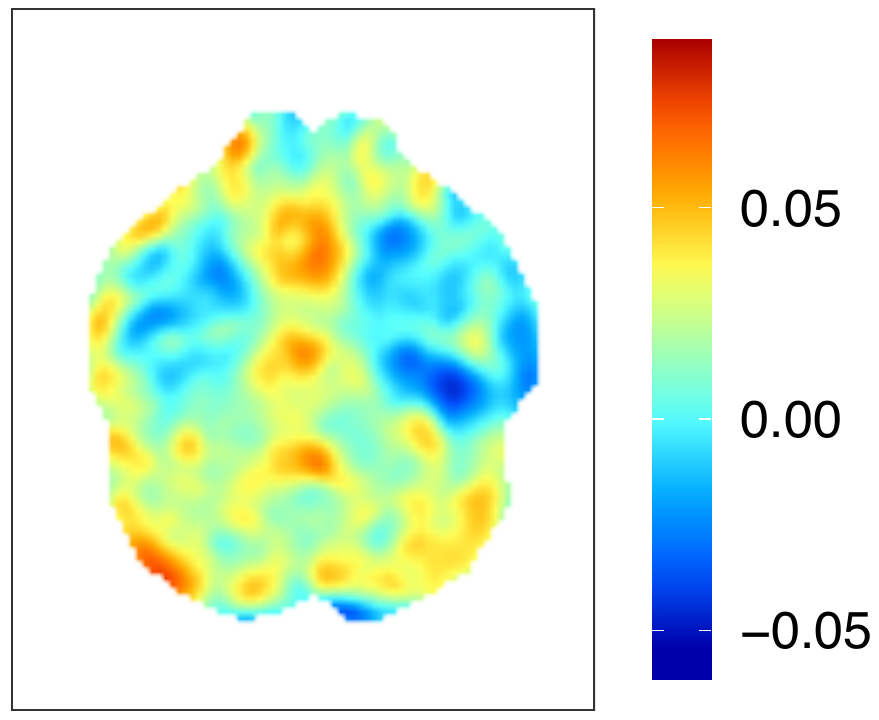} \!\!\!\!\! & \!\!\!\!\\[-5pt]
			\multicolumn{7}{c}{Slice 8}\\[5pt]
			\includegraphics[scale=0.33]{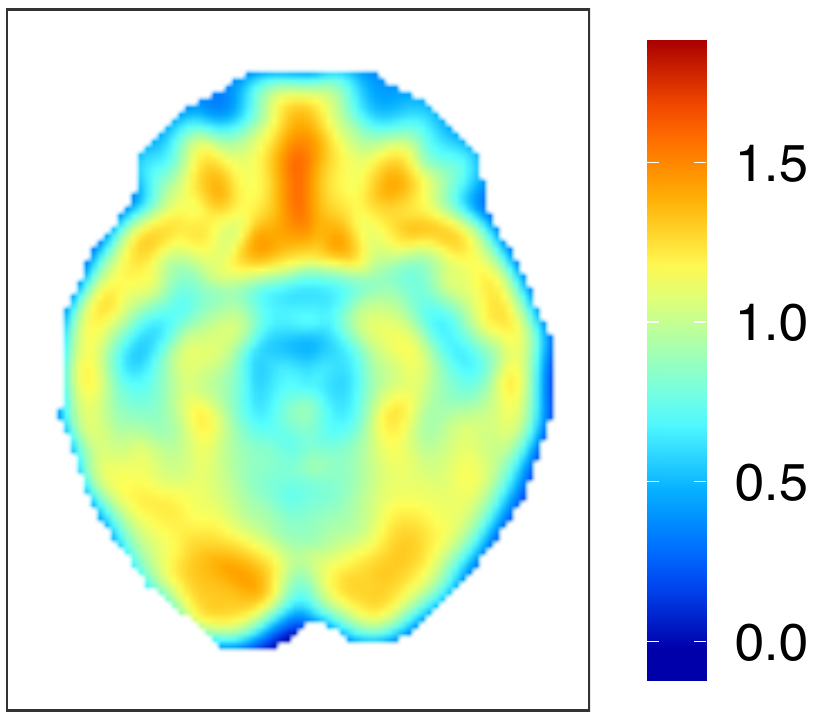} \!\!\!\!\! & \!\!\!\!
			\includegraphics[scale=0.33]{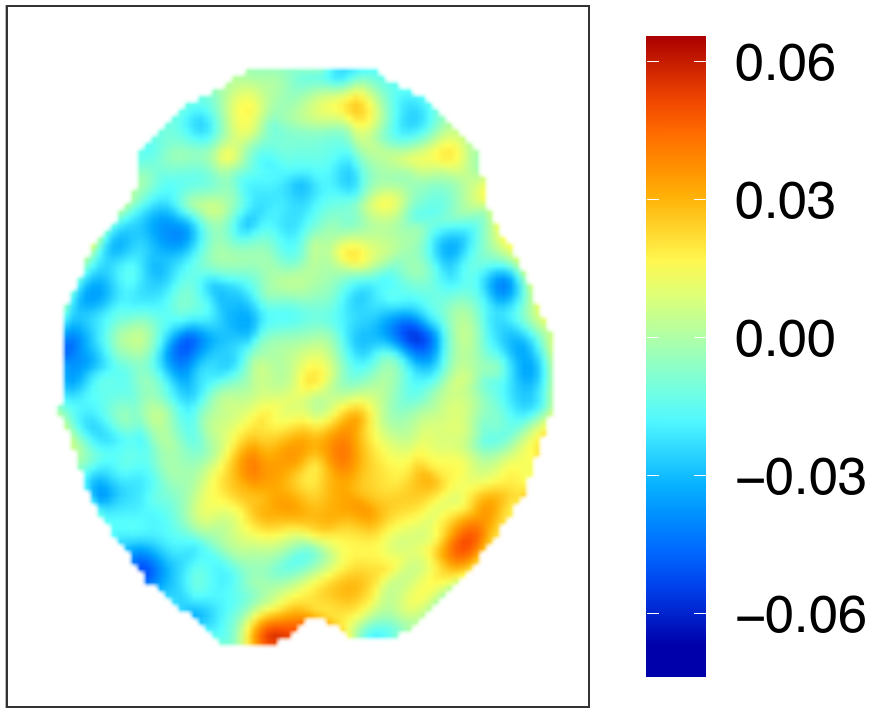} \!\!\!\!\! & \!\!\!\!
			\includegraphics[scale=0.33]{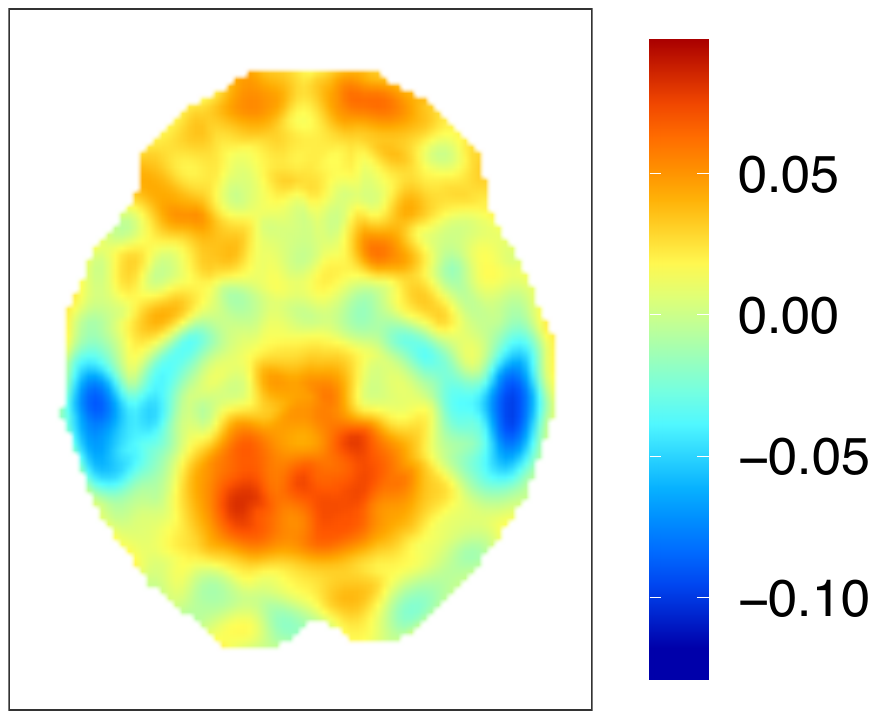} \!\!\!\!\! & \!\!\!\!
			\includegraphics[scale=0.33]{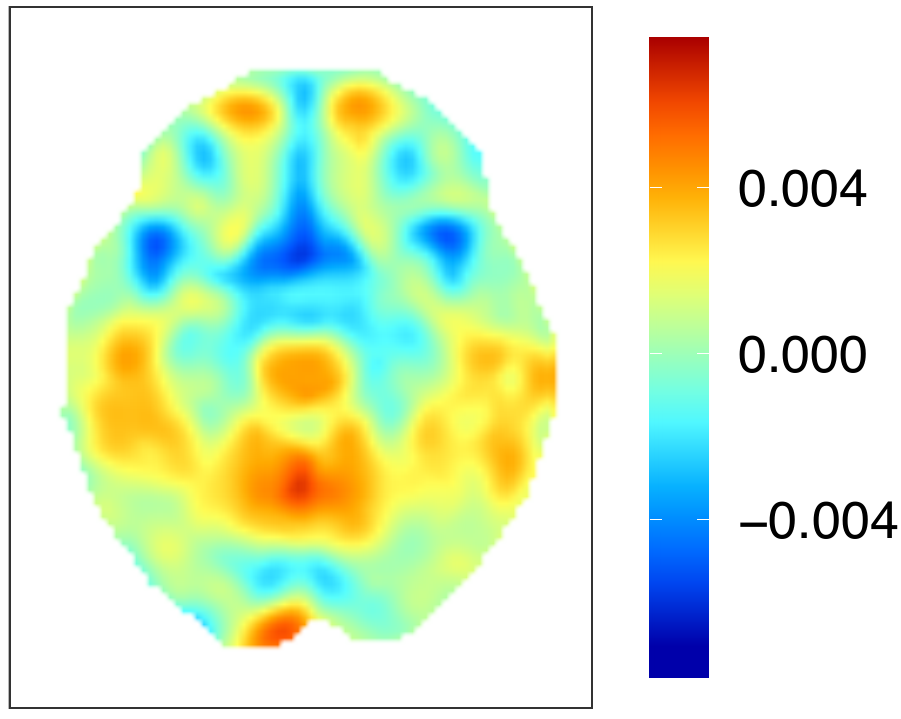}\!\!\!\!\! & \!\!\!\!
			\includegraphics[scale=0.33]{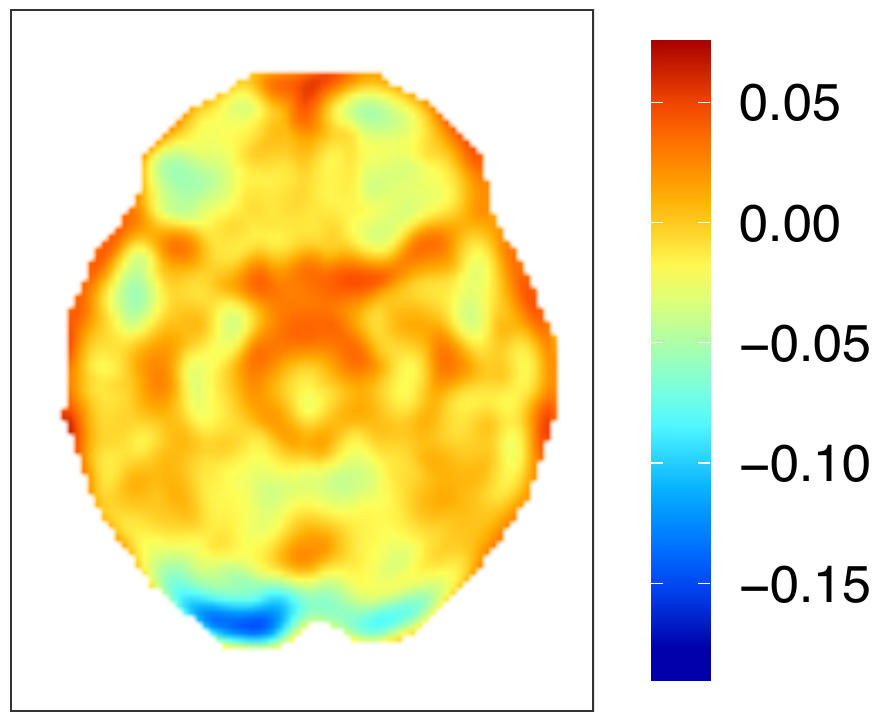} \!\!\!\!\! & \!\!\!\!
			\includegraphics[scale=0.33]{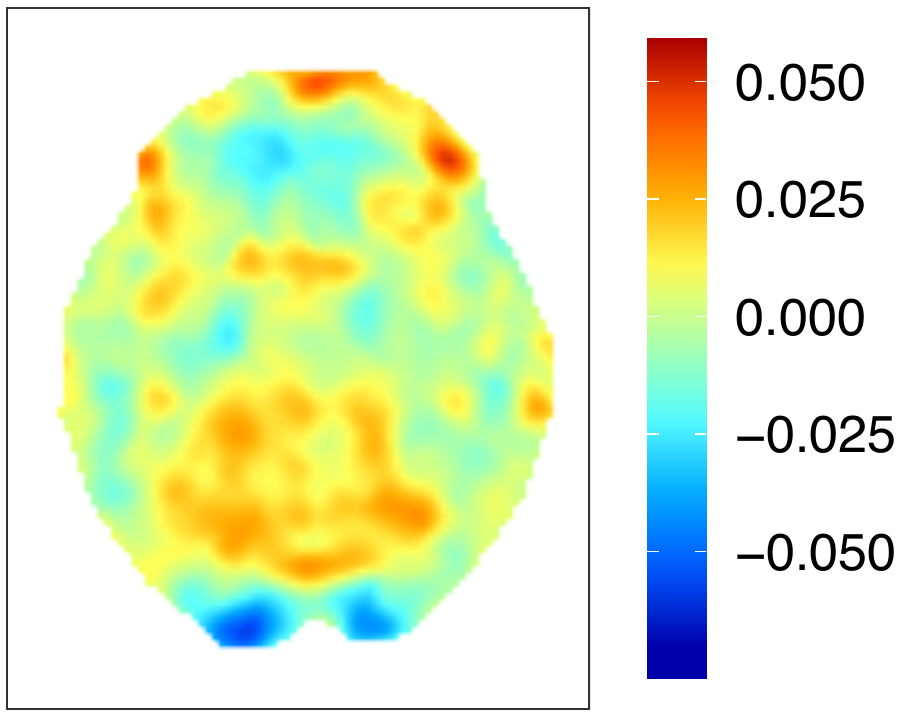} \!\!\!\!\! & \!\!\!\!
			\includegraphics[scale=0.33]{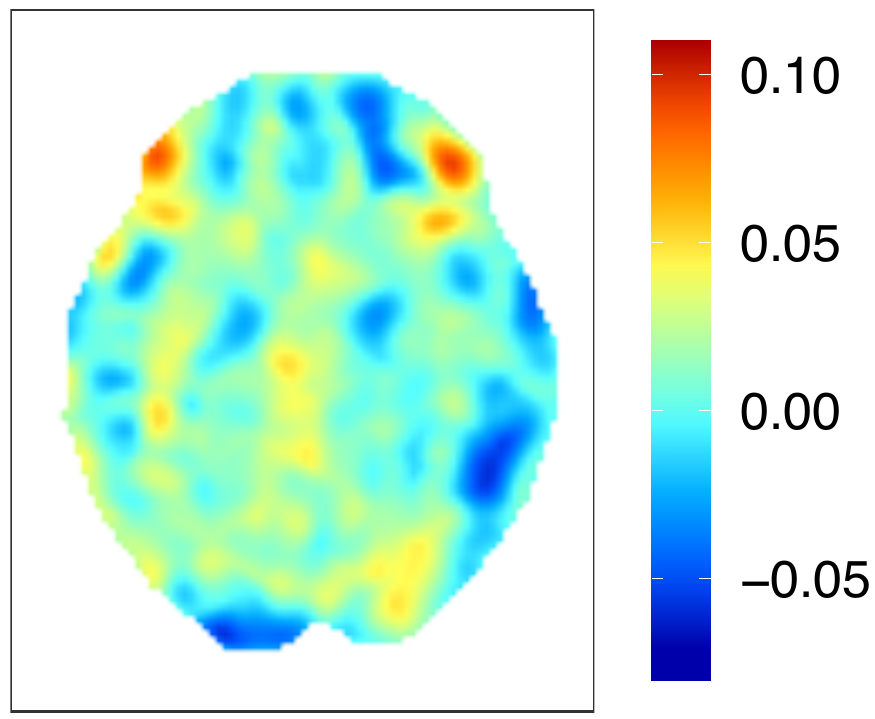} \!\!\!\!\! & \!\!\!\!\\[-5pt]
			\multicolumn{7}{c}{Slice 15}\\[5pt]
			\includegraphics[scale=0.33]{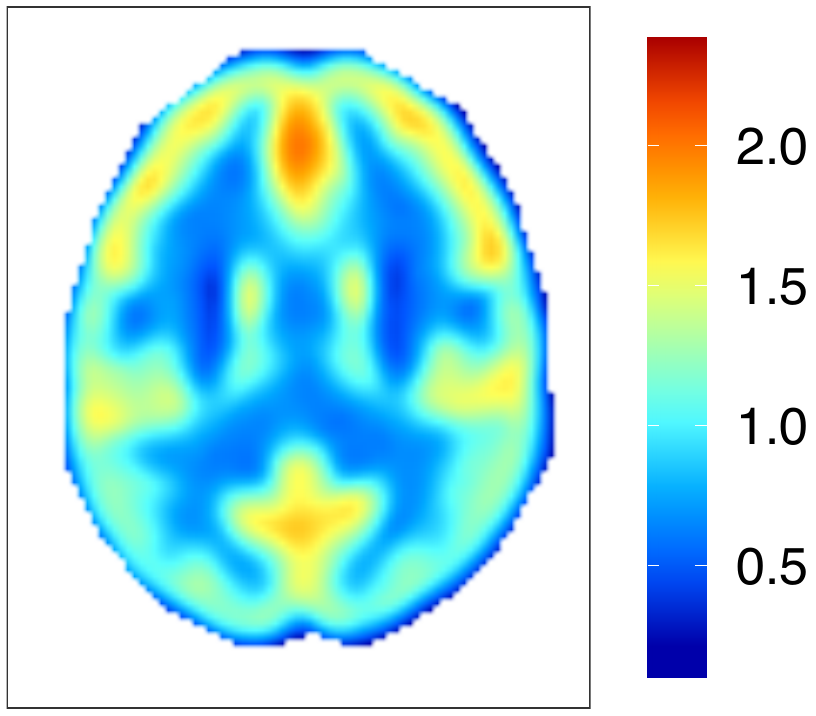} \!\!\!\!\! & \!\!\!\!
			\includegraphics[scale=0.33]{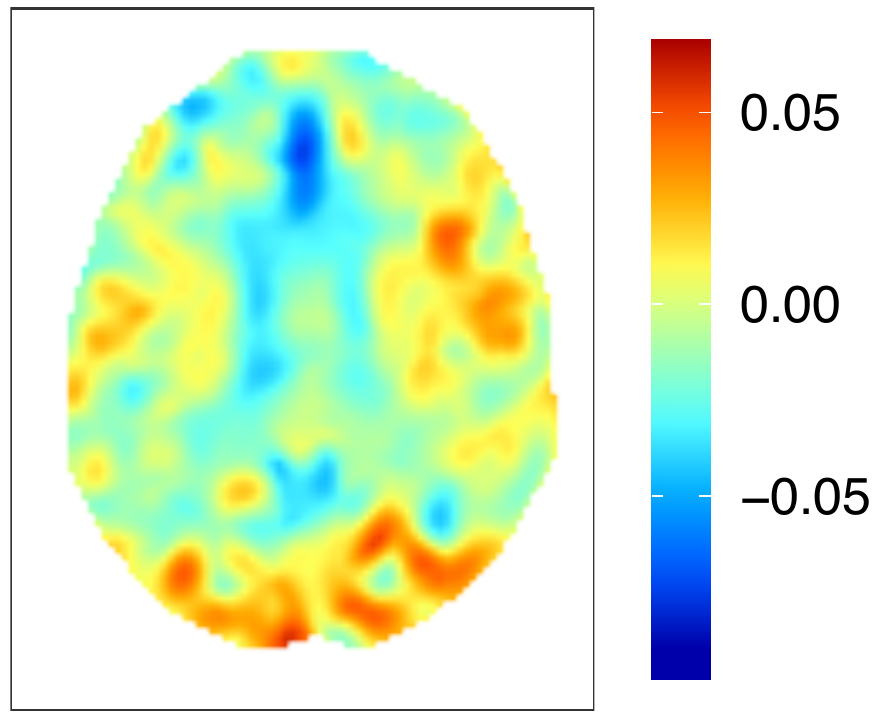} \!\!\!\!\! & \!\!\!\!
			\includegraphics[scale=0.33]{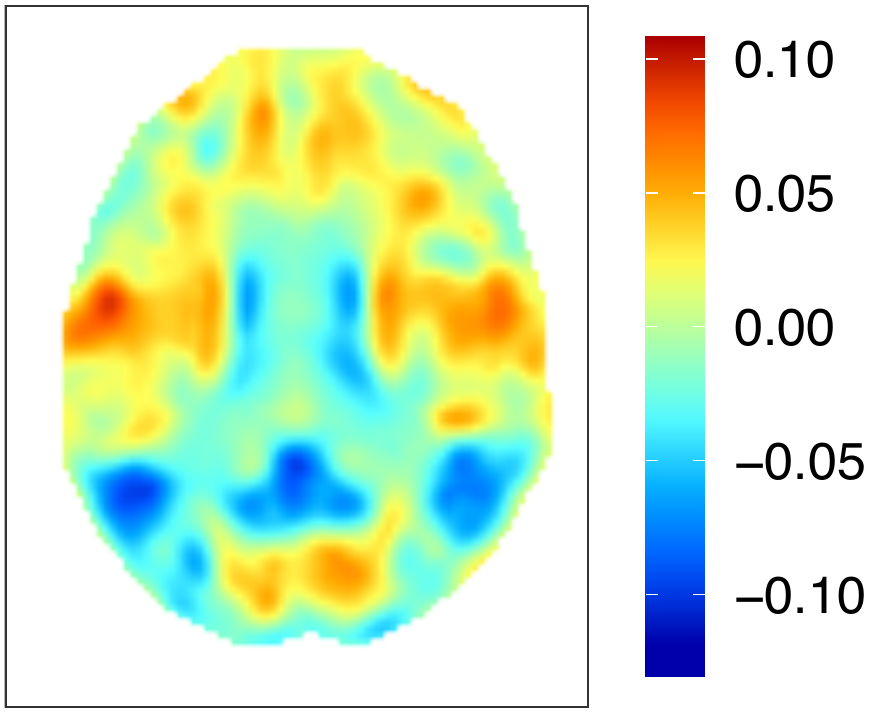} \!\!\!\!\! & \!\!\!\!
			\includegraphics[scale=0.33]{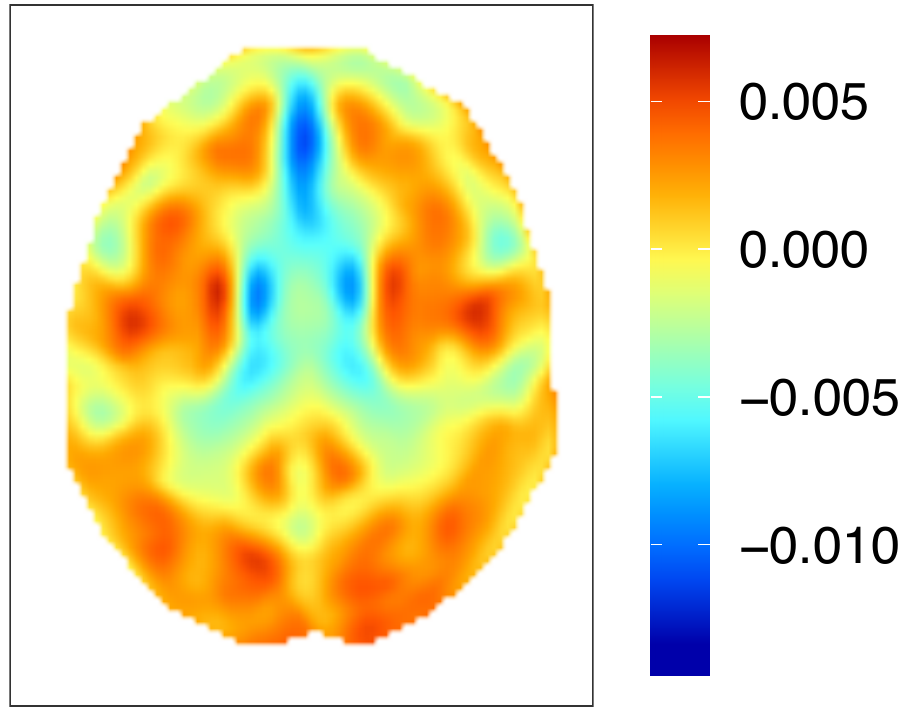}\!\!\!\!\! & \!\!\!\!
			\includegraphics[scale=0.33]{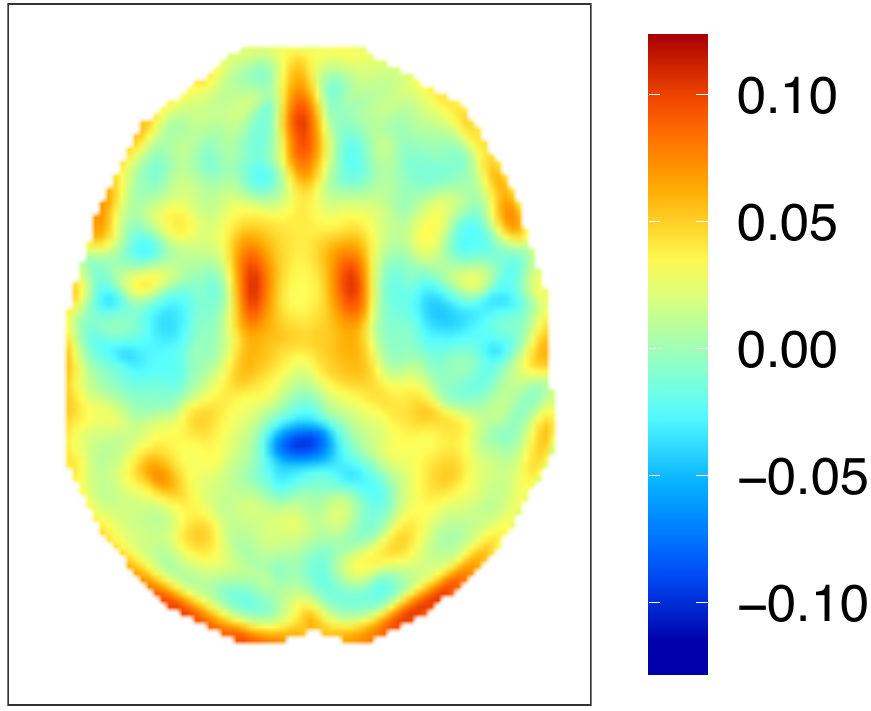} \!\!\!\!\! & \!\!\!\!
			\includegraphics[scale=0.33]{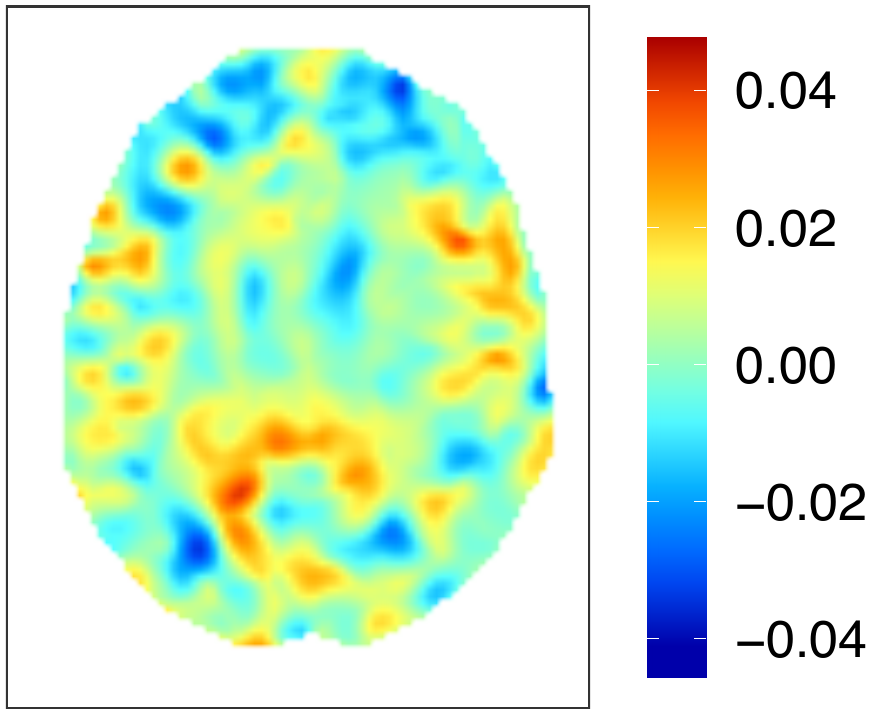} \!\!\!\!\! & \!\!\!\!
			\includegraphics[scale=0.33]{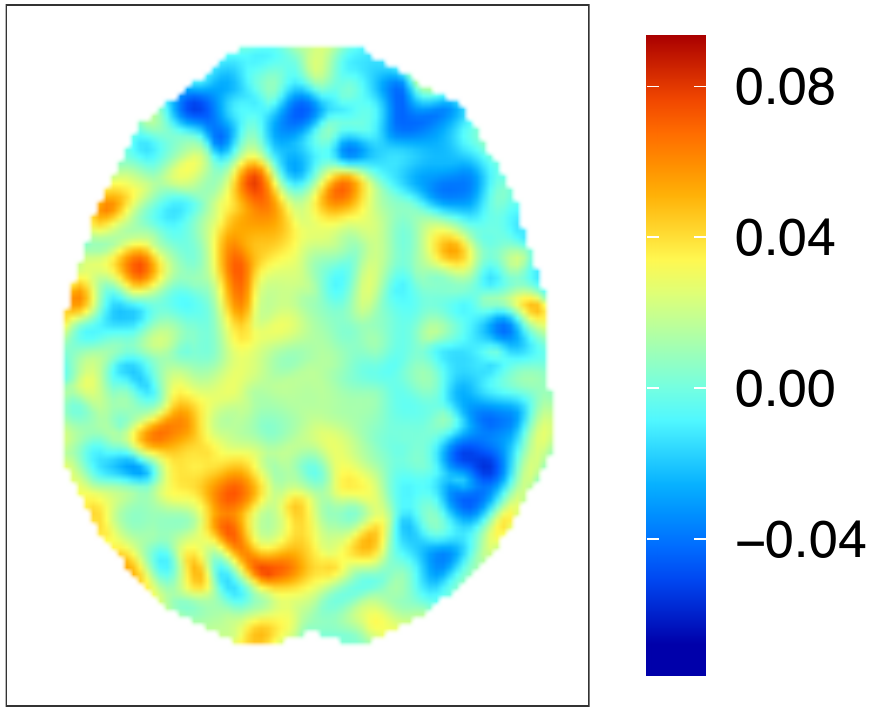} \!\!\!\!\! & \!\!\!\!\\[-5pt]
			\multicolumn{7}{c}{Slice 35}\\[5pt]
		\end{tabular}	
	\end{center}
	\caption{The BPST estimates of the coefficient functions for the ADNI data based on the eighth, 15th and 35th slices, respectively.}
	\label{FIG:APP-EST2}
\end{sidewaysfigure}

\begin{sidewaysfigure}[htbp]
	\begin{center}
		\begin{tabular}{cccccccc} 
			Intercept & MA & AD&Age&Sex&$\textrm{APOE}_1$&$\textrm{APOE}_2$ \\ 
			\includegraphics[scale=0.33]{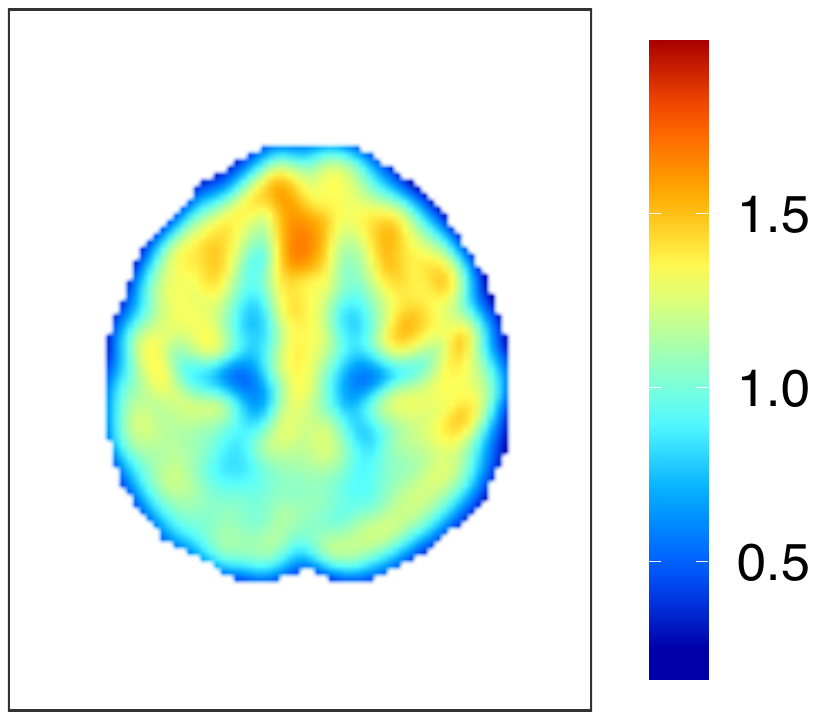} \!\!\!\!\! & \!\!\!\!
			\includegraphics[scale=0.33]{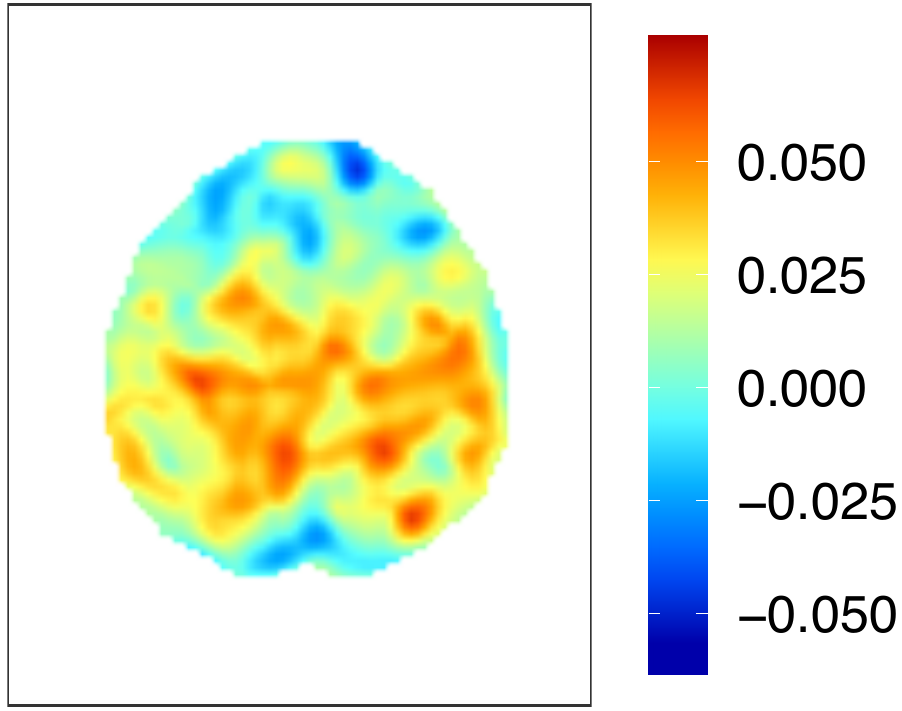} \!\!\!\!\! & \!\!\!\!
			\includegraphics[scale=0.33]{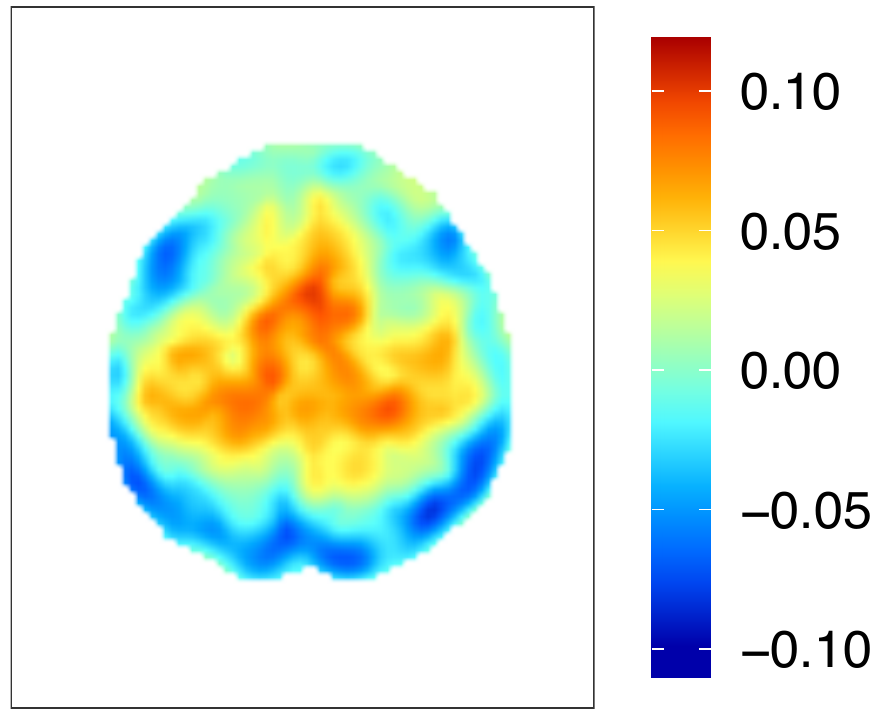} \!\!\!\!\! & \!\!\!\!
			\includegraphics[scale=0.33]{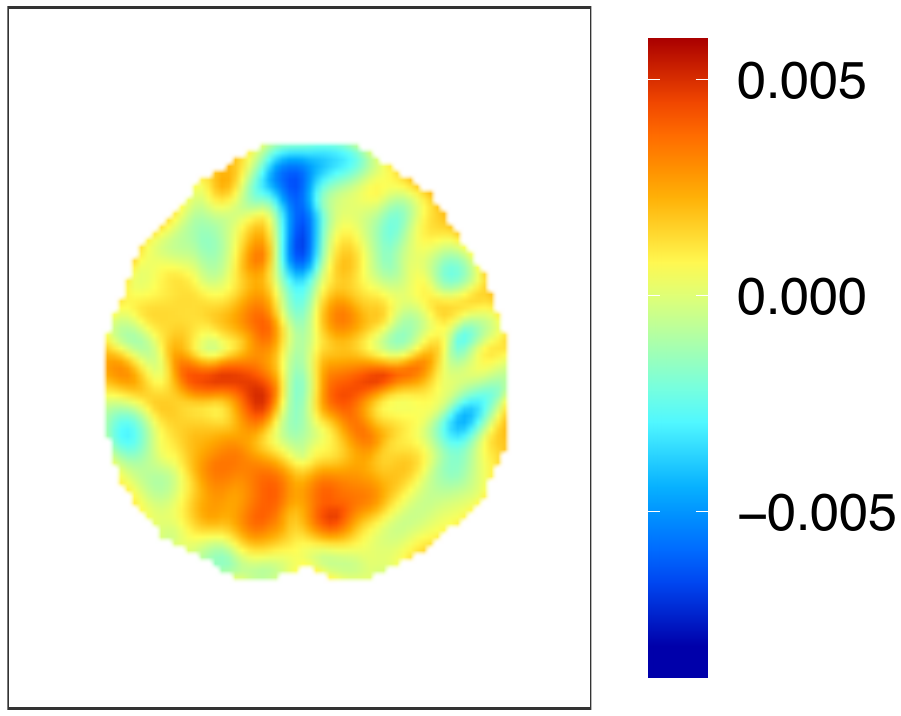}\!\!\!\!\! & \!\!\!\!
			\includegraphics[scale=0.33]{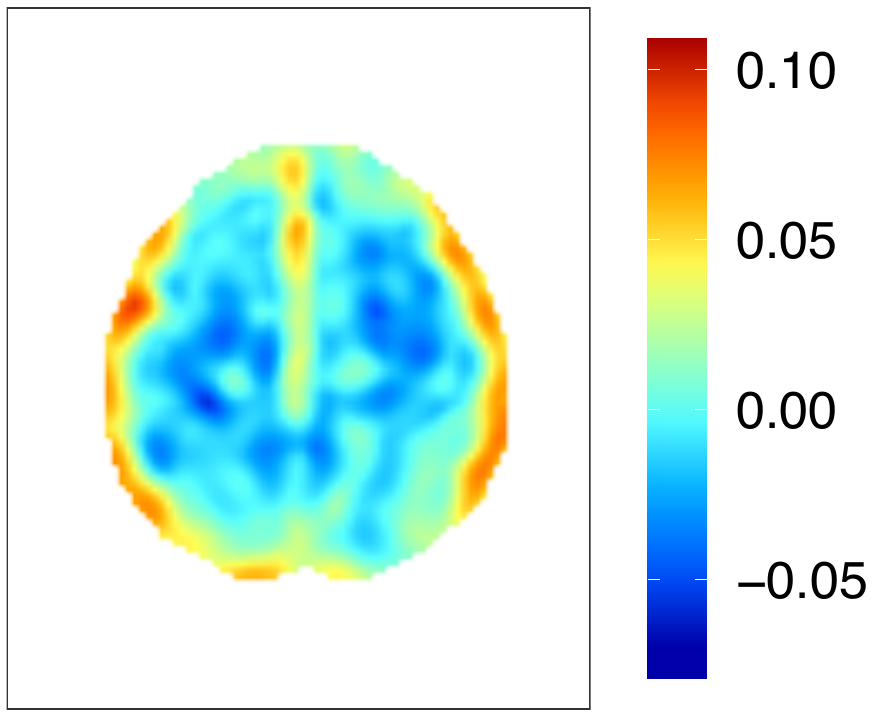} \!\!\!\!\! & \!\!\!\!
			\includegraphics[scale=0.33]{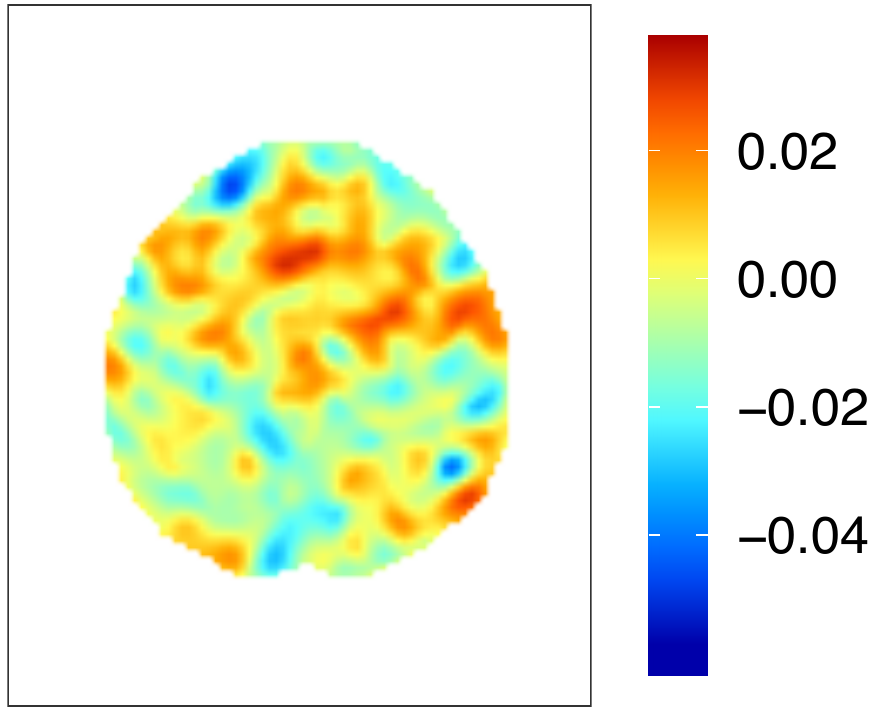} \!\!\!\!\! & \!\!\!\!
			\includegraphics[scale=0.33]{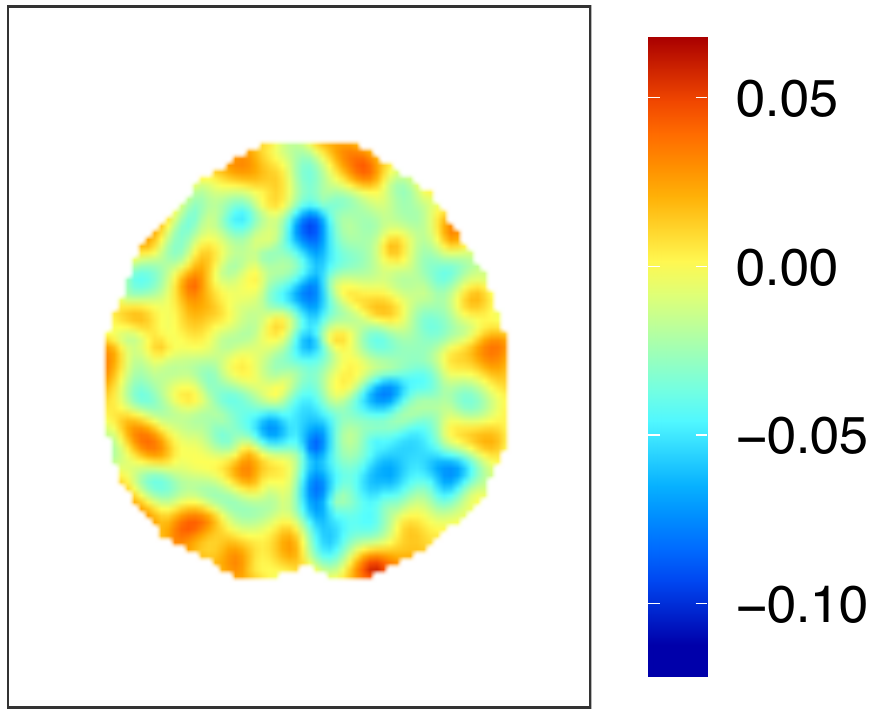} \!\!\!\!\! & \!\!\!\!\\[-5pt]
			\multicolumn{7}{c}{Slice 55}\\[5pt]
			\includegraphics[scale=0.33]{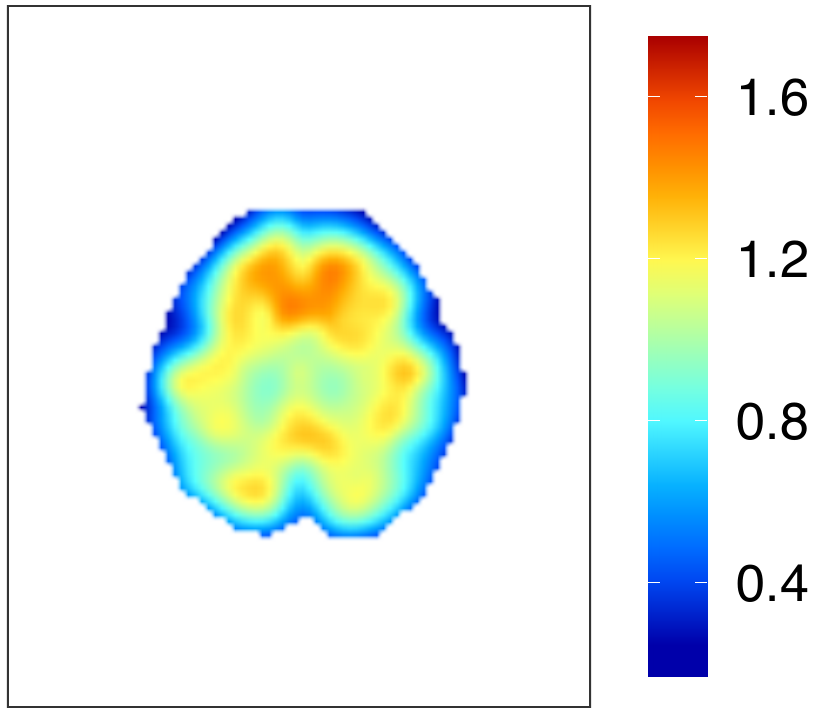} \!\!\!\!\! & \!\!\!\!
			\includegraphics[scale=0.33]{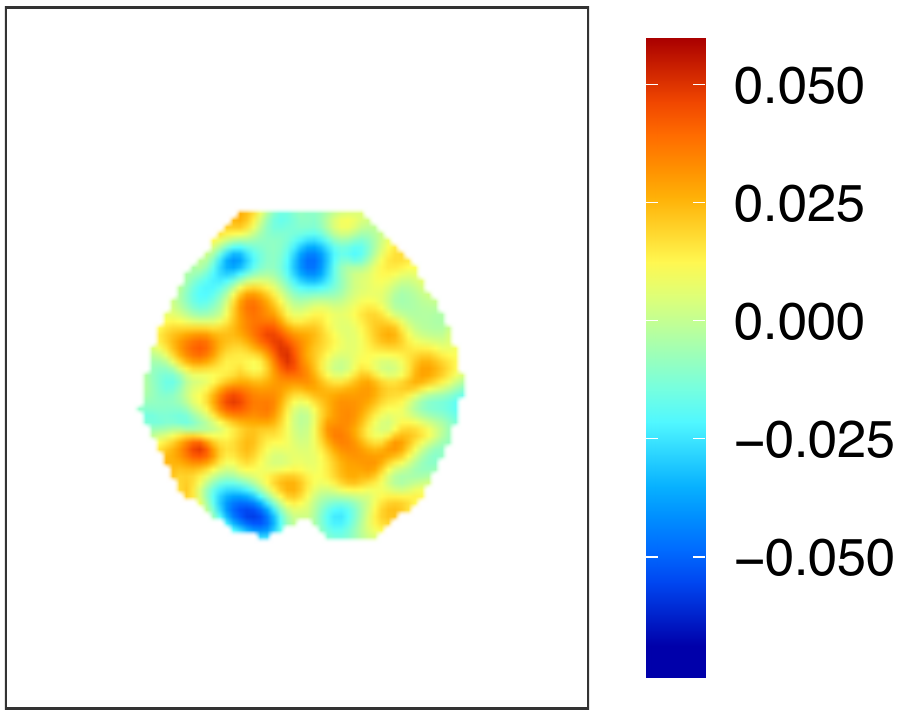} \!\!\!\!\! & \!\!\!\!
			\includegraphics[scale=0.33]{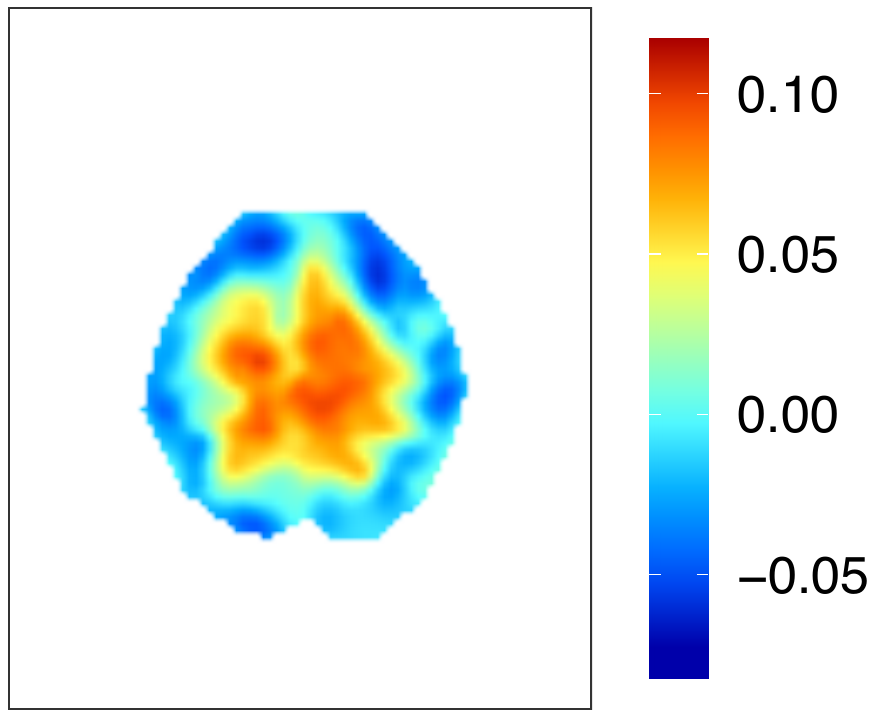} \!\!\!\!\! & \!\!\!\!
			\includegraphics[scale=0.33]{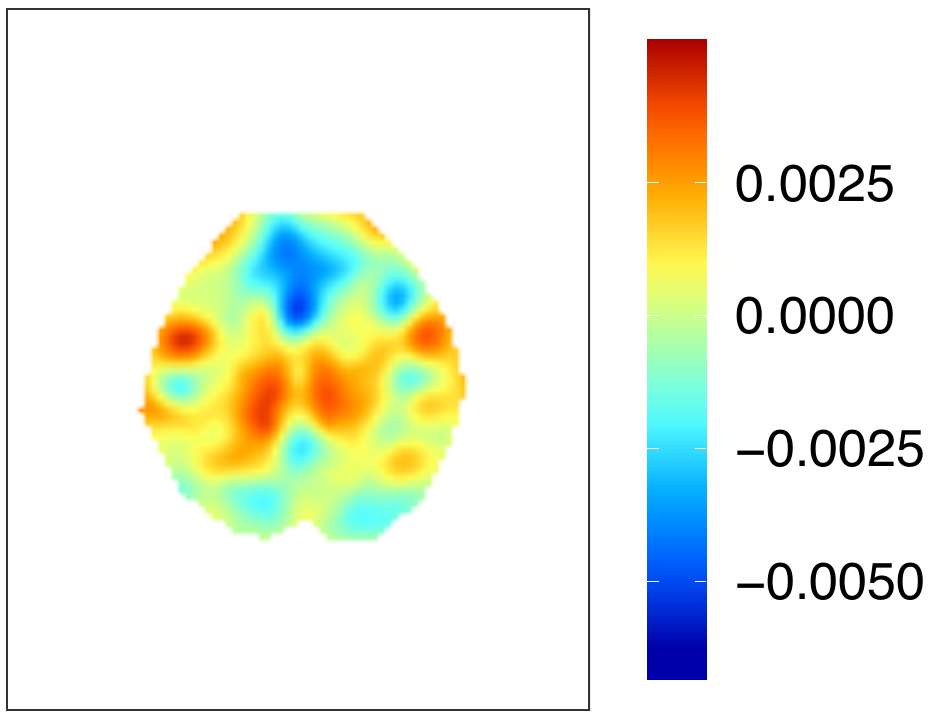}\!\!\!\!\! & \!\!\!\!
			\includegraphics[scale=0.33]{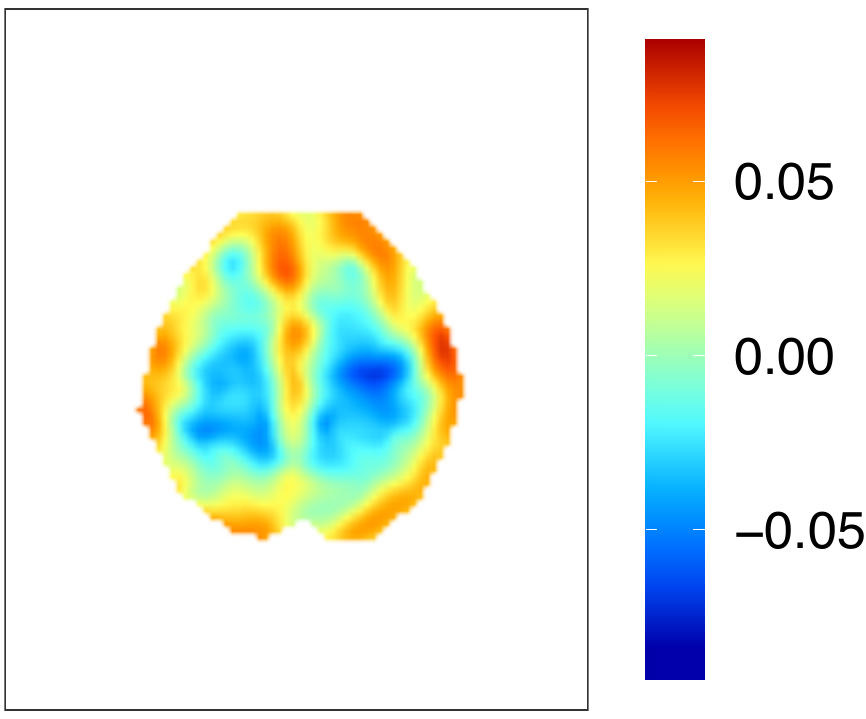} \!\!\!\!\! & \!\!\!\!
			\includegraphics[scale=0.33]{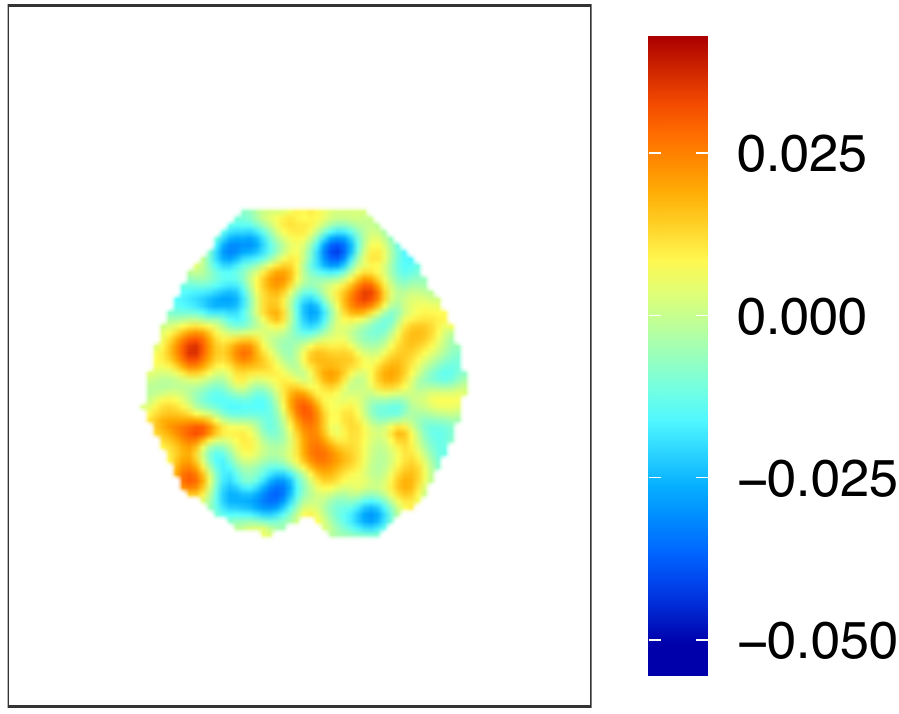} \!\!\!\!\! & \!\!\!\!
			\includegraphics[scale=0.33]{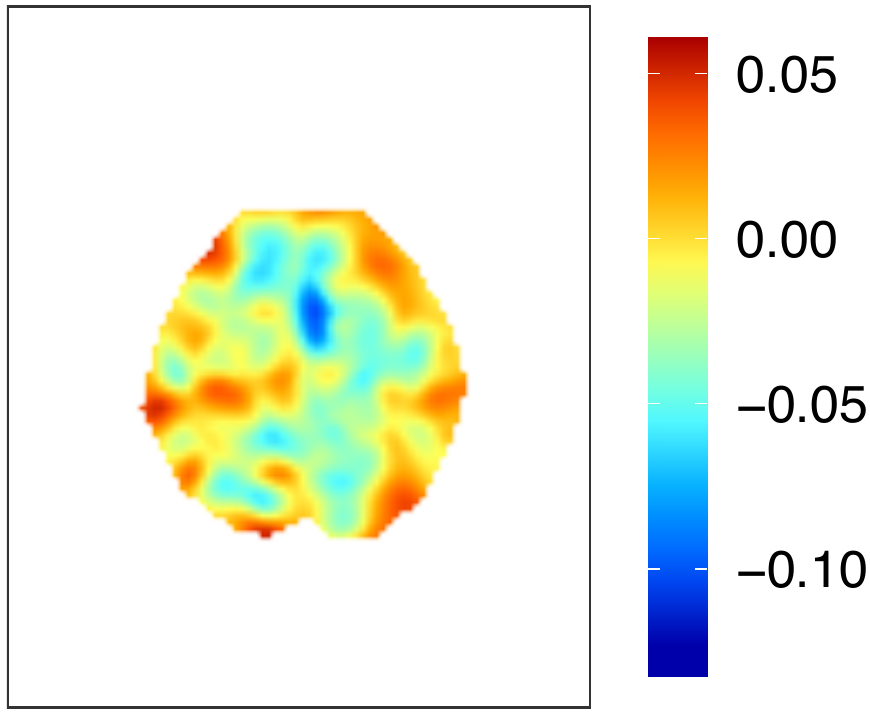} \!\!\!\!\! & \!\!\!\!\\[-5pt]
			\multicolumn{7}{c}{Slice 62}\\[5pt]
			\includegraphics[scale=0.33]{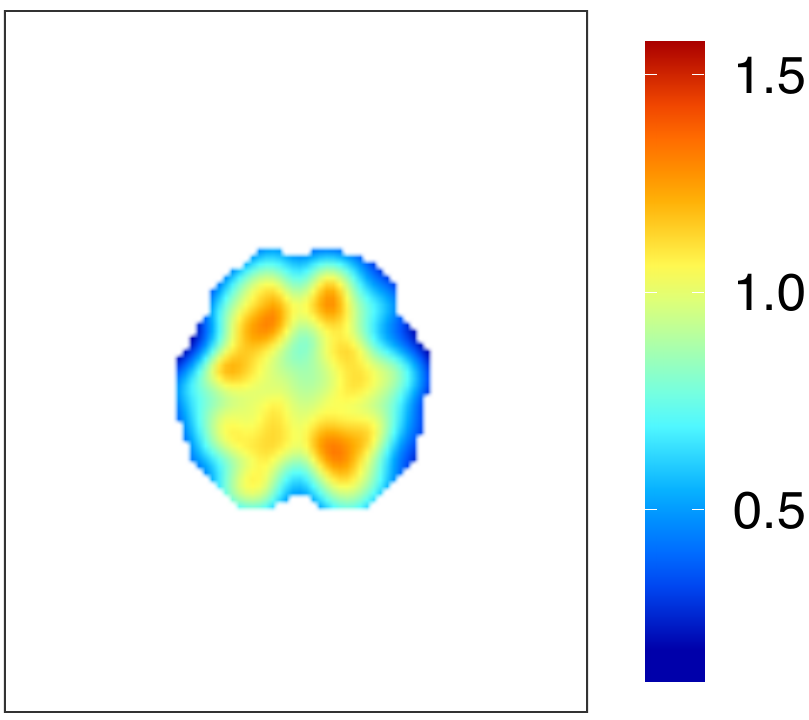} \!\!\!\!\! & \!\!\!\!
			\includegraphics[scale=0.33]{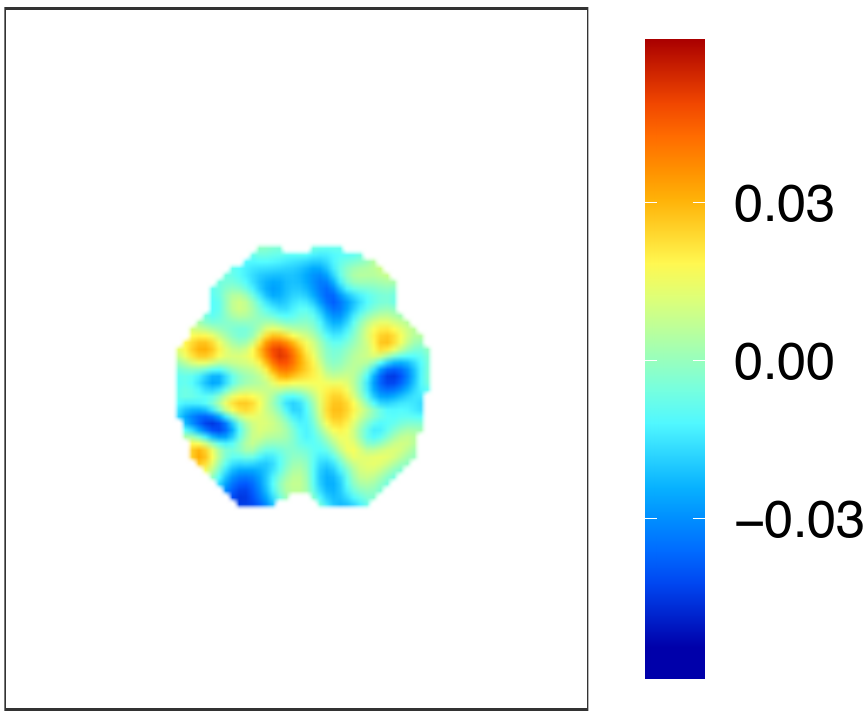} \!\!\!\!\! & \!\!\!\!
			\includegraphics[scale=0.33]{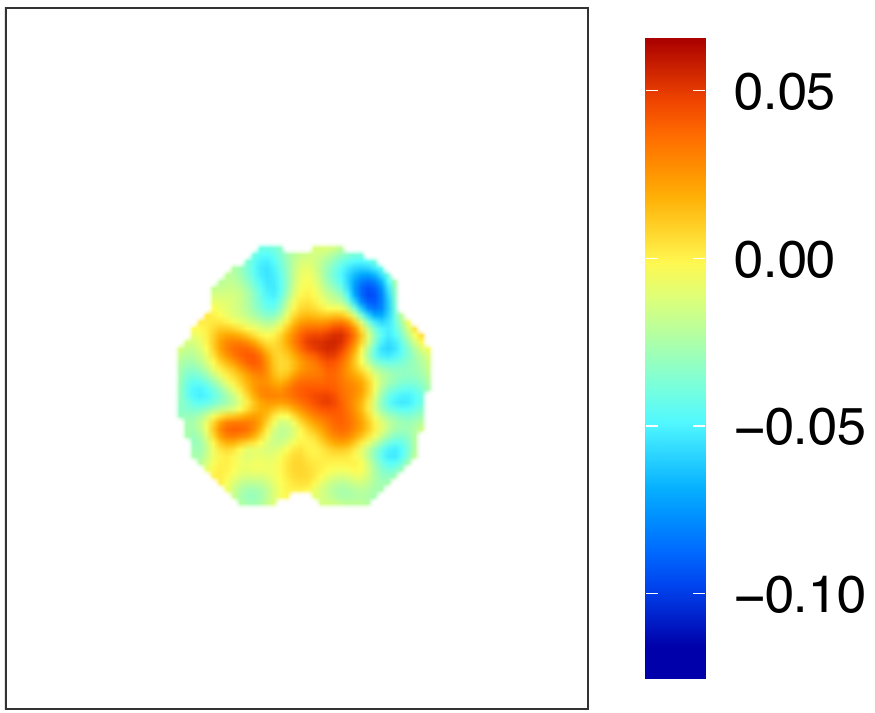} \!\!\!\!\! & \!\!\!\!
			\includegraphics[scale=0.33]{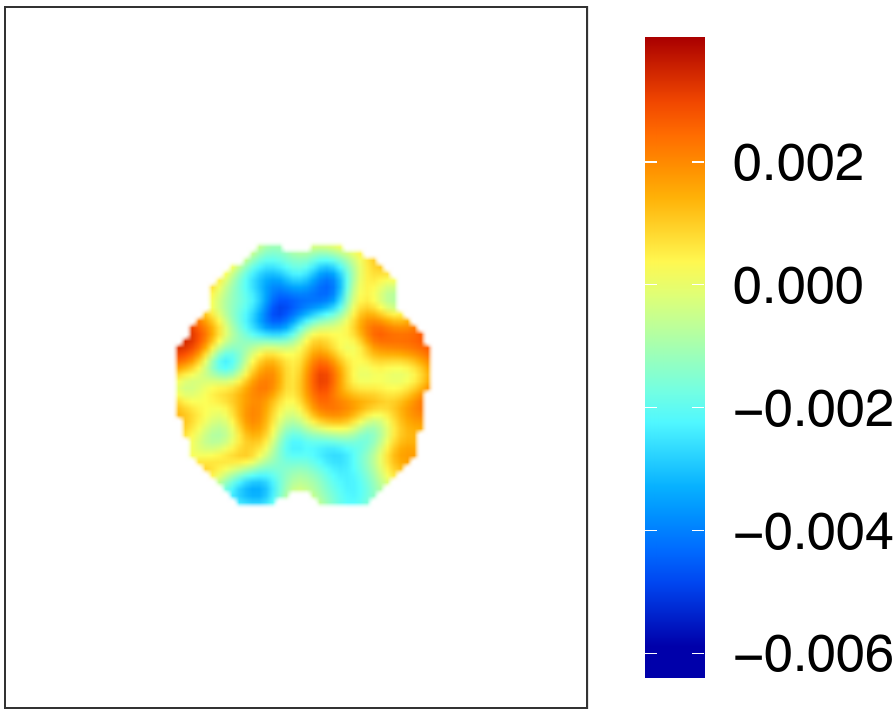}\!\!\!\!\! & \!\!\!\!
			\includegraphics[scale=0.33]{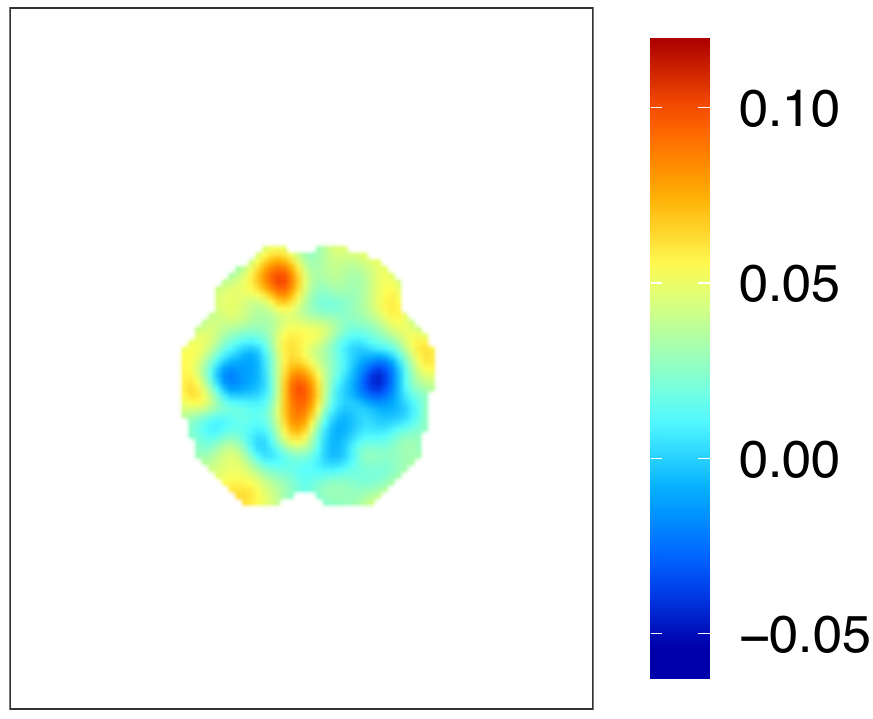} \!\!\!\!\! & \!\!\!\!
			\includegraphics[scale=0.33]{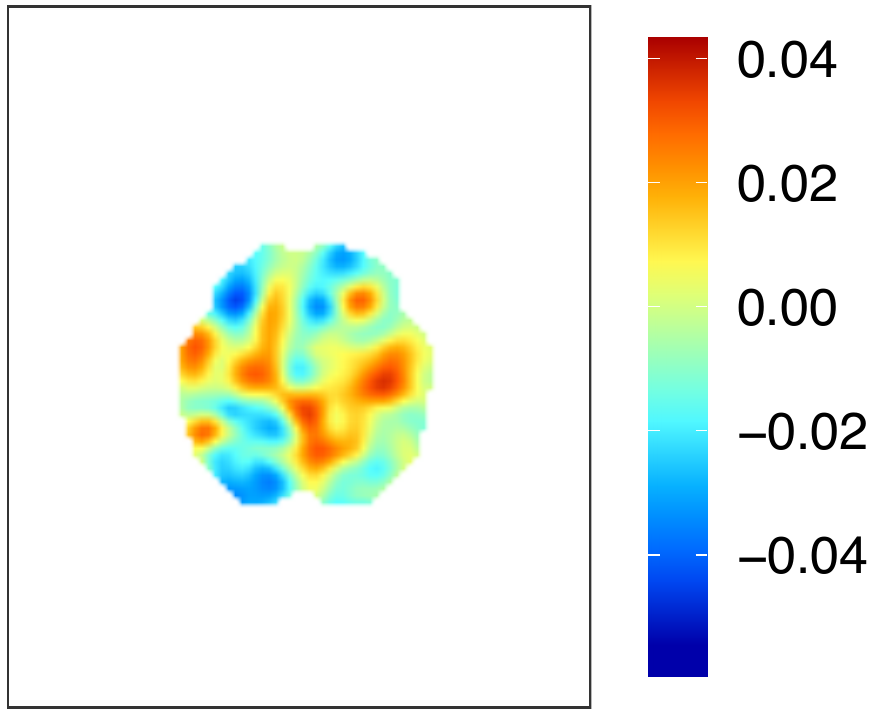} \!\!\!\!\! & \!\!\!\!
			\includegraphics[scale=0.33]{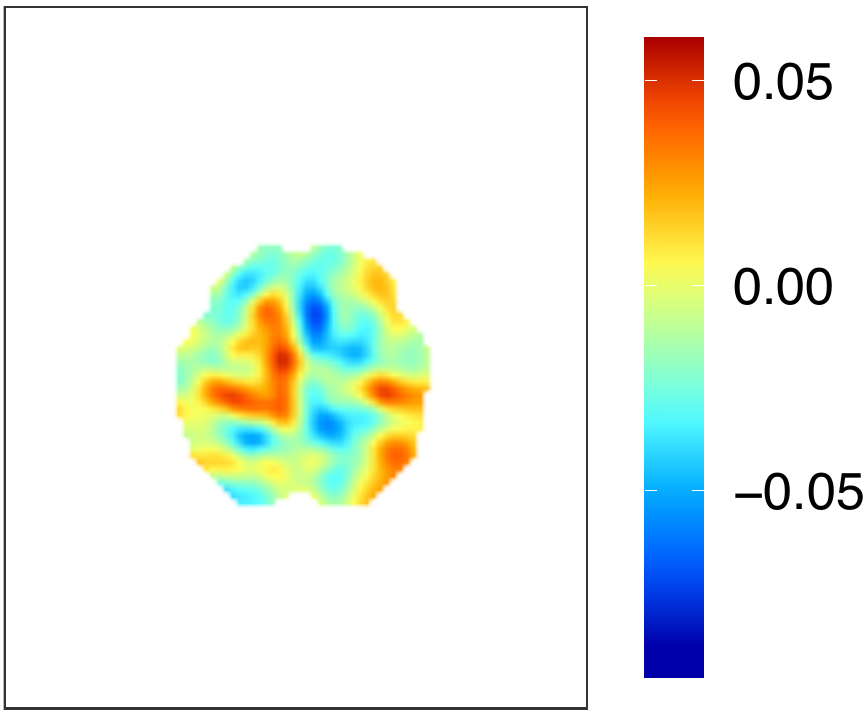} \!\!\!\!\! & \!\!\!\!\\[-5pt]
			\multicolumn{7}{c}{Slice 65}
		\end{tabular}	
	\end{center}
	\caption{The BPST estimates of the coefficient functions for the ADNI data based on the 55th, 62nd and 65th slices, respectively.}
	\label{FIG:APP-EST3}
\end{sidewaysfigure}

\begin{sidewaysfigure}[htbp]
	\begin{center}
		\begin{tabular}{cccccccc} 
			Intercept & MA & AD&Age&Sex&$\textrm{APOE}_1$&$\textrm{APOE}_2$  \\ 
			\includegraphics[scale=0.20]{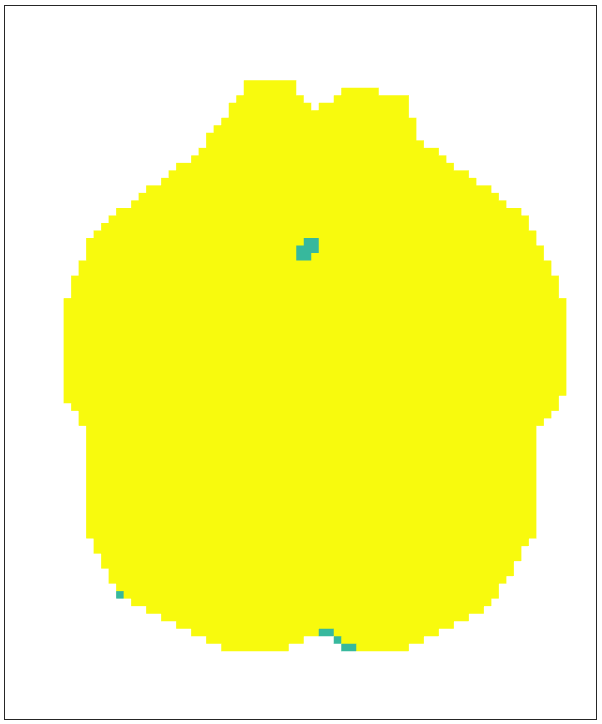} \!\!\!\!\! & \!\!\!\!
			\includegraphics[scale=0.20]{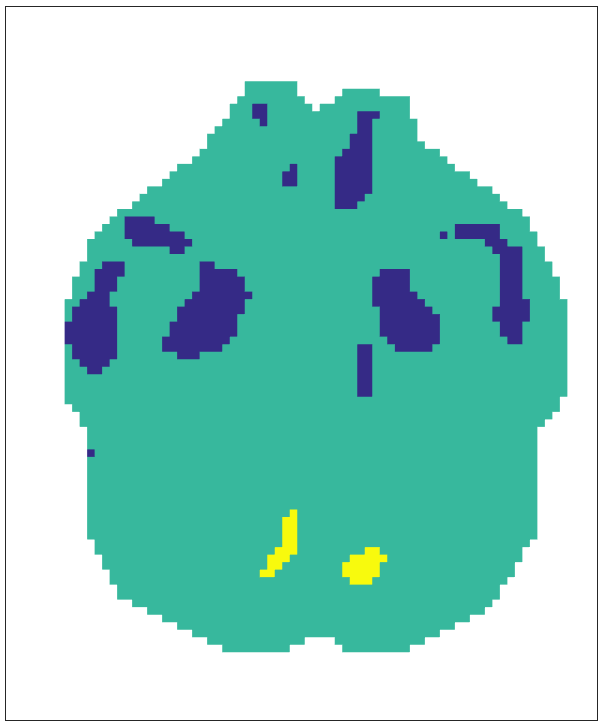} \!\!\!\!\! & \!\!\!\!
			\includegraphics[scale=0.20]{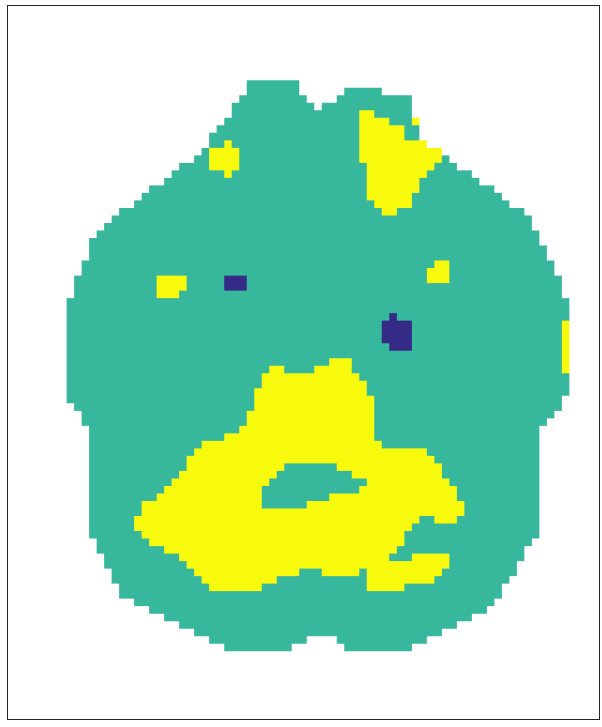} \!\!\!\!\! & \!\!\!\!
			\includegraphics[scale=0.20]{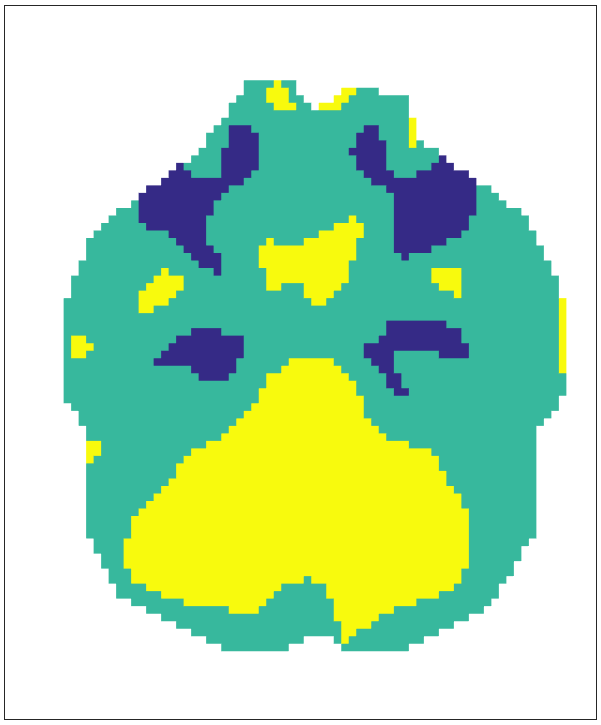}\!\!\!\!\! & \!\!
			\includegraphics[scale=0.20]{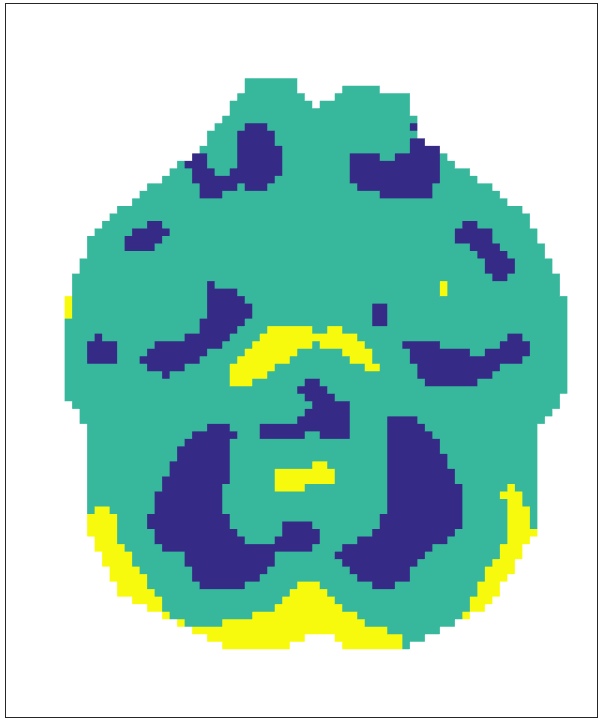} \!\!\!\!\! & \!\!\!\!
			\includegraphics[scale=0.20]{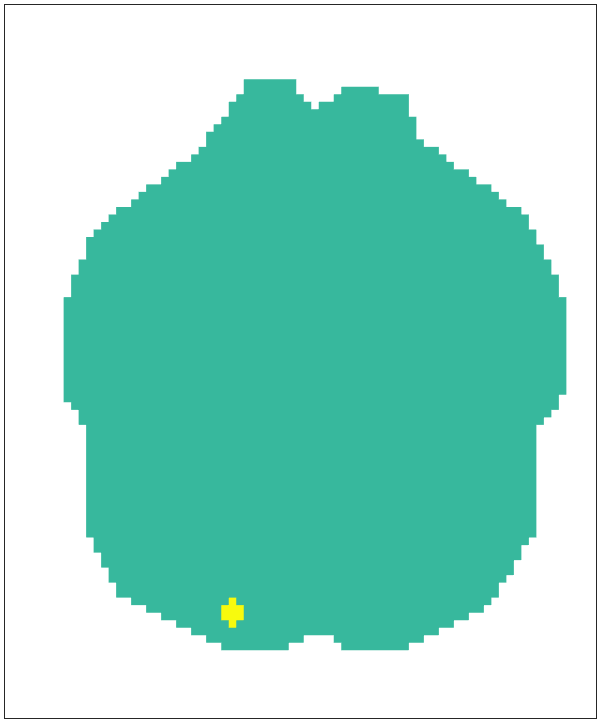} \!\!\!\!\! & \!\!\!\!
			\includegraphics[scale=0.20]{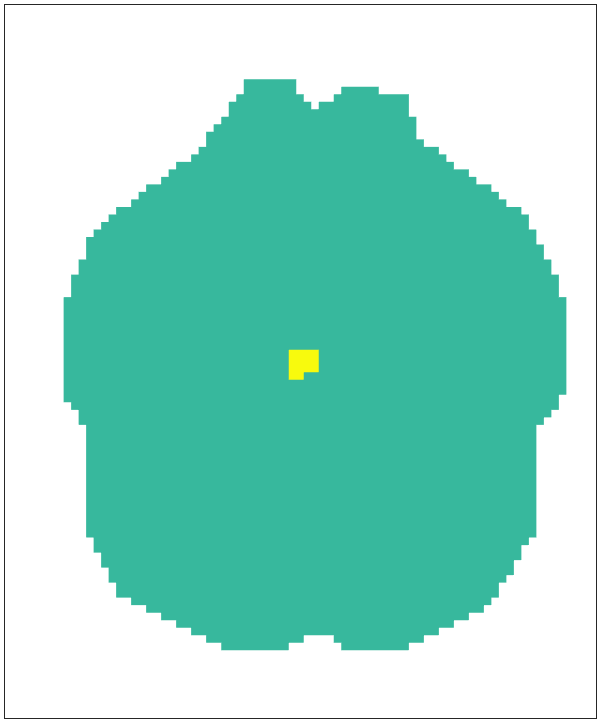} \!\!\!\!\! & \!\!\!\!\\[-7pt]
			\multicolumn{7}{c}{Slice 8}\\[5pt]
			\includegraphics[scale=0.20]{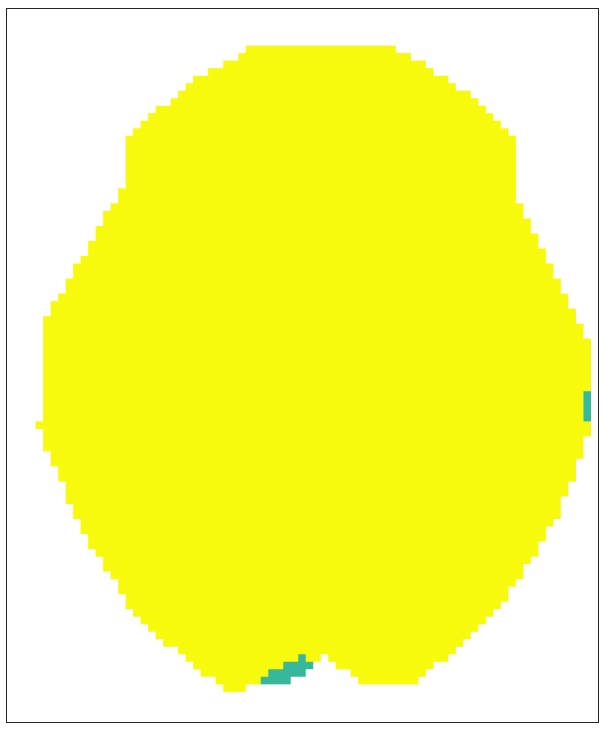} \!\!\!\!\! & \!\!\!\!
			\includegraphics[scale=0.20]{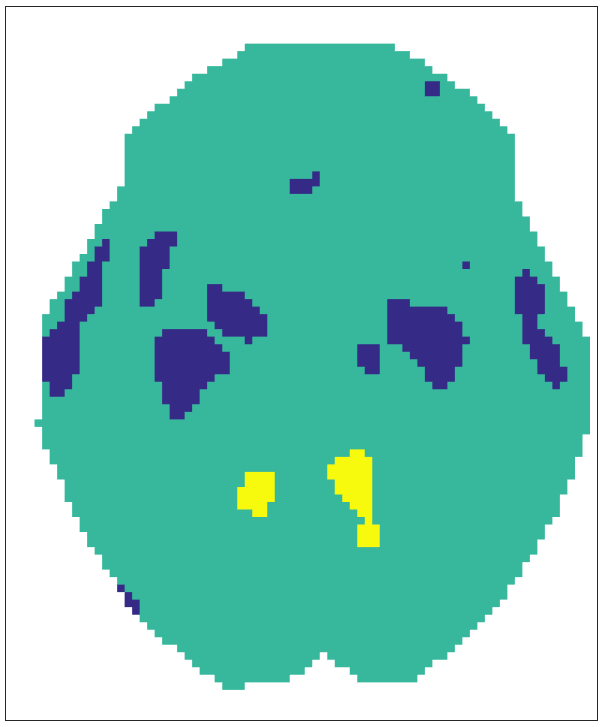} \!\!\!\!\! & \!\!\!\!
			\includegraphics[scale=0.20]{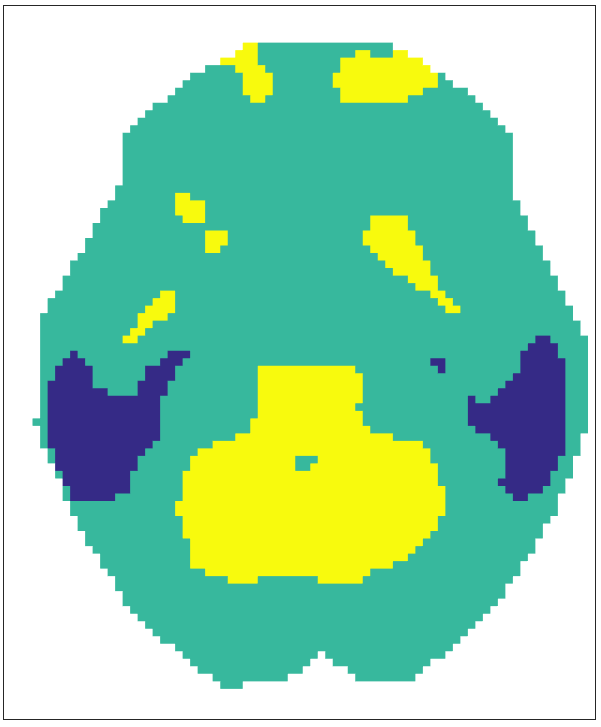} \!\!\!\!\! & \!\!\!\!
			\includegraphics[scale=0.20]{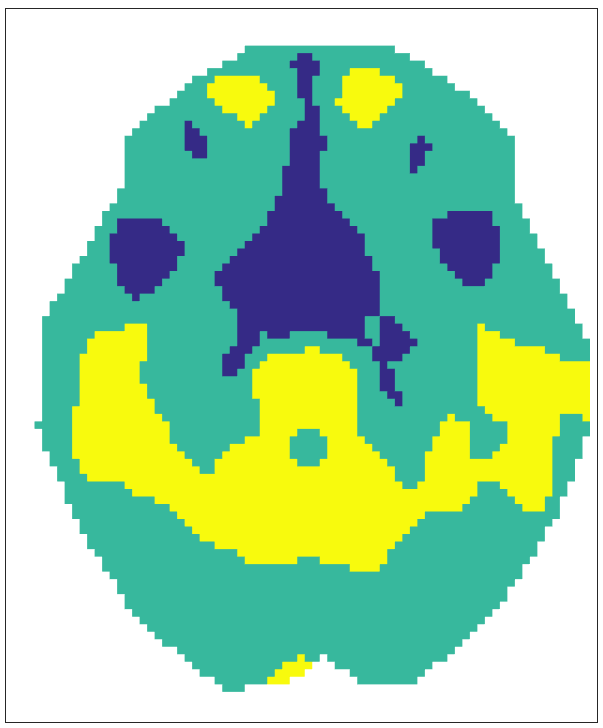}\!\!\!\!\! & \!\!
			\includegraphics[scale=0.20]{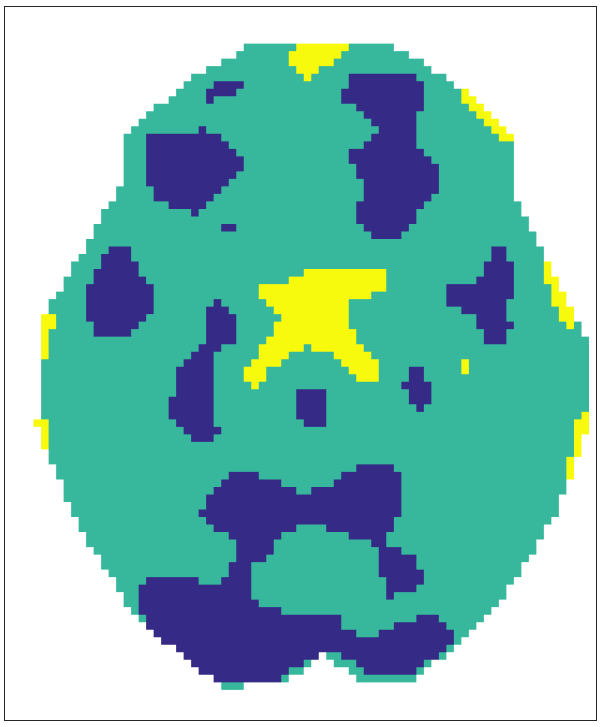} \!\!\!\!\! & \!\!\!\!
			\includegraphics[scale=0.20]{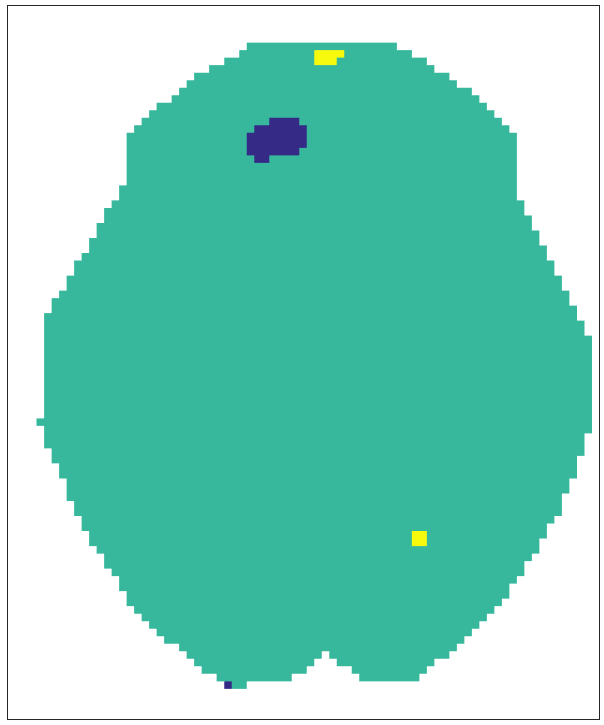} \!\!\!\!\! & \!\!\!\!
			\includegraphics[scale=0.20]{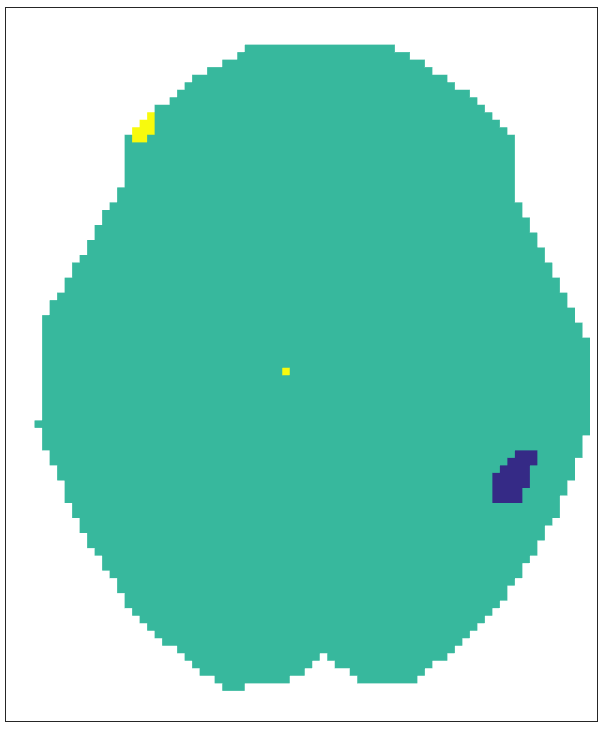} \!\!\!\!\! & \!\!\!\!\\[-7pt]
			\multicolumn{7}{c}{Slice 15}\\[5pt]
			\includegraphics[scale=0.20]{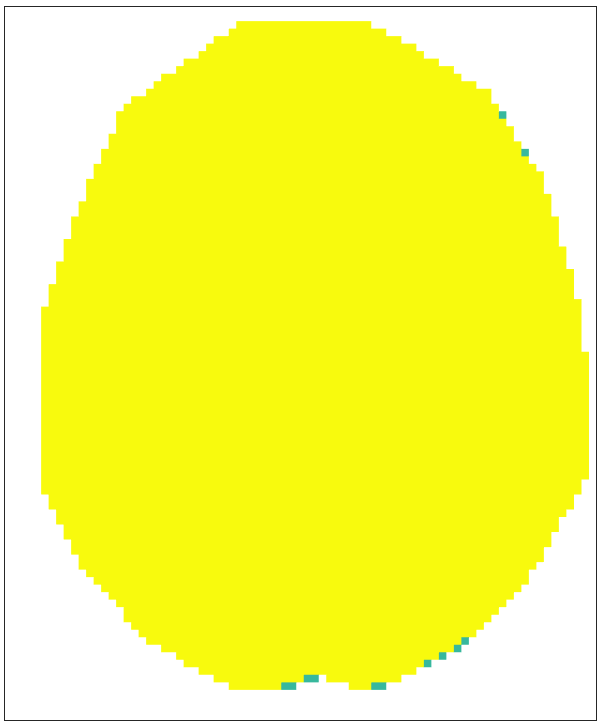} \!\!\!\!\! & \!\!\!\!
			\includegraphics[scale=0.20]{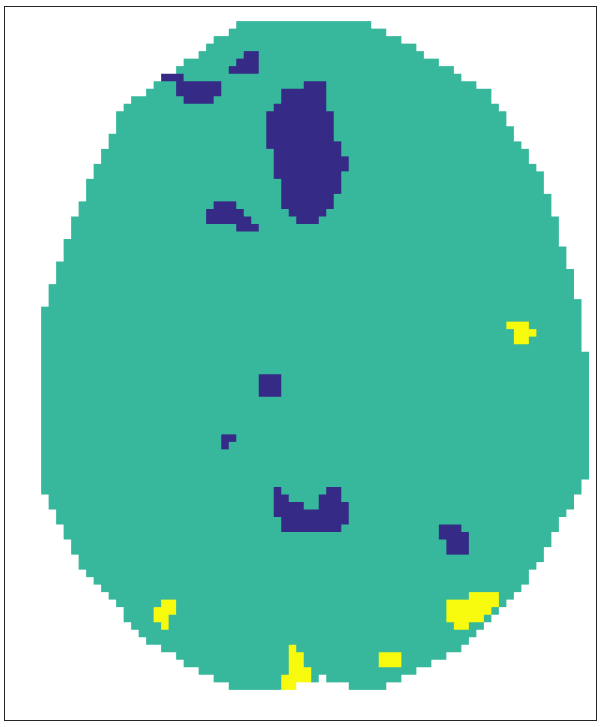} \!\!\!\!\! & \!\!\!\!
			\includegraphics[scale=0.20]{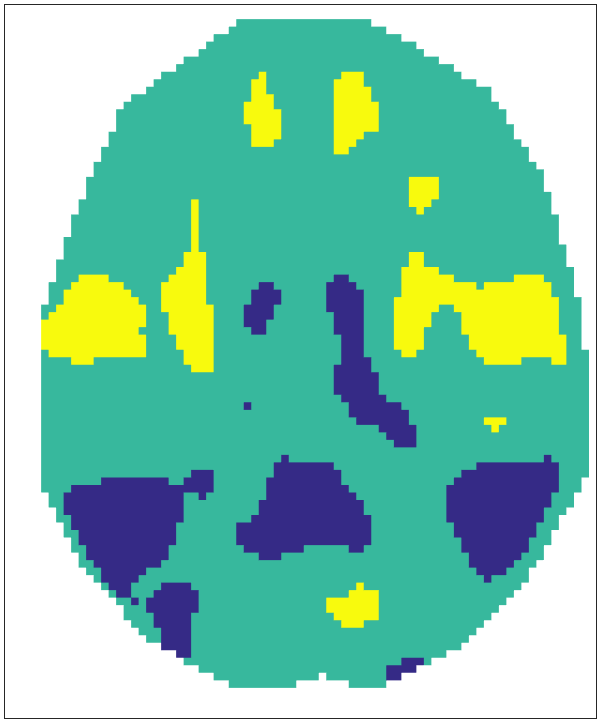} \!\!\!\!\! & \!\!\!\!
			\includegraphics[scale=0.20]{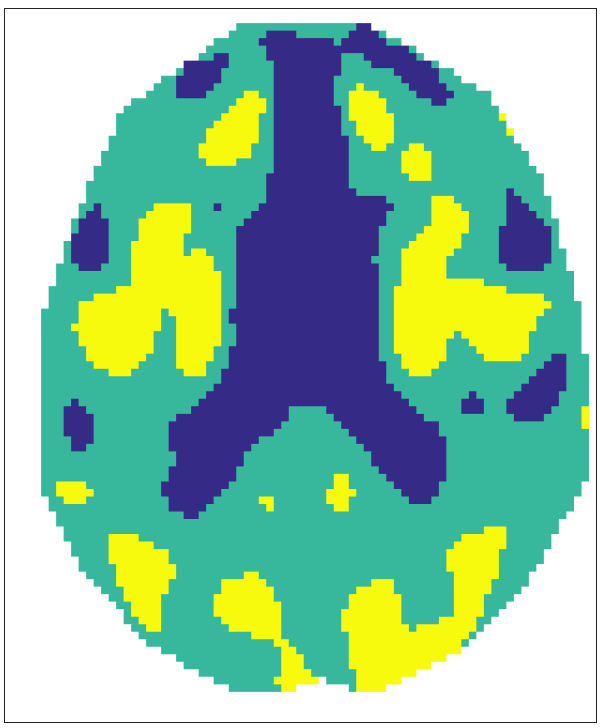}\!\!\!\!\! & \!\!
			\includegraphics[scale=0.20]{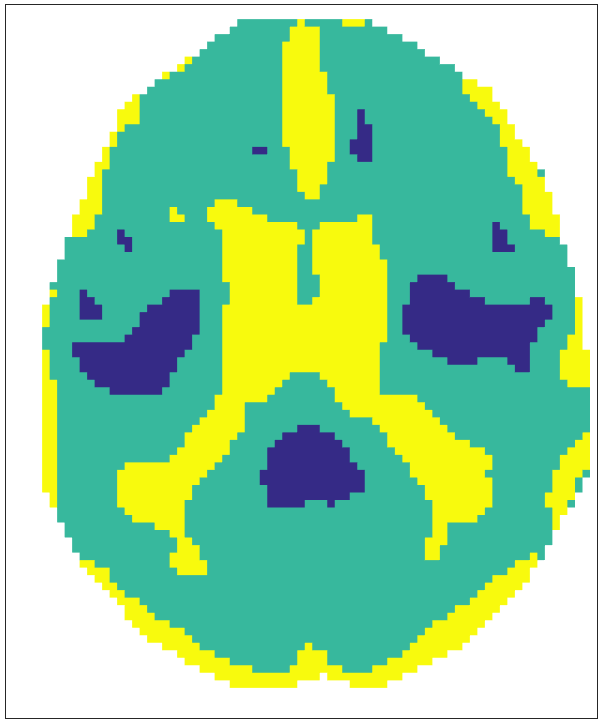} \!\!\!\!\! & \!\!\!\!
			\includegraphics[scale=0.20]{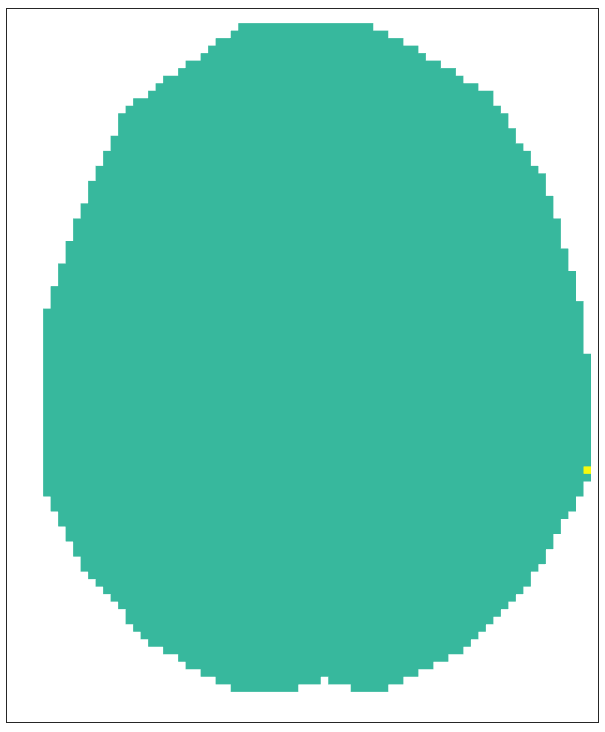} \!\!\!\!\! & \!\!\!\!
			\includegraphics[scale=0.20]{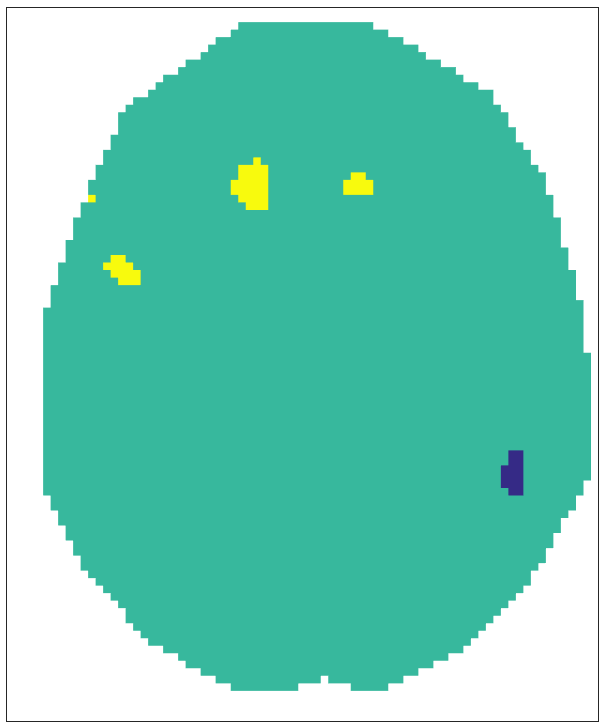} \!\!\!\!\! & \!\!\!\!\\[-7pt]
			\multicolumn{7}{c}{Slice 35}\\[5pt]
		\end{tabular}	
	\end{center}
	\caption{The ``significance'' map (based on the 95\% SCC) for the coefficient functions for the ADNI data. The yellow color and blue color on the map indicate the regions that zero is below the lower SCC or above the upper SCC, respectively.}
	\label{FIG:APP-SCC2}
\end{sidewaysfigure}

\begin{sidewaysfigure}[htbp]
	\begin{center}
		\begin{tabular}{cccccccc} 
			Intercept & MA & AD&Age&Sex&$\textrm{APOE}_1$&$\textrm{APOE}_2$  \\ 
			\includegraphics[scale=0.20]{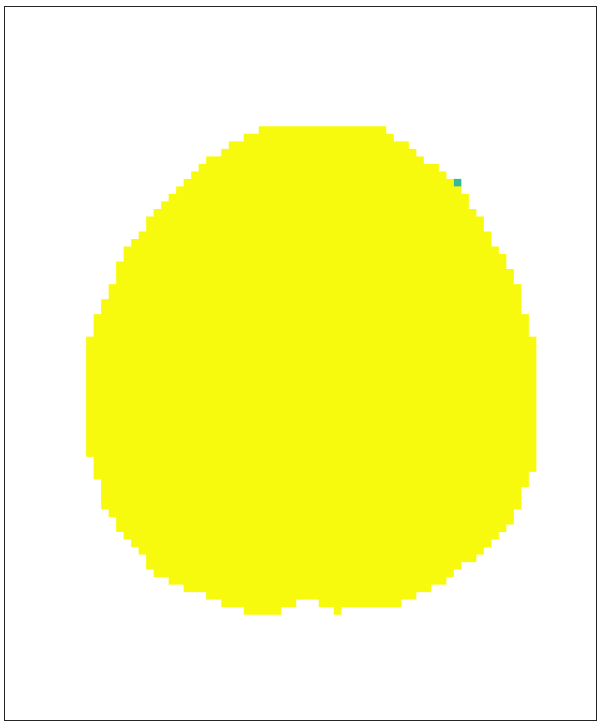} \!\!\!\!\! & \!\!\!\!
			\includegraphics[scale=0.20]{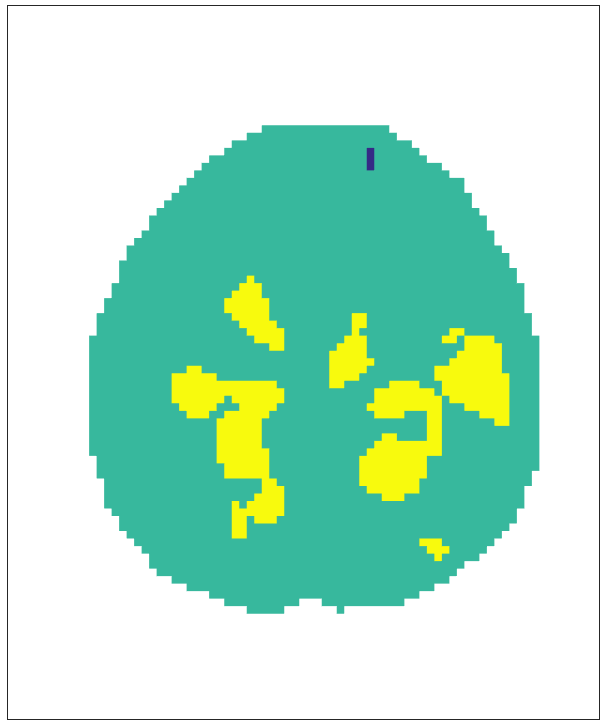} \!\!\!\!\! & \!\!\!\!
			\includegraphics[scale=0.20]{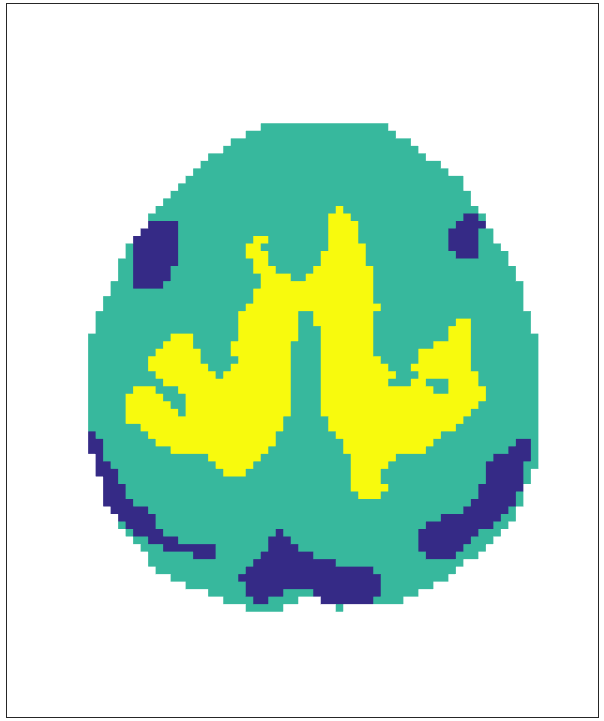} \!\!\!\!\! & \!\!\!\!
			\includegraphics[scale=0.20]{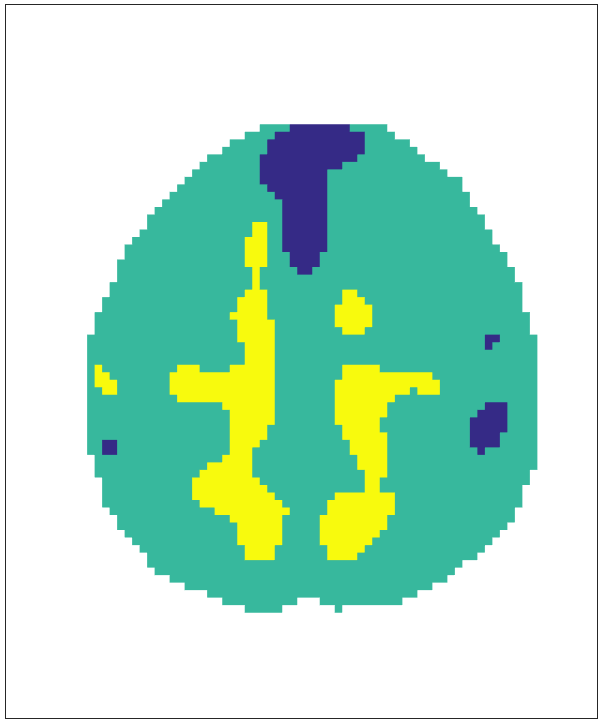}\!\!\!\!\! & \!\!\!\!
			\includegraphics[scale=0.20]{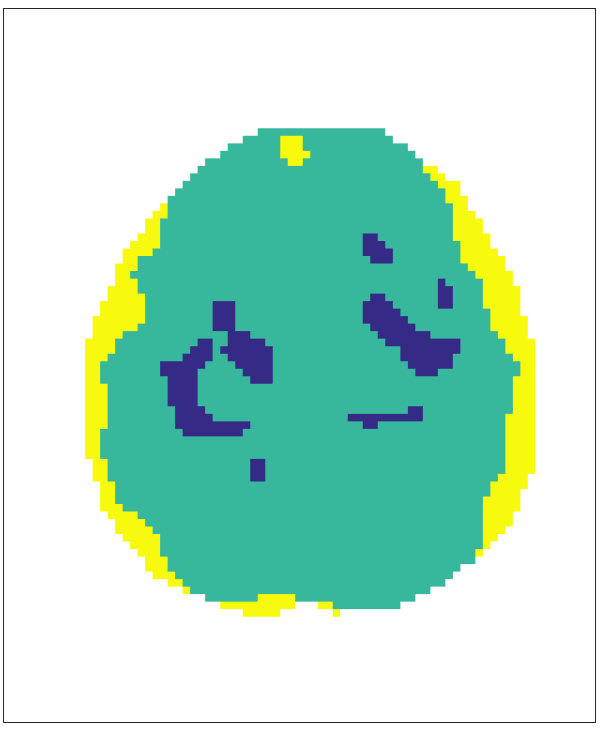} \!\!\!\!\! & \!\!
			\includegraphics[scale=0.20]{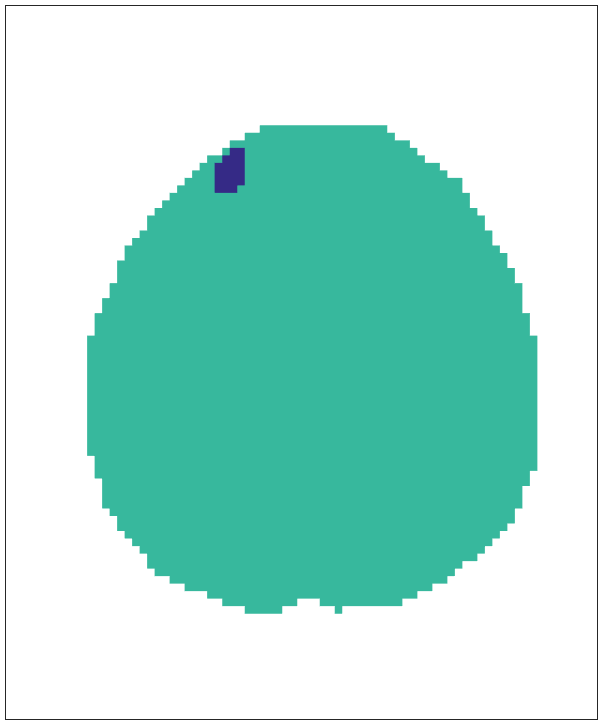} \!\!\!\!\! & \!\!\!\!
			\includegraphics[scale=0.20]{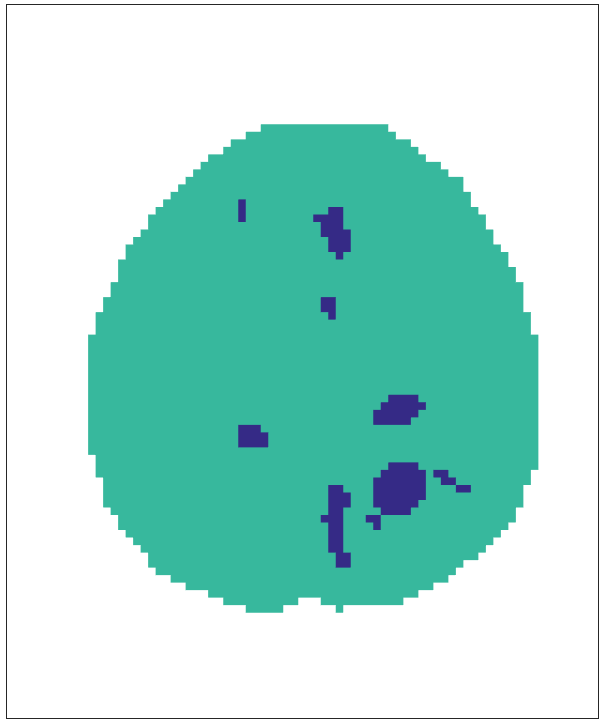} \!\!\!\!\! & \!\!\!\!\\[-7pt]
			\multicolumn{7}{c}{Slice 55}\\[5pt]
			\includegraphics[scale=0.20]{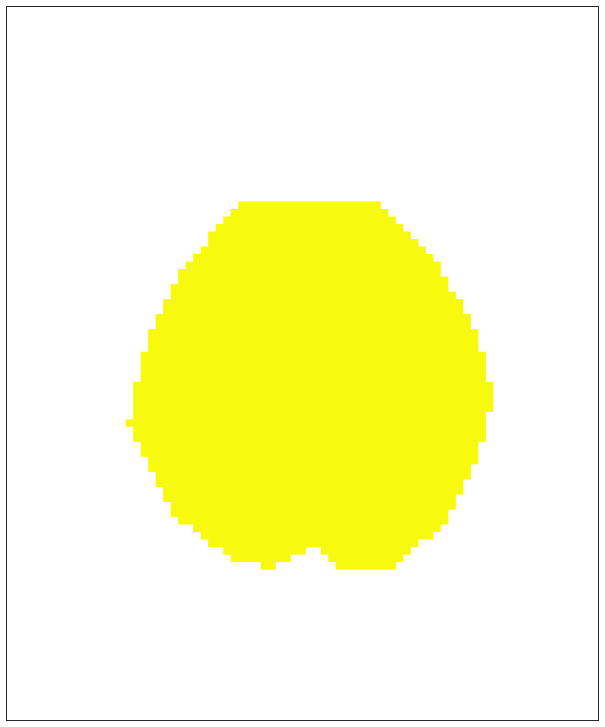} \!\!\!\!\! & \!\!\!\!
			\includegraphics[scale=0.20]{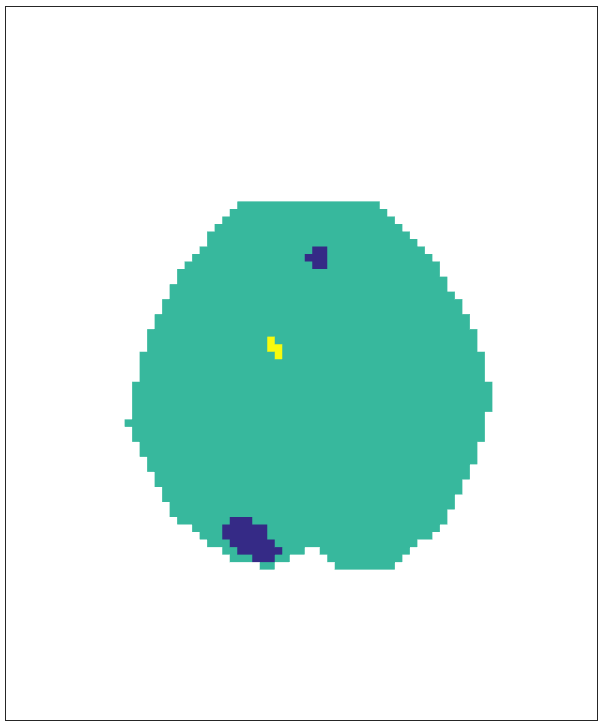} \!\!\!\!\! & \!\!\!\!
			\includegraphics[scale=0.20]{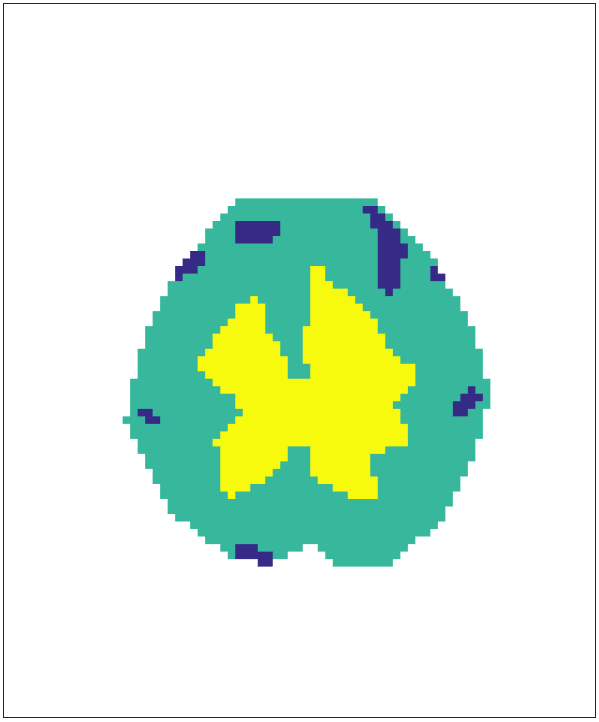} \!\!\!\!\! & \!\!\!\!
			\includegraphics[scale=0.20]{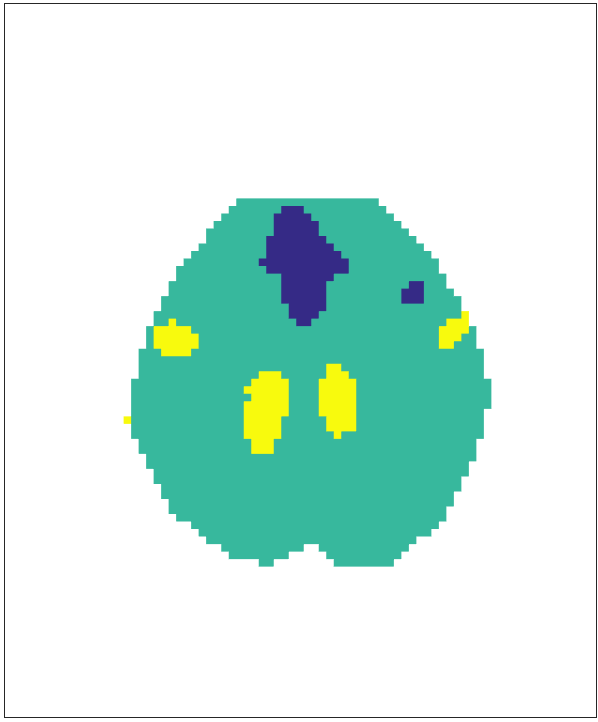}\!\!\!\!\! & \!\!
			\includegraphics[scale=0.20]{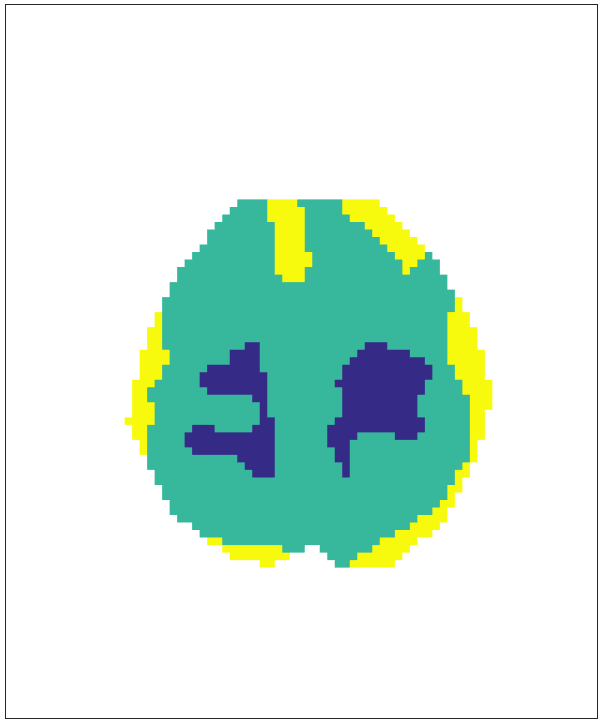} \!\!\!\!\! & \!\!\!\!
			\includegraphics[scale=0.20]{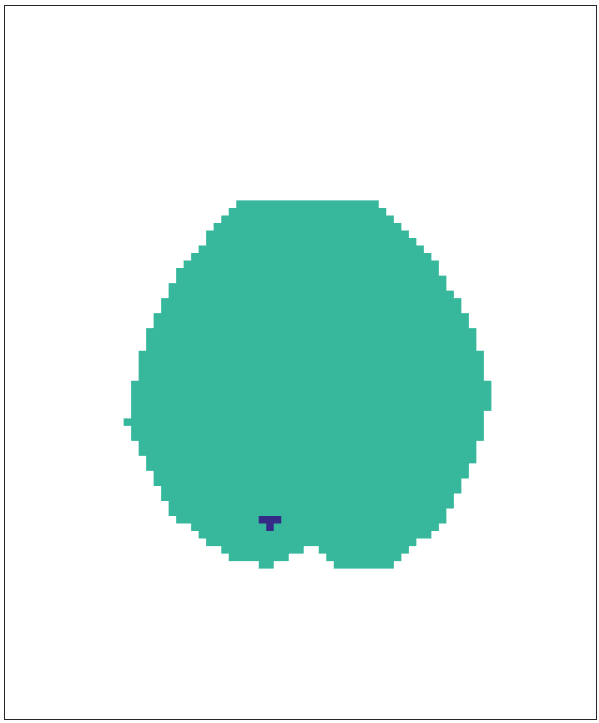} \!\!\!\!\! & \!\!\!\!
			\includegraphics[scale=0.20]{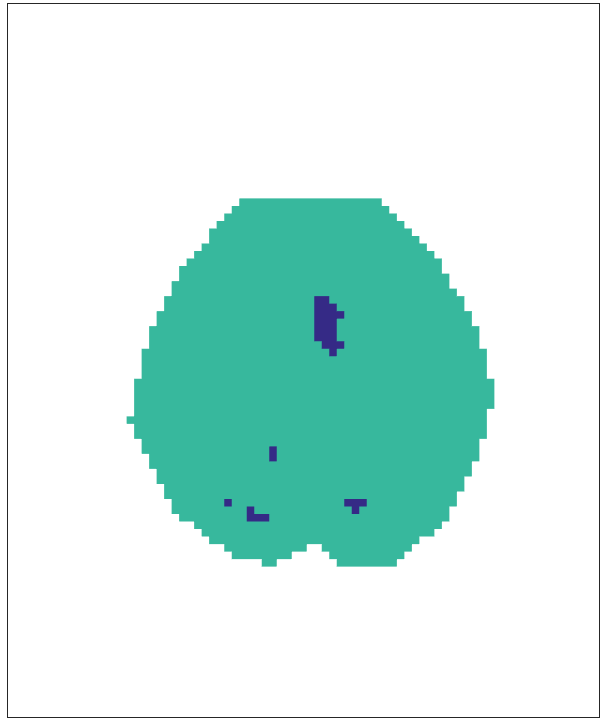} \!\!\!\!\! & \!\!\!\!\\[-7pt]
			\multicolumn{7}{c}{Slice 62}\\[5pt]
			\includegraphics[scale=0.20]{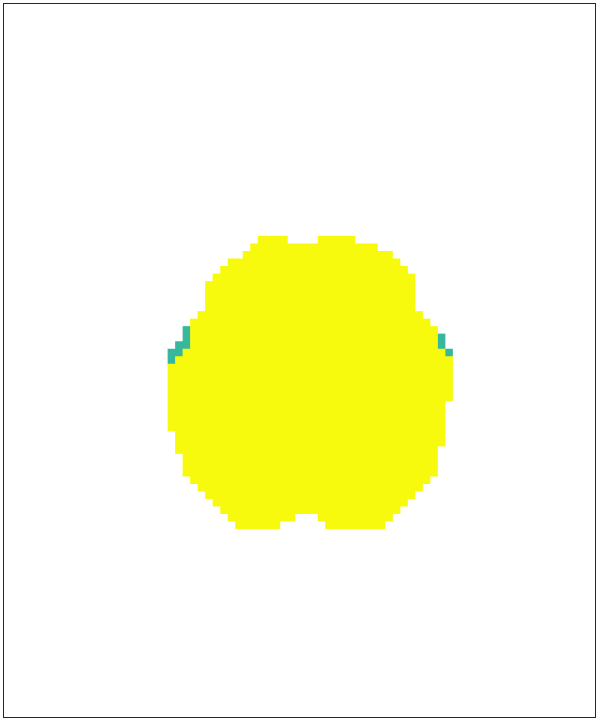} \!\!\!\!\! & \!\!\!\!
			\includegraphics[scale=0.20]{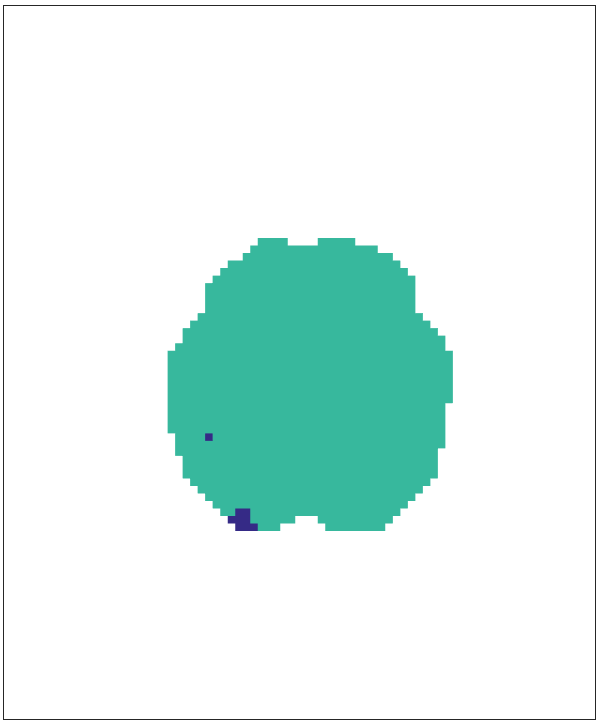} \!\!\!\!\! & \!\!\!\!
			\includegraphics[scale=0.20]{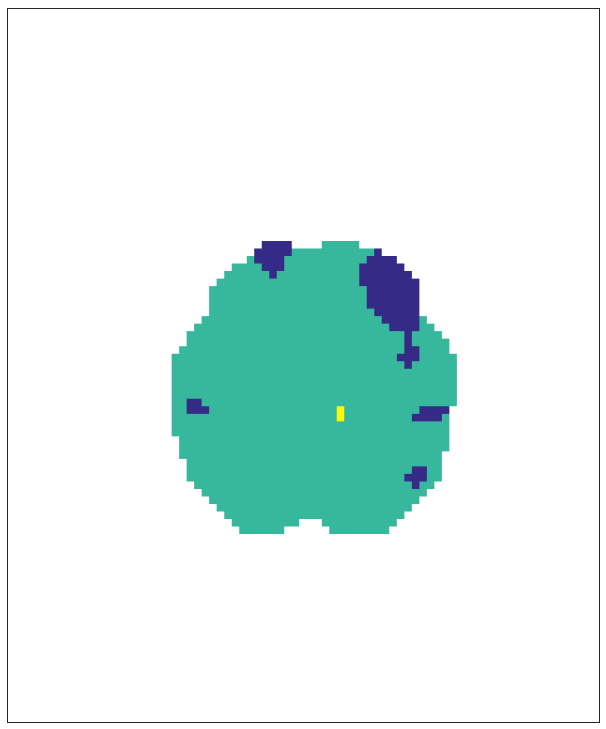} \!\!\!\!\! & \!\!\!\!
			\includegraphics[scale=0.20]{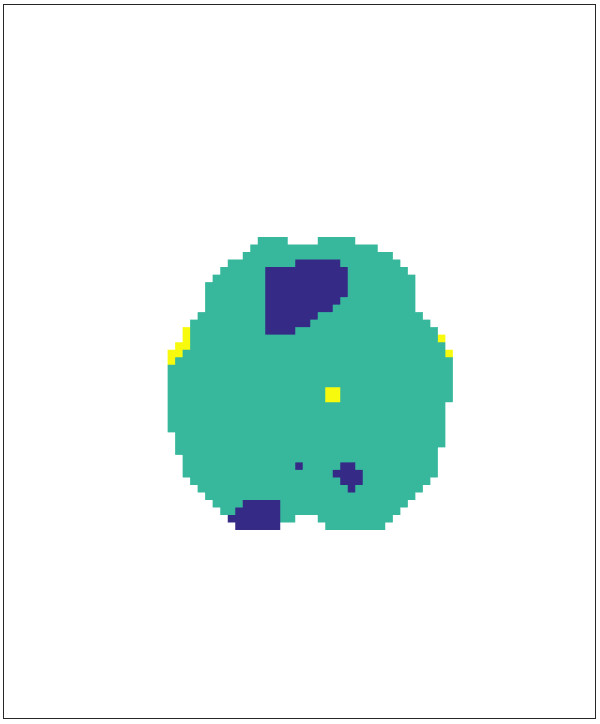}\!\!\!\!\! & \!\!
			\includegraphics[scale=0.20]{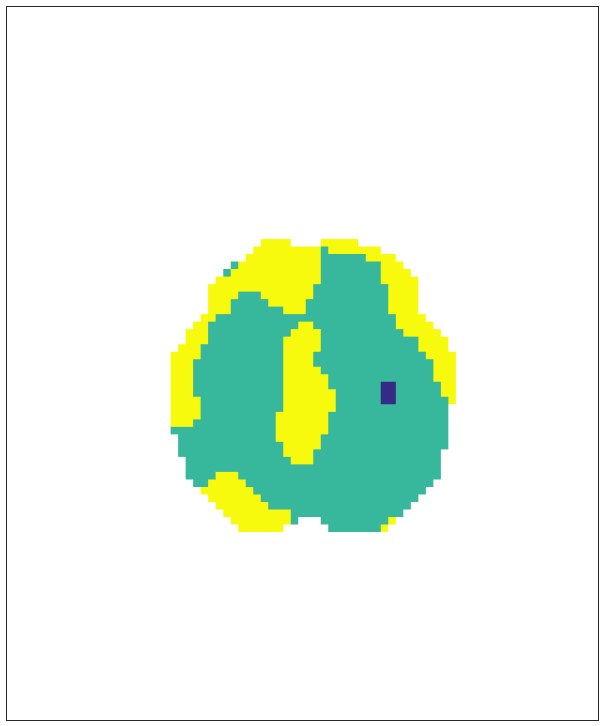} \!\!\!\!\! & \!\!\!\!
			\includegraphics[scale=0.20]{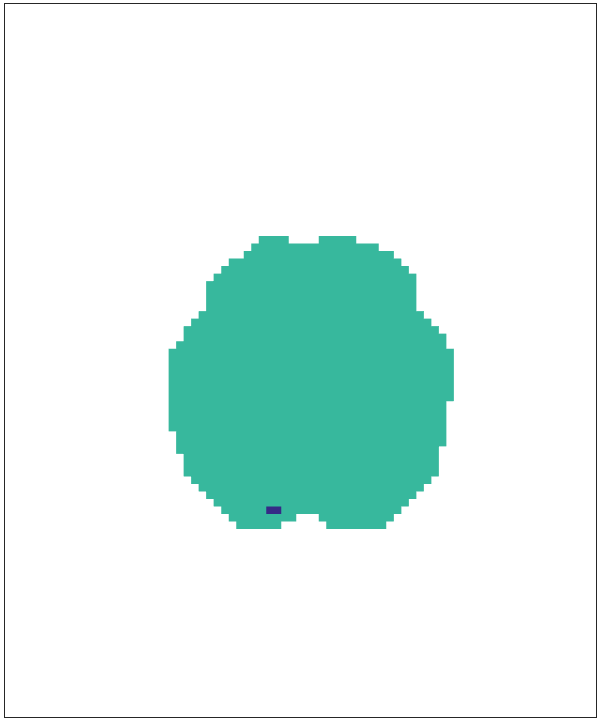} \!\!\!\!\! & \!\!\!\!
			\includegraphics[scale=0.20]{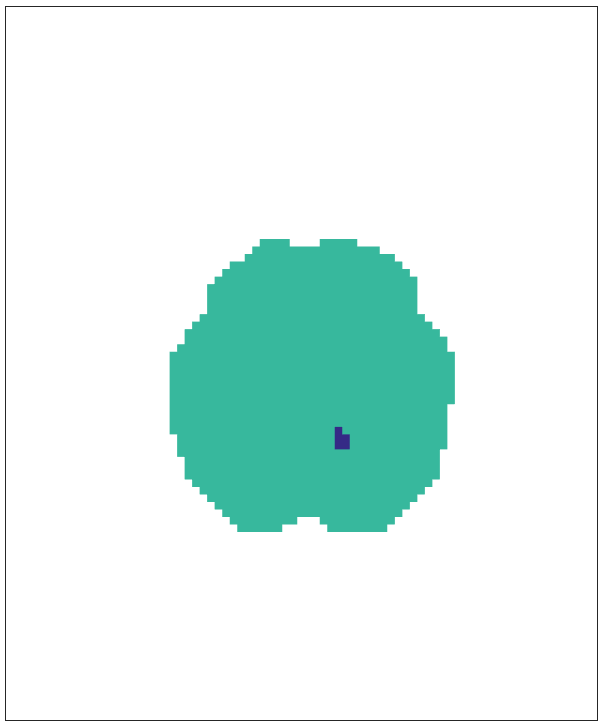} \!\!\!\!\! & \!\!\!\!\\[-7pt]
			\multicolumn{7}{c}{Slice 65}\\[-9pt]
		\end{tabular}	
	\end{center}
	\caption{The ``significance'' map (based on the 95\% SCC) for the coefficient functions for the ADNI data. The yellow color and blue color on the map indicate the regions that zero is below the lower SCC or above the upper SCC, respectively.}
	\label{FIG:APP-SCC3}
\end{sidewaysfigure}

\bibliographystyle{asa}
\bibliography{references}

\end{document}